\documentclass[a4paper,11pt]{book}

\usepackage[applemac]{inputenc}
\usepackage{amsmath}
\usepackage{slashed}
\usepackage{amscd}\usepackage{amstext}\usepackage{amsbsy}\usepackage{amsopn}\usepackage{amsthm}\usepackage{amsxtra}\usepackage{upref}\usepackage{amsfonts}
\usepackage{amssymb}\usepackage{euscript}
\usepackage[]{latexsym}\usepackage{graphicx}\usepackage{color}\usepackage[all]{xy}
\usepackage[frenchb]{babel}
\usepackage{slashed}\usepackage{physics}

\def\O{\Omega}  
  
\def\R{\mathbb{R}} 
\def\C{\mathbb{C}}
\def\N{\mathbb{N}}

\def\Tpq{\mathfrak{T}^{p}_{q}}
\def\g{\mathfrak{g}}

\newtheorem{theorem}{Théorème}[section]
\newtheorem{lemma}[theorem]{Lemme}
\newtheorem{proposition}[theorem]{Proposition}

\newtheorem{defi}{Définition}
\newtheorem{rmq}{Remarque}
\newtheorem{proprietes}[theorem]{Propriétés}
\newtheorem{propriete}[theorem]{Propriété}

\usepackage{geometry}
\title{De la structure de jauge des équations de Maxwell à l'étude des théories au delà du modèle standard grâce à la mesure du couplage trilinéaire du champ de Higgs}
\author{Valdo Tatitscheff\\
	\textit{valdo.tatitscheff@ens.fr}}
\date{03/09/2015}

\begin{document}
\makeatletter
\begin{titlepage}
	\centering
	{\large \textsc{École Normale Supérieure de Paris}}\\
	\textsc{Département de Mathématiques et Applications - Département de Physique}\\
	\vspace{3cm}
	{\large\textbf{	\@date\\
			Mémoire de Licence}}\\
	\vfill
	{\LARGE \textbf{\@title}} \\
	\vspace{7mm}
	{\large \@author} \\
	\vspace{5mm}
	\textbf{Encadrant en mathématique : Thierry Lévy (UPMC)}\\
	\textbf{Encadrant en physique : Roberto Salerno (CNRS, CMS)}
	\vspace{5mm}
	\newenvironment{abstract}%
    {\thispagestyle{empty}\null\vfill\begin{center}%
    \bfseries\abstractname\end{center}}%
    {\vfill\null}
        \begin{abstract}
        Le concept de théorie de jauge a progressivement émergé durant le 20ème siècle, lors de la construction du Modèle Standard de la physique des particules. Ce mémoire aborde divers aspects des théories de jauge. Tout d'abord, une interprétation "moderne" des équations de Maxwell en termes de formes différentielles motive une définition géométrique des théories de jauge comme une étude des connexions sur les fibrés principaux. Après un rapide résumé de l'histoire des théories de jauge en physique des particules, la construction géométrique précédente est appliquée à l'étude de la théorie électrofaible, et du Modèle Standard. Si ce dernier (dont la construction s'est achevée dans les années 1980) a aujourd'hui fait ses preuves expérimentales, il ne demeure pas moins incomplet, au moins d'un point de vue théorique. Le boson de Higgs, observé pour la première fois en 2012 au LHC, est une particle au sujet de laquelle on n'a aujourd'hui que très peu d'informations. Des indications intéressantes quant'à aux biais par lesquels il est le plus judicieux d'étendre le Modèle Standard pourraient être données si l'on observait des différences notoires entre la valeur mesurée de certains paramètres du boson de Higgs, et celle prévue par le Modèle Standard. Le couplage trilinéaire du boson de Higgs fait partie de ces paramètres intéressants. Ainsi, la dernière partie est une étude phénoménologique de l'impact de variations de cette quantité sur les distributions typiquement observables par le détecteur de particules CMS, au LHC.
        \end{abstract}
	\vfill
\end{titlepage}
\makeatother

\chapter*{Introduction}

\subsection*{Les équations de Maxwell}
Notre étude débute au milieu du $XIX^{eme}$ siècle, lorsque James Clerk Maxwell, en se basant principalement sur les travaux de Faraday et Ampère, développe un cadre mathématique à l'électromagnétisme. Il rassemble (en 1965) en vingt équations différentielles à vingt variables \cite{max1}, les lois qui décrivent le comportement des champs électriques et magnétiques, et leur interaction avec la matière. Elles prévoient en particulier l'existence d'une onde, perturbation du champ électromagnétique, se déplaçant dans le vide à une vitesse finie, accessible expérimentalement. Maxwell calcule avec les données de l'époque 310 740 000 $m.s^{-1}$. Pour citer son papier de 1965 : "\textit{The general equations are next applied to the case of a magnetic disturbance propagated through a non-conductive field, and it is shown that the only disturbances which can be so propagated are those which are transverse to the direction of propagation, and that the velocity of propagation is the velocity v, found from experiments such as those of Weber, which expresses the number of electrostatic units of electricity which are contained in one electromagnetic unit. This velocity is so nearly that of light, that it seems we have strong reason to conclude that light itself (including radiant heat, and other radiations if any) is an electromagnetic disturbance in the form of waves propagated through the electromagnetic field according to electromagnetic laws}". C'est une véritable révolution. De plus, ces idées ont directement mené à la relativité restreinte. Nous verrons plus tard en quoi ces équations donnaient en fait également un avant-goût de la physique quantique.\\\\ 
Dans son traité de 1973 reprenant en particulier ces travaux, Maxwell a modifié ses équations en utilisant des notations quaternioniques \cite{max2} ce qui réduit leur nombre à huit. Ce n'est que plus tard qu'Heaviside et Gibbs introduisent les notations vectorielles qui ont perduré jusqu'à aujourd'hui, et dérivent les fameuses "quatre équations de Maxwell".\\\\
L'apport de la géométrie différentielle au 20ème siècle permet enfin de réécrire ces équations sous la forme de deux équations seulement, plus générales que les équations de Maxwell dans le sens où elles sont définies dans un espace-temps courbe (à peu près) quelconque. C'est l'objet du premier chapitre que de dériver ces deux équations grâce au langage géométrique.

\subsection*{De nouvelles dimensions pour géométriser les interactions...}
Nous verrons au cours de notre travail, que pour dériver une géométrie agréable pour décrire des interactions, il est nécessaire de "rajouter des dimensions au dessus de l'espace(-temps)".
En guise d'exemple, considérons le cas particulier de la gravitation. On veut trouver un environnement adapté pour décrire la gravitation ; on demande une géométrie de l'espace(-temps) telle qu'un objet soumis uniquement à son poids suive les géodésiques de cette géométrie. Sur un espace courbe "quelconque", une courbe géodésique entre deux points minimise la longueur du parcours. Ces courbes particulières vérifient une propriété très forte : deux géodésiques tangentes en un point sont forcément confondues. Si on s'intéresse aux trajectoires suivies par deux boulets de canon dans notre espace $\R^{3}$, elles peuvent être tangentes en un point sans toutefois être confondues, par conséquent les deux boulets ne suivent pas les géodésiques de la géométrie euclidienne à trois dimensions. En ajoutant une dimension temporelle, la condition de tangence sur les trajectoires dans l'espace à quatre dimensions ajoute la condition d'égalité des vitesses des deux boulets au point de tangence. On ne peut alors plus affirmer que les deux masses ne suivent pas les géodésiques de cette géométrie. La relativité générale dit même en fait que les chemins suivis par des objets massifs soumis uniquement à la gravité sont des géodésiques pour cette géométrie imposée sur l'espace temps. Ainsi \textbf{rajouter une dimension permet de géométriser la gravitation}.

\begin{figure}[!h]
	\centering
	\includegraphics{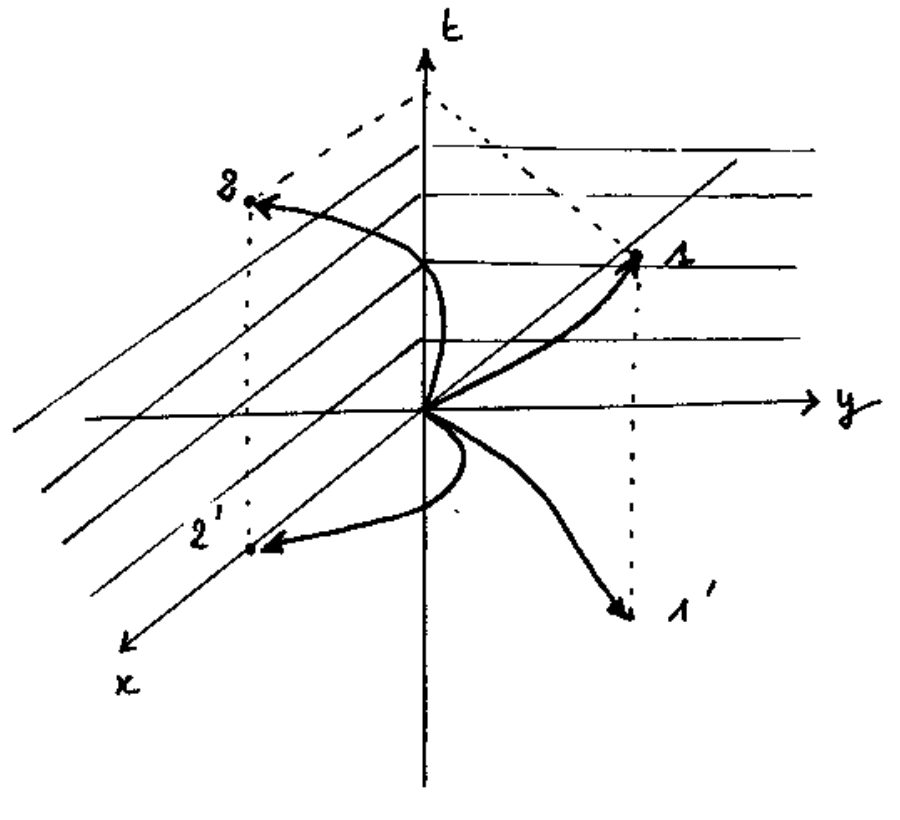}
	\caption{Rajouter une dimension temporelle lève l'indétermination liée à la vitesse du mobile. Ici on a représenté une dimension de la surface de la Terre, l'axe x, avec la normale verticale que constitue l'axe y.}
\end{figure}

Si maintenant les boulets sont chargés, de charges opposées, et plongés dans un champ électromagnétique, les trajectoires seront à nouveau différentes. Nous verrons qu'une solution est de rajouter une \textbf{dimension de charge}. L'espace résultant à cinq dimensions, étudié par Kaluza (1921) puis Klein (1926) et surtout Weyl, a en fait une structure de \textbf{fibré principal} de groupe $U(1)$ que nous définirons dans le second chapitre. Une \textbf{connexion} ou \textbf{potentiel de jauge} munit cet espace d'une bonne géométrie et nous permet de parler de la composante de charge d'une géodésique de cet espace à 5 dimensions. Si une géodésique a une charge q, sa projection sur l'espace temps est la trajectoire (non géodésique) d'un objet de charge q soumis à la force du potentiel de jauge. L'équivariance du potentiel de jauge sous l'action de $U(1)$ assure en fait la conservation de la charge. Les trois ouvrages principalement utilisés pour la partie mathématique sont les livres de S. Bleecker \textit{Gauge theory and variational principles} \cite{bleecker}, de R. Coquereaux \textit{Espaces fibrés et connexions} \cite{coq} et de S. Kobayashi et K. Nomizu \textit{Foundations of Differential Geometry} \cite{kobay}.

\subsection*{Présentation du mémoire}
C'est lors de la XIXème édition du Séminaire Poincaré que j'ai rencontré Yves Sirois, qui présentait "la découverte du boson H au LHC". J'ai pu, grâce à lui, faire mon stage d'un mois au laboratoire Leprince-Ringuet (LLR), sur le campus de l'école Polytechnique, dans la collaboration CMS, et encadré par Roberto Salerno, sur le sujet "Etude de la production double de bosons de Higgs au delà du modèle standard". Dans le cadre du cursus mixte maths-physique de licence de l'ENS, M. Bouttier et M. Kashani Poor m'ont fait rencontrer Thierry Lévy, qui a accepté de m'encadrer pour la partie mathématique de mon mémoire, avec la problématique suivante : "En quoi les équations de Maxwell contiennent-elle intrinsèquement des germes de théories quantiques de champs ?".\\\\
Afin d'essayer de répondre à cette dernière question, et présenter les résultats obtenus à l'issue du stage de physique expérimentale au sein de la collaboration CMS, ce travail, qui résume mes recherches personnelles et mes travaux de stage durant ma dernière année de licence 2014-2015, est divisé en quatre parties successives.\\\\
La première partie traite deux idées principales. Tout d'abord, nous montrons que les équations de Maxwell sont relativistes, ce qui passe par l'introduction d'objets chers à la relativité (voir les notes historiques du début du second chapitre), les tenseurs, dont nous rediscuterons la signification profonde, en lien avec \textbf{l'existence d'un objet}, au sens mathématique, lorsque les propos auront été nourris par des concepts tels que la \textbf{bonne définition au sens de la jauge} pour des champs de particules. Ensuite, il s'agit d'introduire le cadre de géométrie différentielle, les notions allant de pair avec le concept de variété, permettant de réécrire les équations de Maxwell sous leur forme moderne (deux équations). Ce cadre permet, comme nous l'avons dit, de généraliser les équations de Maxwell ; il devient possible de faire de l'électromagnétisme dans un espace-temps courbé.\\\\
Le deuxième partie comprend tout d'abord un retour sur l'histoire de l'émergence du concept de jauge dans les théories des interactions, puis introduit les objets centraux des théories de jauge : espaces fibrés, connexions, courbure. Les équations de Maxwell ne sont qu'un cas particulier de ces théories ; c'est là que nous verrons aussi ce que les équations de Maxwell ont d'intrinsèquement quantique.\\\\
Dans le troisième chapitre, nous proposons une approche (un peu) originale de la théorie standard électrofaible, ou modèle de Glashow-Salam-Weinberg, dans la mesure où elle est toujours traitée, dans les ouvrages de théories quantiques des champs, de manière "physique", avec le poids de l'histoire et des traditions, ce qui ne met pas forcément en valeur la structure de jauge qui, nous l'espérons, sera apparue comme fondamentale avec le regard des chapitres précédents. Le but est donc de dériver la théorie standard électrofaible en partant de la structure de jauge comme fondement absolu. Comme pour l'émergence de la jauge, l'histoire est indissociable de la manière dont s'est construite la théorie que nous connaissons aujourd'hui, c'est pourquoi nous présentons avant tout les grandes étapes de la découverte et de l'étude des interactions faibles, jusqu'aux miracles les plus récents. Après le modèle standard électrofaible (sans la masse), le mécanisme de Higgs est présenté, puis la brisure de la symétrie $SU(2)$, et l'apparition de la masse, mathématiquement, dans la théorie, mais aussi physiquement, dans l'Univers, une nanoseconde après le Big Bang. Enfin, un rapide tableau des grandes familles de théories "Beyond the Standard Model" (BSM) est dressé.\\\\ 
Le quatrième et dernier chapitre, dont la partie précédente constitue le cadre théorique, résume les travaux et résultats du stage au LLR, qui répondent à la problématique : Comment peut-on accéder à la mesure de l'auto-couplage $\lambda$ du champ de Higgs en analysant la productions di-Higgs par fusion de gluons au LHC, dans la détecteur CMS ?\\\\

La présentation pour laquelle j'ai opté n'est sans doute pas la plus efficace possible. Cependant, je voulais faire un exposé le plus accessible possible et surtout, cohérent. J'ai donc rajouté l'introduction de concepts qui préparent à l'étude, pour que l'entrée en matière soit la plus douce possible. Les notes historiques en particulier ne sont pas non plus indispensables à la géométrie et la physique à proprement parler, mais sont importantes pour saisir la motivation des définitions. Dans ce mémoire, j'ai surtout repris des concepts et des idées préexistantes, et mon apport personnel, surtout dans la partie mathématique, passe essentiellement par les remarques dans lesquelles je donne une interprétation qui me tient à cœur des objets utilisés, et par la structure globale de l'exposé. Un fait qui m'a marqué en étudiant les différents concepts, est que l'intuition dans de tels domaines est certes longue à acquérir, mais esthétique et en fait très naturelle en regard des développements du siècle dernier qui ont mené à ces théories. Les raisonnements qui ont mené Weyl à de telles avancées mathématiques semblent d'ailleurs poussés par une grande intuition et un sens physique très fin (introduction du fibré des échelles, du fibré des phases...). La vision qui en résulte, de ce qu'est une interaction, est également intéressante. On comprend mieux la nature d'un boson de jauge, du photon en électromagnétisme par exemple, qui définit simplement un transport canonique de la phase des fonctions d'ondes quantiques dans l'espace-temps, et fait ainsi interagir les charges.\\\\
Le troisième chapitre est une initiative complètement personnelle, que j'ai voulu mener à terme, déjà parce que je trouvais l'idée intéressante, mais aussi car le format de l'exercice - le stage-mémoire du cursus mixte maths-physique de l'ENS - semble tout à fait s'y prêter. J'ai eu une chance inouïe (j'ai été bien aiguillé) de pouvoir faire un mémoire de mathématique sur un sujet extrêmement proche de ce que j'allais étudier durant mon stage de physique. J'ai donc voulu, au lieu de rendre un mémoire de géométrie d'un côté, et un rapport de stage de l'autre, présenter un seul dossier, cohérent, d'où la troisième partie qui sert de "transition" entre le mémoire et le rapport de stage. Les discussions que j'ai pu avoir, les personnes que j'ai pu écouter notamment lors du séminaire Higgs Hunting 2015 au LAL à Orsay, m'ont poussé à poursuivre cette idée. Il en résulte un dossier beaucoup trop long certes, mais de nombreuses sous-parties pourront être passées par les lecteurs déjà familiers avec les concepts. J'espère qu'il sera plus lisible pour les autres.

\subsection*{Remerciements}

Je tiens tout d'abord à remercier mes encadrants, qui ont fait don avec une infinie gentillesse, de beaucoup de temps, de conseils, d'explications, afin de m'aiguiller sur le chemin que j'ai suivi, et dont ce mémoire est l'aboutissement. Je leur suis redevable de tout ce que j'ai pu comprendre de ces belles théories. Merci à Yves Sirois bien sur de m'avoir permis de faire un tel stage au sein du laboratoire, pour ses réponses éclairées et enjouées à mes questions, et ses encouragements. Merci à M. Bouttier et M. Kashani Poor d'avoir rendu ce travail qui m'a fait découvrir des domaines des sciences tout à fait passionnants possible. J'ai également profité d'un accueil plus que convivial au LLR, et ai passé grâce à cela six semaines tout à fait merveilleuses, internationales, et pleines de découvertes ; merci à ceux que j'ai pu embêter avec mes discussions interminables, Luca Mastrolorenzo, Florian Beaudette, Luca Cadamuro, Raphaël Duque, François Bacher, et tous les autres dont la liste complète remplirait bien plus que cette page. Merci enfin à mes parents pour leur relecture bienveillante et pas toujours amusante.

\tableofcontents

\chapter{Les équations de Maxwell généralisées}

\textit{"La libération du carcan de l'espace et du temps est une aspiration du poète et du mystique, mais ce sont les mathématiciens qui l'ont réalisé"} (Eddington)

\section{Équations de Maxwell et relativité}

Nous allons voir dans cette partie que les équations de Maxwell ne sont pas invariantes par action du groupe de Galilée. On introduit les notations tensorielles d'emblée pour, entre autres, alléger les calculs.

\subsection{Calcul tensoriel dans des espaces vectoriels réels de dimension finie}

\paragraph{Approche "intuitive"}

Dans l'espace plat $V=\R^{n}$ muni d'une base $(e_{i})_{i\in [|1,n|]}$, un vecteur $v$ de composantes $v^{i}$ s'écrit $$v=\sum_{1}^{n}v^{i}e_{i}$$ En dimension finie, V est canoniquement isomorphe à son espace dual (l'espace des formes linéaires sur V) noté $V^{*}$. On note $(e^{i})_{i\in [|1,n|]}$ la base duale associée à $(e_{i})_{i\in [|1,n|]}$ ; une forme linéaire quelconque $f \in V^{*}$ s'écrit dans cette base $$f=\sum_{1}^{n}f_{i}e^{i}$$

\begin{rmq}
	Un vecteur $v$ (qu'on s'imagine comme une petite flèche), tout comme une forme linéaire $f$ (qu'on peut visualiser comme un ensemble de lignes de niveau), existe sans même qu'on ait besoin de choisir une base pour pouvoir écrire ses composantes, par conséquent, ses coordonnées ne changent pas n'importe comment quand on décide de regarder "l'objet" d'une manière différente (dans une base différente). C'est ce qui motive les définitions suivantes. Les notions sont introduites dans ce cadre qui n'est qu'un cas très particulier de la théorie présentée par les géomètres italiens Levi-Civita et Ricci dans leur papier de 1900 : \cite{levicivita}. 
\end{rmq}

Soit $P$ = $(P^{i}_{j})_{i,j\in [|1,n|]}$ la matrice de passage de la base $(e_{k})_{k\in [|1,n|]}$ à la base $(e'_{k})_{k\in [|1,n|]}$. Soit $v$ un vecteur de $V$, de coordonnées $(v^{i})_{i\in [|1,n|]}$ dans la base de départ et $(v'^{i})_{i\in [|1,n|]}$ dans la base d'arrivée. On exprime 'les anciennes coordonnées en fonction des nouvelles', c'est-à-dire que \textbf{pour les coordonnées d'un vecteur} les formules de changement de base sont :
\begin{equation}
v^{i}=\sum_{1}^{n}P^{i}_{j}v'^{j}
\end{equation} 
pour tout $i \in [|1,n|]$. Par définition de la matrice de passage, on exprime par contre 'les nouveaux vecteurs de base en fonction des anciens', c'est-à-dire : 
\begin{equation}
e'_{j}=\sum_{1}^{n}P_{j}^{i}e_{i}
\end{equation}
On veut donner une définition rigoureuse de ces propriétés, voici la définition historique des géomètres italiens :
\begin{defi}
On appelle système d'ordre m un ensemble de fonctions des n vecteurs de base et à valeurs dans un espace vectoriel réel en correspondance bijective avec $[|1,n|]^m$.
\end{defi}

\begin{rmq}
	Cette définition est à comprendre au sens suivant : les fonctions servent à décrire un objet, par exemple un vecteur (et ce sont ses coordonnées), dans ce cas il y a n fonctions des n vecteurs de base. Si on veut décrire un endomorphisme, il faut $nm$ fonctions des $n$ vecteurs de base s'il est à valeurs dans un espace de dimension $m$.
\end{rmq}

\begin{defi}
On dit qu'un système d'ordre m est covariant (et ses éléments seront désignés par des symboles $X_{i_1...i_m}$) si les éléments dans la nouvelle base $(e'_{k})_{k\in [|1,n|]}$ s'expriment par rapport à ceux de l'ancienne base $(e_{k})_{k\in [|1,n|]}$ par les formules : $$X'_{i_1...i_m}=\sum_{a_1=1}^{n}...\sum_{a_m=1}^{n}X_{a_1...a_m}P_{i_1}^{a_1}...P_{i_m}^{a_m}$$ où P est la matrice de passage $P=(P_i^j)$ (lignes i et colonnes j).\\
On dit qu'un système d'ordre m est contravariant (et ses éléments seront désignés par des symboles $X^{i_1...i_m}$) si les éléments dans la nouvelle base $(e'_{k})_{k\in [|1,n|]}$ s'expriment par rapport à ceux de l'ancienne base $(e_{k})_{k\in [|1,n|]}$ par les formules : $$X'^{i_1...i_m}=\sum_{a_1=1}^{n}...\sum_{a_m=1}^{n}X^{a_1...a_m}(P^{-1})^{i_1}_{a_1}...(P^{-1})^{i_m}_{a_m}$$ où P est la matrice de passage $P=(P_i^j)$ (lignes i et colonnes j).
\end{defi}
On vient donc de voir que les coordonnées d'un vecteur forment une famille contravariante tandis que les vecteurs de base forment une famille covariante. Les notations utilisées pour les différents objets duaux ne sont pas anodines puisque les cordonnées des vecteurs du dual $V^{*}$ forment une famille covariante tandis l'ensemble des vecteurs de la base duale est une famille contravariante.

Introduisons la convention de sommation d'Einstein qui consiste à supprimer dans l'écriture des équations le signe somme, si cette dernière porte sur un indice répété dans un produit de grandeurs contravariantes (indice en haut) et covariantes (indice en bas). Dans le cadre de cette convention, la dimension de l'espace considéré étant connue et fixée, on notera un vecteur $v$ dans la base $(e_{i})_{i\in [|1,n|]}$ : $v^{i}e_{i}$, une forme linéaire dans la base duale associée $(e^{i})_{i\in [|1,n|]}$ : $f_{i}e^{i}$, etc...

\paragraph{Approche algébrique}

\begin{defi}
Soient $U$ et $V$ deux espaces vectoriels de dimension finie sur $\R$. Soit $M(U,V)$ l'espace vectoriel sur $\R$ dont les vecteurs de base sont les couples $(u,v)_{u \in U, v \in V}$. On regarde alors ces couples comme des objets fondamentaux, aucune opération algébrique n'est a priori définie sur les couples eux-mêmes : par exemple, $M(U,V)$ contient toutes les combinaisons linéaires finies de ces couples, mais pour tout réel $\lambda$ différent de $1$ ou $0$ on a : $\lambda\cdot(u,v)\neq (\lambda u,\lambda v)$ car $\lambda\cdot(u,v)$ représente $\lambda$ fois le vecteur de base $(u,v)$ tandis que $(\lambda u,\lambda v)$ est \textbf{un autre} vecteur de base de $M(U,V)$. $M(U,V)$ est le produit libre et non le produit cartésien de $U$ et $V$.
\end{defi}
\begin{rmq}
	Cette définition, bien qu'aride, définit un espace beaucoup plus gros que $U\times V$. Cette définition se transpose naturellement à des espaces vectoriels quelconques. Considérons par exemple $V$ et $W$ deux espaces vectoriels de dimension $2$ sur le corps à trois élément $\mathbb{F}_3$. L'espace vectoriel $V\times W$ est un espace vectoriel de dimension $4$ tandis que $M(V,W)$ est de dimension $81$ ! \\\\
	Dans le cas d'espaces vectoriels de dimension finie sur $\R$, $M(U,V)$ est de dimension indénombrable. Le passage au quotient permet ensuite, en "tordant" $M(U,V)$, d'obtenir des espaces de dimension finie comme $U\otimes V$. Si on change l'espace par lequel on quotiente, on arrive facilement à d'autres espaces, par exemple $U\times V$.
\end{rmq}
\begin{defi}
Soit $N$ le sous-espace vectoriel de $M(U,V)$ engendré par les éléments de la forme $(u+u',v)-(u,v)-(u',v)$ ou $(u,v+v')-(u,v)-(u,v')$ ou $(ru,v)-r\cdot(u,v)$ ou enfin $(u,rv)-r\cdot(u,v)$. On pose : \framebox{$U\otimes V=M(U,V)/N$}.
\end{defi}

L'image d'un couple $(u,v)$ par la projection canonique de $M(U,V)$ sur $U\otimes V$ est notée $u\otimes v$. On définit l'application bilinéaire canonique de $U\times V$ dans $U\otimes V$ par $\phi(u,v)=u\otimes v$ $\forall u\in U, v\in V$.\\\\

Le théorème suivant motive l'introduction des espaces tensoriels et peut même servir de définition, malheureusement pas constructive, du produit tensoriel de deux espaces vectoriels.
\begin{theorem}
Soit $f$ une application bilinéaire de $U\times V$ dans un espace vectoriel réel de dimension finie $W$. Alors $f$ se factorise de manière unique en $f=\tilde{f}\circ\phi$ avec $\tilde{f}:U\otimes V \rightarrow W$ où $\tilde{f}$ est linéaire.
\end{theorem}	

$$
\xymatrix{
	U\times V \ar[r]^f \ar[d]_\phi & W \\
	U\otimes V \ar[ur]_{\tilde{f}} & 
}
$$
\begin{proof}[Preuve]
	La preuve, technique, est omise ; on peut la trouver par exemple dans \cite{kobay}.
\end{proof}
\begin{proposition}
	Il y a un isomorphisme unique $\phi : V \otimes W \rightarrow W \otimes V$ tel que pour tous $ v \in V$ et $w \in W$, $\phi(v\otimes w) = w \otimes v$.
\end{proposition}

\begin{proof}[Preuve]
	On considère l'application bilinéaire $f:V\times W \rightarrow W\otimes V$ qui a $(v,w)$ associe $w\otimes v$ et qui se factorise dans $V\otimes W$.
\end{proof}

De la même façon, on a la :
\begin{proposition}
	Il y a un unique isomorphisme de $(U\otimes V)\otimes W$ sur $U\otimes (V\otimes W)$ tel que pour tous $u \in U, v\in V$ et $w \in W$, $\phi$ associe $u \otimes (v\otimes w)$ à $(u\otimes v)\otimes w$. 
\end{proposition}

et on prouve également :
\begin{proposition}
	Soient $f_{i}:U_{i}\rightarrow V_{j}, i=1,2$ des applications linéaires. Alors il existe une unique application linéaire $f:U_{1}\otimes U_{2}\rightarrow V_{1}\otimes V_{2}$ telle que pour tous $u_{1}\in U_{1}$ et $u_{2}\in U_{2}$, $f(u_{1}\otimes u_{2})=f(u_{1})\otimes f(u_{2})$.
\end{proposition}

De plus la propriété de linéarité du produit tensoriel de deux vecteurs se propage au produit tensoriel de deux espaces vectoriels.

\begin{propriete}
Le produit tensoriel est distributif pour la somme directe.
\end{propriete}	

Le produit tensoriel de deux espaces vectoriels $V$ et $W$ possède une base induite des bases $(v_{i})_{i \in [|1,m|]}$ et $(w_{i})_{i \in [|1,n|]}$ respectives de $V$ et $W$, donnée par $(v_{i}\otimes w_{j})_{(i,j)\in [|1,m|]\times [|1,n|]}$. En effet $V=\bigoplus_{1}^{m}V_{i}$ et $W=\bigoplus_{1}^{n}W_{i}$ où les $V_{i}$ et les $W_{j}$ sont les espaces engendrés respectivement pour $v_{i}$ et $w_{j}$, pour $i\in [|1,m|]$ et $j\in[|1,n|]$. Alors $V \otimes W=\bigoplus_{i=1, j=1}^{i=m, j=n}V_{i}\otimes W_{j}$ d'après la propriété précédente et la factorisation de l'application linéaire $f_{i,j}:V_{i}\times W_{j} \rightarrow \R$ qui à $(\lambda v_{i},\mu w_{j})$ associe $\lambda\mu$ ($\R$ est vu comme espace vectoriel de dimension $1$)

On définit alors différents espaces tensoriels sur un espace vectoriel fixé $V$ :

\begin{defi}
Pour un entier positif $r$, on appelle $\mathfrak{T}^{r}=V^{\otimes r}$ \textbf{espace tensoriel contravariant de degré r}. Un élément de $\mathfrak{T}^{r}$ est appelé tenseur contravariant de degré $r$. Si $r=1$, $\mathfrak{T}^{1}=V$. Par convention on écrit même $\mathfrak{T}^{0}=\R$. \\
De la même manière, pour tout entier $s$ positif, $\mathfrak{T}_{s}=(V^{*})^{\otimes s}$ est appelé \textbf{espace tensoriel covariant de degré $s$} et ses éléments tenseurs covariants de degré $s$. On a $\mathfrak{T}_{1}=V^{*}$ et par convention $\mathfrak{T}_{0}=\R$.
\end{defi}

Si $(e_{i})_{i\in[|1,n|]}$ est une base de $V$ et si $(e^{i})_{i\in[|1,n|]}$ est la base duale associée, tout tenseur $K$ contravariant d'ordre $r$ s'écrit (en convention d'Einstein) de manière unique :
$$
K = K^{i_{1}...i_{r}}e_{i_{1}}\otimes ... \otimes e_{i_{r}}
$$
et tout tenseur $L$ covariant d'ordre $s$ s'exprime de manière unique par :
$$
L = L_{j_{1}...j_{s}}e^{j_{1}}\otimes ...\otimes e^{j_{s}}
$$
 $K^{i_{1}...i_{r}}$ et $L_{j_{1}...j_{s}}$ sont respectivement les composantes de $K$ et $L$ par rapport à la base $(e_{i})_{i\in[|1,n|]}$.
 
\paragraph{Liens entre les deux approches}

Soient $(e_{i})_{i\in[|1,n|]}$ et $(\tilde{e_{i}})_{i\in[|1,n|]}$ deux bases de $V$ reliées par la transformation $\tilde{e_{i}}=A_{j}^{i}e_{i}$. Le changement de base duale associé dans $V^{*}$ s'écrit $\tilde{e^{i}}=B_{j}^{i}e{i}$ où $B=A^{-1}$.
Si $K$ est un tenseur contravariant d'ordre $k$, on a $\tilde{K}^{i_{1}...i_{r}}=A_{j_{1}}^{i_{1}}...A_{j_{r}}^{i_{r}}K^{i_{1}...i_{r}}$ et de même pour le tenseur $L$ $s$ fois covariant : $\tilde{L}_{i_{1}...i_{s}}=B_{i_{1}}^{j_{1}}...B_{i_{r}}^{j_{r}}L_{i_{1}...i_{s}}$ ce qui correspond bien à ce qui est attendu.

\begin{defi}
	L'espace tensoriel mixte de type $(r,s)$ ou espace tensoriel $r$ fois contravariant et $s$ fois covariant est le produit tensoriel $\mathfrak{T}^{r}_{s}=\mathfrak{T}^{r}\otimes \mathfrak{T}_{s}$. On a le même type de propriétés que pour les tenseurs contravariant ou covariant : expression des coordonnées dans la base induite, formules de changement de base ...
\end{defi}

\paragraph{Produit de deux tenseurs}
Posons $\mathfrak{T}=\bigoplus_{r,s=0}^{\infty}\mathfrak{T}^{r}_{s}$. On munit $T$ d'une structure d'algèbre $\N$-graduée : par la propriété de factorisation universelle du produit tensoriel, il existe une unique application linéaire de $\mathfrak{T}^{r}_{s}\times\mathfrak{T}^{p}_{q}$ dans $\mathfrak{T}^{r+p}_{s+q}$ qui envoie $(v_{1}\otimes...\otimes v_{r}\otimes v^{1}\otimes ... \otimes v^{s}, w_{1}\otimes ... \otimes w_{p}\otimes w^{1}\otimes ... \otimes w^{q})$ sur $(v_{1}\otimes...\otimes v_{r}\otimes v^{1}\otimes ... \otimes v^{s}\otimes w_{1}\otimes ... \otimes w_{p}\otimes w^{1}\otimes ... \otimes w^{q})$

\begin{defi}
	On définit la contraction d'indice comme suit : à chaque couple $(i,j)_{i\in [|1,r|],\ j\in [|1,s|]}$ est associé l'unique application de $\mathfrak{T}^{r}_{s}$ dans $\mathfrak{T}^{r-1}_{s-1}$ qui envoie $v_{1}\otimes ... \otimes v_{r}\otimes v^{1}\otimes ...\otimes v^{s}$ sur $v^{j}(v_{i}) v_{1}\otimes ...\otimes v_{i-1}\otimes v_{i+1}\otimes ... \otimes v_{r}\otimes v^{1}\otimes ... \otimes v^{j-1}\otimes v^{j+1}\otimes...\otimes v^{s}$.
\end{defi}

\paragraph{Interprétation comme applications multilinéaires}
\begin{proposition}
	$\mathfrak{T}_{r}$ est canoniquement isomorphe à l'espace vectoriel des application $r$-linéaires de $V^{\times r}$ dans $\R$.
\end{proposition}
\begin{proposition}
	$\mathfrak{T}^{r}$ est canoniquement isomorphe à l'espace vectoriel des application $r$-linéaires de $(V^{*})^{\times r}$ dans $\R$.
\end{proposition}

\begin{defi}
	Posons $\mathfrak{T}^{0}_{0}(V,W)=F$. Pour $p,q>0$, $\Tpq(V,W)$ est identifié à l'espace des fonctions multilinéaires de $(V^{*})^{\times p} \times V^{\times q}$ dans $W$. $\Tpq(V,\R)$ est noté $\Tpq(V)$. $\forall f \in \Tpq(V)$ $f$ s'écrit (en convention d'Einstein) $f=f^{i_{1}...i_{p}}_{j_{1}...j_{q}}v_{i_{1}}\otimes ... \otimes v_{i_{p}}\otimes v^{j_{1}}\otimes ... \otimes v^{j_{q}}$. 
\end{defi}
\paragraph{Algèbre extérieure}
\begin{defi}
	Définissons $\bigwedge^{k}(V,W)$ comme le sous-espace de $\mathfrak{T}^{0}_{q}(V,W)$ des applications multilinéaires totalement antisymétriques de $E$ dans $F$. On note $\bigwedge^{k}(V)=\bigwedge^{k}(V,\R)$. Soit $\omega \in \bigwedge^{k}(V)$. $\omega$ s'écrit : $\omega=\omega_{i_{1}...i_{k}} v^{i_{1}}\otimes ...\otimes v^{i_{k}}$ où $\omega_{i_{1}...i_{k}}\in \R$ est antisymétrique en les indices $i_1, ..., i_k$.
\end{defi}

Enfin il existe un produit qui munit $\bigwedge(M)=\bigoplus_{k=0}^{\infty}\bigwedge^k(M)$ d'une structure d'algèbre $\N$-graduée.

\begin{defi}[Produit extérieur]
	Pour $\alpha\in\bigwedge^i(E)$ et pour $\beta\in\bigwedge^j(E)$, on définit $\alpha\wedge\beta\in\bigwedge^{i+j}$ par:$$(\alpha\wedge\beta)(u_1, ..., u_{i+j})=\frac{1}{i!j!}\sum_{\sigma\in\mathfrak{S}_{i+j}}(-1)^{\sigma}\alpha(u_{\sigma(1)}, ..., u_{\sigma(i)})\beta(u_{\sigma(i+1)}, ..., u_{\sigma(i+j)})
	$$ Pour $\alpha\in\bigwedge^0(E)$, on pose $\alpha\wedge\beta=\alpha\beta$.
\end{defi}

\begin{rmq}
	Les tenseurs interviennent naturellement beaucoup en physique pour la raison suivante : les objets manipulés, comme les vecteurs, les endomorphismes, les formes, ont une existence intrinsèque ; cependant la manière de les décrire dépend de la base dans laquelle on les regarde. C'est cette propriété fondamentale que respectent les tenseurs. Inversement, si une grandeur suit les mêmes formules de changement de base qu'un tenseur d'ordre $(p,q)$, alors c'est un tenseur d'ordre $(p,q)$ et l'objet existe par delà les bases utilisées pour le représenter. \textbf{Toute loi physique peut en fait s'écrire comme une égalité de tenseurs} puisqu'on cherche à décrire des objets intrinsèques.
\end{rmq}

\subsection{Non invariance par transformations de Galilée}
Dans tout ce paragraphe, on se place dans $\R^{3}$ muni du produit scalaire euclidien et de la métrique associée $g_{ij}$.
Rappelons les quatre équations de Maxwell exprimées sous leur forme classique : 
\begin{equation}
\vec{\nabla} \cdot \vec{E} = \frac{\rho}{\epsilon_{0}} \leftrightarrow \partial_{i}E^{i} = \frac{\rho}{\epsilon_{0}}
\end{equation}
\begin{equation}
\vec{\nabla} \times \vec{B} = \mu_0j+\frac{1}{c^{2}}\frac{\partial \vec{E} }{\partial t} \leftrightarrow c^{2}\epsilon^{ijk}\partial_{i}B_{j} = \mu_0c^2j+\partial_{t}E^{k}
\end{equation}
\begin{equation}
\vec{\nabla} \cdot \vec{B} = 0 \leftrightarrow \partial_{i}B^{i} = 0
\end{equation}
\begin{equation}
\vec{\nabla} \times \vec{E} = - \frac{\partial \vec{B}}{\partial t} \leftrightarrow \epsilon^{ijk}\partial_{i}E_{j} = \partial_{t}B^{k}
\end{equation}
Plaçons nous d'emblée dans le vide. Considérons une transformation de Galilée : Soit ($\mathfrak{R'}$) un référentiel qui s'éloigne du référentiel galiléen ($\mathfrak{R}$), à la vitesse $\vec{V}$ constante.
\begin{figure}[!h]
	\centering
	\includegraphics{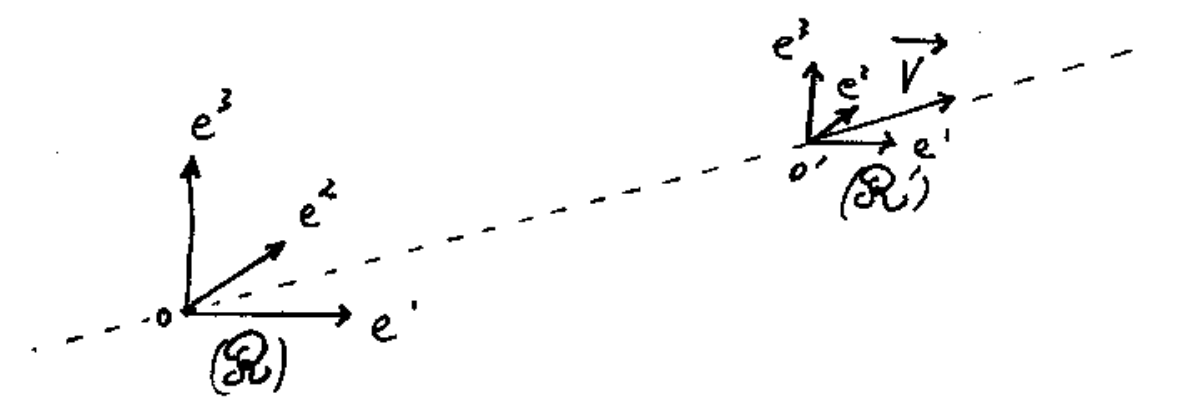}
	\caption{Schéma de la situation}
\end{figure}
On désigne les coordonnées d'un point dans le référentiel ($\mathfrak{R}$) par $(x^{i})_{i\in{1,2,3}}$ et dans ($\mathfrak{R'}$) par $(x'^{i})_{i\in{1,2,3}}$. On a les relations suivantes entre les coordonnées :
  
\[
\left\{
\begin{array}{r c l}
x'^{1} &=& x^{1} + V^{1}t\\
x'^{2} &=& x^{2}+V^{2}t\\
x'^{3} &=& x^{3}+V^{3}t
\end{array}
\right.
\]
 
Le principe de relativité galiléenne affirme que les lois physiques sont invariantes par changement de référentiel galiléen. Une particule chargée de charge $q$ et de vitesse $\vec{u}$ subit, lorsqu'elle est soumise à une champ électromagnétique extérieur, une force dite force de Lorentz qui s'écrit, dans le référentiel ($\mathfrak{R}$) : $\vec{F}=q(\vec{E}+\vec{u}\times\vec{B})$. \\
Dans ($\mathfrak{R'}$) on a donc $\vec{F'}=q(\vec{E'}+\vec{u}\times\vec{B'}+\vec{V}\times\vec{B'})$. Le principe de relativité galiléenne impose $\vec{F'}=\vec{F}$. Cela doit être vrai pour toutes les vitesses $\vec{u}$, d'où $\vec{B}=\vec{B'}$, et par conséquent on déduit $\vec{E'}=\vec{E}-\vec{V}\times\vec{B}$.
\begin{proposition}
	L'équation de Maxwell-Gauss $\partial_{i}E^{i}=0$ n'est pas invariante par action du groupe de Galilée.
\end{proposition}
\begin{proof}[Preuve.]
On a, d'après les équations de Maxwell : 
$$
\partial_{i}E^{i}=0
$$
d'où, d'après ce qui a été énoncé sur les formules de changement de référentiel : 
$$
\partial_{i}'(E'^{i}+\epsilon^{jki}V_{j}B_{k})=0
$$
où $\partial_{i}'$ est la dérivée selon $x'^{i}$ (on a $\partial_{i}'=(\partial_{i}'x^{i})\cdot \partial_{i}=\partial_{i}$). On obtient donc :
$$
\partial_{i}'E'^{i}+V_{j}\epsilon^{jki}\partial_{i}'B_{k}=0
$$
ce qui donne :
$$
\partial_{i}'E'^{i}=V_{k}\epsilon^{ijk}\partial_{i}'B_{j}
$$
Autrement dit :
$$
\vec{\nabla'}\cdot\vec{E'}=\vec{V}\cdot(\vec{\nabla'}\times\vec{B'})
$$
qui est non nul si $\vec{E'}$ n'est pas constant.
\end{proof}

\begin{proposition}
	L'équation de Maxwell-Ampère $c^{2}\epsilon^{ijk}\partial_{i}B_{j} = \partial_{t}E^{k}$ n'est pas invariante par transformation de Galilée.
\end{proposition}
\begin{proof}[Preuve.]
L'équation donne, compte tenu des formules de changement de référentiel :
$$
\epsilon^{ijk}\partial_{i}'B'_{j} = \frac{1}{c^{2}}\partial_{t}(E'^{k}-\epsilon^{ijk}V_{i}B_{j})
$$
d'où :
$$
\epsilon^{ijk}\partial_{i}'B'_{j} = \frac{1}{c^{2}}(\partial_{t}'E'^{k}-v^{i}\partial_{i}'E'^{k}-\epsilon^{ijk}V_{i}\partial_{t}'B_{j}+\epsilon^{ijk}V_{i}V^{l}\partial_{l}'B_{j})
$$
car on passe du système de coordonnées ($x^{1}, x^{2}, x^{3}, t$) au système ($x'^{1}, x'^{2}, x'^{3}, t'$) où $t=t'$, cependant $\partial_{t}=(\frac{\partial x'^{i}}{\partial t})\partial_{x'^{i}} + \partial_{t'}$. 
A priori, on n'a pas 
$$
v^{i}\partial_{i}'E'^{k}+\epsilon^{ijk}V_{i}\partial_{t}'B_{j}=\epsilon^{ijk}V_{i}V^{l}\partial_{l}'B_{j}
$$
et on peut prendre un contre-exemple simple ($\vec{V}=V\vec{e_{1}}$) pour s'en convaincre.
\end{proof}

\begin{rmq}
	Les deux autres équations de Maxwell sont invariantes par transformation de Galilée. \color{red}Ce fait trahit la structure profonde des équations de Maxwell, sur laquelle nous reviendrons en essayant de généraliser ces équations à un espace courbe, avec le moins d'hypothèses possibles. \color{black} 
\end{rmq}

\subsection{Invariance par transformations de Lorentz}
Dans ce paragraphe, on se place dans le cadre naturel de la relativité restreinte : l'espace-temps de Minkowski, c'est-à-dire $\R^{4}$ muni d'une forme bilinéaire symétrique $f$, dont la forme quadratique associée est de signature $(1,3)$, la \textbf{métrique de Minkowski}. Autrement dit, il existe une base $(e_{i})_{i\in[|0,3|]}$ telle que : $f(e_{0},e_{0})=-1$ et $$\forall i \in [|1,3|]\ \ f(e_{i},e_{i})=1 $$ et $f(e_{\mu},e_{\nu})=0$ si $\mu \neq \nu$. On note $\eta$ la matrice de f $(\eta_{\alpha \beta})_{\alpha,\beta \in [|0,3|]}$ telle que $\eta_{\alpha \beta}=f(e_{\alpha},e_{\beta})$. 
\begin{defi}
	On appelle groupe de Lorentz le groupe $\mathcal{O}(1,3)$ des endomorphismes de l'espace vectoriel $\R^{4}$ qui préservent la métrique. Il se représente naturellement comme les matrices A telles que $\eta^{T}A\eta=A$. $\mathcal{O}(1,3)$ a quatre composantes connexes, selon que le sens du temps est préservé (transformations orthochrones) ou non, et selon la conservation du signe des volumes (les transformation de déterminant $1$ préservent ce signe). Le groupe $SO(1,3)^{+}$ des transformations propres orthochrones est le groupe de Lorentz restreint. Le quotient $\frac{O(1,3)}{SO(1,3)^{+}}$ est isomorphe au 'Klein Viergruppe', et en fait : $O(1,3) \simeq SO(1,3)^{+}\rtimes(1,P,T,PT)$, où $P$ et $T$ sont respectivement les opérateurs d'inversion spatiale de renversement du temps : \\
	P = $\begin{pmatrix} 1 & 0 & 0 & 0 \\ 0 & -1 & 0 & 0 \\ 0 & 0 & -1 & 0 \\ 0 & 0 & 0 & -1 \end{pmatrix}$ et T = $\begin{pmatrix} -1 & 0 & 0 & 0 \\ 0 & 1 & 0 & 0 \\ 0 & 0 & 1 & 0 \\ 0 & 0 & 0 & 1 \end{pmatrix}$
\end{defi}	
Pour simplifier, on considère la situation de changement de référentiels précédente, où $\vec{V}$ est dirigée selon l'axe de vecteur directeur $\vec{e_{1}}$. La transformation des coordonnées d'un évènement s'écrit alors : 
\[
\left\{
\begin{array}{r c l}
x'^{0} &=& \gamma(x^{0} - \beta x^{1})\\
x'^{1} &=& \gamma(x^{1} - \beta x^{0})\\
x'^{2} &=& x^{2}\\
x'^{3} &=& x^{3}
\end{array}
\right.
\]
et la transformation des champs associée :
\[
\left\{
\begin{array}{r c l}
E_{x} &=& E_{x}'\\
E_{y} &=& \gamma(E_{y}'+VB_{z}')\\
E_{z} &=& \gamma(E_{z}'-VB_{y}')\\
B_{x} &=& B_{x}'\\
B_{y} &=& \gamma(B_{y}'-\frac{V}{c^{2}}E_{z}')\\
B_{z} &=& \gamma(B_{z}'+\frac{V}{c^{2}}E_{y}')
\end{array}
\right.
\]
Les équations de Maxwell sont invariantes par action du groupe de Lorentz restreint. Cependant, nous ne faisons pas apparaître les preuves ici puisqu'elles seront données dans un cadre beaucoup plus général dans la suite.

\section{Équations de Maxwell sur une variété}

Avant d'introduire les fibrés principaux pour géométriser les interactions électromagnétiques, nous devons définir "proprement" le champ électromagnétique ; et cela passe par une définition si générale qu'elle permet en fait de définir le champ électromagnétique sur une variété quelconque.

\subsection{Variétés, champs de tenseurs et k-formes}
Le champ électrique, le champ magnétique et le champ électromagnétique sont des exemples de formes différentielles sur une variété. Bien que relativement abstraites, les formes différentielles sont une notion unificatrice forte. Commençons par donner la définition très naturelle de variété, correspondant à la généralisation de la description du globe terrestre (par exemple) sous la forme de cartes rassemblées en atlas (il est impossible de décrire toute la Terre de manière correcte (sans la déchirer) à l'aide d'une seule carte plane).

\paragraph{Variétés, espaces tangents et 1-formes}
\begin{defi}[Variété $C^{\infty}$]
Soit $M$ un ensemble muni d'un atlas $(U_{i}, \phi_{i})_{i \in I}$, c'est-à-dire de la donnée d'un recouvrement de M $(U_{i})_{i \in I}$, et de bijections $\phi_{i}:U_{i}\rightarrow\R^{n}$ appelées cartes, telles que l'image de $U_{i}$ par $\phi_{i}$ est ouverte dans $\R^{n}$.  On suppose que pour tout $i, j \in I$, les applications $\phi_{i}\circ\phi_{j}^{-1}:\phi_{j}(U_{i}\bigcap U_{j})\rightarrow \phi_{i}(U_{i}\bigcap U_{j})$ sont lisses, c'est -à-dire $C^{\infty}$. M est muni de la topologie engendrée par les $(\phi_{i})_{i \in I}$ que l'on prend séparée (T2). On appelle dimension de la variété M l'entier $n$. On notera souvent $M^n$ pour désigner la variété et donner d'emblée sa dimension.
\end{defi}
\begin{figure}[!h]
	\centering
	\includegraphics{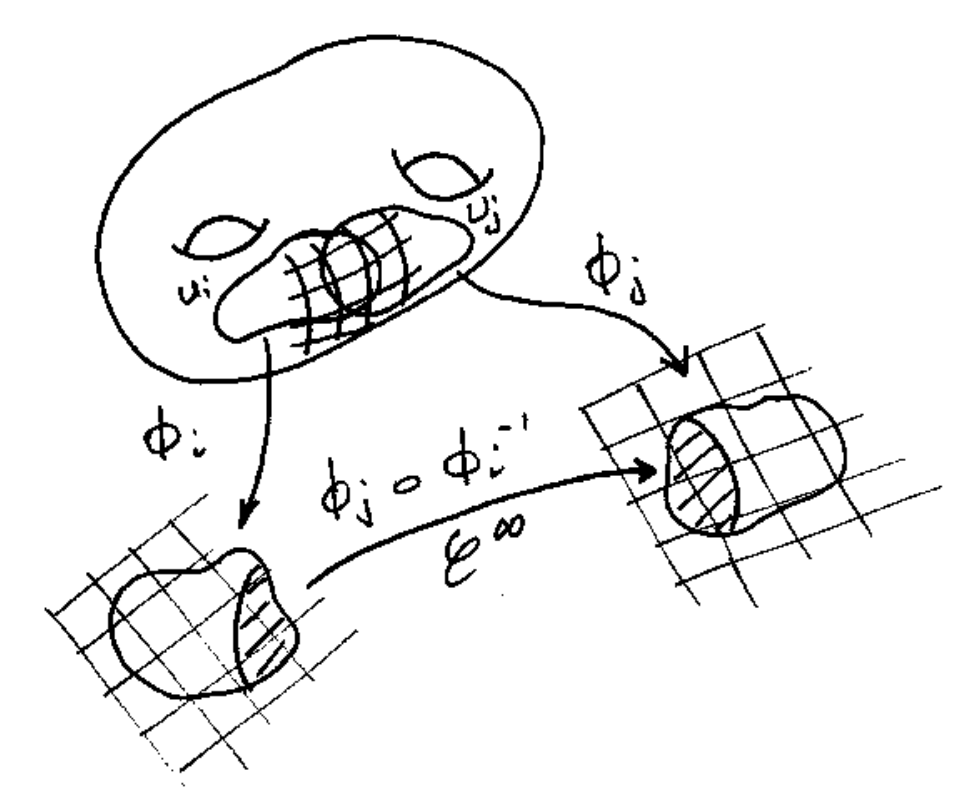}
	\caption{Définition d'une variété différentielle}
\end{figure}

\begin{defi}
	Soit $x\in M^d$. Une courbe passant par $x$ est une application lisse $\gamma:[a,b]\rightarrow M, a<0<b$ telle que $\gamma(0)=x$. On dit que les courbes $\gamma_{1}$ et $\gamma_{2}$ passant par $x$ sont équivalentes si pour une carte $\phi$ sur un voisinage de $x$ on a $(\phi\circ\gamma_{1})'(0)=(\phi\circ\gamma_{2})'(0)$. Une classe d'équivalence de courbes passant par $x$ est appelé \textbf{vecteur tangent en $x$}. L'ensemble de tous les vecteurs tangents en $x$ est noté $T_{x}M$, et est naturellement isomorphe à $\R^d$ en tant qu'espace vectoriel. Soit $f\in \mathcal{C}^{\infty}(M,\R)$ que nous noterons désormais $\mathcal{C}^{\infty}(M)$. On appelle \textbf{dérivée de f le long de $Y_{x}$} avec $\gamma'(0)=Y_{x}\in T_{x}M$ la dérivée $(f\circ\gamma)'(0)$ et on note $Y_{x}[f]$.
\end{defi}

On peut à présent définir ce qu'est un champ de vecteurs sur la variété $M$ :
\begin{defi}
	Posons $TM=\bigcup_{x\in M}T_{x}M$. Un champ de vecteurs autonome sur $M$ est une fonction $Y:M\rightarrow TM$ telle que $\forall x\in M, Y_{x} \in T_{x}M$. De plus il faut que ce champ varie de manière lisse, au sens suivant : $\forall f \in \mathcal{C}^{\infty}(M)$, on impose que $x\mapsto Y_{x}[f]$ soit dans $\mathcal{C}^{\infty}(M)$. On note cette fonction $Y[f]$, c'est la dérivée de $f$ le long du champ de vecteur $Y$. Notons $\Gamma(TM)$ l'ensemble des champs de vecteurs sur $M$. Nous n'utiliserons que des champs de vecteurs autonomes et appellerons champ de vecteurs un champ de vecteurs autonome.
\end{defi}

La donnée d'une carte $(U,\phi)$ au voisinage de $x\in M$ induit une base naturelle du plan tangent en tout $y\in U$ :

\begin{defi}
	Soit $\phi:U\rightarrow \R^{n}$ une carte définie au voisinage de $x$. Les \textbf{champs de vecteurs coordonnées} sont définis par : $$(\partial_{i})_{x} = \frac{d}{dt}\phi^{-1}(\phi(x)+te_{i})_{(t=0)}$$ où $e_{i}$ est le i-ème vecteur de la base canonique de $\R^{n}$. Soit $Y\in \Gamma(TM)$. Restreint à $U$, on peut toujours écrire $Y=a^{i}\partial_{i}$ et on a $a^{i}\in \mathcal{C}^{\infty}(M)$.
\end{defi} 

On énonce maintenant sans démonstration le théorème fondamental suivant qui donne l'existence du flot d'un champ de vecteurs sur une variété.

\begin{theorem}
	Soit $X$ un champ de vecteurs autonome sur $M$. Pour tout $x_0\in M$, il existe un ouvert $I$ de $\R$ contenant $0$, un ouvert $U$ de $M$ contenant $x_0$ et une application flot local $\phi : I\times U \rightarrow M$ $\mathcal{C}^{\infty}$ en ses deux variables, c'est-à-dire vérifiant $\forall x \in U$, $\phi(0,x)=x$ et $\forall x \in U$ l'application $t\rightarrow \phi(t,x)$ est la solution locale de l'équation différentielle $\partial_t\phi(t,x)=X(\phi(t,x))$ dont l'existence est donnée par le théorème de Cauchy-Lipschitz.
\end{theorem}

L'unicité de la solution locale associée au problème de Cauchy résulte aussi de l'unicité dans le théorème de Cauchy-Lipschitz. Si toutes les solutions maximales de l'équation différentielle considérées sont définies sur tout $\R$, le champ de vecteurs est dit complet, et le flot est défini sur tout $\R$ pour la variable temps. Le théorème de sortie des compacts donne directement le théorème suivant.

\begin{theorem}
	Tout champ de vecteur à support compact est complet.
\end{theorem}

On va maintenant voir comment un champ de vecteurs est transporté par un difféomorphisme.
\begin{defi}
	Soit $f:M\rightarrow N$ un difféomorphisme lisse et $x\in M$. On définit le \textbf{poussé en avant} de $f$ en $x$ comme $f_{*x}:T_{x}M\rightarrow T_{f(x)}N$ par $f_{*x}(\gamma'(0))=(f\circ\gamma)'(0)$. On définit ainsi $f_{*}Y$ le champ de vecteurs sur $N$ poussé en avant de $Y$ champ de vecteurs sur $M$ par $f$.
$$
\xymatrix{
	TM \ar[r]^{f_{*}} \ar[d]_\pi & TN \ar[d]_{\pi'}\\
	M \ar[r]^f & N
	}
$$
\end{defi}

On peut définir de la même manière les $1$-formes sur $M$ :

\begin{defi}
Notons $\mathfrak{T}^{p}_{q}(M)=\bigcup_{x\in M}\mathfrak{T}^{p}_{q}(T_{x}M)$. Une $1$-forme sur $M$ est une fonction $\alpha:M\rightarrow\mathfrak{T}^{0}_{1}(TM)$ telle que $\forall x \in M$ on a $\alpha_{x}\in\mathfrak{T}^{0}_{1}(T_{x}M)$. On demande aussi que cette application varie de manière lisse sur $M$ : $\forall Y \in \Gamma(TM)$ la fonction $\alpha(Y)$ donnée par $\alpha(Y)(x)=\alpha_{x}(Y_{x})$ doit être dans $\mathcal{C}^{\infty}(M)$. On note $\O^{1}(M)$ l'ensemble des $1$-formes sur $M$.
\end{defi}

\paragraph{Généralisation du gradient d'une fonction} 
Soit $f$ une fonction $\mathcal{C}^{\infty}$ sur $\R^{n}$. La dérivée directionnelle de $f$ dans la direction $v$ est le produit scalaire du gradient de $f$ avec le vecteur $v$ : $vf=\vec{\nabla}\cdot\vec{v}$. On veut définir, pour tout $f\in\mathcal{C}^{\infty}(M)$ un objet noté $df$ qui joue le rôle du gradient dans $\R^{n}$.

Le gradient d'une fonction est un champ de vecteurs, donc on aurait envie que $df$ soit un champ de vecteurs sur $M$. Cependant le problème vient du produit scalaire défini naturellement sur $\R^{n}$, mais pas sur notre variété $M$. L'objet qui donne la manière de prendre le produit scalaire de deux vecteurs tangents est appelé \textbf{métrique}, et nous le définirons rigoureusement plus tard. Cependant, il est avantageux d'avoir la différentielle $df$ de $f$ de manière indépendante d'une quelconque métrique sur $M$ : de nombreuses métriques sont par exemple solutions de l'équation d'Einstein, et il n'y a pas de choix canonique. En essayant de conserver les propriétés du gradient, on définit alors \textbf{la différentielle} ou \textbf{dérivée extérieure} de la manière suivante.
\begin{defi}
	La différentielle de $f\in \mathcal{C}^{\infty}(M)$ est la $1$-forme $df$ qui à un champ de vecteurs $Y$ sur $M$ associe la dérivée de $f$ le long de $Y$. Évaluée en $x\in M$ on a donc $df(Y)_{x}=Y_{x}[f]$. On vérifie immédiatement qu'il s'agit bien d'une $1$-forme.
\end{defi}
La fonction $d:\mathcal{C}^{\infty}\rightarrow \O^{1}(M)$ qui a une fonction associe sa dérivée extérieure vérifie les propriétés suivantes : Pour $f,g,h \in \mathcal{C}^{\infty}, \alpha \in \R$ on a $d(f+g)=df+dg$ et aussi $d(\alpha f)=\alpha df$, autrement dit $d$ est $\R$-linéaire, on a évidemment $(f+g)dh=fdh+gdh$ et ce qui s'appelle la \textbf{règle de Leibniz} $d(fg)=fdg+gdf$.

\paragraph{Vecteurs cotangents}
Un champ de vecteurs associe à chaque point $x\in M$ un vecteur de l'espace tangent en $x$ $T_{x}M$. De la même manière, une $1$-forme associe à chaque point $x\in M$ un objet appelé \textbf{vecteur cotangent}. Un vecteur cotangent est une forme linéaire sur $T_{x}M$. On peut le visualiser comme des hyperplans "de niveau" parallèles au voisinage du point considéré (l'image réciproque des entiers relatifs par la forme linéaire). Lorsqu'on prend l'image d'un vecteur tangent $\vec{v}$ par le vecteur cotangent df considéré, $df(\vec{v})$ est en quelque sorte le nombre d'hyperplans que croise le vecteur $\vec{v}$, avec une orientation, parce que $df(\vec{v})$ peut être négatif. Il faut donc "marquer" l'hyperplan correspondant à un $+1$. 

\begin{figure}[!h]
	\centering
	\includegraphics{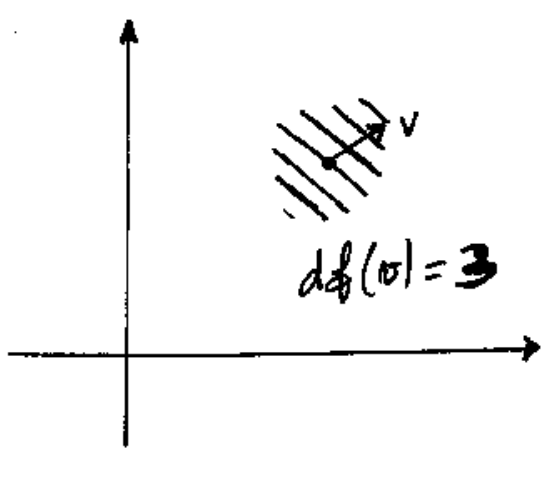}
	\caption{Un vecteur cotangent peut être vu comme un ensemble de lignes de niveau}
\end{figure}

\paragraph{Champs de tenseurs, k-formes}
\begin{defi}
	Un champ de tenseurs de type $(p,q)$ sur $M$ est une fonction $S:M\rightarrow \Tpq(M)$ telle que $S_{x}\in\Tpq(T_{x}M)$ et telle que pour tous $Y_1, ..., Y_q \in \Gamma(TM), \alpha_1, ..., \alpha_p \in \O^{1}(M)$, la fonction $S(\alpha_1, ..., \alpha_p, Y_1, ..., Y_p)$ donnée par $$S(\alpha_1, ..., \alpha_p, Y_1, ..., Y_p)(x)=S(\alpha_{1x}, ..., \alpha_{px}, Y_{1x}, ..., Y_{px})$$ est dans $\mathcal{C}^{\infty}(M)$. L'espace de tous les champs de tenseurs de type $(p,q)$ sur $M$ est noté $\Tpq(M)$.
\end{defi}
\begin{defi}
	Une $k$-forme sur $M$ est un champ de tenseurs $\omega$ de type $(0,k)$ sur $M$ tel que $\forall x\in M$ $\omega_{x} \in \bigwedge^{k}(M)$.
\end{defi}
On note $\O^k(M)$ l'ensemble des $k$-formes sur la variété $M$.
On peut multiplier les formes différentielles entre elles par un produit extérieur point par point :
\begin{defi}
	Pour $\alpha \in \O^i(M)$ et $\beta \in \O^j(M)$, on définit $\alpha\wedge\beta\in\O^{i+j}(M)$ par $(\alpha\wedge\beta)_{x}=\alpha_{x}\wedge\beta_{x}$.
\end{defi}

On peut exprimer localement les $k$-formes dans une base particulière : si $\phi:U\rightarrow\R^n$ est une carte, alors localement, une base des champs de vecteurs est $(\partial_{i})_{i\in[|1,n|]}$. Une base des $k$-formes sur $U$ est alors donnée par $(dx^i)_{i\in[|1,n|]}$ où $dx^i(\partial_j)=\delta_i^j$. Alors toute $k$-forme $\omega\in\O^k(M)$ s'écrit sur $U$ : $$ \omega = \frac{1}{k!}\omega_{i_{1}...i_{k}}dx^{i_1}\wedge...\wedge dx^{i_k}$$ où $\omega_{i_{1}...i_{k}}=\omega(\partial_{i_1}, ..., \partial_{i_k}) \in \mathcal{C}^{\infty}(M)$.

\paragraph{Dérivée extérieure, tiré en arrière}
\begin{defi}
	On a vu que si $f\in\mathcal{C}^{\infty}(M)$, alors $df\in\O^1(M)$ est définie par $df(Y)=Y[f]$ pour tout champ de vecteurs sur $M$. Pour $\omega\in\O^k(M)$, on définit $d\omega$ comme la $(k+1)$-forme qui s'exprime, restreinte à $U$, par 
	$$ d\omega=\frac{1}{k!}d(\omega_{i_1...i_k})\wedge dx^{i_1}\wedge...\wedge dx^{i_k}$$ $$ d\omega=\frac{1}{k!}\partial_i(\omega_{i_1...i_k})dx^i\wedge dx^{i_1}\wedge...\wedge dx^{i_k}$$
\end{defi}

\begin{rmq}
	Il existe une définition de la dérivée extérieure qui ne fait pas appel à un système local de coordonnées : pour tous $X_1, ..., X_{k+1} \in \Gamma(TM)$, on a $$d\omega(X_1, ..., X_{k+1})=\sum_{i=1}^{k+1}(-1)^{i+1}X_i[\omega(X_1,...,\tilde{X_i},...,X_{k+1})]$$ $$+ \sum_{1\neq i < j \neq n}(-1)^{i+j}\omega([X_i, X_j], X_1, ..., \tilde{X_i}, ..., \tilde{X_j}, ..., K_{k+1})$$ où il faut omettre les termes surmontés d'un tilde, et où $[A, B]$ est le crochet de Lie des champs de vecteurs $A$ et $B$, défini par $[A, B]_{x}[f]=A_x[B[f]]-B_x[A[f]]$.
\end{rmq}

\begin{proprietes}
	On a pour $\alpha\in\O^i(M)$ et $\beta\in\O^j(M)$ : $$d(\alpha\wedge\beta)=d\alpha\wedge\beta + (-1)^i\alpha\wedge d\beta$$ et $$d^2=d\circ d=0$$
\end{proprietes}	

\begin{defi}
	On dit qu'une forme différentielle $\omega$ est \textbf{exacte} si elle est la dérivée extérieure d'une autre forme $\epsilon$ i.e $\omega=d\epsilon$. On dit qu'elle est \textbf{fermée} si $d\omega=0$.
\end{defi}

De la même manière qu'on peut pousser en avant les champs de vecteurs par un difféomorphisme, on peut tirer en arrière les k-formes :

\begin{defi}
	Soit $f:M\rightarrow N$ un difféomorphisme lisse, et soit $\omega \in \O^{k}(N)$. Le tiré en arrière $f^{*}\omega$ de $\omega$ par f est dans $\O^{k}(M)$ et est défini par $(f^{*}\omega)_{x}(Y_1, ..., Y_k)=\omega_{f(x)}(f_{*x}(Y_1), ..., f_{*x}(Y_k))$ pour $Y_1, ..., Y_k \in \Gamma(TM)$.
\end{defi}
$$
\xymatrix{
	\O^{k}(M) \ar[d]_\pi & \O^{k}(N) \ar[l]_{f^{*}} \ar[d]_{\pi'}\\
	M \ar[r]^f & N
}
$$

Le tiré en arrière est en fait très naturel pour les formes sur M car : 
\begin{proposition}
	On a, sous les mêmes hypothèses, $d(f^{*}\omega)=f^{*}d\omega$. De plus $f^{*}(\alpha\wedge\beta)=f^{*}\alpha\wedge f^{*}\beta$ et on a la fonctorialité $(f\circ g)^{*}\omega = g^{*}f^{*}\omega$.
\end{proposition}
\begin{proof}[Preuve]
	C'est direct à partir des définitions du tiré en arrière et du wedge.
\end{proof}

\paragraph{Liens entre la dérivée extérieure et les opérateurs différentiels}
\begin{propriete}
	Pour une fonction $f\in\mathcal{C}^{\infty}(\R^{3})$, on a $\vec{\nabla}\times(\vec{\nabla f})=0$ et pour $\vec{A}$ un champ de vecteurs sur $\R^3$, on a $\vec{\nabla}\cdot(\vec{\nabla}\times\vec{A})=0$.
\end{propriete}
Cette propriété des opérateurs différentiels n'est pas sans rappeler que pour la dérivée extérieure, $d\circ d=0$.

On a déjà vu que la dérivée extérieure jouait le même rôle que le gradient, au choix d'une métrique près. Pour $f\in\mathcal{C}^{\infty}(\R^{3})$, on peut écrire :
$$ df = (\partial_xf)dx+(\partial_yf)dy+(\partial_zf)dz$$Calculons alors dans le cas le plus général la dérivée extérieure d'une $1$-forme quelconque : $$\omega=\omega_xdx+\omega_ydy+\omega_zdz$$ ce qui donne : $$d\omega = (\partial_y\omega_z-\partial_z\omega_y)dy\wedge dz +(\partial_z\omega_x-\partial_x\omega_z)dz\wedge dx +(\partial_x\omega_y-\partial_y\omega_x)dx\wedge dy
$$
Autrement dit, la dérivée extérieure d'une $1$-forme est essentiellement le rotationnel, à condition de pouvoir se ramener à un vecteur, ce qui est possible avec une métrique et l'opérateur de Hodge que nous définirons plus tard. Pour une $2$-forme quelconque :
$$ \omega=\omega_{xy}dx\wedge dy+\omega_{yz}dy\wedge dz+\omega_{zx}dz\wedge dx$$un calcul rapide donne :
$$d\omega=(\partial_z\omega_{xy}+\partial_z\omega_{xy}+\partial_z\omega_{xy})dx\wedge dy\wedge dz$$ qui est la divergence un tout petit peu déguisée ! On peut résumer le tout en notant : 
$$ Gradient \leftrightarrow d:\O^{0}(\R^3)\rightarrow\O^1(\R^3)$$
$$ Rotationnel \leftrightarrow d:\O^{1}(\R^3)\rightarrow\O^2(\R^3)$$
$$ Divergence \leftrightarrow d:\O^{2}(\R^3)\rightarrow\O^3(\R^3)$$

\subsection{Crochet de Lie de deux champs de vecteurs}

\begin{figure}[!h]
	\centering
	\includegraphics{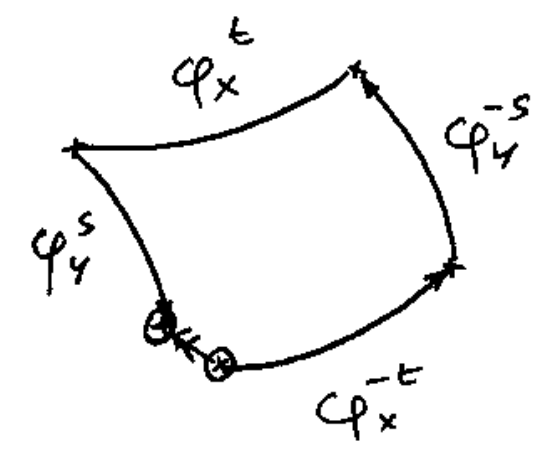}
	\caption{Le crochet de Lie permet de mesurer le défaut de commutation de deux champs de vecteurs}
\end{figure} 

Soit $M^d$ une variété, $x_0\in M$. Soient deux champs de vecteurs X et Y définis localement en $x_0$ et $\phi_X$, $\phi_Y$ leurs flots respectifs. On veut décrire à quel point les flots de ces champs de vecteurs commutent. On définit l'application $$\Phi : \R^2 \rightarrow M$$ localement au voisinage de $(0,0)$ par : $$\Phi(s,t)=\phi_Y^s\circ\phi_X^t\circ\phi_Y^{-s}\circ\phi_X^{-t}(x_0)$$
Comme pour tous les $s$, $t$ dans un voisinage adapté de $0$, $\Phi(0,t)=\Phi(s,0)=0$, on a $$\partial_1\Phi(0,0)=\partial_2\Phi(0,0)=0$$ La dérivée seconde $d^2\Phi(0,0)$ est donc bien définie et $d^2\Phi(0,0)\dot (s,t)=st\partial^2_{12}\Phi(0,0)$.
On pose : $$[X,Y](x_0)=\partial^2_{12}\Phi(0,0)$$
C'est bien un champ de vecteur lisse au sens défini plus haut. En considérant la courbe $\gamma:t\rightarrow \Phi(t,t)$ définie sur un voisinage convenable de $0$, de dérivée nulle en $0$, on a $[X,Y](x_0)=\frac{1}{2}\frac{d^2}{dt^2}\gamma(0)$.

Par exemple, sur $\R^2$, les champs de vecteurs $X=(1,0)$ et $Y=(0,x)$ ne commutent pas. On peut calculer leur crochet : $[X,Y](0)=(0,-1)$.

\begin{figure}[!h]
	\centering
	\includegraphics{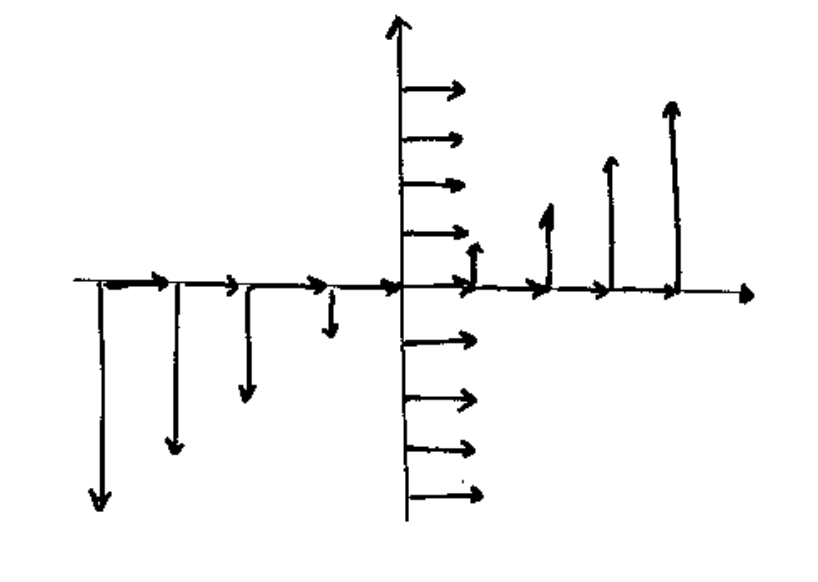}
	\caption{Deux champs de vecteurs qui ne commutent pas}
\end{figure}

\begin{theorem}[Stabilité par poussé en avant]
	Soit $f:M\rightarrow N$ un difféomorphisme de variétés, et $X$, $Y$ deux champs de vecteurs sur M. Alors $[f_*X,f_*Y]=f_*[X,Y]$.
\end{theorem}

\begin{proof}[Preuve]
	Soit $x_0\in M$. On appelle $r(t):sign(t)\sqrt{|t|}$. On considère la courbe $x(t)=\Phi(r(t),r(t))$ qui donne $x'(0)=[X,Y](x_0)$. Alors : $$f\circ x(t)=f\circ\phi_Y^{r(t)}\circ f^{-1}\circ f\circ\phi_X^{r(t)}\circ f^{-1}\circ f\circ\phi_Y^{-r(t)}\circ f\circ\phi_X^{-r(t)}\circ f^{-1}\circ f(x)$$ $$f\circ x(t)=\phi_{f_*Y}^{r(t)}\circ\phi_{f_*X}^{r(t)}\circ\phi_{f_*Y}^{-r(t)}\circ\phi_{f_*X}^{r(t)}(f(x))$$ ce qui donne la formule voulue en dérivant en 0.
\end{proof}

\begin{theorem}
	On a $[X,Y](x_0)=\frac{d}{ds}_{|s=0}((\phi_Y^s)_*X(x))$
\end{theorem}

\begin{proof}[Preuve]
	$$\phi_Y^s\circ\phi_X^t\circ\phi_Y^{-s}=\phi_{(\phi_Y^s)_*X}^t$$
	$$\Phi(s,t)=\phi^t_{(\phi_Y^s)_*X}\circ\phi_X^{-t}(x_0)$$
	$$\partial_2\Phi(s,0)=(\phi_Y^s)_*X(x)-X(x)$$
\end{proof}

En coordonnées on a donc : $([X,Y])_i=\partial_j(Y_i)X_j - \partial_j(X_i)Y_j$ avec la convention de sommation d'Einstein. Cela nous permet d'exprimer le crochet de Lie d'une troisième manière, encore différente, en termes de dérivées le long de champs de vecteurs. On se rend ici compte de l'intérêt de la notation $\partial_i$ pour les champs de vecteurs canoniques locaux.

\begin{theorem}
	Soient $X$ et $Y$ deux champs de vecteurs sur $M$. Le crochet $[X,Y]$ est l'unique champ de vecteurs défini sur $M$ tel que pour toute fonction $f$ lisse sur $M$, on ait $[X,Y]_x[f]=X_x[Y[f]]-Y_x[X[f]]$. 
\end{theorem}

\begin{proof}[Preuve]
	On fait le calcul direct dans les cartes pour arriver au résultat.
\end{proof}

Le crochet vérifie certaines propriétés algébriques intéressantes ; par exemple, pour trois champs de vecteurs $X$, $Y$, et $Z$ sur $M$, on a (égalité de Jacobi) : $$[X,[Y,Z]]+[Y,[Z,X]]+[Z,[X,Y]]=0$$

\subsection{La première paire d'équations}
On cherche à réécrire la première paire d'équations de Maxwell en trouvant une formulation qui les généralise à n'importe quelle variété. Il s'agit des équations données par : $$\vec{\nabla}\cdot\vec{B}=0$$ et $$\vec{\nabla}\times\vec{E}+\partial_t\vec{B}=0$$ 
\paragraph{Cas statique}
On a vu que la divergence est le produit extérieur sur les $2$-formes dans $\R^3$ et que le rotationnel le produit extérieur sur les $1$-formes dans $\R^3$. Ainsi plutôt que de considérer les champs électrique $\vec{E}=(E_x, E_y, E_z)$ et magnétique $\vec{B}=(B_x, B_y, B_z)$ comme des vecteurs, on va les considérer respectivement comme une $1$-forme et une $2$-forme :
$$
B=B_xdy\wedge dz +B_ydz\wedge dx+B_zdx\wedge dy
$$
$$
E=E_xdx+E_ydy+E_zdz
$$
Ainsi la première paire d'équations de Maxwell s'écrit simplement $dE=0$, $dB=0$.
\paragraph{Cas général}
Dans le cas où les champs ne sont pas statiques, il faut les penser comme des objets de l'espace temps et non plus seulement de l'espace. On se place dans $\R^4$ avec les coordonnées standards $(x^0, x^1, x^2, x^3) = (t,x,y,z)$. On prend toujours :
$$
E=E_xdx+E_ydy+E_zdz
$$
$$
B=B_xdy\wedge dz +B_ydz\wedge dx+B_zdx\wedge dy
$$
On définit alors le champ électromagnétique $F$ comme une $2$-forme sur $\R^4$, par :
$$
F=B+E\wedge dt
$$
qu'on peut décomposer sur la base canonique des $2$-formes :
$$
F=\frac{1}{2}F_{\mu\nu}dx^\mu\wedge dx^\nu
$$
et sous forme matricielle, en prenant la base canonique de $\R^4\otimes\R^4$ donnée par la famille : $(dx^i\otimes dx^j)_{i,j\in[|0,3|]})$ si $((dx^j)_{j\in[|0,3|]})$ est la base canonique de $\bigwedge(\R^4)$ : $$F = \begin{pmatrix} 0 & -E_x & -E_y & -E_z \\ E_x & 0 & B_z & -B_y \\ E_y & -B_z & 0 & B_x \\ E_z & B_y & -B_x & 0 \end{pmatrix}$$
\begin{theorem}
	La première paire d'équations de Maxwell s'écrit juste : $$dF=0$$
\end{theorem}

\begin{proof}[Preuve]
	Tout d'abord on a : $$dF=d(B+E\wedge dt)$$ Or si $\omega=\omega_Idx^I$ où $I$ parcourt l'ensemble des multi-indices (par exemple pour $\omega=F$, $\omega_{I}=\frac{1}{2}F_{\mu\nu}$ et $dx^I=dx^\mu\wedge dx^\nu$) alors $d\omega=\partial_\mu\omega_Idx^\mu\wedge dx^I$. On peut donc décomposer $d\omega$ en une partie spatiale $d_s\omega=\partial_i\omega_Idx^i\wedge dx^I$ où $i$ parcourt $[|1,3|]$ et une partie temporelle $dt\wedge\partial_t\omega=\partial_0\omega_Idx^0\wedge dx^I$. Par conséquent :
	$$
	dF=d_sB+dt\wedge\partial_tB+(d_sE+dt\wedge\partial_tE)\wedge dt
	$$
	$$
	dF=d_sB+(\partial_tB+d_sE)\wedge dt
	$$
	d'où le système d'équations :
	\[
	\left\{
	\begin{array}{l}
	d_sB = 0 \\
	\partial_tB + d_sE = 0
	\end{array}
	\right.
	\]
	qui est équivalent aux deux premières équations de Maxwell.	
\end{proof}
\begin{rmq}
		La grande force du langage des formes différentielles est sa généralité, qui permet de définir un champ magnétique $F$ comme une $2$-forme sur n'importe quelle variété $M$. Les deux premières équations de Maxwell la contraignent juste à être fermée. On peut voir l'espace temps comme un fibré en droite sur la sous-variété d'espace, cependant ce fibré n'est pas nécessairement trivial : autrement dit on ne peut pas forcément écrire $M = \R\times S$ où $S$ est la variété d'espace. Les notions de champ électrique et de champ magnétique peuvent alors perdre leur sens ; seule subsiste celle de champ électromagnétique.
\end{rmq}
\begin{rmq}
	Si l'espace temps $M$ est trivial, il peut s'écrire $M=\R\times S$ de bien des façons ! En relativité, les différents repères inertiels correspondent à différentes trivialisations de cette nature, reliées par les transformations de Lorentz. C'est pourquoi le champ électrique et le champ magnétique sont mélangés lorsqu'ils subissent une transformation de Lorentz.
\end{rmq}
\begin{rmq}
	La première paire d'équations n'implique pas le mesure de distances dans l'espace-temps. On n'a pas encore parlé de la métrique de Minkowski, seulement de $\R^4$ ! C'est pour cela que ces équations sont invariantes par transformation de Galilée. La deuxième paire d'équations, elle, a besoin d'une métrique de manière implicite.
\end{rmq}

\subsection{Métrique, élément de volume et opérateur de Hodge}
\paragraph{Le cas de l'espace plat}
\begin{defi}
	Une métrique sur $E$ où $E$ est un espace vectoriel réel est une $2$-forme $g\in \mathfrak{T}^0_2(E)$ telle que $g$ est \textbf{symétrique} et \textbf{non dégénérée} (c'est-à-dire que si $\forall v \in E$, $g(u,v)=0$ alors $u=0$).
\end{defi}
\begin{proposition}
	Il existe une base $(e_i)_{i\in[|1,n|]}$ orthonormale pour $g$ au sens suivant : $$g(e_i,e_j)=\pm \delta_{ij} \forall i, j\in [|1,n|]$$
\end{proposition}
\begin{proof}[Preuve]
	La forme $g$ bilinéaire symétrique $g$ étant non nulle, et sachant que :
	$$
	g(x,y)=\frac{1}{2}(g(x+y,x+y)-g(x,x)-g(y,y))
	$$
	on en déduit qu'il existe un $x\in E$ tel que $g(x,x)$ est non nul (autrement dit il existe un $x$ non isotrope). Par conséquent $\R x\cap x^{\perp}=\emptyset$ donc $\R x \oplus x^{\perp} = E$. On applique alors le même raisonnement à $x^{\perp}$ et on obtient un base orthogonale $(e_i)_{i\in [|1,n|]}$ qui permet d'écrire, en posant $g(e_{i},e_i)=\alpha_i$ : 
	$$
	g(x,x) = \alpha_1(x^{1})^2+...+\alpha_n(x^{n})^2
	$$
	Et comme on est dans $\R$, on en déduit le résultat voulu.
	\end{proof}
	
\begin{defi}
	Un élément de volume sur $E$ relativement à $g$ est un $\mu\in\bigwedge^n(E)$ donné par $e^1\wedge...\wedge e^n$ si $(e_i)_{i\in [|1,n|]}$ est une base orthonormale. On a $e^i\in\bigwedge^1(E)$ et $\mu(e_1, ..., e_n)=1$. $\mu$ dépend d'un facteur $\pm1$ du choix d'une base orthonormée. Un classe d'équivalence de bases pour ce critère est appelé une \textbf{orientation} de $E$. Si $\mu$ est une orientation de $E$, une base $(v_i)_{i\in[|1,n|]}$ est dite \textbf{orientée positivement} si $\mu(v_1, ..., v_n)=1$.
\end{defi}

\begin{defi}
	Une métrique $g$ sur $E$ induit une métrique $\tilde{g}\in\mathfrak{T}^0_2(\bigwedge^k(E))$ définie comme suit. Soit $(v_i)_{i\in[|1,n|]}$ une base de $E$, et posons $g_{ij}=g(v_i, v_j)$. Soit $g^{ij}$ l'élément de matrice $(i,j)$ de l'inverse de $(g_{ij})_{i,j\in[|1,n|]}$. Pour $\alpha, \beta \in \bigwedge^k(E)$, on définit $\tilde{g}(\alpha,\beta)$ en termes de composants (pour la base $(v_i)_{i\in[|1,n|]}$) par :
	$$ 
	\tilde{g}(\alpha,\beta)=\frac{1}{k!}g^{i_1j_1}...g^{i_kj_k}\alpha_{i_1...i_k}\beta{j_1...j_k}
	$$
	$\tilde{g}(\alpha, \beta)$ est indépendant du choix de la base. Si $\alpha, \beta \in \bigwedge^0(E)$, on pose $\tilde{g}(\alpha, \beta)=\alpha\beta$
\end{defi}
\begin{theorem}
	Soient $g$ une métrique sur $E$ de dimension $n$ et $\mu$ un élément de volume de $E$ relativement à $g$. Il y a un unique isomorphisme $*:\bigwedge^k(E)\rightarrow\bigwedge^{n-k}(E)$ tel que $$\forall \alpha,\beta\in\bigwedge^k(E)\ \ \alpha\wedge*\beta=\tilde{g}(\alpha,\beta)\mu$$
\end{theorem}

\begin{proof}[Preuve]
	Pour $\gamma\in\bigwedge^{n-k}(E)$, définissons $\phi_\gamma :\bigwedge^k(E)\rightarrow\R$ par $\phi_\gamma(\alpha)\mu=\alpha\wedge\mu$. Si $\phi_\gamma(\alpha) = 0$ $\forall \alpha\in\bigwedge^k(E)$ alors $\gamma=0$. Par conséquent, $\gamma \mapsto \phi_\gamma$ définit une application linéaire injective de $\bigwedge^{n-k}(E)$ dans $\bigwedge^{k}(E)^{*}$ (espace dual). Comme $dim(\bigwedge^{n-k}(E))=dim(\bigwedge^{k}(E))$ c'est en fait un isomorphisme. Par conséquent, pour tout $\beta\in\bigwedge^k(E)$ il existe $\gamma\in\bigwedge^{n-k}(E)$ tel que $\phi_\gamma(\alpha)=\tilde{g}(\alpha, \beta)$ $\forall \alpha\in\bigwedge^k(E)$. $\gamma$ est renommé $*\beta$. Donc :
	$$
	\alpha\wedge*\beta=\alpha\wedge\gamma=\phi_\gamma(\alpha)\mu=\tilde{g}(\alpha,\beta)\mu
	$$
	L'opérateur $*$ est clairement un isomorphisme.
\end{proof}

Soit $g$ une métrique arbitraire sur $E$, soit $g^{ij}$ définie comme ci-dessus relativement à la base arbitraire de $E$ $(v_i)_{i\in[|1,n|]}$. Pour $\omega\in\bigwedge^k(E)$ de composantes $\omega_{i_1...i_k}$, on définit : 
$$
\omega^{i_1...i_k}=g^{i_1j_1}...g^{i_kj_k}\omega_{i_1...i_k}
$$
Posons $|g|=|det(g_{ij})|$. Si $(e_i)_{i\in[|1,n|]}$ est une base orthonormale pour $g$, posons $g(e_i,e_j)=\eta_{ij}$ et $(-1)^g=det(\eta_{ij})$. Si $v_j=a_j^ie_i$, alors $g_{ij}=a^k_ia^l_j\eta_{kl}$ ou en terme de matrice : $G=A^T\eta A$. Par conséquent, $|detA|=|g|^{1/2}$, donc $\mu\equiv e^1\wedge...\wedge e^n=|g|^{1/2}v^1\wedge...\wedge v^n$ si la base $(v_i)_{i\in[|1,n|]}$ est orientée positivement par rapport à $\mu$.
On pose finalement $$\epsilon_{i_1...i_n}=\sum_{\sigma\in\mathfrak{S}_n}(-1)^\sigma e^{i_1}(e_{\sigma(1)})...e^{i_n}(e_{\sigma(n)})$$ On peut calculer le résultat suivant :
\begin{theorem}
	Si $g$ est une métrique sur $E$ d'orientation $\mu\in\bigwedge^n(E)$, et soit $(v_i)_{i\in[|1,n|]}$ une base orientée positivement de $E$. Pour $\omega=\omega_{i_1...i_k}v^{i_1}\otimes...\otimes v^{i_k}\in\bigwedge^k(E)$, on a :$$ *\omega=|g|^{1/2}\frac{1}{k!}\omega^{j_1...j_k}\epsilon_{j_1...j_kj_{k+1}...j_n}v^{j_{k+1}}\otimes...\otimes v^{j_n}$$ 
\end{theorem}
De la même façon, un calcul donne le résultat suivant :
\begin{theorem}
	Pour $\omega\in\bigwedge^k(E)$ on a $**\omega=(-1)^g(-1)^{k(n-k)}\omega$
\end{theorem}

\paragraph{Sur une variété}

\begin{defi}
	Une métrique sur une variété $M^m$ est un champ de tenseurs $g\in\mathfrak{T}^0_2$ tel que $g_x$ qui est une métrique sur $T_{x}M$ soit symétrique et non dégénérée pour tout $x\in M$.
\end{defi}

\begin{rmq}
	La signature est localement constante donc constante sur les composantes connexes de $M$.
\end{rmq}

Une variété munie d'une métrique est appelée \textbf{variété semi-riemannienne}. Si $g$ est définie positive (ou définie négative) on parle de \textbf{variété riemannienne}. Si la signature de la métrique est du type $(n-1,1)$, on dit que la métrique est lorentzienne, et on parle de \textbf{variété lorentzienne}.

Les notions que nous avons défini grâce à la métrique en espace plat se transposent directement à la variété.

\begin{rmq}
	La métrique désigne un moyen d'associer à chaque vecteur une forme linéaire et inversement (ce que correspond à monter et descendre des indices). C'est grâce à la métrique qu'on peut par exemple assimiler un produit extérieur de deux vecteurs dans $\R^3$ au (pseudo -) vecteur produit vectoriel des deux (c'est l'étoile de Hodge). La première paire d'équations de Maxwell ne dépend donc de la métrique que si on veut l'écrire sous sa forme classique, on n'en a besoin que pour assimiler $E$ et $B$ à des vecteurs. Naturellement $E$ est plutôt une $1$-forme, qu'on peut visualiser comme des hyperplans parallèles en un point (sortes de lignes de niveau). On l'assimile à un vecteur orthogonal à ces hyperplans, grâce à la métrique qui définit la notion même d'orthogonalité : $<u|v>=g_{\alpha\beta}u^\alpha v^\beta$
\end{rmq}	

\begin{defi}
	Une $n$-forme $v$ nulle part nulle sur une variété $M^n$ est appelée une orientation de $M$. La paire $(M,v)$ est appelée \textbf{variété orientable}
\end{defi}

\begin{rmq}
	L'existence de cette forme n'est pas assurée : la surface de Klein est par exemple une variété de dimension 2 non orientable
\end{rmq}
\begin{defi}
	Si $M^n$ est une variété orientée avec une métrique $g$, il y a un élément de volume canonique sur $M$ défini comme suit. Pour chaque carte $\phi_\alpha:U_\alpha\rightarrow\R^n$, on pose $g_{\mu\nu}=g(\partial_\mu, \partial_\nu)$ et on définit $vol=|det(g_{\mu\nu})|^{1/2}dx^{1}\wedge...\wedge dx^{n}$.
\end{defi}	

\begin{rmq}
	Cette définition est cohérente. En effet si $(U', \phi')$ est une autre carte telle que $U\cap U'$ n'est pas vide, en définissant $vol'=|det(g'_{\mu\nu})|^{1/2}dx'^{1}\wedge...\wedge dx'^{n}$. Sur l'intersection on a $dx'^{\mu}=\frac{\partial x'^{\nu}}{\partial x^{\mu}}dx^{\mu}$. Par conséquent, $dx'^{1}\wedge...\wedge dx'^{n}=(detT)dx^{1}\wedge...\wedge dx^{n}$.
	Par ailleurs :
	$$
	g'_{\mu\nu}=g(\partial'_\mu, \partial'_\nu)=g(\frac{\partial x^{\alpha}}{\partial x'^\mu}\partial_\alpha,\frac{\partial x^{\nu}}{\partial x'^{\nu}}\partial_\beta)=(T^{-1})_\mu^\alpha(T^{-1})_\nu^\beta g_{\alpha\beta}
	$$
	Les deux cartes étant orientées, on en déduit donc :
	$$
	|detg'_{\mu\nu}|^{1/2}=(detT)^{-1}|detg_{\mu\nu}|^{1/2}
	$$
	et donc on a bien $vol$=$vol'$.
\end{rmq}

\subsection{La deuxième paire d'équations}

Il y a une forte symétrie entre les deux couples d'équations de Maxwell ;  dans le vide on obtient même le deuxième à partir du premier en substituant $-\vec{B}$ à $\vec{E}$ et $\vec{E}$ à $\vec{B}$ ! 
\[
\left\{
\begin{array}{r c l}
\vec{\nabla}\cdot\vec{B} &=& 0\\
\vec{\nabla}\times\vec{E} &=& -\frac{\partial\vec{E}}{\partial t}\\
\end{array}
\right.
\]
\[
\left\{
\begin{array}{r c l}
\vec{\nabla}\cdot\vec{E} &=& \rho\\
\vec{\nabla}\times\vec{B} &=& \frac{\partial\vec{E}}{\partial t}+\vec{j}\\
\end{array}
\right.
\]
Cependant, pour le premier couple nous avons interprété $\vec{E}$ comme une $1$-forme et $\vec{B}$ comme une $2$-forme puisqu'on considérait respectivement leur rotationnel et leur divergence, or le rotationnel et la divergence peuvent être interprétés comme des restrictions de l'opérateur de différentiation extérieure. Or dans le deuxième couple, c'est cette fois la \textbf{divergence de $\vec{E}$} et le \textbf{rotationnel de $\vec{B}$} que l'on regarde, comme s'il fallait cette fois \textbf{interpréter $\vec{E}$ comme une $2$-forme et $\vec{B}$ comme une $1$-forme} ! C'est ce que fait l'opérateur "$*$" de Hodge, à condition d'avoir défini une métrique et un choix d'orientation.

\paragraph{Dans l'espace de Minkowski}

On considère tout d'abord le cas "habituel", où la variété $M$ est l'espace-temps de Minkowski $\R^4$ muni d'une métrique lorentzienne. On peut alors décomposer le champ électromagnétique en champs électrique et magnétique : $$ F=B+E\wedge dt$$ La métrique de Minkowski : $$\eta(v,w)=-v^0w^0+v^1w^1+v^2w^2+v^3w^3$$ permet alors de définir un opérateur de Hodge "$*$" et de considérer $*F$ donnée sous forme matricielle dans $\R^4\otimes\R^4$ par :
$$*F = \begin{pmatrix} 0 & B_x & B_y & B_z \\ -B_x & 0 & E_z & -E_y \\ -B_y & -E_z & 0 & E_x \\ -B_z & E_y & -E_x & 0 \end{pmatrix}$$ 
Le deuxième couple d'équations de Maxwell fait aussi intervenir la distribution de charges et de courants, qu'il faut exprimer en termes de formes différentielles : la densité de courant relativiste $$\tilde{j}=\rho\partial_0+ j^1\partial_1+ j^2\partial_2+ j^3\partial_3$$ canoniquement associée à la 1-forme $$j=-\rho dx^0+ j_1dx^1+ j_2dx^2+ j_3dx^3$$ est en réalité l'étoile d'une $3$-forme que nous appellerons désormais courant : $$J=-j_1dx^2\wedge dx^3\wedge dt-j_2dx^3\wedge dt\wedge dx^1-j_3dt\wedge dx^1\wedge dx^2+\rho dx^1\wedge dx^2\wedge dx^3$$ car : 
\[
\left\{
\begin{array}{r c l}
*dx\wedge dy\wedge dz &=& -dt\\
*dx\wedge dy\wedge dt &=& -dz\\
*dz\wedge dx\wedge dt &=& -dy\\
*dy\wedge dz\wedge dt &=& -dx
\end{array}
\right.
\]

\begin{rmq}
	Après discussion avec M. Lévy, il est en effet plus naturel de définir le courant comme une $3$-forme que comme une $1$-forme puisque chacune des composantes a alors la bonne nature pour être intégrée sur un volume d'espace-temps : par exemple, $j_1$ représente une quantité qui traverse une surface élémentaire $dx^2\wedge dx^3$ pendant $dt$. De même, si on intègre $\rho$ sur un volume spatial on obtient la charge totale dans ce volume.
\end{rmq}
\begin{theorem}
	Le deuxième couple d'équations de Maxwell s'écrit :  $d*F=J$
\end{theorem}
\begin{proof}[Preuve]
	On a : $$*F=-B_xdx\wedge dt-B_ydy\wedge dt-B_zdz\wedge dt+E_xdy\wedge dz+E_ydz\wedge dx+E_zdx\wedge dy$$ et par conséquent : 
	$$ d*F=(\partial_xE_x+\partial_yE_y+\partial_zE_z)dx\wedge dy\wedge dz + (\partial_yB_x-\partial_xB_y+\partial_tE_z)dx\wedge dy\wedge dt $$
	$$+(\partial_xB_z-\partial_zB_x+\partial_tE_y)dz\wedge dx\wedge dt+(\partial_zB_x-\partial_xB_z+\partial_tE_y)dy\wedge dz\wedge dt$$
	ce qui donne bien les deux dernières équations de Maxwell.

\end{proof}		
Cette deuxième équation très synthétique, peut servir de définition des équations de Maxwell sur des variétés. Si $J$ est une $3$-forme interprétée comme la $3$-forme courant électromagnétique, alors la condition $d*F=J$ a un sens sur la variété, et constitue la deuxième et dernière équation de Maxwell.
Résumons maintenant les idées auxquelles nous sommes arrivés, en rajoutant de plus en plus de structure sur la variété afin d'arriver finalement aux équations de Maxwell "classiques".
\begin{center}
	\fbox{\begin{minipage}{140mm}
			Soit $M$ une variété quelconque. 
			\begin{enumerate}
			\item On définit le \textbf{champ électromagnétique $F$ comme une $2$-forme sur $M$}, et le \textbf{courant $J$ comme une $1$-forme sur $M$}.
			\item La première équation de Maxwell s'écrit $dF=0$.
			\item A condition de supposer que \textbf{$M$ est une variété semi-riemannienne orientée} on peut alors écrire la deuxième équation de Maxwell comme $d*F=J$.
			\item Pour pouvoir introduire les champ électrique et magnétique, il faut supposer que $M=\R\times S$ où $S$ est la variété d'espace, et écrire $F=B+E\wedge dt$. De même on écrit $J=j+\rho dx^1\wedge dx^2\wedge dx^3$. La première équation de Maxwell se sépare alors en 
			$d_SB=0$ et $\partial_tB+d_SE=0$.
			\item  Si on suppose de plus que l'espace $S$ est de dimension $3$, que la métrique sur $M$ est lorentzienne, et que la restriction de la métrique à l'espace $S$ est riemannienne, en notant $*_S$ l'opérateur étoile de Hodge sur les formes différentielles sur $S$, on a $*F=*_SE-*_SB\wedge dt$ d'où $d*F=*_S\partial_tE\wedge dt+d_S*_SE-d_S*_SB\wedge dt$ donc $*d*F=-\partial_tE+*_Sd_S*_SE\wedge dt +*_Sd_S*_SB$ et en écrivant $*d*F=*J$ on obtient la deuxième équation de Maxwell sous sa forme "historique" :
			\[
			\left\{
			\begin{array}{r c l}
			*_Sd_S*_SE &=& \rho\\
			-\partial_tE+*_Sd_S*_SB &=& j\\
			\end{array}
			\right.
			\]
			
			\end{enumerate}
		\end{minipage}}
\end{center}

\chapter{Une théorie géométrique des interactions}

Dans ce chapitre, nous allons donner une approche plus intrinsèque et plus générale aux équations de Maxwell. Nous aboutirons à la forme 'récente' de ces équations, sans trop rentrer dans les détails pour l'équation inhomogène, le but étant plutôt de présenter un cadre géométrique agréable pour la description des interactions élémentaires et d'acquérir une certaine intuition à ce sujet, que de mener tous les raisonnements de la manière la plus abstraite possible. Les structures que nous allons développer mettent en valeur la nature de jauge des théories modernes des interactions : la relativité générale qui, même si nous l'évoquerons peu, est décrite de manière conventionnelle avec ces idées, l'électrodynamique quantique (QED), la théorie électrofaible (modèle de Glashow-Weinberg-Salam) et la chromodynamique quantique (QCD). Néanmoins, nous ne nous arrêterons pas sur la quantification des théories classiques, qui sort du cadre de ces travaux. Dans un premier temps, nous présenterons les idées décisives qui ont mené à l'émergence des concepts que nous étudions en les replaçant dans leur concept historique, puis nous proposerons une approche géométrique aux théories des interactions.\\\\
\textit{"Une géométrie ne peut pas être plus vraie qu'une autre, elle peut seulement être plus commode"} (Poincaré, 1902)

\section{Rapports entre les théories de jauge et les interactions}

Le concept de champ de jauge gouverne aujourd'hui la physique des hautes énergies. Il est même universellement accepté qu'\textbf{une bonne théorie des interactions fondamentales doit être une théorie de jauge}. Cependant, lors de l'élaboration de ces théories, ce n'était pas encore le cas, et c'est seulement grâce à des idées remarquables de Weyl, Dirac, Aharanov et Bohm, Yang et Mills que ces points de vu ont pu émerger.
\subsection{Invariance de jauge, apparition des charges et indices pour la quantification}
Dans toute théorie d'interaction, il est question de charges - au sens large. On sait par exemple que la charge électrique est conservée par les interactions électromagnétiques. Pour l'interaction faible, on parle d'isospin faible et d'hypercharge faible, pour l'interaction forte, de couleurs.
On ne peut pas vraiment parler d'interaction si on ne parle pas de charges, qui codent l'intensité de l'interaction pour la particule considérée.\\
C'est le théorème de Noether (1918) qui est en fait à l'origine du lien entre invariance de jauge et existence de charges conservées.\\\\
Pour un observable $O$ dans un système de hamiltonien $H$, l'équation du mouvement dans la description de Heisenberg s'écrit : 
$$i\hbar\frac{dO(t)}{dt}=[O(t),H]$$
On considère un état $\ket{\psi}$. On fait subir au système une transformation unitaire $U$ : $\ket{\psi}'=U\ket{\psi}$. L'observable O est alors transformée selon $O'=UOU^\dagger$. Comme nous le verrons plus tard, on peut écrire $U=e^{i T}$ où T est une matrice hermitienne ($T=T^\dagger$) correspondant au vecteur tangent à la courbe $\alpha\rightarrow e^{i\alpha T}$ en $\alpha=0$. La dérivée étant encore évaluée en $0$ : 
$$\frac{dO'}{d\alpha}=i[T,O]$$
car $(1+i\alpha T)P(1-i\alpha T)=O+i\alpha[T,O]$ au premier ordre en $\alpha$ ; et si le hamiltonien est invariant sous la transformation U, $\frac{dH'}{d\alpha}=0$ car $H'=H$ pour tout $\alpha$ et en prenant O=H dans la formule ci-dessus, donc on obtient que $T$ est une intégrale première de l'évolution hamiltonienne du système.

Dérivons le théorème de Noether dans le cas particulier des théories des champs : soit un système de lagrangien $\mathcal{L}$ fonction d'un nombre fini n de champs $\phi_r$ (dépendant de manière régulière de 4 variables $x^\alpha$) et de leur dérivées $\phi_{r,\alpha}=\frac{\partial\phi}{\partial x^{\alpha}}$. Sous une transformation unitaire on suppose que l'on a :
$$\phi_r\rightarrow \phi_r'(x)=\phi_r(x)+\delta \phi_r(x)$$
Or de manière très générale : 
$$\delta\mathcal{L}=\frac{\partial\mathcal{L}}{\partial \phi_r}\delta\phi_r + \frac{\partial\mathcal{L}}{\partial \phi_{r,\alpha}}\delta\phi_{r,\alpha}$$
Si les champs sont des solutions 'physiques', ils vérifient l'équation d'Euler-Lagrange : 
$$\frac{\partial\mathcal{L}}{\partial \phi_r}=\partial_\alpha(\frac{\partial\mathcal{L}}{\partial \phi_{r,\alpha}})$$
ce qui permet de réécrire : 
$$\delta\mathcal{L}=\partial_\alpha(\frac{\partial\mathcal{L}}{\partial \phi_{r,\alpha}}\delta\phi_r)$$
Si maintenant on suppose que le lagrangien est invariant sous cette transformation, c'est-à-dire $\delta\mathcal{L}=0$, en posant $f^\alpha=\frac{\partial\mathcal{L}}{\partial \phi_{r,\alpha}}\delta\phi_r$ on obtient l'équation de conservation : 
$$\partial_\alpha f^\alpha=0$$
et posant $F^\alpha=\int d^3\vec{x}f^\alpha(\vec{x},t)$, $F^0$ est une intégrale première de l'évolution. Remarquons que $$F^0=c\int d^3\vec{x}\pi_r(x)\delta\phi_r(x)$$
Ainsi l'invariance de lagrangien sous des transformations continues impose la conservation de certaines quantités lors de l'évolution.

Pour des champs complexes se transformant comme $\phi_r'=e^{i\epsilon}\phi_r$ et $\phi_r^{\dagger'}=e^{-i\epsilon}\phi_r^\dagger$ on obtient : 
$$F^0=i\epsilon c\int d^3\vec{x}[\pi_r(x)\delta\phi_r(x)-\pi_r^\dagger(x)\delta\phi_r^\dagger(x)]$$
On pose : 
$$Q=-i\frac{q}{\hbar}\int d^3\vec{x}[\pi_r(x)\delta\phi_r(x)-\pi_r^\dagger(x)\delta\phi_r^\dagger(x)]$$
Dans la théorie quantifiée, on a
$$
[\pi_s(x),\phi_r(y)]=i\hbar\delta_{rs}\delta^4(x-y) \\
$$
$$
[\pi_s^\dagger(x),\phi_r^\dagger(y)]=i\hbar\delta_{rs}\delta^4(x-y) \\
$$
$$
[\pi_s^\dagger(x),\phi_r(y)]=0 \\
$$
$$
[\pi_s(x),\phi_r^\dagger(y)]=0
$$
où $\pi_r(x)=\frac{1}{c^2}\frac{d\phi^\dagger}{dt}$.
Cela conduit à la relation :
$$[Q,\phi_r(x)]:-\frac{iq}{\hbar}\int d^3\vec{x'}[\pi_s(x'),\phi_r(x)]\phi_s(x')=-q\phi_r(x)$$
qui permet de voir que si $\ket{Q'}$ est vecteur propre de $Q$ avec valeur propre $Q'$, alors $\phi_r(x)\ket{Q'}$ est vecteur propre de $Q$ avec valeur propre $(Q'-q)$, comme si $\phi_r$ avait retiré une charge $q$ en agissant sur l'état. C'est un germe de la quantification des théories des champs.

\subsection{L'idée de connexion et l'héritage de Levi-Civita}
Les théories d'Einstein de la relativité restreinte, et générale, entre 1905 et 1916, ont eu un énorme impact sur la mathématique et les mathématiciens \cite{varada}. Élie Cartan ou Hermann Weyl (qui - en bon élève de Hilbert - avait déjà fait de grandes contributions à la théorie spectrale des opérateurs différentiels, mais avait développé une manière très géométrique de penser la physique, en se basant sur les travaux des géomètres italiens du début du siècle : Ricci, Levi-Civita ...) en particulier, ont été bouleversés par l'affirmation que la gravitation n'est rien d'autre que la manifestation de la courbure de l'espace-temps. Weyl était fasciné par les moyens selon lesquels la géométrie différentielle permettait d'expliquer la nature, et avait déjà publié \textit{Die Idee der Riemanschen Fläschen} en 1913.\\\\ 
L'une des principales difficultés à faire de la physique ou de la mathématique sur des variétés est qu'il n'existe pas de système de coordonnées universel ; il faut montrer que les résultats obtenus ne dépendent pas des coordonnées choisies. Comme nous l'avons vu antérieurement, la solution consiste à exprimer les lois de la physique en termes de tenseurs et de leurs dérivées. Cependant, cela implique de pouvoir comparer des vecteurs tangents en deux point différents. Il faut en fait exhiber un isomorphisme pour tout couple d'espaces tangents, et c'est l'idée même de connexion. Sur des variétés riemanniennes, de tels isomorphismes peuvent être fixés en définissant un \textbf{transport parallèle}, autrement dit une manière canonique de transporter un vecteur tangent en un point $p$ en un vecteur tangent en un point $p'$ en suivant la courbe $\gamma$ qui va de $p$ à $p'$. Cet objet permet alors de définir la différentiation de champs vectoriels et tensoriels le long de la courbe.\\\\
\textbf{La connexion de Levi-Civita} est un cas particulier intéressant : soit $(M,g)$ une variété riemannienne (g est la métrique) plongée dans un espace euclidien. De plus, on demande que $g$ soit la restriction de la métrique de l'espace euclidien à la variété $M$. Soit $p\in M$, $Y$ un champ de vecteurs défini au voisinage de $p$ et $X\in T_pM$. Pour obtenir la dérivée covariante de $p$ de $Y$ dans la direction $X$, on prolonge $Y$ en un champ $Y'$ défini au voisinage de $p$, (dans $E$ cette fois) puis on calcule la dérivée directionnelle de $Y'$ dans la direction $X$ dans E. La dérivée covariante cherchée est la projection de la dérivée directionnelle obtenue sur l'espace tangent en $p$ à $M$. On définit alors le transport parallèle le long d'une courbe $\gamma:[0,1]\rightarrow M$ avec $\gamma(0)=p$ et $\gamma'(0)=X$ de la manière suivante : si les $Y(t)$ sont des vecteurs tangents à $M$ en $\gamma(t)$, il sont les transportés parallèlement de $Y$, à condition que la dérivée covariante de Y(t) dans la direction $\gamma'(t)$ soit nulle pour tout $t\in[0,1]$. Cette condition est équivalente aux équations différentielles du premier ordre : 
$$\frac{dY^\mu}{dt}+\Gamma^\mu_{\nu\lambda}(\gamma(t))\gamma^{'\lambda}(t)Y^\nu(t)=0$$
où on a fait apparaitre les "célèbres" symboles de Christoffel $\Gamma^\mu_{\nu\lambda}$. On peut montrer que la connexion de Levi-Civita conserve la norme des vecteurs tangents.

\subsection{De la géométrie riemannienne à la géométrie de Weyl}
Weyl a tout d'abord remarqué que le concept de connexion était en fait généralisable à une variété abstraite, non plongée dans un espace métrique. Plus encore, il a remarqué qu'\textbf{une connexion peut être définie indépendamment de toute métrique}. De là est née la théorie des connexions affines.\\\\
Weyl a également observé que la conservation de la longueur des vecteurs lors du transport parallèle n'était naturelle ni du point du vue mathématique, ni du point de vue physique. Mathématiquement, on s'autorise à changer la direction lors du transport parallèle, il n'y a aucune raison de ne pas aussi changer la norme : "\textit{Beim Herumfahren eines Vektors längs einer geschlossenen Kurve durch fortgesetzte infinitesimale Parallelverscheibungkehrt dieser im allgemeinen in einer andern Lage zurück ; seine Richtung hat sich geändert. Warum nicht auch seine Länge ?}" \cite{weyl1}. Physiquement, les valeurs de longueurs ne sont fixées qu'après le choix d'une échelle de référence, une unité de longueur. Si la variété de base est l'espace-temps où se trouvent plusieurs observateurs, l'invariance de la norme par transport parallèle est en fait l'affirmation que tous les observateurs ont la même référence pour les longueurs, alors qu'ils peuvent être fort éloignés en temps et en espace. Pour que les observateurs soient d'accord sur une échelle de longueur, il faut que l'information soit transmise entre eux, et l'échelle n'a aucune raison de rester la même lors de ce transport. D'où la remarquable observation de Weyl, que la géométrie riemannienne ne peut pas être considérée comme une vraie géométrie infinitésimale puisque la métrique permet de comparer non seulement la norme de deux vecteurs au même point, mais de deux vecteurs en deux points différents de la variété. "\textit{A truly infinitesimal geometry must recongnise only the principle of transference of a length from one point to another infinitely near to the first}"\cite{varada}.

\paragraph{Le fibré des échelles}

Pour Weyl, la métrique de l'espace-temps est donc définie à une échelle près. En considérant l'espace fibré au dessus de l'espace-temps $M$ dont les fibres sont difféomorphes à $R^{+*}$ (chaque point de l'espace-temps est remplacé par une copie de $R^{+*}$), Weyl s'est aperçu qu'on pouvait également définir une connexion sur cet espace grâce à une 1-forme :
$$A=A_\mu dx^\mu$$ qui est l'analogue de la forme dont les coefficients sont les symboles de Christoffel, en dimension 1. On a donc une \textbf{dérivée covariante} qui est une dérivation qui finalement ne prend en compte que les "véritables" variations de longueur une fois que l'échelle est fixée, et un \textbf{transport parallèle} qui donne un sens à la longueur d'un objet transporté dans l'espace temps. Le transport de longueur est équivalent à l'équation : 
$$\frac{ds}{dt}+A_\mu(\gamma(t))\gamma^{'\mu}(t)s(t)=0$$ dont une solution est : $$s(t)=\exp(-\int_{\gamma_t}A_\mu dx^\mu)=\exp(-\int_{\gamma_t}A)$$
où $\gamma_t$ est la restriction de $\gamma$ à $[0,t]$. Soit maintenant une fonction $$g:M\rightarrow \R^{+*}$$ de changement d'échelle : dans ces nouvelles unités, les échelles sont données par : $$s'(t)=s(t)g(\gamma(t))$$ et pour obtenir la même équation pour $s'$ que pour $s$, en remplaçant $A$ par $A'$, il faut faire la transformation $$A'=A-d(\log(g))$$ 
Weyl a appelé cette étape un \textbf{changement de jauge}, et est ainsi arrivé au \textbf{principe d'invariance de jauge}, qui affirme que \textbf{les lois de la physique doivent être invariantes non seulement lors de transformations de coordonnées, mais aussi lors de changement de jauge}, appelés transformations de jauge.

\subsection{L'électromagnétisme, conséquence de la physique quantique}
En 1925, la toute jeune physique quantique révolutionne l'idée de fibré des échelles de Weyl : l'apparition d'une nouvelle constante, $\hbar$, fixe une \textbf{échelle de longueurs universelle}. Il faut donc abandonner le transfert d'échelles comme source de l'électromagnétisme. L'idée que Weyl développe alors présente \textbf{l'électromagnétisme comme conséquence de la physique quantique}. En physique quantique, les fonctions d'ondes sont définies à une phase près. De la meme manière que pour le fibré des échelles, il n'y a aucune raison de fixer une référence identique en tous les points de l'espace-temps. Il faut a priori laisser cette référence varier, et s'intéresser au transfert de phase axiomatique (transport parallèle) le long de chemins dans l'espace-temps. Pour ce faire, le fibré des échelles et remplacé par le fibré des phases. En chaque point de l'espace-temps, l'ensemble des phases forme un groupe isomorphe au groupe $U(1)$. Weyl propose alors de voir l'électromagnétisme comme une connexion sur le fibré des phases.
La connexion est localement définie par la 1-forme : 
$$-iA_\mu dx^\mu$$
Le transport parallèle de phases est défini par l'équation : 
$$\frac{ds}{dt}-iA_\mu(\gamma(t))\gamma^{'\mu}(t)s(t)=0$$
dont une solution est : 
$$s(t)=\exp(i\int_{\gamma_t}A_\mu dx^\mu)$$
Les transformations de jauge sont définies par : 
$$g:x\in M\mapsto g(x)\in U(1)$$
et on peut écrire, au moins localement si l'espace-temps n'est pas simplement connexe : 
$$g=e^{i\beta}$$ où les fonctions $\beta$ sont lisses. Lors de tels changements de jauge, le potentiel de jauge $A$ varie selon : 
$$A'=A-d(\log g)$$
Ces idées sont discutées dans le papier historique de 1929 \cite{weyl29}.
\begin{rmq}
	Avec une telle approche, c'est-à-dire en présentant l'électromagnétisme comme une théorie de jauge, elle devient une conséquence à la fois de la physique quantique et de la relativité restreinte. Autrement dit, c'est presque déjà une théorie quantique des champs ! Même si, lorsque Maxwell a dérivé ses équations et dans les années qui ont suivi, la structure intrinsèque de la théorie n'était pas bien comprise (en même temps, ni la relativité ni la physique quantique n'étaient connues), \textbf{ces égalités contenaient déjà intrinsèquement l'essentiel de la relativité et de la physique quantique}. La quantification de la théorie classique des champs est donc possible, et conduit à l'électrodynamique quantique. 
\end{rmq}

Dirac a utilisé ce cadre théorique pour son étude des monopoles magnétiques présentée dans son célèbre article de 1931 \cite{dirac}.\\\\

Précisions que Weyl avait eu l'intuition de pouvoir unifier l'électromagnétisme avec la relativité générale grâce au paradigme du fibré des échelles, mais cela devient impossible avec le fibré des phases à cause de l'introduction de la structure quantique pour décrire l'électromagnétisme. Si on veut pouvoir définir les champs de particules dans le cas où l'espace-temps $M$ est courbé par la gravitation, il faut introduire des champs spinoriels sur des variétés courbes. Cela impose des contraintes de nature topologique sur $M$ (qui doit être un "spin manifold"). D'autres contraintes apparaissent, et finalement, l'espace-temps doit être couplé à une \textbf{très petite variété compacte} (c'est la compactification) de dimension 6 qui est la sous-variété réelle d'une variété de Calabi-Yau complexe de dimension 3 (complexe). D'après Varadarajan, au Congrès International de Mathématique de 1986, Witten a commencé sa conférence en remarquant que la tentative de Weyl d'unifier gravitation et électromagnétisme n'avait pas fonctionné parce que les formes de matière inclues dans la théorie n'étaient pas assez diversifiées \cite{witten}. Weyl avait, de manière intéressante, remarqué la même chose puisque dans la préface de la première édition américaine de \textit{Space, Time and Matter} on peut trouver les mots suivants : \textit{Since then, a unitary fiels theory, so it seems to me, should encompass at least three fields : electromagnetic, graviational and electronic. Ultimately the wave fiels of other elementary particles will have to be included too - unless quantum physics succeds in interpreting them all as different quantum states of one particle}...

\subsection{Les théories de Yang-Mills}

Dans les années 50, Yang et Mills ont essayé de comprendre la conservation du spin isotopique dans les interactions fortes et de \textbf{l'exprimer d'une manière analogue à la conservation de la charge en électromagnétisme}. L'aboutissement de leur travaux a été leur article de 1954 \cite{yangmi} qui est une avancée décisive dans le développement du concept de jauge. Partant de l'observation que le proton et le neutron, si on néglige les effets électromagnétiques, sont indistinguables, ils reprennent l'idée de 1932 de Heisenberg qui propose que le proton et le neutron sont en réalité deux états d'une même particule, le nucléon. Mathématiquement, on introduit un \textbf{espace interne} pour le nucléon (l'espace de spin isotopique), qui est un espace de Hilbert de dimension 2 avec la représentation standard de $SU(2)$. Les idées qui suivent sont le généralisation de celles de Weyl. Il faut prendre une référence dans cet espace de spin isotopique, qui n'a aucune raison d'être identique en tout point de l'espace-temps. On définit dans cet espace une connexion, donc un transport canonique de spin isotopique. Yang et Mills affirment alors que les lois de la physique doivent être invariantes par rotation dans l'espace de spin isotopique, ce qui conduit aux équation dites de Yang-Mills. La différence essentielle avec les travaux de Weyl est que cette théorie de jauge est \textbf{non-abélienne}. Le lien entre ces idées "physiques" est la géométrie différentielle sous-jacente n'a été fait que dans les années 1970.\\\\
Entre-temps, en 1959, Aharanov et Bohm ont écrit un article fameux, où ils discutent de savoir si \textbf{oui ou non, les potentiels ont un sens physique} \cite{ahabohm}. Ils suggèrent l'affirmative, même en électromagnétisme. Plus précisément, si l'espace-temps n'est pas simplement connexe, ils montrent que des effets électromagnétiques doivent se faire sentir même si le champ électromagnétique est nul, et proposent pour vérifier cela leur célèbre expérience éponyme. Le temps y est découplé, et la variété d'espace est $M=\frac{\R^3}{L}$ où L est une droite de l'espace décrivant un solénoïde infiniment fin. Le champ magnétique est nul à l'extérieur du solénoïde, mais \textbf{la phase des particules qui voyagent à proximité du solénoïde dépend elle du potentiel vecteur}, qui lui, ne l'est pas. Ils prédisent alors que la figure d'interférence observées derrière le solénoïde doit dépendre de l'intensité qui parcourt ce dernier. L'expérience est faite, et est concluante. Ce n'est donc pas par hasard que les équations de Yang-Mills font intervenir les potentiels de jauge directement, et non seulement le champ de force.

\section{Groupes de Lie et actions de groupe}
\subsection{Généralités sur les groupes de Lie}
\begin{defi}
	Soit G une variété de dimension n et un groupe tel que les opérations de multiplication $G\times G\rightarrow G$ donnée par $(g,g')\mapsto gg'$ et d'inversion $G\rightarrow G$ donnée par $g\mapsto g^{-1}$ soient lisses. Alors G est appelé groupe de Lie de dimension n
\end{defi}
Les exemples les plus courants de groupes de Lie les groupes de matrices complexes comme le groupe linéaire complexe $GL_n(\C)$, le groupe unitaire $U_n(\C)$, le groupe unitaire unimodulaire $SU_n(\C)$ ou réels comme $GL_n(\R)$, $O_n(\R)$, $SO_n(\R)$...

\begin{defi}
	Soit $L_g:G\rightarrow G$ l'application définie par $L_g(g')=gg'$ (multiplication à gauche) qui est évidemment un difféomorphisme. Soit e l'élément unité de G, et soit $A\in T_eG$. On définit le champ de vecteur invariant à gauche $\overline{A}$ engendré par A, par $\overline{A}_g=(L_g)_*(A)$.
\end{defi}

\begin{figure}[!h]
	\centering
	\includegraphics{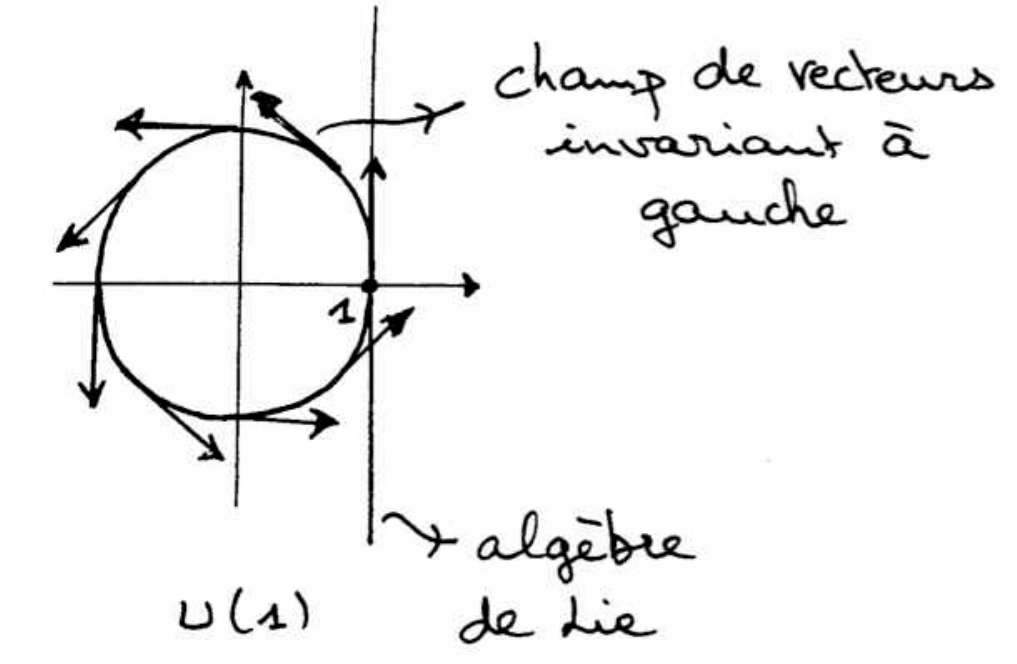}
	\caption{Le groupe de Lie $U(1)$}
\end{figure}

\paragraph{L'algèbre de Lie associée à un groupe de Lie}

Grâce à la structure de groupe qui permet de définir de manière canonique des champs de vecteurs à partir d'un vecteur tangent fixé, le crochet de Lie se transpose naturellement en une application crochet sur $T_eG = \g$ : soient $A,B\in\g$, on définit $[A,B]\in\g$ par $[A,B]=[\overline{A},\overline{B}]_e$. On a les propriétés immédiates : $$[A,B]=-[B,A]$$ et $$[A,[B,C]]+[B,[C,A]]+[C,[A,B]]=0$$
Le crochet $[.,.]:\g\times \g\rightarrow \g$ munit l'espace tangent en l'identité de $G$ d'une structure d'algèbre ; l'algèbre $\g$ est appelée algèbre de Lie du groupe $G$.

On peut prouver que $\overline{A}$ est un champ de vecteurs complet sur $M$. 

Soit $(\phi_t)$ le flot du champ de vecteurs $\overline{A}\in\g$.
Prouvons que $\phi_{s+t}(e)=\phi_{t}(e)\phi_{s}(e)$ :

Soit $s\in \R$ fixé et posons $\gamma_1(t)=\phi_{s+t}(e)$ ainsi que $\gamma_2(t)=\phi_{s}(e)\phi_{t}(e)$. On a $\gamma_1'(t)=\overline{A}_{\phi_{s+t}(e)}$ et $\gamma_2'(t)=(L_{\phi_{s}(e)})_*(\overline{A}_{\phi_{t}(e)})=\overline{A}_{\phi_{s}(e)\phi_{t}(e)}$. Donc $\gamma_1$ et $\gamma_2$ sont des courbes intégrales du même champ de vecteurs $\overline{A}$ qui coïncident en $0$, et donc sur tout leur ensemble de définition. Donc $\phi_.(e):\R\rightarrow G$ est un morphisme de groupes. \\\\Soit une courbe $\sigma:\R\rightarrow G$ (c'est aussi un morphisme de groupes), alors $\psi_t:G\rightarrow G$, défini par $\psi_t(g)=g\sigma(t)$, est un groupe de difféomorphismes de $G$ à un paramètre tel que : $$\overline{B}_g=\frac{d}{dt}\psi_t(g)_{|t=0}$$ définit le champ de vecteurs invariant à gauche déterminé par $B=\overline{B}_e$. Il y a donc une correspondance biunivoque entre $A\in\g$ et $\gamma$.

\begin{defi}
	On définit l'application exponentielle : $$exp:\g\rightarrow G$$ par $$exp(A)=\gamma(1)$$ Notons que $\gamma(t)=exp(tA)$ et $\phi_t(g)=g\gamma(t)=gexp(tA)$.
\end{defi}

\begin{figure}[!h]
	\centering
	\includegraphics{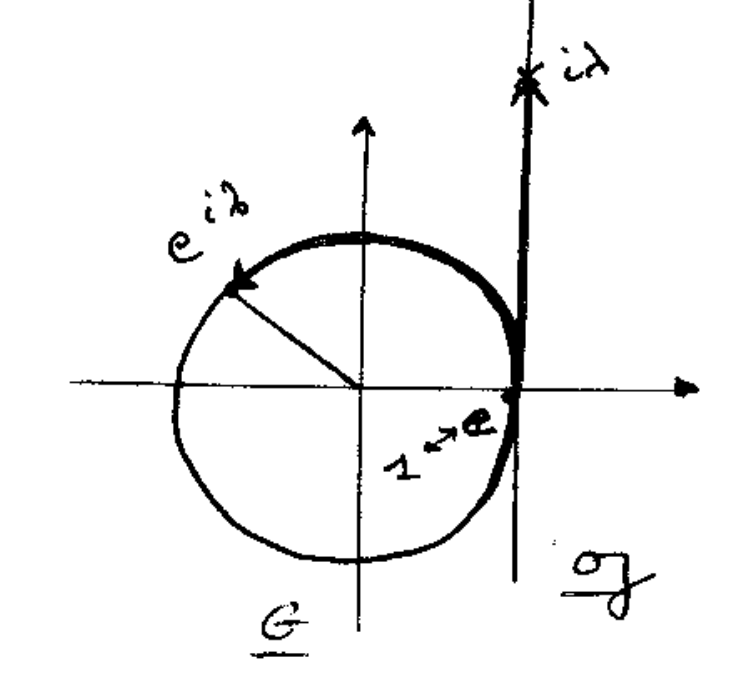}
	\caption{Rapports entre groupe de Lie et algèbre de Lie}
\end{figure}

Définissons maintenant les sous-groupe de Lie d'un groupe de Lie.

\begin{defi}
	Soit $G$ un groupe de Lie. Un sous-groupe de Lie $H$ de $G$ est un sous-groupe de $G$ qui est aussi une sous-variété de $G$ (une sous-variété $N$ de dimension $k$ d'une variété $M$ de dimension $d$ est une variété de dimension $k$ munie de la topologie induite par la topologie de $M$, c'est-à-dire que $V$ est un ouvert de $N$ si et seulement si $V$ s'écrit $V'\cap N$ où $V'$ est un ouvert de $M$ ; on demande aussi que l'application tangente $i_{*x}$ de l'inclusion soit injective pour tout $x\in N$).
\end{defi}

Les morphismes $\gamma:\R\rightarrow H$ peuvent être vus comme des morphismes $\gamma:\R\rightarrow G$, donc $exp:\mathfrak{h}\rightarrow H$ est la restriction de $exp:\g\rightarrow G$.

\paragraph{La forme de Maurer-Cartan}
Cet objet permet d'associer, de manière canonique, à un vecteur tangent à $G$ en $g\in G$, une "direction", c'est-à-dire un vecteur de l'algèbre de Lie. Il s'agit d'une 1-forme $\omega$ définie sur $TG$ et à valeurs dans l'algèbre de Lie $\g$ telle que $\forall v\in T_gG$ : $$\omega(v)=(L_{g^{-1}})_{*}v $$
La forme de Maurer-Cartan appartient donc à $\bigwedge^{1}(G)\otimes\g$.\\\\
Cette idée de pouvoir définir \textit{de manière canonique} une direction lorsqu'on se déplace dans le groupe est rendue possible par la structure de groupe, et fondamentale ! C'est une propriété fondamentale qui justifie, entre autres, la définition des fibrés principaux.

\paragraph{Distribution d'espaces tangents}

La proposition suivante révèle les contraintes fortes qu'impose la structure de groupe ; Posons tout d'abord quelques définitions pour pouvoir l'énoncer.
\begin{defi}
	Soit $M$ une variété . Une distribution $H$ de dimension $k$ de sous-espaces tangents à $M$ est la donnée, pour tout $x\in M$ de sous-espaces tangents $H(x)$ variant de manière lisse au sens suivant : localement, $H$ doit être engendré par $k$ champs de vecteurs linéairement indépendants. On dit que la distribution est \textbf{intégrable} si elle peut être redressée dans les cartes (dire que H est intégrable en $x_0$ revient à dire : il existe $X_1, ..., X_k$ k champs de vecteurs engendrant H localement en $x_0$ et une carte $(U,\phi)$ en $x_0$ telle que $\phi_*(X_i)=e_i$ où $e_i$ est le i-ème vecteur canonique de $\R^k$). Autrement dit, il faut que ce soit la distribution des sous-espaces tangents à une sous-variété de $M$. On dit que la distribution est \textbf{involutive} si elle est stable par crochet, c'est-à-dire, si $X$ et $Y$ sont deux champs de vecteurs à valeur dans $H$, alors $[X,Y]$ est encore un champ de vecteur à valeurs dans $H$. 
\end{defi}

\begin{figure}[!h]
	\centering
	\caption{Une distribution est intégrable si elle peut être redressée dans des cartes}
	\includegraphics{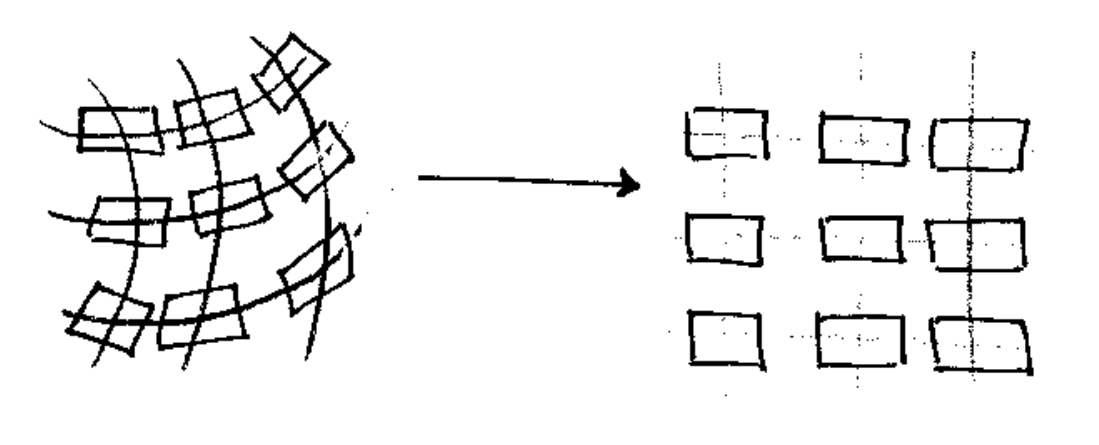}
\end{figure}

Le théorème suivant permet de faire le lien entre ces deux notions ; nous l'énonçons sans preuve.

\begin{theorem}[Frobenius]
	Une distribution de sous-espaces est intégrable si et seulement si elle est involutive.
\end{theorem}

On a le théorème suivant concernant les groupes de Lie : 

\begin{theorem}
	Soient $G$ et $G'$ deux groupes de Lie, et $F:G\rightarrow G'$ un morphisme lisse de groupes de Lie ($F$ doit être lisse dans les cartes en plus d'être un morphisme de groupes). Alors $F_{*e}:\g\rightarrow\g'$ est une fonction linéaire telle que $F_{*e}([A,B])=[F_{*e}A,F_{*e}B]$. Autrement dit, $F_{*e}$ est un morphisme d'algèbres de Lie. 
\end{theorem}

\begin{proof}[Preuve]
	On a : $$ F\circ L_g(g')=F(gg')=F(g)F(g')=(L_{F(g)}\circ F)(g')$$ Donc $$F_{*g}(\overline{A}_g)=F_{*g}(L_{g*}A)=L_{F(g)*e'}(F_{*e}A)=(\overline{F_{*e}A)}_{F(g)}$$ D'après le théorème 3.6, on a le résultat voulu car on a prouvé que $F_*(\overline{A})=\overline{(F_{*e}A)}$.
\end{proof}

Ce théorème implique, en l'appliquant à l'inclusion naturelle de $H$ dans $G$, que le crochet sur $\mathfrak{h}$ est juste la restriction du crochet à $\g$.\\
\fbox{Les espaces tangents à $H$ forment une distribution de sous-espaces qui est involutive, donc intégrable.}

\paragraph{Représentations adjointes}
\begin{defi}
	Soit $g\in G$. Soit $Ad_g:G\rightarrow G$ \textbf{l'isomorphisme} lisse dit \textbf{adjoint} donné par $Ad_g(g')=gg'g^{-1}$. Le théorème précédent implique l'existence d'un isomorphisme induit sur $\g$ noté $\mathfrak{Ad}_g:\g\rightarrow\g$ c'est-à-dire $\mathfrak{Ad}_g=Ad_{g*e}$. On a donc un morphisme $\mathfrak{Ad}:G\rightarrow GL(\g)$ et d'après le théorème précédent, un autre morphisme induit $\mathfrak{ad}:\g\rightarrow\mathfrak{Gl(\g)}$.
\end{defi}

\begin{theorem}
	Soient $A,B\in \g$. On a : $$ \mathfrak{ad}(A)(B)=\frac{\partial^2}{\partial s\partial t}_{|s,t=0}(exp(tA)exp(sB)exp(-tA)=[A,B]$$
\end{theorem}

\begin{proof}[Preuve]
	Soit $\phi_t$ le flot de $\overline{A}$. On a vu que $\phi_t(g)=gexp(tA)$. On a alors : 
	$$[A,B]=[\overline{A},\overline{B}]_e=\frac{d}{dt}_{|t=0}(\phi_{-t})_*(\overline{B}_{\phi_t(e)})=\frac{d}{dt}_{|t=0}(\phi_{-t})_*(\frac{d}{ds}_{|s=0}\phi_t(e)exp(sB)) $$ 
	$$ [A,B]=\frac{\partial^2}{\partial s \partial t}(exp(tA)exp(sB)exp(-tA))=\frac{d}{dt}\mathfrak{Ad}(exp(tA))(B)=\mathfrak{ad}(A)(B)$$
\end{proof}

Pour les représentations matricielles des groupes, ce théorème montre que l'objet commutateur de deux matrices, et la notion définie dans cette partie, coïncident.

Soit $V$ un espace vectoriel de dimension finie $m$. En considérant $GL(V)$ comme un groupe de matrices, il est simple de montrer que c'est bien un groupe de Lie. $\mathfrak{Gl}(V)$ peut être identifié à l'ensemble des endomorphismes de $V$, la correspondance étant donnée par : $A\rightleftharpoons\frac{d}{dt}(I+tA)_{|t=0}$. 
Pour $A\in \mathfrak{Gl}(V)$ posons : $$ Exp(A)=I+A+\frac{A^2}{2!}+...$$
La somme converge puisqu'on est dans un espace de Banach, et $$Exp((t+s)A)=Exp(tA)Exp(sA)$$ En particulier, $Exp(tA)Exp(-tA)=I$ donc $Exp(tA)\in GL(V)$. $Exp$ est donc l'application exponentielle pour $GL(V)$.

On a alors $$[A,B]=\frac{\partial^2}{\partial s\partial t}(Exp(tA)Exp(sB)Exp(-tA))_{|s,t=0}=AB-BA$$

\paragraph{Constantes de structure}
\begin{defi}
	Soit $(e_i)_{i\in[|1,n|]}$ un base de l'algèbre de Lie $\g$ de $G$. Les constantes de structure $c_{ij}^k\in\R$ sont définies par $[\frac{e_a}{i},\frac{e_b}{i}]=c_{ab}^c\frac{e_c}{i}$ (en convention d'Einstein). \\L'antisymétrie du crochet entraine $c_{ij}^k=-c_{ji}^k$ et l'identité de Jacobi : $$\forall h,i,j,k\in [|1,n|] : c_{im}^hc_{jk}^m + c_{km}^hc_{ij}^m + c_{jm}^hc_{ki}^m=0$$
\end{defi}
Par exemple l'algèbre de Lie de $SU(2)$, dans sa représentation standard, est engendrée par les matrices : $$
\frac{i}{2}\sigma_1=\frac{1}{2}\begin{pmatrix}
0 & i \\
i & 0 \\
\end{pmatrix}
\frac{i}{2}\sigma_2=\frac{1}{2}\begin{pmatrix}
0 & 1 \\
-1 & 0 \\
\end{pmatrix}
\frac{i}{2}\sigma_3=\frac{1}{2}\begin{pmatrix}
i & 0 \\
0 & -i \\
\end{pmatrix}
$$
où $\sigma_1, \sigma_2, \sigma_3$ sont les matrices de Pauli, qui vérifient $[\sigma_i,\sigma_j]=2i\epsilon^{ijk}\sigma_k$ où $\epsilon^{ijk}$ est le tenseur totalement antisymétrique avec $\epsilon^{123}=1$. Les constantes de structure de $\mathfrak{su}(2)$ sont donc $f^{abc}=2i\epsilon^{abc}$.

\paragraph{L'algèbre de Lie de $SU(n)$}
Le calcul de l'algèbre de Lie de $SU(n)$ est relativement simple, et surtout, ces groupes sont fréquemment utilisés en physique des particules ($SU(2)$ est le groupe de jauge des interactions faibles, $SU(3)$ celui des interactions fortes).

Notons $GL_n(\C)$ l'ensemble des matrices $n\times n$ à coefficients complexes. Pour $A\in GL_n(\C)$, on note $A^\dagger$ la matrice conjuguée de la transposée de $A$. On a : $$SU(n)=\{A\in GL_n(\C)|AA^\dagger=1 ; \det A=1\}$$ 

Soit $t\mapsto A(t)$ une courbe dans $U(n)$ avec $A(0)=I$. Alors $0=\partial_{t|t=0}(I)=\partial_{t|t=0}(A(t)A(t)^\dagger)$ donc $$A'(0)+A'(0)^\dagger=0$$

Par ailleurs, si $B\in\mathfrak{gl}_n(\C)$ vérifie $B+B^\dagger=0$, alors $Exp(B)(ExpB)^\dagger=Exp(B)Exp(B^\dagger)=I$ donc $Exp(B)\in U(n)$ et en prenant la dérivée en $0$ de $t\mapsto Exp(tB)$, on obtient $B\in \mathfrak{u}(n)$.\\

On a donc montré que \fbox{$\mathfrak{u}(n)=\{B\in\mathfrak{gl}(n)|B+B^\dagger=0\}$.}\\\\

L'algèbre de Lie $\mathfrak{su}(n)$ de $SU(n)$ est la sous-algèbre de $\mathfrak{u}(n)$ des matrices de trace nulle, puisque $\det(ExpB)=e^{\tr B}$. En effet, en posant $f(t)=\det(Exp(tB))$ on aboutit à l'équation différentielle $f'(t)=\tr(B)f(t)$ qui donne le résultat voulu.

\subsection{Actions de groupe}

\begin{defi}
	Soit $G$ un groupe et $X$ un ensemble. Une action de $G$ sur $X$ est une application $a:G\times X\rightarrow X$ telle que :
	\begin{enumerate}
		\item $\forall g,g'\in G, x\in X$ $a(g',a(g,x))=a(g'g,x)$
		\item $\forall x\in X$ $a(e,x)=x$
	\end{enumerate}
\end{defi}

On note souvent $a(g,x)=g\cdot x$ pour une action à gauche. On s'intéresse particulièrement aux cas où $G$ est un groupe de Lie, et où l'ensemble $X$ est muni d'une structure de variété.\\
Pour $x\in X$, l'ensemble $\{g\cdot x|g\in G\}$ est l'orbite de $x$ sous $G$.\\
Un action est dite \textbf{libre} si, $\forall x \in G, \forall g \in G, g\cdot x=x \Leftrightarrow g=e$.\\
Une action est dite \textbf{propre} si $\forall K \subset G$ compact, $\{g\in G|gK\cap K\neq\emptyset\}$ est compact. Notons que si $G$ est compact (ce qui est le cas pour les théories des interactions), l'action est nécessairement propre.\\\\
On a le théorème suivant que nous énonçons sans preuve. 
\begin{theorem}
	Soit $a$ une action libre et propre d'un groupe de Lie $G^{d_G}$ sur une variété $P^{d_P}$. Alors il existe une variété $M$ de dimension $d_P-d_G$ et une application lisse $\pi:P\rightarrow M$ dont la différentielle est une surjection en tout point, vérifiant : $\forall z_0 \in M$, $z_0$ il existe un voisinage U de $z_0$ et un difféomorphisme $\psi=(\pi,\theta):\pi^{-1}(U)\rightarrow U\times G$ vérifiant $\psi\circ a(g,x)=(\pi(x),g\theta(x))$.
\end{theorem}

Le triplet $(P,M,\pi)$ possède une structure de \textbf{fibré principal}, que nous allons maintenant étudier.

\section{Espace fibrés}
La "philosophie" des espaces fibrés est très générale, et permet de réinterpréter de nombreux objets d'une manière originale. Le vocabulaire est par contre assez spécifique.

Soit $\pi$ une application quelconque de $P$ (espace de départ) dans $M$ (image de $P$ par $\pi$). On dit que $P$ est \textbf{l'espace total} et que $M$ est \textbf{la base}, $\pi$ est \textbf{la projection} (et n'a a priori aucune raison d'être surjective pour l'instant), enfin, on appelle $\pi^{-1}(\{x\})$ \textbf{la fibre au dessus de $x$}.

\subsection{Généralités}
Toute application peut se réinterpréter en termes de fibration. On peut choisir un inverse de $\pi$ à droite $\sigma:M\rightarrow P$ en demandant, pour $x\in M$, que $\sigma(x)\in\pi^{-1}(x)$. $\sigma$ est alors appelée \textbf{section locale} de la projection $\pi$.\\\\

\begin{figure}[!h]
	\centering
	\caption{Une application réinterprétée en termes de fibration}
	\includegraphics{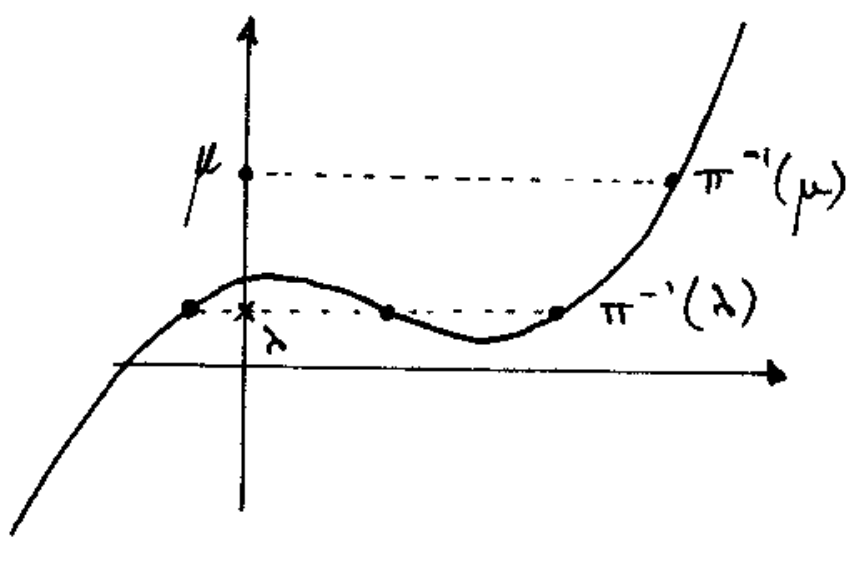}
\end{figure}

Dans la catégorie des variétés différentiables, toutes les applications ne vont pas être des fibrations au sens de la géométrie différentielle. Il faut se restreindre à des structures lisses pour la géométrie différentielle (c'est-à-dire qu'on veut des structures $\mathcal{C}^\infty$).

\begin{defi}
	On dit qu'une application lisse $\pi:P\rightarrow M$ où $P$ et $M$ sont des variétés, est une fibration, si toutes les fibres sont difféomorphes. On dit alors que la fibre type est $F$ où $F$ est difféomorphe à toutes les fibres.
\end{defi}

\begin{defi}
	On dit qu'une fibration $\pi:P\rightarrow M$ de fibre type $F$ est localement triviale si pour tout $x\in M$, il existe un voisinage de $U$ de $x$ dans $M$ tel que $\pi^{-1}(U)$ est difféomorphe à $U\times F$.
\end{defi}

Cette propriété de locale trivialité permet de se représenter localement l'espace fibré comme un 'cylindre' au dessus de la base :

\begin{figure}[!h]
	\centering
	\caption{Fibration localement triviale}
	\includegraphics{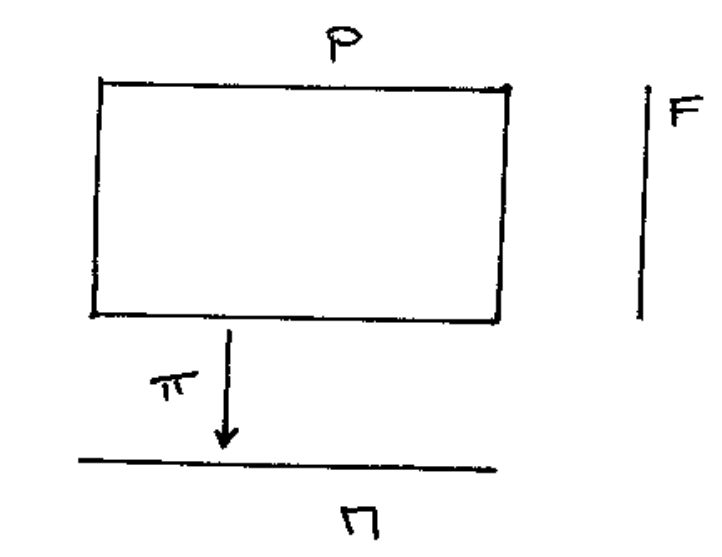}
\end{figure}

\begin{rmq}
	Le plus souvent, la fibration est qualifiée par un terme précisant sa structure ; par exemple on parlera de fibré vectoriel si les fibres sont des espaces vectoriels, dont l'un des cas particulier est le fibré en droite, ou de fibré principal, auquel nous allons maintenant nous intéresser en particulier.
\end{rmq}

\subsection{Espace fibrés principaux}
\paragraph{Généralités}
\begin{defi}
	Une fibration $(P,M,\pi)$ a une structure de fibré principal si les trois conditions suivantes sont vérifiées : 
	\begin{enumerate}
		\item $(P,M,\pi)$ est une fibration localement triviale
		\item Un groupe de Lie $G$ agit à droite sur $P$, de manière lisse, et libre et transitive sur chaque fibre.
		\item Toutes les fibres sont difféomorphes à $G$.
	\end{enumerate}
\end{defi}

Une géométrie adaptée à la description de l'électromagnétisme, que nous allons étudier, repose sur un fibré principal de groupe $U(1)$, qu'on peut localement représenter par :

\begin{figure}[!h]
	\centering
	\includegraphics{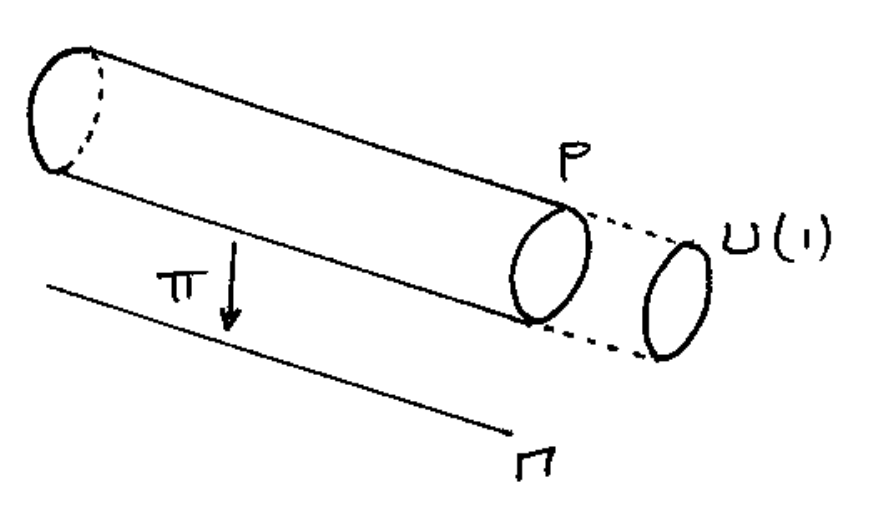}
	\caption{Trivialisation locale du fibré $U(1)$ de l'électromagnétisme sur un espace-temps de dimension $1$}
\end{figure}

\begin{rmq}
	Chacune des fibres est bien difféomorphe à $G$, mais pas de manière canonique ! Expliciter un difféomorphisme revient à faire un choix d'une \textbf{trivialisation locale} : Soit $x\in M$. Il existe un voisinage ouvert $U$ de $x$ dans $M$ et une application $T_U:\pi^{-1}(U)\rightarrow U\times G$ telle que $T_U(x)=(\pi(x),s_U(x))$ où $s_U:\pi^{-1}(U)\rightarrow G$ a la propriété $\forall p\in \pi^{-1}(U) \forall g\in G$ $s_U(pg)=s_U(p)g$. $T_U$ est appelé trivialisation locale, ou dans le langage historique de la physique, \textbf{choix de jauge}. Autrement dit, lorsqu'on choisit une trivialisation locale, on marque un point de la fibre qu'on peut ensuite identifier à l'élément unité, ce qui donne par le même biais une section locale.
\end{rmq}

On retrouve la structure obtenu à la fin de la sous-partie précédente ! Un fibré principal $\pi:P\rightarrow M$ de groupe $G$ est donc \textbf{localement} le produit cartésien de la base par le groupe. On définit les fonctions de transitions qui contiennent l'information nécessaire pour obtenir l'espace total en recollant les différents morceaux trivialisés.

\begin{defi}
	Soient $T_U:\pi^{-1}(U)\rightarrow U\times G$ et $T_V:\pi^{-1}(V)\rightarrow V\times G$ deux trivialisations locales d'un fibré principal avec groupe $G$. La fonction de transition de $T_U$ à $T_V$ est la carte $g_{UV}:U\cap V\rightarrow G$ définie, pour $x\in\pi^{-1}\in U\cap V$, par $g_{UV}(x)=s_U(p)s_v(p)^{-1}$.
\end{defi}
\begin{proprietes}
	$g_{UV}(x)$ est indépendant du choix de $p\in\pi^{-1}(x)$. De plus : 
	\begin{enumerate}
		\item $g_{UU}(y)=e$ pour tout $y\in U$
		\item $g_{VU}(y)=g_{UV}^{-1}(y)$ pour tout $y\in U\cap V$
		\item $g_{UV}(y)g_{VW}(y)g_{WU}(y)=e$ pour tout $y\in U\cap V \cap W$
	\end{enumerate}
\end{proprietes}

\begin{proof}[Preuve]
	 $s_U(pg)s_V(pg)^{-1}=s_U(p)g(s_V(p)g)^{-1}=s_U(p)gg^{-1}s_V(p)^{-1}=s_U(p)s_V(p)^{-1}$
\end{proof}
$P$ peut en fait être défini comme l'espace obtenu par l'union disjointe $(U\times G)\cup (V\times G)\cup ...$ en identifiant le point $(x,g)\in U\times G$ avec $(x,g')\in V\times G$ si $g=g_{UV}(x)g'$. Les propriétés précédentes montrent qu'il s'agit d'une relation d'équivalence.

\begin{defi}
	Une section locale d'un fibré principal $\pi:P\rightarrow M$ de groupe $G$ comme est une application lisse de $U$ dans $P$ où $U$ est un ouvert de $M$, telle que $\pi\circ\sigma=id_U$
\end{defi}

\begin{figure}[!h]
	\centering
	\includegraphics{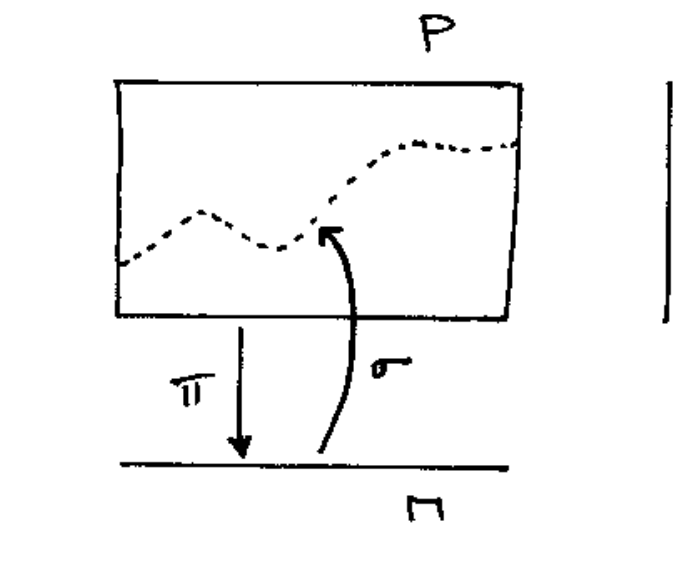}
	\caption{Section locale d'un fibré}
\end{figure}

Un fibré principal admet de telles sections. De plus, comme évoqué dans la remarque précédente :

\begin{theorem}
	Il y a un isomorphisme canonique entre les trivialisations locales et les sections locales.
\end{theorem}

\begin{proof}[Preuve]
	Soit $\sigma:U\rightarrow P$ une section locale. On définit alors $T_U:\pi^{-1}(U)\rightarrow U\times G$ par $$T_U(\sigma(x)g)=(x,g)$$ Réciproquement, étant donnée une trivialisation locale $T_U:\pi^{-1}(U)\rightarrow U\times G$ on définit la section locale $\sigma:U\rightarrow P$ par $$\sigma(x)=T_U^{-1}(x,e)$$
\end{proof}

\begin{rmq}
	Si $T_U$ est une trivialisation locale avec $U=M$, $T_M$ est appelée trivialisation globale ; un fibré principal est dit trivial si une telle application existe. 
\end{rmq}

\paragraph{Le fibré des repères linéaires}
L'exemple suivant de fibré principal a en fait motivé la définition même de ces objets. Il donne une intuition au sujet de ces structures, et fournit un vocabulaire utile.\\\\
Soit $M$ une variété différentiable de dimension $n$. En chaque point $x\in M$ on a défini un espace vectoriel tangent, de dimension $n$. On note $G_x$ l'ensemble des bases de cet espace tangent. Choisir un point de $G_x$ revient à choisir n vecteurs indépendant de $T_xM$. 
La projection $\pi$ associe à une base de l'espace tangent en $x$ la point $x$ lui-même.\\ On définit l'action libre et propre, transitive sur chaque fibre, de $GL(n)$ sur l'espace total $P=\bigcup_{x} G_x$, par : $$\forall p\in P, g\in G,  p\cdot g=gp$$ où $gp$ est la multiplication matricielle 'formelle' de $g$ et de la base $p$, au sens suivant : supposons qu'on se soit fixé une base de $T_xP$. Alors on peut représenter matriciellement $g\in \pi^{-1}(x)$ dans cette base ; le repère $p\cdot g$ est alors la base de $T_xP$ dont la matrice \underline{dans la même base} s'écrit $gp$. Cette définition ne dépend pas de la base fixée, donc cela a un sens de considérer le produit $p\cdot g$ sans fixer une origine dans la fibre. On vient aussi de voir que toutes les fibres sont difféomorphes à $GL(n)$.\\\\
Montrons que l'espace total est naturellement muni d'une structure de variété:
Soit $d$ la dimension de l'espace de base $M$, et $n$ celle du groupe de Lie $G$. Soit $(U_i,\phi_i)$ un atlas de $M$ et $(V_j,\psi_j)$ un atlas de G. Soit $x\in P$. P étant localement trivialisable, il existe un voisinage ouvert $W$ de $x$ tel que $W$ est difféomorphe à $U\times G$, donc quitte à restreindre $W$ on peut le prendre difféomorphe à $U\times V$ où $(U,\phi)$ et $(V,\psi)$ sont des cartes respectives de $M$ et $G$ en $\pi(x)$ et en $e$. On définit alors $\zeta:U\times V\rightarrow \R^d\times\R^n$ par $\zeta(y)=(\phi(\pi(y)),\psi(y))$.\\\\
Une section du fibré principal correspond à une application lisse de $M$ dans $P$ qui à tout point de $M$ associe un repère de $T_xM$.\\\\
Par analogie, dans la suite, les éléments de $P$ seront souvent appelés de repères, et l'action de $G$ sur $P$ peut être interprétée comme un changement de base.

\paragraph{Vecteurs horizontaux}

\paragraph{Réduction dans les fibrés principaux}
Soit $P(M,G)$ un fibré principal de groupe $G$. On se demande s'il est possible de passer à une autre fibré principal, de base $M$ et de groupe structural $H<G$, qui serait une sous-variété de $P$. Ce n'est pas forcément le cas, si on prend $H=\{1\}<G$, on cherche alors une section globale de $P$, qui n'existe que si $P$ est trivial globalement.\\ 
On a le théorème suivant, dont la démonstration peut par exemple être trouvée  dans \cite{coq}.
\begin{theorem}
	Le choix d'une réduction du fibré principal à un sous-fibré n'est en général pas unique, et est caractérisé par le choix d'une section globale dans un fibré en espaces homogènes associé à $P$, en l'occurrence le fibré associé $P\times_{G}G/H$.
\end{theorem}
Soit $P=P(M,G)$ le fibré des repères d'une variété $M$ qui est de groupe structural $GL(n,\R)$. \textbf{Choisissons} maintenant \textbf{une} réduction à un sous-fibré de groupe structural $SO(n)$, c'est-à-dire sélectionnons une classe de repères, dits orthonormés, tels que $SO(n)$ agit de manière libre et transitive sur cette classe. Par définition, une variété riemannienne est une variété différentiable de dimension n pour laquelle on a choisi une réduction du fibré $FM$ des repères linéaires à un sous-fibré de groupe structural $SO(n)$. On a construit le fibré des repères orthonormés. Le tenseur métrique s'identifie alors avec la section globale du fibré en espaces homogènes $GL(n,\R)/SO(n)$ qui définit la réduction. La dimension de cet espace est $n^2-n(n-1)/2=n(n+1)/2$, et ses éléments peuvent s'identifier à des tenseurs de rang $2$ totalement symétriques.\\\\
La décomposition polaire d'une matrice inversible est un difféomorphisme $$O_n(\R)\times S_n^{++}(\R)\rightarrow GL_n(\R)$$ qui montre que le quotient $GL_n(\R)/O_n(\R)$ est difféomorphe à $S_n^{++}(\R)$, lui-même difféomorphe à l'espace vectoriel des matrices symétriques réelles, qui est un espace vectoriel donc un espace topologique contractile (homotope à un point). La contractilité étant conservée par difféomorphisme, $GL_n(\R)/O_n(\R)$ est contractile, donc le fibré réduit à des fibres contractiles et admet des sections globales par un théorème de la théorie des fibrés. Toute variété "raisonnable" admet donc une métrique riemannienne.\\\\
Le choix d'une structure riemannienne sur une variété différentiable revient à choisir une "forme" pour la variété : selon la métrique, on obtiendra pour une variété difféomorphe à $\mathbb{S}^2$, un objet de la forme d'une sphère, de la forme d'un ballon de rugby ... et chaque réduction possible correspond à une et une seule métrique riemannienne possible pour cette variété différentiable.

\begin{figure}[!h]
	\centering
	\includegraphics{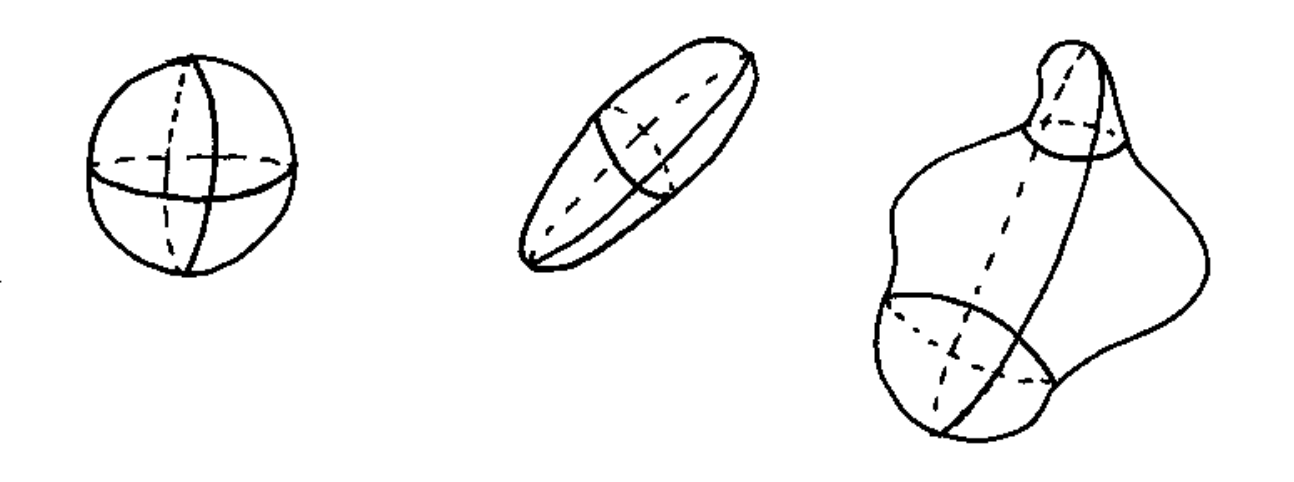}
	\caption{Choisir une réduction revient à choisir une "forme" (cela revient aussi à choisir une métrique)}
\end{figure}

\subsection{Espace fibrés associés}

\paragraph{Motivations} On a vu que on pouvait construire le fibré des repères d'une variété donnée. On peut également construire le fibré tangent, dont chaque fibre est difféomorphe à $\R^d$ où $d$ est la dimension de $M$ et munie d'une structure d'espace vectoriel. On peut, dans le même esprit, construire des fibrés tensoriels. Ces espaces fibrés sont associés au fibré des repères puisque l'expression de leurs éléments dépend du choix d'une base de l'espace tangent. Cette d'association à un fibré principal correspond à la notion de fibré associé.\\\\
Une autre considération qui motive la définition de fibré associé est la recherche d'une généralisation des espaces vectoriels dans lesquels vivent les composantes d'un tenseur : un tenseur est un objet \textbf{intrinsèque}, ce qui impose des conditions lors d'un changement de base. Dans les fibrés associés, on considère aussi des objets intrinsèques, dont les "coordonnées" dépendent de la manière dont on regarde cet objet, c'est-à-dire du choix de base.\\\\
Nous continuerons à utiliser le vocabulaire du fibré des repères, bien qu'en considérant des fibrés principaux a priori quelconques, ce vocabulaire ayant le mérite d'être particulièrement visuel.
\paragraph{Fibrés associés} Soit $P\rightarrow M$ un espace fibré principal, de groupe structural $G$, et soit $\rho$ une action (à gauche) de $G$ sur un ensemble $F$.\\\\
On obtient une relation d'équivalence sur $P\times F$ en disant que :
$(z,f)\in P\times F \Leftrightarrow (z',f')\in P\times F$ si et seulement si il existe un $g\in G$ tel que $z'=zg$ et $f'=\rho(g^{-1})f$. L'ensemble quotient $E=P\times_GF$ est alors appelé fibré associé à $P$ via l'action de $G$ sur $F$.

\begin{figure}[!h]
	\centering
	\includegraphics[scale=0.9]{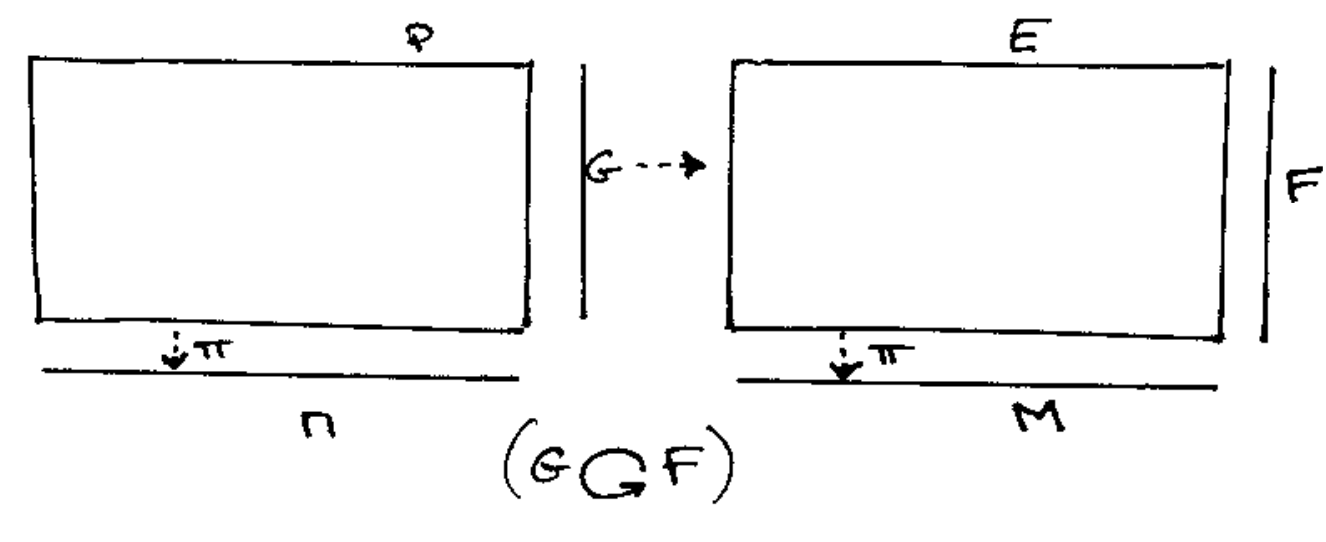}
	\caption{Passage du fibré principal à un fibré associé}
\end{figure}

On identifie donc les points $(z,f)$ et $(zg,\rho(g^{-1})f)$. \\\\f représente les composantes l'objet $(z,f)$ dans la base $z$. 
L'espace fibré associé à un fibré principal correspond donc aux objets 'vecteurs' dans le cas où $\rho$ est la représentation régulière de $GL(n,\R)$ comme l'ensemble des isomorphismes de $\R^{n}$, qu'on décompose sous la forme (repère, coordonnées dans ce repère). Comme évoqué dans cet exemple, le concept de fibré associé prend une signification forte si $g$ est une représentation linéaire de $G$.

\begin{figure}[!h]
	\centering
	\includegraphics{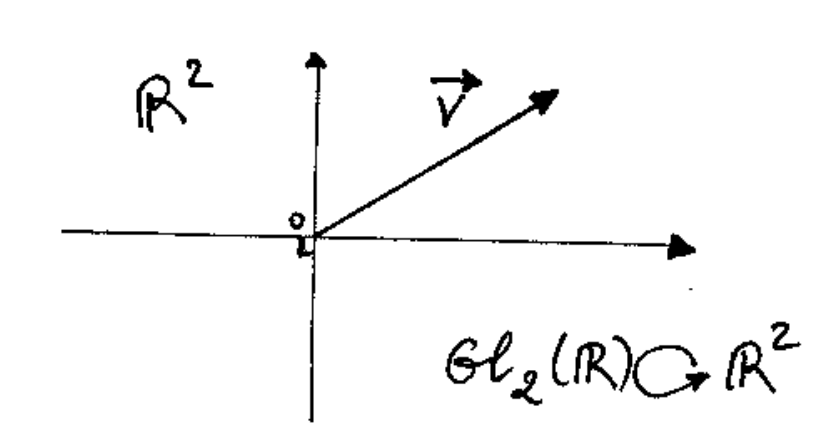}
	\caption{Un élément du fibré associé est "bien défini" pour la symétrie induite par la représentation du groupe structural du fibré principal}
\end{figure}

Les fibrés associés interviennent naturellement en relativité générale, mais aussi dans les théories des autres interactions entre particules !\\
Pour la force forte par exemple, on considère un fibré principal de groupe structural $SU(3)$ au dessus de l'espace-temps de Minkowski. On choisit la représentation standard de $SU(3)$ de dimension $3$ comme sous-groupe de $GL(3,\C)$, et on construit le fibré vectoriel associé à $P$ $P\times_{SU(3)}\C^{3}$. Le spineur qui décrit une particule soumise à l'interaction forte est une section de ce fibré associé.

\section{Connexions, courbure et équation de champ homogène}
\subsection{Connexions dans le fibré principal}
\paragraph{Définition intrinsèque sur le fibré principal}

Soit $P=P(M,G)$ un fibré principal de groupe structural $G$. Notons $d$ la dimension de $M$ et $n$ celle de $G$. Malgré le fait que l'unité (au sens du groupe) dans chaque fibre n'est pas désigné de manière canonique, et qu'il faut donc \textit{faire un choix} pour réaliser un difféomorphisme entre la fibre et le groupe structural, les directions dans le groupe sont définies par l'application exponentielle sans avoir besoin de marquer l'origine ! Par exemple, si $A\in\g$ et si $g\in G$, on pose $$A^{*}_{g}=\frac{d}{dt}_{|t=0}g\cdot exp(tA)$$ qui correspond à "la direction $A$ au point $g$". On peut faire la même chose sur le fibré principal. On parle de \textbf{champ de vecteurs fondamental}.
\begin{defi}
	Soit $A\in\g$. Le champ de vecteurs fondamental $A^{*}$ est défini sur $P$ par $$\forall p\in P, A^{*}_p=\frac{d}{dt}_{|t=0}p\cdot exp(tA)$$
\end{defi}
On peut donc transporter une direction tangente dans le groupe, grâce à la structure de groupe. Ces directions "de groupe" sont dites \textbf{directions verticales} par analogie avec l'exemple de fibré principal de groupe structural $U(1)$ où la variété est traditionnellement dessinée comme une courbe plus ou moins horizontale.\\\\
Plus rigoureusement, les champs de vecteurs fondamentaux engendrent une distribution $V$ de sous-espaces tangents en chaque point de $P$, de dimension $n$ (la dimension de $\g$). On parle de \textbf{distribution verticale} puisqu'elle ne contient que des vecteurs verticaux. Cette distribution est caractérisée par : $$V_p=\{V\in T_pP | \pi_{*}(V)=0\}$$ L'espace tangent en chaque point de $P$ étant de dimension $d+n$, on peut prendre une distribution H de sous-espaces tangents, dite \textbf{distribution horizontale}, de dimension $d$, telle que $T_pP=V_p\oplus H_p$. Ce choix n'est pas canonique. De la même manière qu'un choix de jauge permet de choisir une identité dans les fibres, l'objet connexion définit cette distribution horizontale de sous-espaces.

\begin{defi}
	Une connexion associe à chaque $p\in P$ un sous-espace $H_p\subset T_pP$ tel que $T_pP=V_p\oplus H_p$ où $V_p=\{V\in T_pP | \pi_{*}(V)=0\}$. On demande que la distribution obtenue soit lisse, et que $$(R_g)_*(H_p)=H_{pg}$$ afin que la définition de l'horizontalité sur le fibré principal soit invariante par action du groupe structural. On appelle $V_p$ le \textbf{sous-espace vertical} de $T_p P$ et $H_p$ le \textbf{sous-espace horizontal} de $T_p P$.
\end{defi}

Cette définition des connexions est visuelle, mais malheureusement peu pratique. Présentons d'abord les connexions d'une manière différente, plus opérationnelle, puis nous prouverons l'équivalence entre ces deux définitions.

\begin{defi}
	Une connexion est une $1$-forme $\omega$ sur $P$ (c'est-à-dire $\omega\in\O^1(P)\otimes\g$) à valeurs dans $\g$ (l'algèbre de Lie de $G$) telle que les deux propriétés suivantes sont vérifiées : 
	\begin{enumerate}
		\item Si $A^*$ est le champ fondamental associé à $A$, $$\omega(A^{*}_p)=A$$
		\item Si $g\in G$, $p\in P$, $X\in T_pP$ : $$\omega_{pg}(R_{g*}X)=\mathfrak{Ad}_{g^{-1}}\omega_p(X) \Leftrightarrow R_g^*\omega=\mathfrak{Ad}_{g^{-1}}\omega$$
	\end{enumerate}
\end{defi}
On appelle $\omega$ forme de connexion, et associe à un vecteur tangent sa partie verticale sous la forme d'un vecteur de l'algèbre de Lie $\g$. Dans le paradigme du fibré des repères, un vecteur tangent correspondant à un déplacement dans le fibré principal, la forme de connexion $\omega$ permet de définir la notion de "se déplacer en changeant de fibre sans tourner le repère".
\begin{theorem}
	Les deux définitions ci-dessus sont équivalentes.
\end{theorem}

\begin{proof}[Preuve]
	Soit une connexion définie par une $1$-forme de connexion $\omega$ (deuxième définition). On pose, pour tout $p\in P$, $H_p=Ker(\omega_p)$. $\omega_p$ étant une forme linéaire dont la restriction à $V_p$ est une bijection de $V_p$ dans $\g$, on a bien $H_p\oplus V_p=T_pP$. La distribution de sous-espaces obtenue est bien lisse, car $\omega$ l'est. Si pour $X\in T_pP$, on a $(R_g)_*X\in H_{pg}$, alors par bijectivité de $\mathfrak{Ad}_{g^{-1}}$ on a $X\in H_p$ et inversement, donc $R_{g*}(H_p)=H_{pg}$.\\
	Réciproquement, si $H_p$ est une distribution de sous-espaces tangents horizontaux au sens donné plus haut, définissant une connexion sur le fibré principal $P$, on sait que $X\in T_pP$ se décompose \textit{de manière unique} en parties horizontale et verticale $X=X^V+X^H$, avec $X^V=X^\alpha (A^{*}_{\alpha})_p$ où les $A^{*}_\alpha$ sont les champs de vecteurs fondamentaux que définit un base de l'algèbre de Lie $\g$ du groupe G. On définit $\omega$ par $\omega(v)=X^\alpha A_\alpha$. Il s'agit bien d'une $1$-forme à valeurs dans $\g$.
\end{proof}

\begin{rmq}
	On peut démontrer que tout fibré principal admet une connexion. La démonstration se trouve dans \cite{kobay}
\end{rmq}

\paragraph{Expression locale sur la base}
En physique, les fonctions rencontrées sont toutes des fonctions sur la base et non sur l'espace total : par exemple, en théorie quantique des champs, une fonction d'onde prend comme arguments les quatre variables d'espace-temps. On peut définir la connexion sur la base $M$ du fibré principal $P=P(M,G)$, mais \textit{a priori} seulement de manière locale, puisque seule la trivialité locale est imposée pour le fibré principal. De plus, il faut faire un \textit{choix de jauge}.
\begin{defi}
	Supposons une connexion définie sur l'espace total $P=P(M,G)$ par la $1$-forme de connexion $\omega$. Soit $x\in M$. Soit $T_U$ un choix de jauge, défini sur le voisinage $U$ de $x$ dans $M$, et $\sigma_U$ la section locale de P canoniquement associée à $T_U$. On définit le potentiel de jauge $A\in\O^1(U)\otimes\g$ par : $$\forall y\in U, \forall v\in T_yM, A(v)=\omega((\sigma_U)_*v)$$
	Soient $(X_\alpha)$ un base de $\g$ et $(\partial_\mu)$ une carte locale de $M$ en $x$, définie sur $U$ (quitte à le restreindre). On peut écrire : $$A=A^\alpha_\mu X_\alpha dx^\mu$$ qui est l'expression du potentiel de jauge $A$, localement, en $x$.
\end{defi}

\begin{rmq}
	Comme nous le verrons plus tard, le potentiel de jauge intervient de manière fondamentale en physique des particules : si $G=U(1)$, $A$ est le champ de photons, si $G=SU(2)\times U(1)$, $A$ est le champ électrofaible des bosons $W^+, W^-, Z$ et $\gamma$, enfin, si $G=SU(3)$, $A$ est alors le champ de gluons.
\end{rmq}

Cette définition est encore "très naturelle", puisqu'une fois qu'on a défini une transformation des vecteurs tangents à $M$ autour de $x$ en des vecteurs tangents à $P$, il suffit de prendre l'image du transformé par la $1$-forme de connexion sur $P$.
Notons enfin qu'à cause de ce choix de jauge, on perd un peu du caractère intrinsèque qu'avait la connexion sur l'espace total.

\subsection{Connexions dans les fibrés associés}
On considère toujours le fibré principal $P=P(M,G)$ de base $M$ (de dimension $d$) et de fibre type $G$, groupe de Lie de dimension $n$, avec la projection $\pi$ de $P$ sur $M$.\\\\ Soit $\rho$ une représentation du groupe structural $G$ sur un espace vectoriel $V$ de dimension $p$, et $P\times_\rho V$ l'espace fibré associé au fibré principal $P$ via la représentation $\rho$ de $G$ sur $V$.
\begin{rmq}
	On peut prolonger $\rho$ à $\g$ de la manière suivante : soit $X\in\g$ et $\gamma:[-1,1]\rightarrow G$ telle que $\frac{d}{dt}_{|t=0}(\gamma(t))=X$. On pose alors $\rho(X)=\frac{d}{dt}_{|t=0}(\rho(\gamma(t)))$.
\end{rmq}
Soit $(X_\alpha)$ une base de $\g$. $(\rho(X_\alpha))$ est un endomorphisme de $V$.
Dans toute la suite nous garderons la convention de désigner par :
\begin{enumerate}
	\item $\mu, \nu, \rho, ...$ les indices de base (variété $M$ de dimension $d$, espace-temps)
	\item $i, j, k, ...$ les indices de fibre (espace vectoriel $V$ de dimension $p$)
	\item $\alpha, \beta, \gamma, ...$ les indices d'algèbre de Lie ($\g$ de dimension $n$)
\end{enumerate}

\begin{defi}
	La représentation $\rho(A)$ du potentiel de jauge $A$ est appelée matrice de connexion et s'écrit localement : $$\rho(A)=A_\mu^\alpha \rho(X_\alpha) dx^\mu$$ C'est une matrice dont les coefficients sont des $1$-formes sur $M$.
\end{defi}
Dans la suite, sauf si le contraire est mentionné, on notera encore $A$ la matrice de connexion. En posant $T_\alpha=\rho(X_\alpha)$, on écrit : $$A^i_j=A^\alpha(T_\alpha)^i_j$$ où $A^\alpha=A^\alpha_\mu dx^\mu$. Alors $A^i_j=A^i_{j\mu}dx^\mu=A^\alpha_\mu(T_\alpha)^i_jdx^\mu$.\\\\
Les nombres $A^i_{j\mu}$ sont les \textit{coefficients de connexion relativement aux bases $(e_i)$ de $V$ et $(\partial_\mu)$ de $M$}.

\subsection{Courbure intrinsèque}
\paragraph{L'algèbre graduée $\O(P)\otimes\g$}

On a vu que $\omega\in\O^{1}(P)\otimes\g$. Il est en fait possible de considérer des $k$-formes à valeurs dans l'algèbre de Lie. L'ensemble des formes $\g$-valuées est naturellement muni d'une structure d'algèbre graduée grâce aux objets suivants.

\begin{defi}
	Soient $\phi\in\O^j(P)\otimes\g$, $\psi\in\O^i(P)\otimes\g$. On définit $[\phi,\psi]\in\O^{i+j}(P)\otimes\g$ par :
	$$[\phi,\psi](X_1, ..., X_{i+j})=\frac{1}{i!j!}\sum_{\sigma}(-1)^{\sigma}[\phi(X_{\sigma(1)}, ..., X_{\sigma(i)}), \psi(X_{\sigma(i+1)}, ..., X_{\sigma(i+j)})]$$
\end{defi}

Dérivons l'expression de ces objets en coordonnées. Soit $(E_\alpha)$ est une base de $\g$. Pour $\phi\in\O^{k}\otimes\g$, on écrit : 
$$\phi=\phi^\alpha \otimes E_\alpha$$
et le crochet de deux telles formes devient : 
\fbox{$[\phi,\psi]=c_{\alpha\beta}^{\gamma}(\phi^{\alpha}\wedge\psi^{\beta})\otimes E_{\gamma}$}.

Énonçons à présent deux théorèmes qui découlent de cette structure d'algèbre graduée de l'ensemble des formes sur $P$ $\g$-valuées $\bigoplus_k\O^{k}(P)\otimes\g$.
\begin{theorem}
	Soient $\phi\in\O^i(P)\otimes\g$, $\psi\in\O^j(P)\otimes\g$ et $\rho\in\O^k(P)\otimes\g$. On a : 
	\begin{enumerate}
		\item $[\psi,\phi]=(-1)^{ij}[\phi,\psi]$
		\item $(-1)^{ik}[[\phi, \psi], \rho]+(-1)^{kj}[[\rho,\phi],\psi]+(-1)^{ji}[[\psi,\rho],\psi]=0$
	\end{enumerate}
\end{theorem}

\begin{proof}[Preuve]
	Pour le 1., on a $[\phi,\psi]=c_{\alpha\beta}^{\gamma}(\phi^{\alpha}\wedge\psi^{\beta})\otimes E_{\gamma}=(-1)^{ij}c_{\alpha\beta}^{\gamma}(\psi^{\alpha}\wedge\phi^{\beta})\otimes E_{\gamma}=(-1)^{ij}[\psi,\phi]$.\\
	Pour le 2. on obtient l'égalité grâce à l'identité de Jacobi pour $\g$.
\end{proof}

\begin{theorem}
	Soient $\phi\in\O^{i}(P)\otimes\g$ et $\psi\in\O^j(P)\otimes\g$. On a : $d[\phi,\psi]=[d\phi,\psi]+(-1)^i[\phi,d\psi]$.
\end{theorem}

\begin{proof}[Preuve]
	C'est une conséquence directe de $d(\phi^\alpha\wedge\psi^\beta)=d\phi^\alpha\wedge\psi^\beta+(-1)^j\phi^\alpha\wedge d\psi^\beta$
\end{proof}

\paragraph{Différentielle extérieure covariante}

\begin{defi}
	Soit $\omega$ une $1$-forme de connexion sur le fibré principal $P=P(M,G)$ et $\alpha\in\O^{k}(P)\otimes\g$. Définissons l'horizontalisée $\omega^H$ de cette $k$-forme $\g$-valuée par son action sur $k$ champs de vecteurs $X_1, ..., X_k$, respectivement de partie horizontale $X_1^H, ..., X_k^H$: 
	$$\omega^{H}(X_1, ..., X_k)=\omega(X_1^H, ..., X_k^H)$$
\end{defi}

Autrement dit, on ne s'intéresse qu'au déplacement horizontal, et pas au déplacement dans le fibre ! 

\begin{defi}
	Soit $\alpha\in\O^{k}(P)\otimes\g$. La différentielle extérieure covariante de la $k$-forme $\alpha$ est la $(k+1)$-forme $d^H\alpha$ définie par : 
	$$d^H\alpha(X_1, ..., X_k)=d\alpha(X_1^H, ..., X_k^H)$$
\end{defi}

\begin{rmq}
	La différentielle extérieure covariante dépend de la connexion $\omega$ : $d^H=d^H_\omega$. Cependant, gardant ceci en tête, nous omettrons l'indice de connexion pour ne pas alourdir les notations, jusqu'à l'étude des lagrangiens à la fin du chapitre. 
\end{rmq}

\begin{defi}
	La courbure $\O^\omega$ est la différentielle covariante de la $1$-forme de connexion. $$\O^\omega=d^H\omega$$
\end{defi}

Cette expression de la courbure n'est pas facile à manier, heureusement, elle s'exprime sous une autre forme beaucoup plus "pratique", donnée par l'équation de structure. 

\paragraph{Équation de structure}

L'équation de structure donne l'égalité : $$\O^\omega=d\omega+\frac{1}{2}[\omega,\omega]$$
Prouvons trois lemmes pour aboutir au résultat.

\begin{lemma}
	Soit $X$ un champ de vecteurs sur $M$. Il existe un unique champ de vecteur horizontal $\tilde{X}$ sur $P$ tel que $\pi^{*}\tilde{X}=X$. C'est le relèvement horizontal de $X$  (qui est nécessairement $G$-invariant).
\end{lemma}

\begin{proof}[Preuve]
	La connexion $\omega$ définit une distribution $H$ de sous-espaces tangents tels que pour tout $p\in P$, on a un isomorphisme entre $H_p$ et $T_{\pi(p)}M$. Cette distribution est lisse, donc le champ de vecteur sur $P$ défini par isomorphisme d'espaces vectoriels de manière ponctuelle, est bien lisse lui aussi.\\
	On a : $\pi_*(R_{g*}\tilde{X}_p)=(\pi\circ R_g)_*(\tilde{X}_p)=\pi_*(\tilde{X}_p)=X_{\pi(p)}$. Donc $R_{g*}\tilde{X}_p=\tilde{X}_{pg}$.
\end{proof}

Le lemme suivant montre la stabilité du crochet de Lie par passage au champ de vecteur fondamental invariant à gauche sur $P$ associé à un vecteur de $\g$.

\begin{lemma}
	Soient $A$ et $B$ dans $\g$. Alors $[A,B]^*=[A^*,B^*]$ en tant que champs de vecteurs sur $P$.
\end{lemma}

\begin{proof}[Preuve]
	Soit $\phi_t:P\rightarrow P$ l'application donnée, pour $t$ un réel défini sur un voisinage de $0$, donné par : $\phi_t(p)=pexp(tA)$. $\phi$ est le flot du champ de vecteur $A^*$. Toutes les dérivées étant évaluées en $0$, on a :
	$$ [A^*,B^*]_p=\frac{d}{dt}\phi_{t*}^{-1}(B_{\phi_t(p)}^*)=\frac{d}{dt}\frac{d}{ds}\phi_t(p)exp(sB)exp(tA)^{-1}$$ donc :
	$$[A^*,B^*]_p=\frac{d}{dt}\frac{d}{ds}pexp(tA)exp(sB)exp(tA)^{-1}=\frac{d}{dt}\frac{d}{ds}pexp(s\mathfrak{Ad}_{exp(tA)}B)$$
	ce qui donne :
	$$[A^*,B^*]_p=\frac{d}{ds}pexp(s\frac{d}{dt}[\mathfrak{Ad}_{exp(tA)}B])=\frac{d}{ds}pexp(s[A,B])=[A,B]_p^*$$
\end{proof}

\begin{lemma}
	Soient $A\in\g$ et $X$ un champ de vecteurs sur $M$. 
	Alors $[A^*,\tilde{X}]=0$ où $\tilde{X}$ est le relèvement horizontal de $X$.
\end{lemma}

\begin{proof}[Preuve]
	$\tilde{X}$ est $G$-invariant donc :
	$$[A^*,\tilde{X}]_p=\frac{d}{dt}\phi_{t*}^{-1}(\tilde{X}_{\phi_t(p)})=\frac{d}{dt}\tilde{X}_p=0$$
\end{proof}
 On peut maintenant décliner l'équation de structure. On va faire la preuve dans trois cas particulier, puis nous pourrons conclure par linéarité des différents objets impliqués dans l'équation.

\begin{enumerate}
	\item Soient $X$ et $Y$ deux champs de vecteurs horizontaux sur $P$. On a :
	$$d\omega(X^H,Y^H)=d\omega(X,Y)=d\omega(X,Y)+[\omega(X),\omega(Y)]$$ car $\omega(X)=\omega(Y)=0$.
	\item Soient $X$ et $Y$ deux champs de vecteurs verticaux sur $P$, $G$-invariants. On peut alors écrire $X=A^*$ et $Y=B^*$ où $A$ et $B$ sont deux vecteurs de $\g$. Alors d'après la formule de la remarque 5, $d\omega(A^*,B^*)=A^*[\omega(B^*)]-B^*[\omega(A^*)]-\omega([A^*,B^*])$ et comme $\omega(B^*)=B$ et $\omega(A^*)=A$ (ce sont des constantes) on a $d\omega(A^*,B^*)=-\omega([A^*,B^*])=-\omega([A,B]^*)=-[A,B]=-[\omega(A^*),\omega(B^*)]$. On a aussi prouvé : 
	$$d\omega(X^H,Y^H)=d\omega(X,Y)=d\omega(X,Y)+[\omega(X),\omega(Y)]$$
	\item On cherche a montrer la formule dans le dernier cas nécessaire, c'est-à-dire où l'un des champs de vecteurs est vertical (on peut le supposer champ fondamental d'un certain vecteur de l'algèbre de Lie) et ou l'autre est horizontal (on peut le supposer relèvement horizontal d'un champ de vecteurs sur la base). Supposons $X=A^*$ vertical ($A\in\g$) et $Y=\tilde{B}$ horizontal. \\
	$d\omega(A^*,\tilde{B})=A^*[\omega(\tilde{B})]-\tilde{B}[\omega(A^*)]-\omega([A^*,\tilde{X}])=0$ car $\omega(\tilde{B})=0$, $\omega(A^*)=A$, et $[A^*,\tilde{X}]=0$ d'après les lemmes précédents. On a donc la même égalité qu'avant puisque les deux membres de l'égalité s'annulent.
\end{enumerate}
Comme dit précédemment, le résultat s'étend donc à tous les champs de vecteurs $X$ et $Y$ sur $P$. On a montré : 
\begin{center}
	\fbox{La $2$-forme de courbure est donnée par $\O^\omega=d\omega+\frac{1}{2}[\omega,\omega]$}
\end{center}

Voici un théorème dont la preuve est facilitée par cette équation de structure.

\begin{theorem}
	Soit $g\in G$. On a $R_g^*\O^\omega=\mathfrak{Ad}_{g^{-1}}\O^\omega$.
\end{theorem}

\begin{proof}[Preuve]
	$R_g^*\omega=\mathfrak{Ad}_{g^{-1}}\omega$ et $R_g^*[\phi,\psi]=[R_g^*\phi,R_g^*\psi]$ donc $$R_g^*\Omega^\omega=R_g^*d\omega+\frac{1}{2}[R_g^*\omega,R_g^*\omega]=d(\mathfrak{Ad}_{g^{-1}}\omega)+\frac{1}{2}[\mathfrak{Ad}_{g^{-1}}\omega,\mathfrak{Ad}_{g^{-1}}\omega]$$ donc $$R_g^*\Omega^\omega=\mathfrak{Ad}_{g^{-1}}(d\omega)+\frac{1}{2}\mathfrak{Ad}_{g^{-1}}[\omega,\omega]$$ ce qui permet de conclure.
\end{proof}

Le théorème suivant (identité de Bianchi) donne l'expression de la différentielle de la $2$-forme de courbure. C'est la partie sans champ des équations de Maxwell.

\paragraph{Équation de champ homogène}

\begin{theorem}[Identité de Bianchi ou équation de champ homogène]
	Si $\omega$ est une $1$-forme de connexion sur $P$, de courbure $\O^\omega$, alors 
	\begin{center}
		\fbox{$d^H\O^\omega=0$ et même $d\O^\omega=[\O^\omega, \omega]$}
	\end{center}
\end{theorem}

\begin{proof}[Preuve]
	$$d\O^\omega=d(d\omega+\frac{1}{2}[\omega,\omega])=d^2\omega+\frac{1}{2}[d\omega,\omega]-\frac{1}{2}[\omega,d\omega]$$
	Comme $d^2\omega=0$ et $[\omega,d\omega]=-[d\omega,\omega]$ : 
	$$d\O^\omega=[d\omega,\omega]$$
	et puisque $[[\omega,\omega],\omega]=0$ : 
	$$d\O^\omega=[d\omega+\frac{1}{2}[\omega,\omega],\omega]=[\O^\omega,\omega]$$
	Comme $\omega$ s'annule sur les vecteurs horizontaux, on a $d^H\O^\omega=0$.
\end{proof}

\subsection{Expression locale de la courbure sur la base}

Soit $A$ le potentiel de jauge associé à la connexion $\omega$ via la trivialisation locale $T_U$, donc $A\in\O^1(U)\otimes\g$. Si $\sigma_u$ est la section locale de $P$ associée à $T_U$, $A=\sigma_u^*\omega$. Définissons de même la courbure localement sur la base par $\O_u=\sigma_u^*\O^\omega$.

\begin{theorem}
	$\O_u=d\omega_u+\frac{1}{2}[\omega_u,\omega_u]$
\end{theorem}

\begin{proof}[Preuve]
	$\O_U=\sigma_u^*(\Omega^\omega)=\sigma_u^*(d\omega+\frac{1}{2}[\omega,\omega])=d(\sigma_u^*\omega)+\frac{1}{2}[\sigma_u^*,\sigma_u^*]=d\omega_u+\frac{1}{2}[\omega_u,\omega_u]$
\end{proof}

Pour commencer à faire la transition vers l'expression de ces objets dans les fibrés vectoriels associés, considérons une représentation linéaire $\rho$ du groupe $G$ sur un espace vectoriel de dimension finie. Notons encore $\rho$ le prolongement de la représentation à l'algèbre de Lie $\g$. On a alors la caractérisation suivante de $[.,.]$ :
\begin{theorem}
	Soient $\phi\in\O^i(M,\rho(\g))$ et $\psi\in\O^j(M,\rho(\g))$. On considère ici $\phi$ et $\psi$ comme des matrices de formes différentielles ($\R$-valuées) sur $M$. Alors : 
	$$[\phi,\psi]=\phi\wedge\psi-(-1)^{ij}\psi\wedge\phi$$
	où $\phi\wedge\psi$ est la multiplication matricielle de $\phi$ et $\psi$ avec les entrées multipliées par le "wedge" $\wedge$.
\end{theorem}

\begin{proof}[Preuve]
	Pour $A,B\in\rho(\g)$, on a $[A,B]=AB-BA$. Donc :
	$$[\phi,\psi](X_1, ..., X_{i+j})=\frac{1}{i!j!}\sum_{\sigma}(-1)^\sigma[\phi(X_{\sigma(1)}, ..., X_{\sigma(i)}),\psi(X_{\sigma(i+1)}, ..., X_{\sigma(i+j)})]$$
	$$[\phi,\psi](X_1, ..., X_{i+j})=(\phi\wedge\psi-(-1)^{ij}\psi\wedge\phi)(X_1, ..., X_{i+j})$$
\end{proof}

Sous les mêmes hypothèses, on a alors que $\frac{1}{2}[\omega,\omega]=\omega\wedge\omega$ donc : 
$$\O^\omega=d\omega+\omega\wedge\omega$$ et $$\O_U=d\omega_U+\omega_U\wedge\omega_U$$

Il est important de connaitre l'expression du changement du potentiel de jauge et de la courbure lors d'un changement de choix de jauge.

\begin{theorem}
	Soient $T_U:\pi^{-1}(U)\rightarrow U\times G$ et $T_V:\pi^{-1}(V)\rightarrow V\times G$ deux trivialisations locales, et $g_{uv}:U\bigcap V\rightarrow G$ la fonction de transition de l'une à l'autre définie par $\sigma_v(x)=\sigma_u(x)g_{uv}(x)$ où $\sigma_u$ et $\sigma_v$ sont les sections locales canoniquement associées, respectivement, à $T_U$ et $T_V$. Alors en notant $A_u$ et $A_v$ les potentiels de jauge correspondants, on a, pour $Y_x\in T_xM$ : 
	$$A_v(Y_x)=L^{-1}_{g_{uv}(x)*}(g_{uv*}(Y_x))+\mathfrak{Ad}_{g_{uv}(x)^{-1}}(A_u(Y_x))$$
\end{theorem}

\begin{proof}[Preuve]
	Soit $Y$ un vecteur tangent à $M$ en $x$ et $\gamma$ une courbe représentant $Y$ c'est-à-dire telle que $\gamma'(0)=Y$. Alors (les dérivées étant évaluées en $0$) :
	$$\sigma_{v*}(Y)=\frac{d}{dt}\sigma_v(\gamma(t))=\frac{d}{dt}[\sigma_u(\gamma(t))g_{uv}(\gamma(t))]$$
	qui donne : $$\sigma_{v*}(Y)=\sigma_u(x)\frac{d}{dt}[g_{uv}(\gamma(t))] + \frac{d}{dt}[\sigma_u(\gamma(t))]g_{uv}(x) $$
	d'où : 
	$$\sigma_{v*}(Y)=\frac{d}{dt}[\sigma_v(x)g_{uv}(x)^{-1}g_{uv}(\gamma(t))]+R_{g_{uv}(x)*}\sigma_{u*}(Y)=[L^{-1}_{g_{uv}(x)*}g_{uv*}(Y)]^*_{\sigma_v(x)}+R_{g_{uv}(x)*}\sigma_{u*}(Y)$$
	ce qui donne :
	$$\omega_v(Y)=\omega(\sigma_{v*}Y)=L^{-1}_{g_{uv}(x)*}g_{uv*}(Y)+\mathfrak{Ad}_{g_{uv}(x)^{-1}}\omega_u(Y)$$
\end{proof}

La courbure varie de façon beaucoup plus simple lors d'un changement de section locale : 
\begin{theorem}
	Sur $U\bigcap V$, on a \fbox{$\O_v=\mathfrak{Ad}_{g_{uv}^{-1}}\O_u$}. Si le groupe est représenté de manière linéaire sur un espace vectoriel de dimension finie, on a \fbox{$\O_v=g_{uv}^{-1}\O_ug_{uv}$} (groupe de matrices).
\end{theorem}

\begin{proof}[Preuve]
	On a vu que : 
	$$\sigma_{v*}(Y)=[L^{-1}_{g_{uv}(x)*}g_{uv*}(Y)]^*_{\sigma_v(x)}+R_{g_{uv}(x)*}\sigma_{u*}(Y)$$
	Par définition, la $2$-forme de courbure est bilinéaire et s'annule si l'un des vecteurs qu'elle prend en argument est vertical. Donc :
	$$\O_v(X,Y)=\O^\omega(\sigma_{v*}(X),\sigma_{v*}(Y))=\O^\omega(R_{g_{uv}(x)*}\sigma_{u*}(X), R_{g_{uv}(x)*}\sigma_{u*}(Y))$$
	donc :
	$$\O_v(X,Y)=\mathfrak{Ad}_{g_{uv}(x)^{-1}}\O_u(X,Y)$$
\end{proof}

Enfin, l'identité de Bianchi a pour expression locale, moyennant une trivialisation $T_U:\pi^{-1}(U)\rightarrow U\times G$ : 
$$d\O_u=[\O_u,\omega_u]$$ 
et pour un groupe 'représenté matriciellement' : $$d\O_u=\O_u\wedge\omega_u-\omega_u\wedge\O_u$$
Cela provient directement du fait que le crochet de deux formes $\g$-valuées est préservé par tiré en arrière.

\begin{rmq}
	En électromagnétisme, de groupe de jauge abélien $U(1)$, le théorème précédent montre que l'expression de la courbure (c'est-à-dire le tenseur électromagnétique) ne dépend pas du choix de jauge. De plus, comme $[\O_u,\omega_u]=0$, l'équation de champ homogène s'écrit juste $d\O_u=0$.
	Pour des théories de jauge non abéliennes, l'expression de la courbure dépend a priori de la trivialisation locale, et c'est pour cela qu'on considèrera plutôt la courbure comme la $2$-forme bien définie sur $P$, plutôt que son pull-back par une section locale.
	Enfin, \textbf{La courbure n'est plus une fonction linéaire de la connexion} puisque intervient le terme $\frac{1}{2}[\omega,\omega]$.
	C'est ce terme qui est à l'origine de l'interaction des bosons de jauge avec eux-même dans les théories quantiques des champs non abéliennes.
\end{rmq}

\subsection{Expression de la courbure dans les fibrés vectoriels associés}
\paragraph{Dérivée covariante d'une fonction scalaire sur le fibré principal}
Soit $f$ une fonction lisse de $P=P(M,G)$ dans $\R$. Sa différentielle extérieure covariante est : $$d^Hf(X)=df(X^H)$$ On peut, grâce à la linéarité de l'application linéaire tangente $df$, réécrire son expression en : $$d^Hf(X)=df(X)-d^Vf(X), X\in\Gamma(P)$$ où $$d^Vf(X)=df(X^V)$$ 
Soit $Y$ un champ de vecteurs local sur $M$ en $x\in M$.\\
Moyennant le choix d'une section locale $s$, on peut transporter la différentielle verticale $d^V$ sur $M$ ; en effet, pour $x\in M$: $$(s^*(d^Vf))_x(Y_x)=(d^Vf)_{s(x)}((s_*Y)_{s(x)})=df_{s(x)}((s_*Y)_{s(x)}^V)=df_{s(x)}((s_*y)^\alpha (\tilde{X}_\alpha)_{s(x)})$$ où $(s_*y)^\alpha$ est la fonction de $\mathcal{C}^\infty(P)$ coordonnée de $s_*Y$ selon le champ de vecteur fondamental $\tilde{X}_\alpha$ défini par $(\tilde{X}_{\alpha})_{p}=R_{g*}X_\alpha$, pour $p=(s(x),g)$ $\{X_\alpha\}_{\alpha\in[|1,n|]}$, telle que $s_*Y=(s_*y)^\alpha\tilde{X}_\alpha$. Donc en prenant $x\in M$, $p=s(x)=(x,e)\in P$ pour la trivialisation locale canoniquement associée à $s$, on a :
$$(s^*(d^Vf))_x(Y_x)=df_{s(x)}(\omega((s_*Y)_{s(x)}))=df_{s(x)}((s^*\omega)_x(Y_x))=df_{s(x)}(A_x(Y_x))$$ où $A$ est le potentiel de jauge local : $A=s^*\omega$. Ainsi, en écrivant localement $Y=y^\mu\partial_\mu$ avec $\partial_\mu$ les d champs de vecteurs coordonnées sur $M$ : $$(s^*(d^Vf))_x(Y_x)=df_{s(x)}(A_\mu^\alpha(x)y^\mu X_\alpha)=A_\mu^\alpha(x) y^\mu(x) df_{s(x)}(X_\alpha)$$ donc finalement : \fbox{$(d^Hf)_{s(x)}((s_*Y)_{s(x)})=df_{s(x)}((s_*Y)_{s(x)})-A_\mu^\alpha(x) y^\mu(x) df_{s(x)}(X_\alpha)$}

\vspace{15mm}

Passons à présent à une section $v$ d'un fibré $P\times_\rho V$ vectoriel associé à $P=P(M,G)$ via la représentation $\rho$ de $G$, dont nous omettrons dès maintenant la notation. On peut relier de manière canonique cette section $v$ de ce fibré associé à une fonction de $P$ dans $V$, à condition d'avoir choisi une section locale $s$ du fibré principal (c'est-à-dire à condition d'avoir fait un choix de jauge). En effet en écrivant, pour $x\in M$, $$v(x)=[(s(x),w(x))]$$ avec $s(x)\in P$ et $w(x)\in V$ et $[(s(x),w(x))]=\{(s(x)g,g^{-1}w(x)), g\in G\}$, il suffit de regarder $h:P\rightarrow V$ définie par $h(p)=g(s(x),g)=g^{-1}w(x)$ si $p=s(x)g$. On peut écrire : $$h(pg)=g^{-1}h(p)$$ et en explicitant les indices d'espace vectoriel : $$e_ih^i(pg)=e_i(g^{-1})^i_jh^j(p)$$ où $\{e_i\}_{i\in[|1,n|]}$ forme une base de $V$.
Ces fonction $h^i$, $i\in[|1,p|]$ sont donc définies sur $P$ à valeurs dans $\R$ et d'après le calcul précédent, pour un champ de vecteurs $Y$ sur $M$ défini localement en $x\in M$: $$(d^Hh^i)(s_*Y)=dh^i(s_*Y)-A_\mu^\alpha y^\mu dh^i_p(X_\alpha)$$
Soit maintenant un chemin $\gamma=[-1,1]\rightarrow G$ tel que $\gamma(0)=e$ et $\frac{d\gamma}{dt}(0)=X_\alpha$. 
$$dh^i_e(X_\alpha)=\frac{d}{dt}_{|t=0}h^i(s(x)g(t))=(\frac{d}{dt}_{|t=0}g^{-1}(t))h^i(s(x))$$
Or : $$\frac{d}{dt}_{|t=0}(g(t)g^{-1}(t))=\frac{d}{dt}_{|t=0}(R_{g^{-1}(t)}(g(t)))=\frac{d}{dt}_{|t=0}(R_{g^{-1}(t)}(e))+R_{g^{-1}(e)*}\frac{d}{dt}_{|t=0}(g(t))$$ ce qui donne $$\frac{d}{dt}_{|t=0}(g(t)g^{-1}(t))=\frac{d}{dt}_{|t=0}(g^{-1}(t))+\frac{d}{dt}_{|t=0}(g(t))=0$$ ce qui montre que \fbox{$\frac{d}{dt}_{|t=0}(g^{-1}(t))=-X_\alpha$}. Si on substitue ce résultat dans l'expression de $dh^i_e(X_\alpha)$, on obtient : $dh^i_e(X_\alpha)=-(X_\alpha)^i_j h^j(s(x))$, ce qui permet de conclure : $$(d^Hh^i)(s_*Y)=dh^i(s_*Y)+A_\mu^\alpha y^\mu(X_\alpha)^i_j h^j(s(x))$$ ou encore, puisque la base $(e_i)_{i\in[|1,n|]}$ est choisie indépendamment de $x$ : $$(d^H(e_ih^i))(s_*Y)=e_idh^i(s_*Y)+e_iA_{j\mu}^i h^j(s(x))dx^\mu(y^\nu\partial_\nu)$$ 
Notons : \fbox{$d^H(e_ih^i)=e_id(h^i)+e_iA^i_jh^i$}

\begin{rmq}
	Considérons par exemple un spineur de Dirac en électrodynamique quantique. Ce spineur est une section d'un fibré vectoriel $P\times_\rho V$ réel de dimension $4$, associé au fibré principal $P=P(M,U(1))$ où $M$ est la variété d'espace-temps, via la représentation de $U(1)$ donnée par : $$\rho:e^{i\alpha}\in U(1)\rightarrow e^{i\alpha}id$$ où $id$ est la matrice identité $4\times 4$. Ca permet de justifier le fait que ce spineur est un objet \textbf{qui est bien défini} au sens de la jauge (il \textbf{existe},comme un vecteur dans un espace vectoriel). Ayant choisi une jauge locale en $x_0$ $s:M\rightarrow G$, on peut donc écrire les spineurs, pour $y\in M$ dans un bon voisinage de $x_0$, comme : $$\psi:x\in M\rightarrow \psi(x)\in V$$ puisque : $$\tilde{\psi}(x)=[s(x),\psi(x)]$$ puisqu'on connait la "référence" $s$ et qu'on sait que pour un changement de jauge $g$ local en $x$ : $g:x\in M\rightarrow e^{i\alpha(x)}\in U(1)$ telle que la nouvelle jauge locale est $s'(x)=s(x)e^{i\alpha}(x)$, $\psi$ varie comme : $$\psi'(x)= e^{-i\alpha}\psi(x)$$
	Lorsqu'on a choisi une jauge, la relation entre les sections $\tilde{\psi}$ du fibré associé $P\times_\rho V$ et les fonctions $f$ sur la base $M$, à valeurs dans $V$, et se transformant comme $g^{-1}\cdot f$ lors du changement de jauge $g$ est biunivoque.
\end{rmq}

On pose, pour se rapprocher des notations utilisées en physique : $\psi^i(x)=h^i(s(x))$, ce qui donne : $$d^H(e_i\psi^i)=e_id\psi^i+e_iA^i_j\psi^j$$

\begin{rmq}
	Cette formule donne : \fbox{$d^H(e_i)=e_jA^j_i$}
\end{rmq}

Soit $\psi$ une section quelconque de $E$ au voisinage de $x\in M$. Localement, $\psi$ s'écrit $\psi(x)=e_i(x)\psi^i(x)$. Alors : $$d^H(v)=e_jA^j_i\psi^i+e_id\psi^i=(e_jA^j_{i\mu}\psi^i+e_i\partial_\mu \psi^i)dx^\mu$$
et, en ré-indiçant : $$d^H(v)=e_i(A^i_{j\mu}\psi^j+\partial_\mu \psi^i)dx^\mu$$
On commence à voir apparaître la 'substitution minimale de jauge' qui apparait en théories quantiques des champs pour assurer l'invariance de jauge du Lagrangien.\\
On définit la différentielle covariante dans la direction $\xi=\xi^\nu\partial_\nu$ par : $$\partial^H_\xi \psi = e_i(A^i_{j\mu}\psi^j+\partial_\mu \psi^i)\xi^\mu$$ et on note $\partial^H_\mu \psi$ la différentielle covariante de $\psi$ dans la direction $\partial_\mu$. On a : $$\partial^H_\mu \psi= e_i(A^i_{j\mu}\psi^j+\partial_\mu \psi^i)$$ 
Pour prendre des notations plus proches de celles employées dans les théories physiques, on pose : $$\psi^i_{;\mu}=\psi^i_{,\mu}+A^i_{j\mu}\psi^j$$ ce qui permet d'écrire $\partial^H_\mu \psi=e_i\psi^i_{;\mu}$ et où on a adopté la notation $\psi^i_{,\mu}=\partial_\mu \psi^i$.\\

Remarquons enfin que :
$d^H(e_i\psi^i)=d^H(e_i)+e_id\psi^i$ ce qui montre que $d^H$ est une dérivation pour les sections du fibré associé.
\\

Un champ de particules est une section d'un fibré associé, c'est-à-dire qu'un spineur varie d'une "bonne manière" lors d'un changement de jauge. Cette invariance est au cœur des théories des interactions. La différentielle de ce spineur doit avoir du sens elle aussi, et par conséquent doit varier de la même façon par changement de jauge. \\\\
Lorsqu'on change de jauge avec $g:M\rightarrow G$, le champ varie selon $\psi'(x)=g(x)^{-1}\psi(x)$. On sait que les coefficients de la connexion varient selon : $$A'(x)=g(x)^{-1}A(x)g(x)+g(x)^{-1}dg(x)$$ Pour alléger, nous écrirons $\psi'=g^{-1}\psi$ et $A'=g^{-1}Ag+g^{-1}dg$ mais toutes ces quantités sont bien, comme toujours, définies comme fonctions sur $M$, a priori non constantes.
Ainsi :
\begin{theorem}
	La différentielle covariante d'un spineur est une $1$-forme à valeurs dans le fibré associé, autrement dit, on conserve l'équivariance de jauge.
\end{theorem}
\begin{proof}[Preuve]
	$$d^{H}(\psi')=d(g^{-1}\psi)+(g^{-1}Ag+g^{-1}dg)(g^{-1}\psi)=(dg^{-1})\psi+g^{-1}d\psi+g^{-1}A\psi+g^{-1}(dg)g^{-1}\psi$$
	Or $$d(g^{-1}g)=g^{-1}dg+(dg^{-1})g=0$$ donc $$d(g^{-1})=-g^{-1}(dg)g^{-1}$$
	Par conséquent : $$d^{H}(g^{-1}\psi)=g^{-1}d^H\psi$$
\end{proof}

\begin{rmq}
	En électrodynamique quantique, on fait la 'substitution minimale' $\partial_\mu\rightarrow (\partial_\mu+\frac{iq}{\hbar c}A_\mu)$ si bien que le potentiel vecteur est exactement, à un facteur de proportionnalité près, la connexion dans un fibré vectoriel associé à $\R^4\times U(1)$ via la représentation standard de $U(1)$. Quitte à changer d'unités, on peut écrire : $$\partial_\mu\rightarrow (\partial_\mu+iA_\mu)$$
	or, pour la représentation déjà évoquée de $U(1)$ sur l'espace des spineurs, la connexion est dans $\mathfrak{u}(1)\equiv i\R\otimes id$. Les coefficients du potentiel vecteur sont donc à un facteur $i$ près, les coefficients de connexion. 
\end{rmq}

\paragraph{Loi de transformation des coefficients de connexion}
Soit $x\in M$. On a vu qu'on pouvait exprimer localement la matrice de connexion dans une base donnée par une famille libre de $p$ sections locales en $x$. Soient $(e_i)$ et $(e'_i)$ deux tels choix de sections locales en $x$. Le passage de l'un à l'autre est donné par une fonction $\Lambda:(M,x)\rightarrow Aut(V)$ où $\Lambda$ est définie sur un bon voisinage de $x$. Ainsi on aura : $$e'_i=\Lambda_i^je_j$$ On note $A_j^i$ et $A_j^{'i}$ les coefficients de la matrice de connexion par rapport à ces deux bases.
D'une part : $$d^H (e'_i)=e'_jA_i^{'j}=e_k\Lambda^k_jA^{'j}_i$$ et d'autre part : $$d^H(e_j\Lambda^j_i)=e_kA^k_j\Lambda^j_i+e_kd\Lambda^k_i$$ ce qui donne, sous forme matricielle : $$A'=\Lambda^{-1}A\Lambda + \Lambda^{-1}d\Lambda$$
Auparavant, on a vu que $$A_v(Y_x)=L^{-1}_{g_{uv}(x)*}(g_{uv*}(Y_x))+\mathfrak{Ad}_{g_{uv}(x)^{-1}}(A_u(Y_x))$$ et ces formules coïncident parfaitement.

\paragraph{Différentielle covariante des sections-k-formes}
On peut montrer que la différentielle covariante est une dérivation et vérifie : 
$$d^H(e_iv^i_{\mu ... \nu} dx^\mu\wedge...\wedge dx^\nu)=d^H(e_iv^i_{\mu...\nu})\wedge dx^\mu\wedge...\wedge dx^\nu$$ C'est l'unique dérivation de l'algèbre graduée $$\bigoplus \O^k(M)\otimes\Sigma(P\times_\rho V)$$ où $\Sigma(P\times_\rho V)$ est l'ensemble des sections du fibré associé $P\times_\rho V$.
L'avantage de ce formalisme est de permettre des calculs plus efficaces ; par exemple, soit $d^H\psi\in\O^1(M)\otimes\Sigma(E)$ où $\psi$ est une section d'un fibré vectoriel associé : $$d^H(e_i\psi^i_{;\mu} dx^\mu)=d^H(e_i\psi^i_{;\mu})\wedge dx^\mu=d^H(e_i)\psi^i_{;\mu}\wedge dx^\mu+e_id(\psi^i_{;\mu})\wedge dx^\mu$$ donc $$d^H(e_i\psi^i_{;\mu} dx^\mu)=e_jA_{i\nu}^j\psi^i_{;\mu}dx^\nu\wedge dx^\mu+e_i\partial_\nu(\psi^i_{;\mu})dx^\nu\wedge dx^\mu$$
$$d^H(e_i\psi^i_{;\mu} dx^\mu)=e_i\psi^i_{;\mu;\nu}dx^\nu\wedge dx^\mu$$

\paragraph{Définition de la courbure dans les fibrés associés}
On définit l'opérateur de courbure comme $F:\Sigma(E)\rightarrow\O^2(M)\otimes\Sigma(E)$ agissant sur les sections de E par : \fbox{$F=(d^H)^2$}. 
\begin{theorem}
	L'opérateur de courbure est $\mathcal{C}^\infty(M)$-linéaire
\end{theorem}

\begin{proof}[Preuve]
	En effet pour $\psi$ une section du fibré associé et $f\in\mathcal{C}^\infty(M)$ : $$F(f\psi)=d^H(d^H(f\psi))=d^H(d^H(\psi)f+\psi df)$$ $$F(f\psi)=F(\psi)f-d^H(\psi)d(f)+d^H(\psi)df+\psi d^2f$$ puisque $d^2f=0$ et car pour $\phi\in\O^k(M)\otimes\Sigma(P\times_\rho V)$ et $\psi\in\O(M)\otimes\Sigma(P\times_\rho V)$ on a $$d^H(\phi\wedge\psi)=d^H(\phi)\wedge\psi+(-1)^{k}\phi\wedge d^H(\psi)$$ donc finalement $$F(f\psi)=fF(\psi)$$
\end{proof}

Calculons les coefficients de l'opérateur de courbure qui est, comme on vient de le montrer, linéaire.
$$Fe_i=d^H(d^He_i)=d^H(e_jA_{i\mu}^jdx^\mu)=d^H(e_j)\wedge A_{i\mu}^jdx^\mu+e_jd(A_{i\mu}^jdx^\mu)$$ $$ Fe_i=e_kA_{j\nu}^kA_{i\mu}^jdx^\nu\wedge dx^\mu+e_jA^j_{i\mu,\nu}dx^\nu\wedge dx^\mu$$
Finalement, en notant $Fe_i=F_i^je_j$ on obtient \fbox{$F_i^j=d(A_i^j)+A^j_k\wedge A^k_i$}

\paragraph{Équation de structure pour la courbure} 
Soit $\sigma=e_i\sigma^i_\mu e^\mu$ un élément quelconque de $\O^1(M)\otimes\Sigma(P\times_\rho V)$. Ici, $(e^\mu)_{\mu\in[|1,n|]}$ est une base locale en $x$ du fibré cotangent, avec $\frac{1}{2}f^\mu_{\nu\rho}e^\nu\wedge e^\rho=de^\mu$ Alors : $$d^H\sigma=e_j(\sigma^j_{\mu;\nu}-\frac{1}{2}\sigma^j_\rho f^\rho_{\nu\mu})e^\nu\wedge e^\mu$$ (il faut rajouter le terme $e_i\psi^i_{;\mu}de^\mu$ au calcul semblable précédent) et $$d^H\sigma(e_\tau,e_u)=e_j(\sigma^j_{\mu;\nu}-\frac{1}{2}\sigma^j_\rho f^\rho_{\nu\mu})(\delta^\nu_\tau\delta^\mu_u-\delta^\nu_u\delta^\mu_\tau)=e_j(\sigma_{u;\tau}^j-\sigma_{\tau;u}^j-\sigma_\rho^jf_{\tau u}^\rho)$$ ce qui donne $$d^H\sigma(e_\tau,e_u)=d^H_{e_\tau}\sigma(e_u)-d^H_{e_u}\sigma(e_\tau)-\sigma([e_\tau,e_u])$$
Si comme avant, on a $\sigma=d^H\psi$, alors : $$F_{\tau u}\psi=F\psi(e_\tau,e_u)=d^H_{e_\tau}d^H_{e_u}\psi-d^H_{e_u}d^H_{e_\tau}\psi-d^H_{[e_\tau,e_u]}\psi$$ donc : \fbox{$F_{\mu\nu}=[d^H_\mu,d^H_\nu]-d^H_{[\mu,\nu]}$}

\paragraph{Identité de Bianchi pour la courbure}
De la même façon que dans les fibrés principaux, l'identité de Bianchi s'obtient à partir de l'équation de structure : $$dF=d(dA+A\wedge A)$$ donne $$dF=d^2(A)+dA\wedge A-A\wedge dA$$ et puisque $F=dA+A\wedge A$ : 
\begin{center}
	\fbox{$dF=F\wedge A-A\wedge F$}
\end{center}

Si le groupe est abélien, on arrive juste à $dF=0$ c'est-à-dire la première équation de Maxwell.

\paragraph{Expression plus familière (pour les théories physiques) des coefficients de courbure}

L'équation de structure de la courbure donne : $$F_{\mu\nu}=[d^H_\mu,d^H_\nu]-d^H_{[\mu, \nu]}$$
Si $(e_\mu)$ est un repère local en $x$, on a donc : 
$$F_{\mu\nu}=[d^H_\mu,d^H_\nu]$$
Soit $v\in V$, alors : $$F_{\mu\nu}v=d^H_\mu(d^H_\nu v) - d^H_\nu(d^H_\mu v)$$
Or : $$(d^H_\mu(d^H_\nu v))^i=\partial_\mu(d^H_\nu v)^i+A^i_{j\mu}(d^H_\nu v)^j$$ et en développant : 
$$(d^H_\mu(d^H_\nu v))^i=\partial_\mu\partial_\nu v^i+\partial_\mu(A^i_{j\mu}v^j)+A^i_{j\mu}\partial_\nu v^j + A^i_{j\mu}A^j_{\tau\nu}v^{\tau}$$
Le calcul donne : $$(F_{\mu\nu}v)^i=(\partial_\mu A^i_{j\nu})v^j-(\partial_\nu A^i_{j\mu})v^j+(A^i_{k\mu}A^k_{j\nu}-A^i_{k\nu}A^i_{j\mu})v^j$$
Rappelons que : $$A=A^\alpha_\mu X_\alpha dx^\mu=A^i_{k\mu}e_i\otimes e^j\otimes dx^\mu$$
avec $A^i_{k\mu}=A^\alpha_\mu (X_\alpha)^i_k$. On peut écrire : 
$$A^i_{k\mu}A^k_{j\nu}-A^i_{k\nu}A^i_{j\mu}=A^\alpha_{\mu}(T_\alpha)^i_kA^\beta_{\nu}(T_\beta)^k_j-A^\gamma_{\nu}(T_\gamma)^i_kA^\delta_{\mu}(T_\delta)^k_{j}$$
d'où : 
$$A^i_{k\mu}A^k_{j\nu}-A^i_{k\nu}A^i_{j\mu}=A^\alpha_{\mu}A^\beta_{\nu}(T_\alpha T_\beta)^i_j-A^\beta_{\nu}A^\alpha_{\mu}(T_\beta T\alpha)^i_j$$
et :
$$A^i_{k\mu}A^k_{j\nu}-A^i_{k\nu}A^i_{j\mu}=A^\alpha_{\mu}A^\beta_{\nu}[T_\alpha, T_\beta]^i_j=A^\alpha_{\mu}A^\beta_{\nu}(f_{\alpha\beta}^\gamma T_\gamma)^i_j$$
donc finalement : 
$$(F_{i}^{j})_{kl}=\partial_\mu A^j_{i\mu}-\partial_\nu A^j_{i\mu}+A_\mu^\alpha A_\nu^\beta f_{\alpha \beta}^{\gamma}(T_\gamma)^j_i$$
ou encore : 
$$F^{\gamma}_{\mu\nu}=\partial_\mu A_\nu^\gamma-\partial_\nu A_\mu^\gamma+A_\mu^\alpha A_\nu^\beta f_{\alpha\beta}^\gamma$$

\subsection{Lagrangiens et invariance de jauge}
Dans cette partie, nous allons définir ce qu'est un lagrangien, en termes de fibrés associés, puis nous donnerons une version "géométrique" des idées qui sont en fait à l'origine de l'introduction des connexions en physique. Considérons toujours avec le fibré principal $P=P(M,G)$ à qui est associé un fibré vectoriel $P\times_\rho V$ via la représentation $\rho$ de $G$ sur un espace vectoriel $V$. Au lieu de regarder exactement des sections du fibré associé, nous utiliserons des fonctions lisses $f$ définies sur $P$ et à valeur dans $V$, telles que : 
$$\forall g\in G,\ f(pg)=g^{-1}f(p)$$
Notons $\mathcal{C}(P,V)$ l'espace de ces fonctions. Il est clair qu'il y a un isomorphisme (non canonique puisqu'il faut choisir une jauge) entre $\mathcal{C}(P,V)$ et l'ensemble des sections du fibré associé, ce qui rend cette approche équivalente à celle par des fibrés associés. Posons également la notation $\overline{\O}^k(P,V)$ pour les $k$-formes sur $P$ à valeurs dans $V$ qui vérifient une telle propriété d'équivariance.
\paragraph{Définition du lagrangien}
\begin{defi}
	L'espace de $1$-jets des applications de $P$ dans $V$ est $$J(P,V)=\{(p,v,\theta), p\in P, v\in V, \theta\in\mathcal{L}(T_pP,V)\}$$ $J(P,V)$ est muni d'une structure canonique de variété.
\end{defi}

\begin{defi}
	Un lagrangien est une application $$L:J(P,V)\rightarrow\R$$ tel que pour tous $(p,v,\theta)\in J(P,V)$ et pour $g\in G$, on ait $$L(p,v,\theta)=L(pg,g^{-1}v,g^{-1}\theta\circ R_{g*})$$
\end{defi}

\begin{theorem}
	Soit un lagrangien $L:J(P,V)\rightarrow \R$. Il existe une fonction $$\mathfrak{L}_0:\mathcal{C}(P,V)\rightarrow\mathcal{C}^{\infty}(M)$$ définie, pour $x\in M$, $\psi\in\mathcal{C}(P,V)$ et $p\in P$ avec $\pi(p)=x$, par :
	$$\mathfrak{L}_0(\psi)(x)=L(p,\psi(p),d\psi_p)$$
\end{theorem}

\begin{rmq}
	La définition du lagrangien montre que moyennant un choix de jauge, le lagrangien est une fonction bien définie sur la base et sur des sections du fibré associé. C'est simplement une approche moins générale, mais équivalente dès que le choix a été fait.
\end{rmq}

On dira qu'un lagrangien $L:J(P,V)\rightarrow \R$ est $G$-invariant si : $$L(p,g\cdot v,g\cdot\theta)=L(p,v,\theta)$$

Presque tous les lagrangiens rencontrés dans des théories physiques sont $G$-invariants.

\paragraph{Invariance de jauge du lagrangien et dérivée covariante}
Cette section fait écho à l'introduction de ce chapitre, dans laquelle nous avons commencé à faire émerger l'idée de dérivée covariante pour l'établissement de l'électrodynamique quantique. Nous allons voir que la fonction :
$$\mathfrak{L}_0:\mathcal{C}(P,V)\rightarrow\mathcal{C}^{\infty}(M)$$
n'est pas forcément bien définie au sens de la jauge, et comment l'utilisation de dérivées covariantes résout ce problème.\\\\
Soit $f$ un automorphisme de jauge, c'est-à-dire une fonction lisse de $P$ dans $P$ telle que : $$\forall p\in P,\ \forall g\in G,\ f(pg)=f(p)g$$
f est canoniquement reliée à une fonction $\tau$ de $P$ dans $G$ : 
$$\tau(pg)=g^{-1}\tau(p)g$$ par la relation $f(p):p\tau(p)$ puisqu'alors : 
$$f(pg)=pg\tau(pg)=pgg^{-1}\tau(p)g=p\tau(p)g=f(p)g$$
Remarquons que $$f^*\psi(p)=\psi(f(p))=\psi(p\tau(p))=\tau(p)^{-1}\psi(p)$$ autrement dit $f^*\psi=\tau^{-1}$. On veut calculer $d(\tau^{-1}\psi)_p$. Soit $X\in T_pP$ et $\gamma:\R\rightarrow P$ tel que $\gamma'(0)=X$. Dans le calcul les dérivées sont toujours évaluées en 0.
$$d(\tau^{-1}\psi)(X)=\frac{d}{dt}\tau^{-1}(\gamma(t))\psi(\gamma(t))=\frac{d}{dt}\tau^{-1}(p)\psi(\gamma(t))+\frac{d}{dt}\tau^{-1}(\gamma(t))\psi(p)$$
$$d(\tau^{-1}\psi)(X)=\tau^{-1}(p)d\psi(X)+\frac{d}{dt}\tau^{-1}(\gamma(t))\tau(p)\tau^{-1}(p)\psi(p)$$
$$d(\tau^{-1}\psi)(X)=\tau^{-1}(p)d\psi(X)+R_{\tau(p)*}(\tau^{-1})_{*p}(X)\tau(p)^{-1}\psi(p)$$
donc finalement : 
$$d(\tau^{-1}\psi)_p=\tau^{-1}(p)d\psi(p)+R_{\tau(p)*}(\tau^{-1})_{*p}\tau(p)^{-1}\psi(p)$$
Par conséquent : $\mathfrak{L}_0(f^*\psi)(x)=L(p,(f^*\psi)(p),d(f*\psi)_p)$ donne : 
$$\mathfrak{L}_0(f^*\psi)(x)=L(p,\tau^{-1}(p)\psi(p),\tau^{-1}(p)d\psi(p)+R_{\tau(p)*}(\tau^{-1})_{*p}\tau(p)^{-1}\psi(p))$$
Autrement dit, on n'a pas forcément $\mathfrak{L}_0(\psi)=\mathfrak{L}_0(f^{-1*}\psi)$ !
\begin{defi}
	Soit $L:J(P,V)\rightarrow\R$ un lagrangien. Soit $\mathcal{C}$ l'espace des connexions sur P. Définissons : $$\mathfrak{L}:C(P,V)\times\mathcal{C}\rightarrow\mathcal{C}^\infty(M)$$ par : 
	$$\forall x\in M,\ p\in\pi^{-1}(x),\ \psi\in C(P,V),\ \omega\in\mathcal{C},\ \mathfrak{L}(\psi,\omega)(x)=L(p,\psi(p),d^H_\omega\psi_p)$$
\end{defi}

Vérifions que $\mathfrak{L}$ est bien défini. Par définition, $$d^H_\omega\psi\in\O^{1}(P,V)$$ et comme on l'a déjà vu sous une forme légèrement différente, $$R_{g}^*(d^H_\omega\psi_{pg})=g^{-1}d^H_\omega\psi_p$$ puisque 
$$R_{g}^*(d\psi)^H=(R_{g}^*d\psi)^H=(d(R_{g}^*\psi))^H=(g^{-1}d\psi)^H=g^{-1}d^H_\omega\psi$$
donc 
$$L(pg,\psi(pg),d^H_\omega\psi_{pg})=L(pg,g^{-1}\psi,g^{-1}d^H_\omega\psi_p\circ R_{g^{-1}*})$$
ce qui donne bien le résultat voulu.

\begin{lemma}
	Soit $f$ un automorphisme de jauge $(f(pg)=f(p)g)$ et si $\tau:P\rightarrow G,\ \tau(pg)=g^{-1}\tau(p)g$ est canoniquement associé à $f$ par la relation $f(p)=\tau(p)$, pour $X\in T_pP$, on a $$f_*(X)=(L^{-1}_{\tau(p)*}\tau_*(X))^*_{f(p)}+R_{\tau(p)*}(X)$$
\end{lemma}

\begin{proof}[Preuve]
	Soit $\gamma$ une courbe représentant $X$. Les dérivées étant évaluées en $0$ :
	$$f_*(X)=\frac{d}{dt}f(\gamma(t))=\frac{d}{dt}\gamma(t)\tau(\gamma(t))=\frac{d}{dt}p\tau(p)\tau(p)^{-1}\tau(\gamma(t))+\frac{d}{dt}R_{\tau(p)}(\gamma(t))$$
	$$f_*(X)=\frac{d}{dt}f(p)\tau(p)^{-1}\tau(\gamma(t))+R_{\tau(p)*}(\gamma(t))$$
\end{proof}

Une conséquence directe de ce lemme est que si $\phi$ est une k-forme sur $P$ $V$-valuée, et si f est un automorphisme de jauge associé à $\tau$ comme dans le lemme, on a $$f^*\phi=\tau^{-1}\phi$$ 

\begin{theorem}
	Si $L$ est $G$-invariant, $\mathfrak{L}$ est invariant de jauge, au sens suivant : Si $f$ est un automorphisme de jauge, alors : 
	$$\mathfrak{L}(f^*\psi,f^*\omega)=\mathfrak{L}(\psi,\omega)$$
\end{theorem}

\begin{proof}[Preuve]
	$$\mathfrak{L}(f^*\psi,f^*\omega)(x)=L(p,(f^*\psi)(p),f^*(d\psi_p+\omega_p\psi(p)))$$ d'où :
	$$\mathfrak{L}(f^*\psi,f^*\omega)(x)=L(p,(f^*\psi)(p),\tau(p)^{-1}d^H_\omega\psi_p)$$
	et par G-invariance de L :
	$$\mathfrak{L}(f^*\psi,f^*\omega)(x)=\mathfrak{L}(\psi,\omega)(x)$$
\end{proof}

\subsection{Le principe de moindre action}
Soient toujours le fibré principal $P=P(M,G)$ de groupe structural $G$, et $\rho:G\rightarrow GL(V)$ une représentation de $G$. Soit $h$ une métrique sur $M$ variété d'espace-temps, qu'on suppose orientée pour qu'il y ait une forme volume $\mu$ bien définie associée à $h$. Soit $L$ un lagrangien $G$-invariant, et $\omega$ une connexion fixée.
\begin{defi}
	Soit $U$ un ouvert de $M$ d'adhérence compacte. Soit $\psi\in C(P,V)$. L'action de $\psi$ sur $U$ est : $$\mathfrak{S}_U^\omega(\psi)=\int_{U}\mathfrak{L}^\omega(\psi)\mu\ \in\R$$
\end{defi}

\begin{defi}[Principe de moindre action]
	Soit $\psi\in\mathcal{C}(P,V)$. On dit que $\psi$ est stationnaire pour $\mathfrak{L}^\omega$ si pour tout ouvert $U$ de $M$ à support compact, et pour tout $\sigma\in\mathcal{C}(P,V)$ dont l'image du support par la projection canonique est dans $U$, on a : 
	$$\frac{d}{dt}\mathfrak{S}_U^\omega(\psi+t\sigma)_{|t=0}$$
\end{defi}

\begin{rmq}
	Il est possible de montrer, à condition d'introduire de la structure sur l'espace $V$, en particulier une métrique, l'équivalence entre : $\psi$ vérifie le principe de moindre action et : $\psi$ vérifie l'équation d'Euler-Lagrange : 
	$$\delta^H(\frac{\partial L}{\partial(d^H\psi)})+\frac{\partial L}{\partial \psi}=0$$
	où on définit $\frac{\partial L}{\partial(d^H\psi)}$ comme une 1-forme sur $P$ à valeurs dans $V$ telle que $\frac{\partial L}{\partial(d^H\psi)}\circ R_{g*}=g^{-1}\frac{\partial L}{\partial(d^H\psi)}$, où $\delta^H$, définie grâce à l'étoile de Hodge généralisée, est la codifférentielle covariante qui transforme les $k$-formes sur $P$ en $(k-1)$-formes sur $P$, et où $\frac{\partial L}{\partial \psi}$ est dans $\mathcal{C}(P,V)$. 
\end{rmq}

On arrive après quelques manipulations, aux équations de champ inhomogènes : définissons 
$$\mathfrak{L}_{self}:\mathcal{C}\rightarrow\mathcal{C}^\infty(M)$$ la self-action d'une connexion, par $$\mathfrak{L}_{self}(\omega)=-\frac{1}{2}(\overline{h}k)(\O^\omega,\O^\omega)$$ à l'aide d'une métrique $\overline{h}$ sur $M$ et d'une métrique $k$ sur $\g$.
Si $L$ est un lagrangien, l'action combinée de $\psi$ et $\omega$ est définie par : $$(\mathfrak{L}+\mathfrak{L}_{self})(\psi,\omega)=\mathfrak{L}(\psi,\omega)+\mathfrak{L}_{self}(\omega)$$
Si on généralise le principe de moindre action en disant que la paire $(\psi,\omega)$ est stationnaire pour $(\mathfrak{L}+\mathfrak{L}_{self})$ si, pour tout ouvert $U$ de $M$ à support compact, pour tout $\sigma\in \mathcal{C}(P,V)$ et pour toute forme $\tau\in\overline{\O}^1(P,\g)$, on a :
$$\frac{d}{dt}_{|t=0}\int_U(\mathfrak{L}+\mathfrak{L}_{self})(\psi+t\sigma,\omega+t\tau)\mu=0$$
ce qui permet en passant de dire que la structure de l'espace de connexions $\mathcal{C}$ sur P est un espace affine, pour lequel l'espace vectoriel sous-jacent est l'espace des 1-formes sur P à valeurs dans $\g$ équivariantes pour la jauge.
On arrive alors au théorème suivant :
\begin{theorem}
	La paire $(\psi,\omega)$ est stationnaire pour $(\mathfrak{L}+\mathfrak{L}_{self})$ si et seulement si les deux conditions suivantes sont vérifiées :
	\begin{enumerate}
		\item $\delta^H(\frac{\partial L}{\partial(d^H\psi)})+\frac{\partial L}{\partial \psi}=0$
		\item $\delta^H\O^\omega=J^\omega(\psi)$
	\end{enumerate}
\end{theorem}
où $J^\omega(\psi)$ est la courant associé à la paire $(\psi,\omega)$

\chapter{Théorie électrofaible et boson BEH}
Après avoir présenté la structure des théories de jauge de manière générale, nous allons en prendre un cas particulier afin de dériver le modèle standard électrofaible (ou modèle de Glashow - Salam - Weinberg). De la même manière que l'histoire de l'apparition des concepts permet de mieux saisir les idées profondes des théories de jauge, le modèle standard électrofaible est indissociable des développements historiques qui y ont mené. Nous présentons tout d'abord les points les plus importants qui ont marqué la découverte puis l'étude de cette interaction fondamentale, puis nous dériverons la théorie à partir des idées du chapitre précédent - c'est-à-dire que nous prendrons comme postulat de base que \textbf{la théorie électrofaible a une structure de jauge, de groupe structural $SU(2)\times U(1)$}. Enfin, nous présenterons les idées actuelles qui visent à dépasser ce modèle standard de la physique des particules : les théories "Beyond the Standard Model" (BSM) et les conséquences qu'elles ont sur les quantités mesurables.\\\\
\textit{"Or, pour l'esprit scientifique, tracer nettement une frontière, c'est déjà la dépasser."} (G. Bachelard)

\section{Contexte historique et construction de la théorie physique}
\subsection{La découverte de l'interaction faible et le postulat du neutrino}
\textbf{Henri Becquerel découvre la radioactivité en 1896.} La formulation de ses résultats dans les comptes rendus de l'Académie des Sciences du 24 février 1896 est restée célèbre. Voici comment il décrit ses travaux : \\
"On enveloppe une plaque photographique Lumière, au gélatinobromure, avec deux feuilles de papier noir très épais, tel que la plaque ne se voile pas par une exposition au Soleil, durant une journée. On pose sur la feuille de papier, à l'extérieur, une plaque de la substance phosphorescente, et l'on expose le tout au Soleil, pendant plusieurs heures. Lorsqu'on développe ensuite la plaque photographique, on reconnaît que la silhouette de la substance phosphorescente apparaît en noir sur le cliché. Si l'on interpose entre la substance phosphorescente et le papier une pièce de monnaie, ou un écran métallique percé d'un dessin à jour, on voit l'image dè ces objets apparaître sur le cliché. On peut répéter les mêmes expériences en interposant entre la substance phosphorescente et le papier une mince lame de verre, ce qui exclut la possibilité d'une action chimique due à des vapeurs qui pourraient émaner de la substance échauffée par les rayons solaires. On doit donc conclure de ces expériences que la substance phosphorescente
en question émet des radiations qui traversent le papier opaque à la lumière et réduisent les sels d'argent."\\
La radioactivité a ensuite été observée par \textbf{Pierre et Marie Curie} dans le thorium, et dans des "nouveaux" éléments : le polonium et le radium.

En 1899, Rutherford classe la radioactivité en deux catégories différentes, en se basant sur la pénétrabilité des émissions : la \textbf{radioactivité alpha}, stoppée par de minces feuilles de papier ou d'aluminium, et la \textbf{radioactivité bêta}, pouvant pénétrer plusieurs millimètres d'aluminium.

En 1900, Paul Villard découvre des radiations encore plus pénétrantes que Rutherford identifie en 1903 à un nouveau type d'émissions, appelées \textbf{rayons gamma}.

Toujours en 1900, Becquerel prouve que les particules émises par radioactivité bêta sont des électrons, en mesurant le quotient de leur masse par leur charge. Un an plus tard, Rutherford et Frédéric Soddy montrent que la radioactivité bêta implique la transmutation des atomes impliqués, en atomes d'un autre élément chimique. Après plus d'analyses, en 1913, Soddy et Kazimierz Fajans proposent indépendamment leur \textbf{loi de déplacement par radioactivité}, qui relie la radioactivité bêta à un déplacement d'une case à droite dans le tableau périodique, et la radioactivité alpha à un déplacement de deux cases à gauche.\\\\

En 1914, James Chadwick étudie le spectre énergétique de l'électron émis par radioactivité bêta. Compte tenu de la récente théorie de la relativité restreinte, il attend un pic correspondant à l'énergie de différence de masse entre le noyau avant désintégration et le noyau après désintégration. Or voilà qu'il mesure un spectre continu, qui s'étend d'environ zéro jusqu'à cette valeur de différence de masse, comme si de l'énergie était perdue durant la désintégration. Un autre problème est la conservation du moment angulaire : par exemple, le spin du noyau $^{14}_7N$ est un entier, contrairement à ce que prévoit la désintégration 
$$^{14}_6C \rightarrow ^{14}_7N + e^{-}$$
Ces résultats alors inexplicables sont confirmés par de nombreuses expériences entre 1920 et 1927.\\\\
La révolution vient d'une lettre écrite par Pauli en 1930, pleine d'humour, où il suggère, presque en s'excusant d'une telle hypothèse, qu'une petite particule neutre, de petite masse (moins de 1\% de la masse du proton), et contenue dans le noyau tout comme les protons et les électrons, pourrait être émise durant la désintégration. Il la baptise neutron.

\begin{figure}[!h]
	\centering
	\includegraphics[scale=0.85]{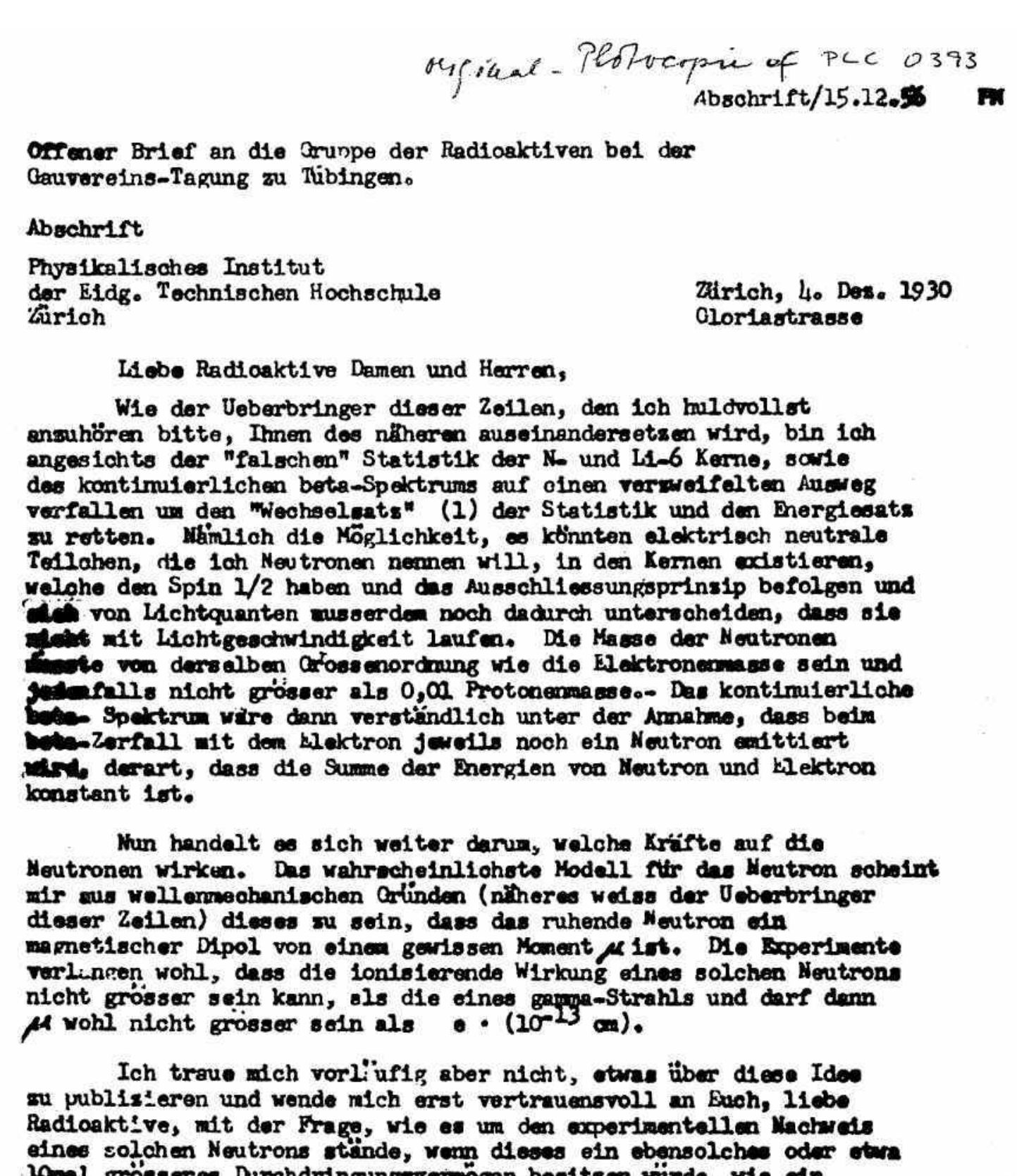}
	\caption{Début de la lettre de Pauli}
\end{figure}

En 1932, après la découverte de ce qui est appelé aujourd'hui \textbf{neutron} par J. Chadwick en étudiant les travaux des Joliot-Curie, Fermi renomme (en accord avec Pauli) la particule de Pauli \textit{neutrino}, ou petit neutron\cite{wikbeta}.

\subsection{La théorie de Fermi}
En 1933, Fermi propose une théorie de contact à quatre fermions de l'interaction faible, qui explique la radioactivité bêta. Son article est refusé par le journal \textit{Nature} ; il publie donc en italien, et son article n'est traduit en anglais qu'en 1968 par Fred Wilson\cite{bosonchapeau}.
Cette théorie est révolutionnaire, dans la mesure où il s'agit d'une théorie quantique des champs : ni l'électron ni le neutrino ne préexistent dans le noyau, il faut donc pouvoir faire apparaître et disparaitre des particules.

\begin{figure}[!h]
	\centering
	\includegraphics[scale=0.7]{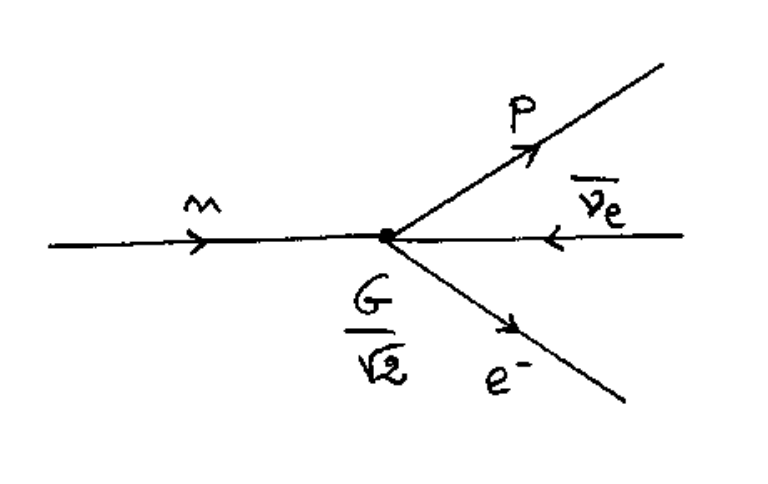}
	\caption{Diagramme de Feynman de la désintégration du neutron dans la théorie de Fermi}
\end{figure}

La théorie de Fermi est en fait très générale, on peut suivre le même état d'esprit et dériver une théorie de Fermi des interactions faibles entre leptons (plus proche de l'approche que nous allons suivre tout à l'heure). Dans cette théorie, le hamiltonien d'interaction s'écrit (avec quelques anachronismes) : 
$$ \mathcal{H}_I^{(F)}=\frac{G}{\sqrt{2}}J^\alpha(x)J_\alpha^\dagger(x) $$
avec $$J_\alpha(x)=\sum_{l}\overline{\psi}_l(x)\gamma_\alpha(1-\gamma_5)\psi_{\nu_l}(x)$$
et
$$J_\alpha^\dagger(x)=\sum_{l}\overline{\psi}_{\nu_l}(x)\gamma_\alpha(1-\gamma_5)\psi_l(x)$$
où la sommation porte sur les différentes familles de leptons.

En effet, les données expérimentales semblent montrer que seules ces combinaisons des champs entrent dans l'expression de l'interaction. Pour la désintégration du muon par exemple, la théorie de Fermi donne un processus du premier ordre qu'on peut représenter par le diagramme de Feynman suivant :
\begin{figure}[!h]
	\centering
	\includegraphics[scale=0.65]{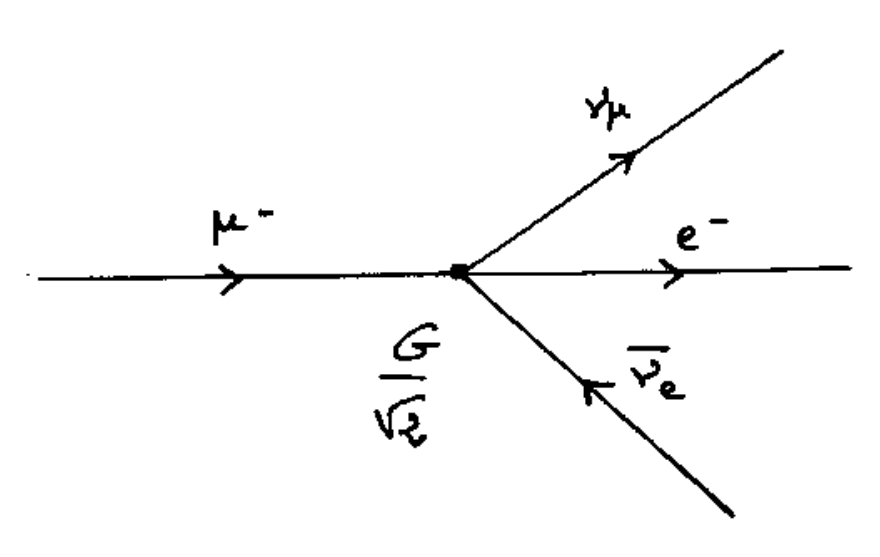}
	\caption{Diagramme de Feynman de la désintégration du muon, dans la théorie de Fermi}
\end{figure}

Cependant, la théorie de Fermi est une approximation de basse énergie d'une théorie plus large, encore à découvrir. Aujourd'hui, on sait que l'interaction est transmise par des particules massives ; il est possible de montrer que la théorie de Fermi correspond à la limite de la théorie IVB (voir ci-dessous) où on considère que les bosons vecteurs sont infiniment lourds. Aux échelles d'énergies accessibles à l'époque, cette approximation est excellente (le projet Manhattan a reposé en entier sur cette théorie, et le résultat montre bien que les idées de Fermi permettaient de décrire les interactions faibles avec une grande précision, du moins à ces échelles d'énergie).

\subsection{L'idée de bosons vecteurs, la découverte du muon et du pion}

En 1935, Yukawa, en étudiant l'interaction forte, prédit l'existence d'une particule médiatrice ; il estime sa masse à partir de la portée de cette interaction : 100 MeV. Il l'appelle méson (du grec mésos qui signifie intermédiaire), puisque sa masse est comprise entre celle de l'électron (0,51 MeV) et celle du nucléon (0,94 GeV).

En 1936, alors qu' Anderson et Neddermeyer étudient les rayons cosmiques afin d'y trouver le méson $\pi$ prédit par le modèle de Yukawa, ils observent une particule de masse 106 MeV qu'ils appellent \textit{mésotron}. L'existence de cette particule est confirmée par Street et Stevenson dans une chambre à bulle, un an plus tard.

Cependant, elle ne semble pas participer à l'interaction forte ; par exemple, elle a une très grande pénétrabilité dans la matière.

Le véritable méson $\pi$ est découvert en 1947 par Cecil F. Powell, toujours dans les rayons cosmiques, dans une chambre à bulles installée sur les sommets de la Cordillère des Andes. Il a quant'à lui les propriétés prédites par Yukawa.

Le mésotron est alors renommé en \textbf{méson $\mu$}. Lors de l'élaboration du modèle standard dans les années 70, le terme méson a été assigné aux particules formées d'un quark et d'un anti-quark, comme le $\pi$, qu'on appelle désormais méson $\pi$, ou pion, et qui forme un triplet d'isospin fort $\pi^-$, $\pi^+$ et $\pi^0$. Ainsi, le méson $\mu$, qui est un \textit{lepton} (du grec leptos : faible) puisqu'il n'est pas constitué de quarks, et a été renommé \textbf{muon}.

L'idée que les interactions faibles sont dues à l'échange de bosons massifs semble avoir été proposée par Klein en 1938 \cite{Klein}.

\subsection{L'énigme $\Theta-\tau$}
En 1949, Powell identifie deux nouvelles particules cosmiques, l'une qu'il appelle méson $\tau^+$ (ce n'est pas le lepton $\tau$) et qui se désintègre en trois pions selon : 
$$\tau^+\rightarrow \pi^+\pi^+\pi^-$$
et l'autre qu'il nomme $\Theta^+$, qui se désintègre en deux pions selon : 
$$\Theta^+\rightarrow\pi^+\pi^0$$
Le calcul des propriétés de ces particules montrent qu'elles sont indistinguables si ce n'est par leur mode de désintégration. Cela ne pose pas de problème en soi puisqu'il pourrait s'agir de deux désintégrations différentes de la même particule. Cependant, la conservation de la parité semble alors violée : en effet la parité du $\pi$ étant $(-1)$, le système de deux pions a une parité paire tandis que la système de trois pions une parité impaire. Par conséquent, il ne peut pas s'agir de la même particule.\\\\
Lors d'un séminaire à Rochester en avril 1956, la question est évoquée et discutée par (entre autres) Feynman, Block, Lee et Yang. Lee et Yang proposent notamment leur idée de "parity doubling", selon laquelle certaines particules peuvent exister dans deux états de parité différente. Les différentes interventions mènent à une sorte de tremblement de terre : et si la conservation de la parité n'était pas une loi fondamentale ? Cette idée, vraiment révolutionnaire puisque la conservation de la parité avait été observée avec un très grand degré de précision dans les interactions électromagnétiques et fortes, et n'avait encore jamais été mise en défaut, a mené au célèbre papier de Lee et Yang \cite{LeeYang} : \textit{Question of Parity Conservation in Weak Interactions}, dans lequel ils proposent une expérience permettant de tirer la question au clair, en étudiant la désintégration du $^{60}Co$. Madame Wu, à l'université de Berkeley, réalise cette expérience en 1957 pourtant difficile sur le plan technique, et observe bien la violation de la symétrie de parité comme cela avait été suggéré \cite{Wu}. L'interaction faible, pour laquelle on commence à peine à obtenir des données expérimentales, semble bien contre intuitive...

\subsection{Des avancées majeures en physique des neutrinos}
De 1953 à 1956, Cowan et Reines tentent un expérience (qui aboutit en 1956) qui vise à montrer que \textbf{le neutrino est bien une particule libre}.
En utilisant le flux (alors hypothétique) d'antineutrinos produits par un réacteur nucléaire ($5.10^{13}$ antineutrinos par seconde), le but était de faire interagir les particules (s'il en était) avec des protons de l'eau d'une "piscine" selon la réaction : $$\overline{\nu}_e+p\rightarrow n + e^+$$
décrite par :
\begin{figure}[!h]
	\centering
	\includegraphics[scale=0.8]{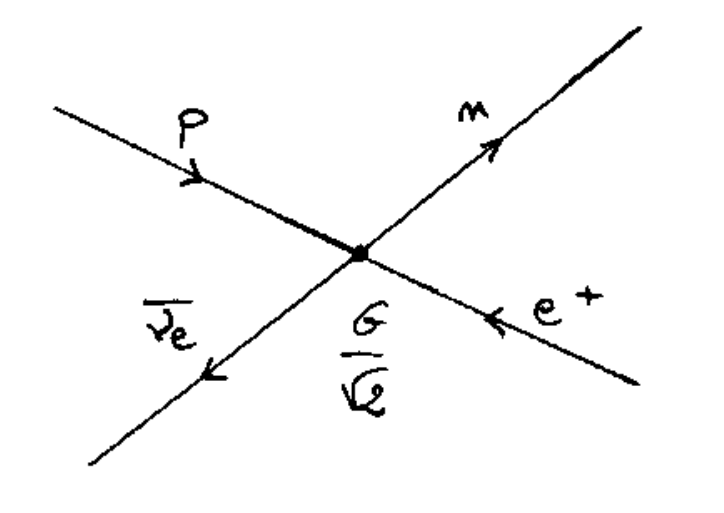}
	\caption{Diagramme de Feynman de l'interaction anti-neutrino/proton dans la théorie de Fermi}
\end{figure}

suivie peu après de $$e^+ + e^- \rightarrow 2\gamma$$ calculé grâce au diagramme :
\begin{figure}[!h]
	\centering
	\includegraphics[scale=0.8]{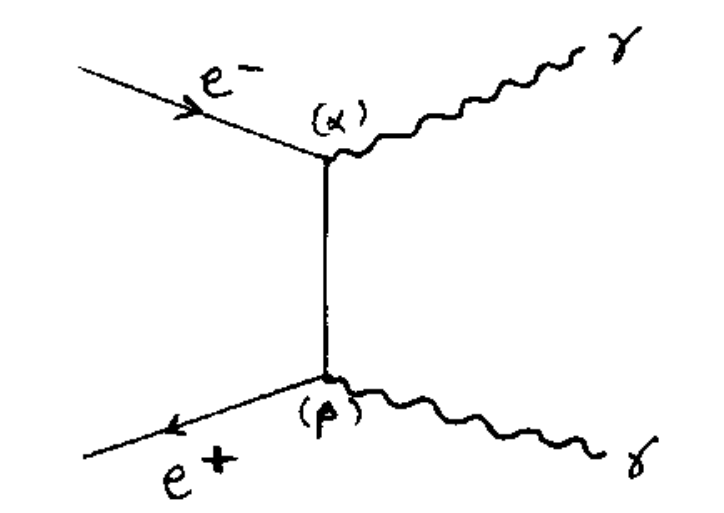}
	\caption{Diagramme de Feynman à l'ordre le plus bas de l'annihilation positron-électron en électrodynamique quantique}
\end{figure}

Les photons ainsi produits ayant une énergie considérable, il est possible de les détecter en utilisant des scintillateurs, qui produisent des douches de photons moins énergétiques dont on mesure l'énergie totale à l'aide de photomultiplicateurs. En fait cette expérience n'était pas assez concluante, et du cadmium (qui est un très bon absorbeur de neutrons) a à terme été utilisé pour obtenir davantage d'informations mais le principe de l'expérience reste identique : 
$$n+^{108}Cd\rightarrow ^{109}Cd^*\rightarrow ^{109}Cd+\gamma$$
Les résultats ont été publiés en 1956 \cite{neutrino}, et Reines a reçu en 1995 le prix Nobel de physique en leurs noms (Cowan a disparu en 1974).

\subsection{La théorie IVB leptonique}
\paragraph{L'apparition de la théorie IVB}
L'idée de cette théorie est de décrire les interactions faibles comme un \textbf{échange de bosons vecteurs massifs}. Schwinger achève la structure de cette théorie en 1957. Nous allons la présenter de manière un peu différente, puisqu'en 1957 manquaient de très nombreuses particules (comme le neutrino muonique par exemple) avec lesquelles nous sommes aujourd'hui familiers, et parce qu'il est simple est satisfaisant d'utiliser et de mettre en valeur la symétrie existant entre les trois familles de leptons $e$, $\mu$ et $\tau$. On suppose donc connus le neutrino muonique, découvert en 1962 grâce à l'étude de la désintégration du muon $$ \mu^-\rightarrow e^- + \overline{\nu}_e + \nu_\mu$$ ainsi que le lepton $\tau$ découvert par Perl (qui d'ailleurs a reçu la deuxième moitié du prix Nobel de 1995 pour cette contribution) grâce à des collisions $e^-e^+$ à l'accélérateur SLAC de Stanford \cite{tau}. Voici l'abstract de cet article : "We have found events of the form $e^++e^-\rightarrow e+\mu$+missing energy, in which no other charged particles or photons are detected. Most of these events are detected at or above a center-of-mass energy of 4 GeV. The missing-energy and missing-momentum spectra require that at least two additional particles be produced in each event. We have no conventional explanation for these events." L'explication est alors apparue comme : 
$$e^++e^-\rightarrow \tau^+ + \tau^- \rightarrow e + \mu + 4\nu $$ 
la masse du $\tau$, maintenant mesurée précisément vers 1777 GeV, explique l'échelle d'énergie de la résonance. Le neutrino tauique n'a été véritablement observé qu'en 2000 par la collaboration DONUT de Fermilab, cependant l'expérience LEP du CERN a pu montrer dès 1990 qu'il n'existait qu'au plus trois familles de neutrinos légers.\\\\
Rappelons qu'en électrodynamique quantique, le hamiltonien d'interaction s'écrit : $$\mathcal{H}_{QED}(x)=-e\overline{\psi}(x)\gamma^\alpha\psi(x)A_\alpha(x)$$qu'on peut réécrire en : $$\mathcal{H}_{QED}(x)=(-e)J^\alpha A_\alpha$$ où $$J^\alpha=\overline{\psi}(x)\gamma^\alpha\psi(x)$$ est le courant spinoriel conservé de l'équation de Dirac. Autrement dit, le hamiltonien d'interaction est le produit du courant par le champ de jauge, celui qui transporte la force, moyennant une constante de couplage $(-e)$ adimensionnée (on est dans un système d'unité adaptée ou $c=1$ et $\hbar=1$ mais la constante de structure fine, proportionnelle à $e^2$, est bien adimensionnée).
En 1957, les quelques expériences déjà réalisées, et étudiées, tendent à montrer que les courants mis en jeu dans les interactions faibles sont ceux données par la théorie de Fermi.

\paragraph{Construction de la théorie IVB}
Par analogie avec la QED a donc été proposé le hamiltonien d'interaction suivant :
$$ \mathcal{H}_I(x)=g_WJ^{\alpha\dagger}W_\alpha(x)+g_WJ^\alpha(x)W_\alpha^\dagger(x)$$
puisque interviennent dans la théorie de Fermi leptonique deux courants spinoriels faibles (et où $g_W$ est une constante de couplage adimensionnée). Cette interaction conserve les nombres leptoniques définis par : 
$$
\left( \begin{array}{c}
	N(e)=N(e^-)-N(e^+)+N(\nu_e)-N(\overline{\nu}_e) \\
	N(\mu)=N(\mu^-)-N(\mu^+)+N(\nu_\mu)-N(\overline{\nu}_\mu) \\
	N(\tau)=N(\tau^-)-N(\tau^+)+N(\nu_\tau)-N(\overline{\nu}_\tau)\\
\end{array} \right.
$$
Cette théorie est également satisfaisante car elle n'est pas invariante par symétrie de parité. En effet, considérons le courant : $$J_\alpha(x)=\sum_{l}\overline{\psi}_l(x)\gamma_\alpha(1-\gamma_5)\psi_{\nu_l}(x)$$ Sous la transformation $P:(t,\vec{x})\rightarrow(t,-\vec{x})$, $J_\alpha(x)$ n'est pas invariant car $\gamma_5$ est un pseudo-scalaire. Il en est de même pour $J_\alpha(x)^\dagger$. 

La théorie IVB a des conséquences frappantes. Supposons tout d'abord que les neutrinos ont une masse nulle. $P_L=\frac{1-\gamma_5}{2}$ est le projecteur sur les états d'hélicité négative, et pour toute particule de masse nulle, ces états sont des états propres de l'opérateur d'hélicité. Par conséquent, l'opérateur : $$\psi_{\nu_l}^L=\frac{1}{2}(1-\gamma_5)\psi_{\nu_l}$$ ne peut annihiler que des neutrinos d'hélicité négative et créer que des antineutrinos d'hélicité positive. 

Pour des particules massives, ce sont "presque" des états propres à condition qu'elles soient ultra-relativistes, donc qu'elles se comportent presque comme des particules non massives. En posant : $$\psi^L_{l}(x)=P_L\psi_l(x)$$ on peut alors réécrire : $$J_\alpha(x)=4\sum_l \overline{\psi}_l^L(x)\gamma_\alpha\psi_{\nu_l}^L(x)$$
Le spin des leptons massifs est aussi contraint par l'interaction faible.\\

Les bosons vecteurs massifs peuvent être décrits par l'équation de Proca : $$\Box W^\alpha(x)+m_W^2 W^\alpha(x)=0$$ 
La quantification du champ est menée de manière "canonique" : 
$$W^{\alpha +}(x)=\sum_k\sum_r \sqrt{\frac{1}{2V\omega_k}}\epsilon_r^\alpha(\vec{k})a_k(\vec{k})e^{-ikx}$$
$$W^{\alpha -}(x)=\sum_k\sum_r \sqrt{\frac{1}{2V\omega_k}}\epsilon_r^\alpha(\vec{k})b_k^\dagger(\vec{k})e^{ikx}$$
et conduit à un propagateur de la forme : $$iD_F^{\alpha\beta}(k,m_W)=\frac{i(-g^{\alpha\beta}+\frac{k^\alpha k^\beta}{m_W^2})}{k^2-m_W^2+i\epsilon}$$

\paragraph{Applications de la théorie IVB}
La théorie IVB permet de calculer \textbf{au premier ordre} des sections efficaces, des temps de demi-vie, ou des grandeurs reliées à ces quantités. Pour la désintégration d'une particule : $$P\rightarrow P_1' + ... + P'_N$$ l'élément de matrice s'écrit, en assignant la sommation sur les états finaux à la variable f et sur les différents leptons à la variable l (la lettre i rapporte à l'état initial) :
$$S_{fi}=\delta_{fi}+(2\pi)^4\delta^{(4)}(\sum p'_f -p)\sqrt{\frac{1}{2VE}}\prod_f\sqrt{\frac{1}{2VE'_f}}\prod_l\sqrt{2m_l}\mathcal{M}$$ 
On obtient l'expression du taux de désintégration différentiel pour le processus dont il est question où la particule $P'_1$ à un moment $\vec{p'_1}$ à $d\vec{p'_1}$ près :
$$d\Gamma=(2\pi)^4\delta^{(4)}(\sum p'_f -p)\frac{1}{2E}\prod_l\sqrt{2m_l}\prod_f\frac{d\vec{p_f'}}{(2\pi)^32E'_f}|\mathcal{M}|^2$$ 
et le taux de désintégration $\Gamma$ en intégrant sur les variables finales. Si il y a plusieurs modes de désintégration possibles, on défini le branching ratio d'un mode par : $$B=\frac{\Gamma}{\sum \Gamma}$$ Le temps de vie d'une particule est : $$\tau=\frac{1}{\sum\Gamma}=\frac{B}{\Gamma}$$
La désintégration du muon $$\mu^-(p,r)\rightarrow e^-(p',r')+\overline{\nu_e}(q_1,r_1)+\nu_\mu(q_2,r_2)$$ est décrite au premier ordre par le graphe de Feynman :

\begin{figure}[!h]
	\centering
	\includegraphics[scale=0.9]{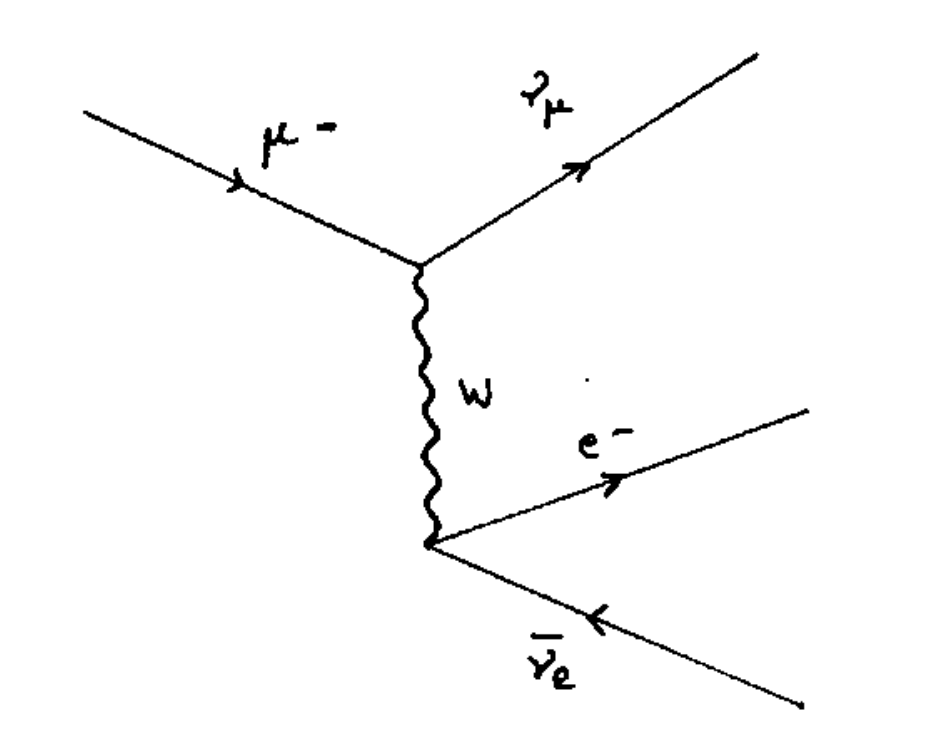}
	\caption{Diagramme de Feynman au premier ordre de la désintégration du muon dans la théorie IVB}
\end{figure}

correspondant à l'amplitude de Feynman
$$\mathcal{M}=-g_W^2[\overline{u}(\vec{p'})\gamma^\alpha(1-\gamma_5)v(\vec{q_1})]\frac{i(-g^{\alpha\beta}+\frac{k^\alpha k^\beta}{m_W^2})}{k^2-m_W^2+i\epsilon}[\overline{u}(\vec{q_2})\gamma^\beta(1-\gamma_5)u(\vec{p})]$$
Dans la limite où $m_W\rightarrow\infty$ (prendre une masse finie entraîne des corrections en $10^{-6}$), on obtient l'expression de l'amplitude de Feynman suivante (la même amplitude aurait été dérivée de la théorie de contact de Fermi) : 
$$\mathcal{M}=-\frac{iG}{\sqrt{2}}[\overline{u}(\vec{p'})\gamma^\alpha(1-\gamma_5)v(\vec{q_1})][\overline{u}(\vec{q_2})\gamma^\beta(1-\gamma_5)u(\vec{p})]$$ avec $$\frac{G}{\sqrt{2}}=(\frac{g_W}{m_W})^2$$
Le calcul donne : $$\Gamma=\frac{G^2m_\mu^5}{192\pi^3}$$
Le branching ratio de cette désintégration étant de 98,6\%, la mesure du temps de demi-vie du muon permet de déduire : $$G=(1,16637±0.00002)\times10^{-5}\ GeV^{-2}$$

C'est en fait la mesure du temps de demi-vie du muon qui a permis (historiquement) la mesure de la constante de couplage $g_W$. Après avoir fixé cette constante, la théorie IVB devient prédictive, c'est-à-dire qu'elle permet de calculer la section efficace de divers processus, comme : 

\begin{figure}[!h]
	\centering
	\includegraphics{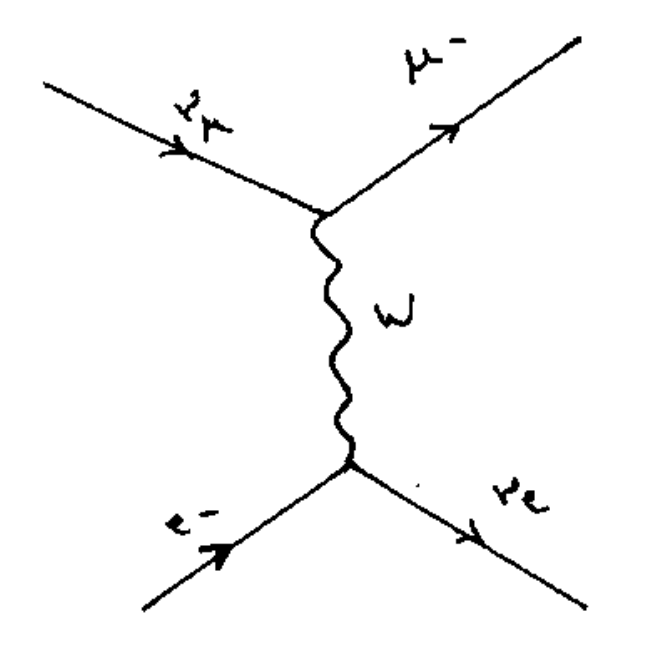}
	\caption{"Inverse muon decay" au premier ordre}
\end{figure}

dit "inverse muon decay". L'expérience est en très bon accord avec les prévisions.

\paragraph{Problèmes soulevés par la théorie IVB}
Malgré les immenses apports de la théorie IVB, certains problèmes subsistent. Par exemple, la théorie ne peut pas décrire des procédés comme $$\nu_\mu + e^- \rightarrow \nu_\mu + e^-$$ qui sont pourtant permis par les lois de conservation. En effet, les courants de la théorie IVB couplent un lepton chargé avec un lepton neutre. La contribution d'ordre le plus faible (en théorie IVB) au procédé $\nu_\mu + e^- \rightarrow \nu_\mu + e^-$ est donnée par les graphes de Feynman :

\begin{figure}[!h]
	\centering
	\includegraphics[scale=0.85]{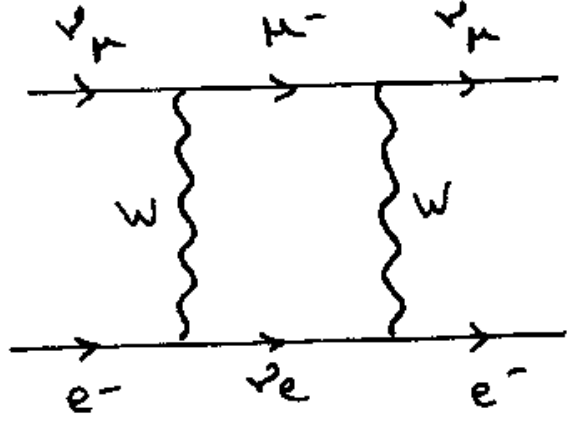}
	\caption{Première contribution au premier ordre}
\end{figure}
\begin{figure}[!h]
	\centering
	\includegraphics[scale=0.85]{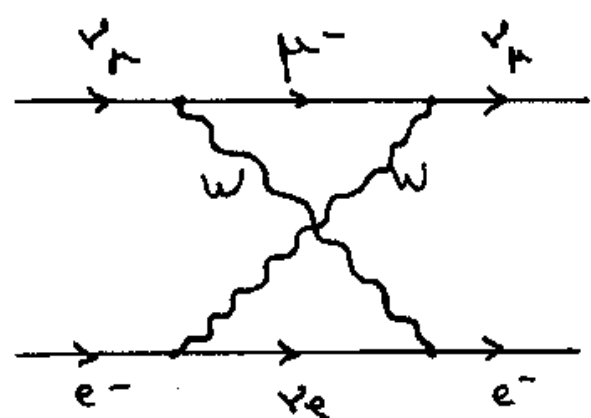}
	\caption{Deuxième contribution au premier ordre}
\end{figure}

qui impliquent le calcul d'intégrales divergentes lorsqu'on les évalue. La théorie IVB n'est pas renormalisable, et il n'est pas possible de s'affranchir de ces infinis.

De plus, les sections efficaces de la diffusion d'un neutrino électronique par des électrons et de la diffusion d'un neutrino muonique par un électron doivent être du même ordre de grandeur, ce qui n'est pas le cas dans le cadre IVB : $\nu_\mu + e^- \rightarrow \nu_\mu + e^-$ correspond à du deuxième ordre, comme s'il manquait un courant couplant des leptons neutres à des leptons neutres et des leptons chargés à des leptons de charge opposée ...
Jusqu'en 1973 cependant, toutes les expériences concordaient avec le fait que l'interaction faible soit transmise par les bosons $W^{±}$.

\subsection{La théorie de jauge électrofaible et les découvertes postérieures}
Compte tenu des courants faibles qui interviennent dans le lagrangien d'interaction de la théorie de Fermi, et grâce aux avancées de la théorie au sujet des interactions fortes, l'idée de Yang et Mills d'identifier proton et neutron à deux états du nucléon est reprise : on rassemble les leptons non neutrinoïques et leur neutrino correspondant dans un \textbf{doublet d'isospin} (symétrie $SU(2)$) si ils sont de bonne chiralité. Dans un premier temps, on suppose les leptons sans masse ce qui permet d'écrire le lagrangien de Dirac comme : 
$$\mathcal{L}=\overline{\psi}_l\slashed{\partial}\psi_l+\overline{\psi}_{\nu_l}\slashed{\partial}\psi_{\nu_l}$$
où la sommation sur les différents leptons est sous-entendue.
Puisque les leptons ne sont pas massifs, on peut projeter sur les différents états de chiralité (qui sont bien définis): 
$$\mathcal{L}=\overline{\psi}_l^L\slashed{\partial}\psi_l^L+\overline{\psi}_l^R\slashed{\partial}\psi_l^R+\overline{\psi}_{\nu_l}^L\slashed{\partial}\psi_{\nu_l}^L+\overline{\psi}_{\nu_l}^R\slashed{\partial}\psi_{\nu_l}^R$$
où les exposants $L$ et $R$ décrivent l'état de chiralité, gauche ou droit, des différents fermions.
On écrit donc : 
$$\mathcal{L}=\overline{\Psi}_l^L\slashed{\partial}\Psi_l^L+\overline{\psi}_l^R\slashed{\partial}\psi_l^R+\overline{\psi}_{\nu_l}^R\slashed{\partial}\psi_{\nu_l}^R$$
où $$\Psi_l^L(x)=\left( \begin{array}{c}
\psi_{\nu_l}^L(x) \\
\psi_l^L(x) \\
\end{array} \right) \ \ \ \ \slashed{\partial}\Psi_l^L(x)=\left( \begin{array}{c}
\slashed{\partial}\psi_{\nu_l}^L(x) \\
\slashed{\partial}\psi_l^L(x) \\
\end{array} \right)$$
Autrement dit, avant la brisure de cette symétrie $SU(2)$, \textbf{un lepton et son neutrino leptonique associé sont deux facettes d'une et une seule meme particule !}.

\begin{rmq}
	Il est nécessaire de construire la théorie de jauge en partant de fermions et de bosons non massifs. Dans le cas contraire, si par exemple on rajoute les termes correspondants à des fermions massifs : $$-m_l\overline{\psi}_l(x)\psi_l(x)$$ on perd l'invariance de jauge. On pourrait alors essayer de remplacer ce terme par quelque chose de plus compliqué, invariant de jauge mais contenant le terme précédent ; seulement on aboutit alors en général à une théorie non renormalisable. Ce ne sont donc pas des solutions acceptables. Le mécanisme de brisure spontanée de symétrie consiste alors à rajouter des champs de particule qui brisent la symétrie de jauge, au sens que le lagrangien reste invariant mais que les solutions de plus basse énergie ne sont pas $SU(2)\times U(1)$-symétriques. En réécrivant le lagrangien dans un système de coordonnées particulier, on aboutit alors à une forme qu'on peut réinterpréter comme : les $e$, $\mu$, $\tau$, $W^+$, $W^-$ et $Z^0$ ont acquis de la masse, le photon $\gamma$ reste non massif et une nouvelle particule scalaire est apparue, qu'on appelle boson BEH.
\end{rmq}

En suivant les idées déjà présentées qui lient l'invariance de jauge et l'existence de courants conservés, on aboutit à deux courants conservés à partir desquels on peut obtenir les deux courants faibles chargés de la théorie IVB, et de manière plus surprenante, à un troisième courant faible conservé, $J_W^3$, qui est la somme de deux termes, dont l'un est proportionnel au courant électromagnétique ! C'est un premier signe de l'unification des interactions faibles et électromagnétique.

L'élaboration théorique de la théorie de jauge électrofaible se termine dans les années 1968-1969, sous l'impulsion de physiciens tels que Glashow, Weinberg, Salam (qui partagent le prix Nobel en 1979), Iliopoulos ou Maiani, notamment en incluant le phénomène de brisure spontanée de la symétrie de jauge SU(2), aboutissement des travaux réalisés par Goldstone, puis Brout, Englert et Higgs. Comme cela a été montré en 1971 par Veltman et 't Hooft, la théorie que l'on obtient ainsi est par contre, et de manière satisfaisante, renormalisable.

\paragraph{Découverte des bosons de jauge électrofaibles}
La théorie prévoit donc, en plus du photon connu depuis plus d'un siècle, et des deux bosons massifs $W^+$ et $W^-$ expliquant par exemple la désintégration $\beta$ ou celle du muon, \textbf{un quatrième boson de jauge}, massif lui-aussi, et électriquement neutre, le $Z^0$.
La masse des bosons $W^+$ et $W^-$ est environ de 80 GeV, celle du $Z^0$ 91 GeV.

Dans un séminaire au CERN le 3 septembre 1973, Paul Musset de la collaboration Gargamelle, une chambre à bulle construite au CERN pour détecter les neutrinos, de 4,8 mètres de long et 2 mètres de diamètre, pesant 1000 tonnes et contenant 12 $m^3$ de fréon $CF_3Br$, présente la première preuve directe de l'existence de courants neutres, avec entre autres un événement leptonique, montrant la trajectoire d'un électron diffusé par un antineutrino arrivant de la gauche \cite{z0lepton}.

\begin{figure}[!h]
	\centering
	\includegraphics[scale=0.51]{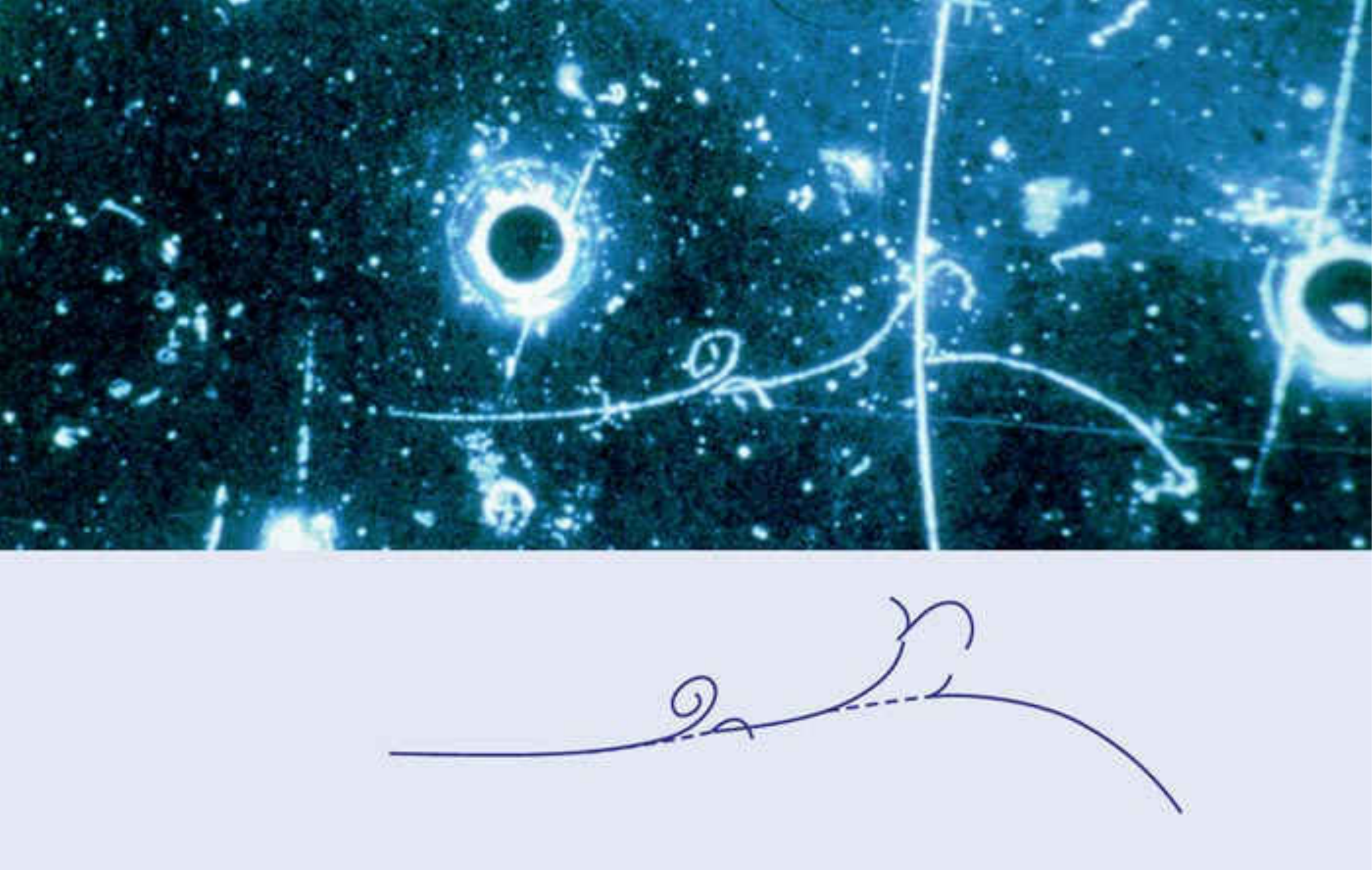}
	\caption{Photographie de la chambre à bulle Gargamelle juste après une collision ayant entrainé un courant neutre : on peut voir l'électron partir de la gauche, perdre de l'énergie (d'où son rayon de courbure plus petit) en émettant un photon (pointillés) très énergétique, qui donne une paire électron-positron (distinguables par leur sens de rotation) ; l'électron ré-émet un photon (deuxième pointillés) qui se scinde en un positron, et un électron de moindre énergie ... }
	\vspace{7mm}
	\centering
	\includegraphics[scale=0.5]{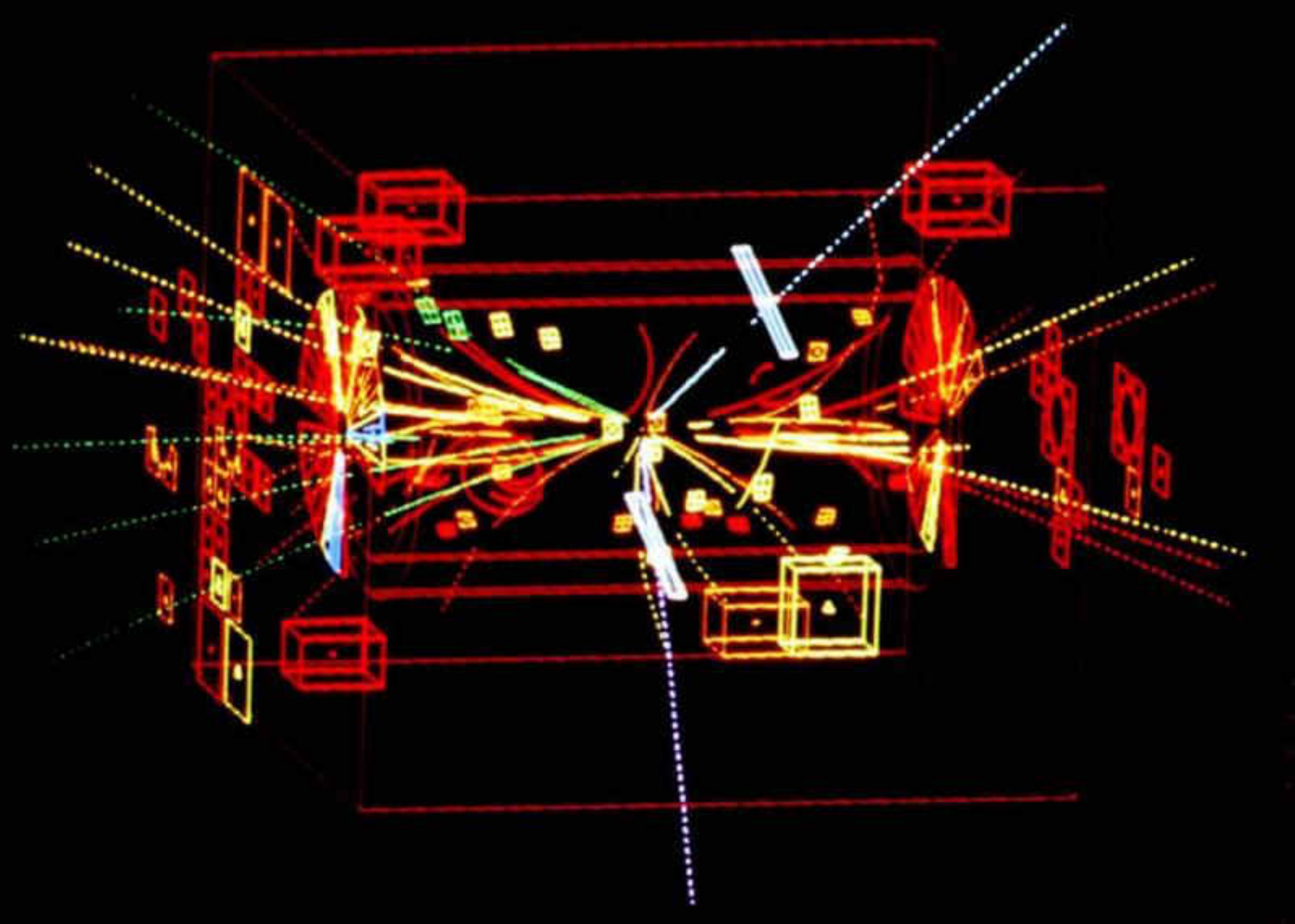}
	\caption{Reconstitution tridimensionnelle du premier évènement Z observé par UA1 le 30 avril 1983. La trajectoire des électrons est représentée par une ligne blanche}
\end{figure}

Dix ans plus tard, le 25 janvier 1983, lors d'une conférence de presse au CERN est annoncée la découverte des bosons W. Le 24 février 1983, la collaboration UA1, l'une des deux expériences sur l'anneau du SPS, super proton-synchrotron, publie un article qui décrit leurs travaux \cite{decwua1}. Le 17 mars 1983, c'est au tour de la collaboration UA2, la deuxième expérience du SPS, de publier un tel papier \cite{decwua2}.

Enfin, la même année (très fructueuse décidément), le $1^{er}$ juin 1983, les physiciens de CERN annoncent avoir observé des bosons $Z^0$, ce qui confirme le modèle de Glashow - Weinberg - Salam. L'événement ci-dessous montre la désintégration du boson $Z^0$ en un électron et un positron, tel qu'elle a été observée le 30 avril 1983. Il s'agit du premier évènement de l'Histoire d'observation directe du $Z^0$ dans un accélérateur.

\paragraph{La découverte du boson BEH}
La dernière pièce manquante du puzzle a pendant longtemps été le boson BEH, prédit par le modèle de Brout, Englert et Higgs qui permet de conférer une masse, tout d'abord aux bosons $W^+$, $W^-$ et $Z^0$, mais aussi également aux leptons, $e$, $\mu$, $\tau$ tout en gardant la caractère renormalisable de la théorie. 
Le boson BEH a finalement été observé dans l'accélérateur le plus récent du CERN, le LHC, qui a commencé à fonctionner en 2008. Cela achève en quelque sorte la période de quête des "particules fondamentales au cœur du modèle standard".\\\\

\begin{figure}[!h]
	\centering
	\includegraphics[scale=0.30]{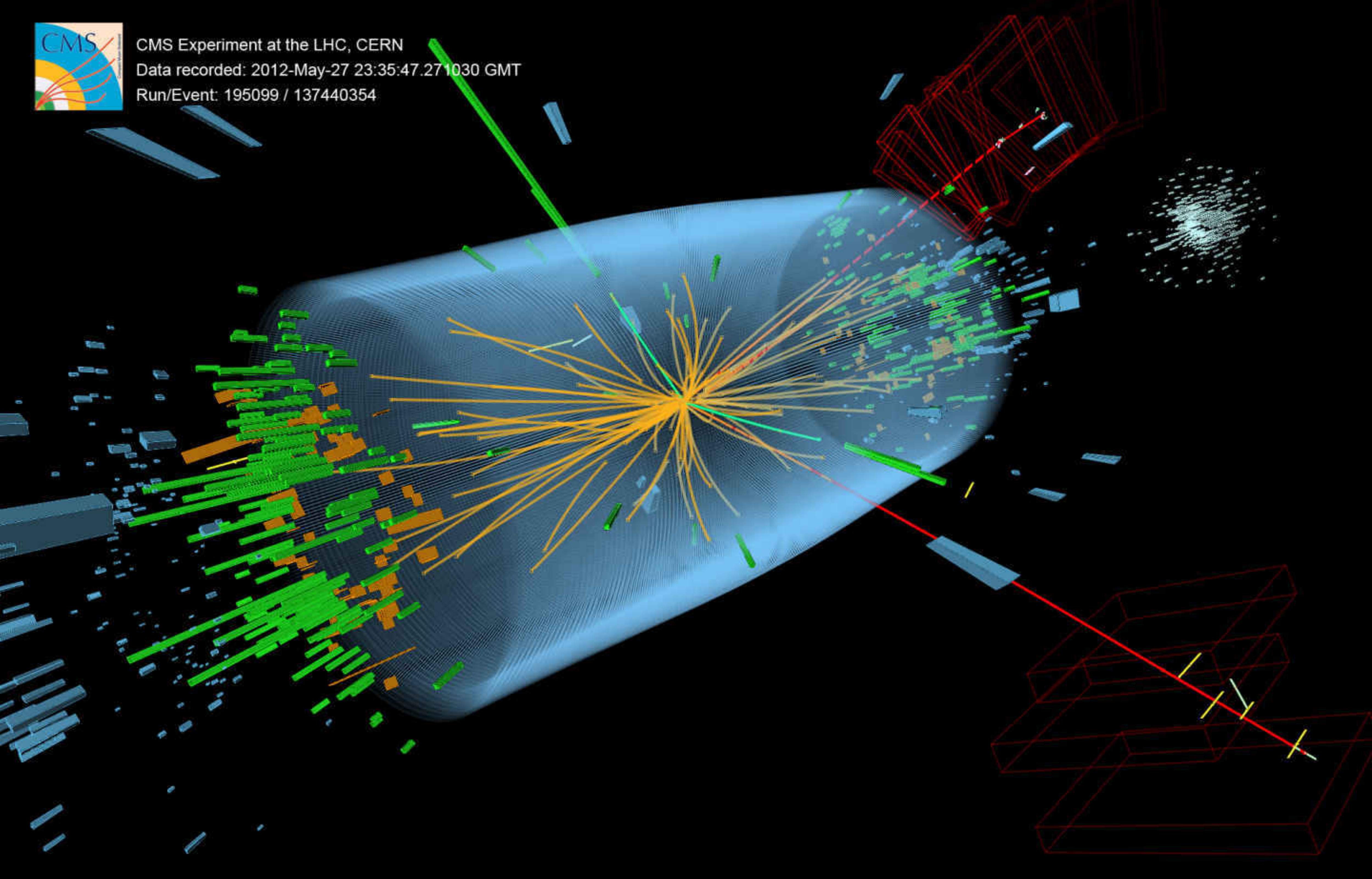}
	\caption{Reconstitution tridimensionnelle d'un évènement candidat Higgs $\rightarrow 4$ leptons, observé en 2012 dans le détecteur CMS. Les deux lignes rouges symbolisent les trajectoires des muons, les deux lignes vertes celles des électrons. Chacune des paires de leptons provient d'un boson $Z^0$, en réalité l'évènement candidat est : $H\rightarrow Z^0+Z^0\rightarrow e^+e^-+\mu^+\mu^-$. Les parallélépipèdes représentent des zones des différents calorimètres, et leur longueur est proportionnelle a la quantité d'énergie qui s'y est déposée.}
\end{figure}

\section{Un approche mathématique de la théorie de jauge électrofaible}
\subsection{Structure de fibré principal et représentations}
Nous allons construire la théorie des interactions faibles pour les leptons est les quarks, mais sans considérer la dynamique forte de ces derniers.\\\\

Les observations expérimentales ont montré qu'on pouvait classer les particules en familles caractérisées par une quantité positive entière ou demi-entière $J$ qui a été appelée \textbf{isospin faible total}, chacune des particules de cette famille possédant de plus une charge supplémentaire, appelée encore \textbf{isospin faible} et notée $I_3^W$. Ces deux invariants ne sont pas sans relation : si pour une particule le $J$ est fixé, $I_3^W$ ne peut prendre que les valeurs comprises entre $J$ et $-J$ en "sautant" de un en un.

\begin{rmq}
	Par exemple, si une particule appartient à une famille d'isospin total $\frac{1}{2}$, alors l'isospin de cette particule ne peut valoir que $\frac{1}{2}$ ou $-\frac{1}{2}$. Dans le cas des leptons électroniques, les neutrinos électroniques et les électrons forment une famille de particules d'isospin faible total $\frac{1}{2}$ ; pour les neutrinos, $(I_3^W)_{\nu_e}=\frac{1}{2}$, et pour les électrons, $(I_3^W)_{e}=-\frac{1}{2}$.
\end{rmq}

En plus de l'isospin, une autre quantité $Y$ apparait comme conservée par les interactions faibles : \textbf{l'hypercharge faible}, pour laquelle on n'observe pas de comportement grégaire. Il existe une relation entre la charge électrique $Q$ (en unités de charge élémentaire) et ces deux charges faibles $I_3^W$ et $Y$ d'une particule : $$Q=I_3^W+Y$$ 
C'est un premier signe que dans la théorie unifiée que nous construisons, les interactions faibles et électromagnétiques sont liées.

L'existence de ces charges conservées donne "l'indice" de la structure de jauge sous-jacente, comme discuté au chapitre précédent. 

Soit donc un fibré principal $P=P(M,G)$ avec $G=SU(2)\times U(1)$ et $M$ la variété d'espace-temps (L'algèbre de Lie $\g$ est de dimension $4$). Pour des raisons que nous expliquerons plus tard, on doit supposer que les leptons et les quarks non massifs pour que la théorie reste invariante de jauge et renormalisable.
Le groupe $SU(2)$ correspond à la symétrie d'isospin faible (on retrouve une structure, dans la classification des états vérifiant cette symétrie, proche de celle qu'on obtient en classifiant les différents états de spin en mécanique quantique, ce qui explique aussi le vocabulaire, bien qu'en réalité ce soit le concept d'isospin fort qui a été proposé antérieurement), et $U(1)$ est le groupe de symétrie d'hypercharge faible (qui ressemble beaucoup à la charge électrique).

\paragraph{Représentations de $SU(2)\times U(1)$}
Expérimentalement encore, on a observé depuis la théorie de Fermi de la désintégration $\beta$ que l'hélicité des particules n'est pas tout à fait indépendante de la manière dont elles interagissent. C'est une autre raison pour supposer les champs de particules sont initialement de masse nulle, pour que leur hélicité soit bien définie.\\
Si $\psi(x)$ est un champ de Dirac non massif (et vérifie donc le principe de moindre action pour le lagrangien de Dirac $\mathfrak{L}=\overline{\psi}\slashed{\partial}\psi$), la partie gauche du champ est $$\psi^L=\frac{1}{2}(1-\gamma_5)\psi$$ et la partie droite $$\psi^R=\frac{1}{2}(1+\gamma_5)\psi$$

Rassemblons les champs de Dirac gauches en doublets d'isospin : 
$$x\in M\rightarrow \begin{pmatrix}
\psi_{\nu_l}^L(x)\\
\psi_l^L(x)
\end{pmatrix}
$$ où $l\in\{e,\mu,\tau\}$ pour les leptons, et 
$$x\in M\rightarrow\begin{pmatrix}
\psi_u^L(x)\\
\psi_{d'}^L(x)
\end{pmatrix}\ ;\ 
x\in M\rightarrow\begin{pmatrix}
\psi_c^L(x)\\
\psi_{s'}^L(x)
\end{pmatrix}\ ;\ 
x\in M\rightarrow\begin{pmatrix}
\psi_t^L(x)\\
\psi_{b'}^L(x)
\end{pmatrix}
$$ pour les quarks (les ' sont des superpositions des états down, strange et bottom).
Ces fonctions sont donc à valeurs dans $\C^4\oplus\C^4$.\\ Construisons une représentation de $SU(2)\times U(1)$ pour que ces fonctions deviennent des sections du fibré associé par cette représentation à $P=P(M,SU(2)\times U(1))$.\\\\
Considérons tout d'abord les doublets d'isospin décrivant les leptons gauches. Soit $(g,h)\in SU(2)\times U(1)$. On prend comme représentation : 
$$\rho:(g,h)\mapsto\rho_{standard}(h)\cdot\rho_{standard}(g)\otimes\begin{pmatrix}
1 & 0 & 0 & 0 \\
0 & 1 & 0 & 0 \\
0 & 0 & 1 & 0 \\
0 & 0 & 0 & 1 
\end{pmatrix}=\rho_{standard}(h)\cdot\rho_{standard}(g)\otimes id_4
$$
(On choisit la représentation de $U(1)$ sous la forme $e^{i\alpha}$ comme en électromagnétisme et $SU(2)$ sous la forme standard de matrices complexes $2\times2$.) Dans la suite nous omettrons de rajouter le $\otimes id_4$ quand il n'y a pas d'ambiguïté.
On prolonge cette représentation à l'algèbre de Lie $\g$, dont une base est alors donnée par : $\frac{i}{2}\tau_1\ \ \frac{i}{2}\tau_2\ \ \frac{i}{2}\tau_3\ \ \frac{i}{2}id_2$ où les $\tau_i$ sont les matrices de Pauli.
$$\tau_1=\begin{pmatrix}
0 & 1\\
1 & 0
\end{pmatrix}\ ;\ \tau_2=\begin{pmatrix}
0 & -i\\
i & 0
\end{pmatrix}\ ;\ \tau_3=\begin{pmatrix}
1 & 0 \\
0 & -1
\end{pmatrix}
$$
Le fibré principal admet une connexion $\omega$, et l'expression des ses coefficients dans ce fibré associé est, dans le cas le plus général : 
$$A^i_{j\mu}\otimes id_4=i\frac{g'}{2}(id_2)^i_jB_\mu\otimes id_4+i\frac{g}{2}(\tau_k)_j^iW^k_\mu\otimes id_4$$ avec les coefficients $B_\mu$ et $W^k_\mu$ réels, pour $\mu\in[|0,3|]$ et $k\in[|1,3|]$, et où $g$ et $g'$ sont des constantes de couplage qui ne font que \textbf{fixer l'échelle relative des interactions faibles et électromagnétiques}. Remarquons que \textbf{la seule quantité vraiment physique est le rapport entre ces deux constantes}.\\

La dérivée covariante s'écrit : 
$$\partial^H_{L\mu}=\partial_\mu\otimes id_8 +i\frac{g'}{2}(id_2)B_\mu\otimes id_4+i\frac{g}{2}(\tau_k)W^k_\mu\otimes id_4$$
Par conséquent, le lagrangien des spineurs $\psi_l^L$ où $l$ est une étiquette de lepton est : 
$$\mathfrak{L}^L=\overline{\Psi_l^L}\slashed{\partial}_L^H\Psi_l^L$$ où la sommation sur $l$ est sous-entendue, où : 
$$\slashed{\partial}_L^H=\partial_{L\mu}^H\gamma^\mu\otimes id_2$$ est la notation slash de Feynman, et où on a noté $$\Psi_l^L=\begin{pmatrix}
\psi_{\nu_l}^L(x)\\
\psi_l^L(x)
\end{pmatrix}$$ les bi-spineurs doublets d'isospin évoqués plus haut.\\\\ 

Pour les spineurs droits qui décrivent les $e,\mu,\tau$, $\psi_{l}^R:M\rightarrow \C^4$ avec $l\in\{e,\mu,\tau\}$, d'isospin nul, on représente $SU(2)\times U(1)$ par : $$\rho :(g,e^{i\alpha})\in SU(2)\times U(1)\mapsto e^{2i\alpha}\otimes id_4$$
donc l'expression des coefficients de connexion est donnée par :
$$A_{\mu}\otimes id_4=ig'B_\mu\otimes id_4$$
si bien que la dérivée covariante s'écrit : $$\partial_{l\mu}^H=\partial_\mu\otimes id_4 + ig'B_\mu\otimes id_4$$
et la partie du lagrangien relative aux spineurs droits est : 
$$\mathfrak{L}^R_l=\overline{\psi_l^R}\slashed{\partial}_l^H\psi_l^R$$

Pour les spineurs droits qui décrivent les $\nu_e,\nu_\mu,\nu_\tau$, $\psi_{\nu_l}^R:M\rightarrow \C^4$ avec $l\in\{e,\mu,\tau\}$, d'isospin nul, on représente $SU(2)\times U(1)$ par la représentation triviale : $$\rho :(g,h)\in SU(2)\times U(1)\mapsto id_4$$
si bien que la dérivée covariante s'écrit : $$\partial_{\nu_l\mu}^H=\partial_\mu\otimes id_4$$
et la partie du lagrangien relative aux spineurs droits est : 
$$\mathfrak{L}^R_{\nu_l}=\overline{\psi_{\nu_l}^R}\slashed{\partial}_{\nu_l}^H\psi_{\nu_l}^R$$

On rassemble les quarks gauches dans les doublets d'isospin évoqués plus haut, notés : $$\Psi_q^L=\begin{pmatrix}
\psi_{q1}^L(x)\\
\psi_{q2}^L(x)
\end{pmatrix}$$
où le $q$ indice les différentes familles de quarks. De la même façon que pour les leptons, on choisit des représentations de $SU(2)\times U(1)$ de telle manière que la dérivée covariante s'écrive  : $$\partial_{q\mu}^H=\partial_\mu\otimes id_4 + \frac{ig}{2}(\tau_k)W_\mu\otimes id_4-\frac{ig'}{6}B_\mu\otimes id_4$$
et on note $$\slashed{\partial}^H_{q}\Psi_q^L=\begin{pmatrix}
\slashed{\partial}^H_{q}\psi_{q1}^L(x)\\
\slashed{\partial}^H_{q}\psi_{q2}^L(x)
\end{pmatrix}$$

Pour les quarks droits : $$\partial_{1\mu}^H=\partial_\mu\otimes id_4-\frac{2ig'}{3}B_\mu\otimes id_4$$ et $$\partial_{2\mu}^H=\partial_\mu\otimes id_4+\frac{ig'}{3}B_\mu\otimes id_4$$

Finalement, le lagrangien total s'écrit : 
$$\mathfrak{L}=\overline{\Psi_q^L}\slashed{\partial}^H_L\Psi_l^L+\overline{\Psi_l^L}\slashed{\partial}^H_q\Psi_q^L+\overline{\psi_l^R}\slashed{\partial}_l^R\psi_l^R+\overline{\psi_{\nu_l}^R}\slashed{\partial}_l^R\psi_{\nu_l}^R+\overline{\psi_{q1}^R}\slashed{\partial}_1^R\psi_{q1}^R+\overline{\psi_{q2}^R}\slashed{\partial}_2^R\psi_{q2}^R$$

Cependant pour que les équations des champs de jauge soient également vérifiées, il faut rajouter à cette densité d'action les termes de self-action des champs de jauge.

\subsection{Quantités conservées}
Le lagrangien $\mathfrak{L}$ est une fonction des champs et de leur dérivées covariantes : 
$$\mathfrak{L}=\mathfrak{L}(\Psi_i^L,\Psi^L_{i;\mu},\psi_i^R, \psi^R_{i;\mu})$$
Pour ne pas alourdir les notations, nous avons écrit de la même manière les dérivées covariantes des spineurs gauches et droits, mais il ne faut pas perdre de vue qu'elles ne sont pas identiques.
$$\delta\mathfrak{L}=\frac{\partial \mathfrak{L}}{\partial \psi^R_i}\delta\psi^R_i+\frac{\partial \mathfrak{L}}{\partial \Psi^L_i}\delta\Psi^L_i+\frac{\partial \mathfrak{L}}{\partial \psi^R_{i;\mu}}\delta\psi^R_{i;\mu}+\frac{\partial \mathfrak{L}}{\partial \Psi^L_{i;\mu}}\delta\Psi^L_{i;\mu}$$
Appliquons au système une transformation d'isospin "pure" telle que les champs sont transformés selon : 
\[
\left\{
\begin{array}{r c l}
\Psi_l^{'L}(x) &=& U(\vec{\xi(x)})\Psi_l^L = exp(i\frac{\xi_\alpha(x)}{2}\tau_\alpha)\Psi_l^L\\
\overline{\Psi_l^{'L}}(x) &=& \overline{\Psi_l^L}(x)U^\dagger(\vec{\xi(x)})=\overline{\Psi_l^L}(x)exp(-i\frac{\xi_\alpha(x)}{2}\tau_\alpha)\\
\Psi_q^{'L}(x) &=& U(\vec{\xi(x)})\Psi_q^L = exp(i\frac{\xi_\alpha(x)}{2}\tau_\alpha)\Psi_q^L\\
\overline{\Psi_q^{'L}}(x) &=& \overline{\Psi_q^L}(x)U^\dagger(\vec{\xi(x)})=\overline{\Psi_q^L}(x)exp(-i\frac{\xi_\alpha(x)}{2}\tau_\alpha)\\
\psi_l^{'R} &=& \psi_l^R\\
\overline{\psi_l^{'R}} &=& \overline{\psi_l^R}\\
\psi_{\nu_l}^{'R} &=& \psi_{\nu_l}^R\\
\overline{\psi_{\nu_l}^{'R}} &=& \overline{\psi_{\nu_l}^R}\\
\psi_{q1}^{'R} &=& \psi_{q1}^R\\
\overline{\psi_{q1}^{'R}} &=& \overline{\psi_{q1}^R}\\
\psi_{q2}^{'R} &=& \psi_{q2}^R\\
\overline{\psi_{q2}^{'R}} &=& \overline{\psi_{q2}^R}\\
\end{array}
\right.
\]
qui correspond à un changement de section locale, donné par : 
$$g:x\in M\mapsto exp(-i\frac{\xi_\alpha(x)}{2}\tau_\alpha)$$
lors duquel la matrice de connexion varie de : 
$$A'=exp(i\frac{\xi_\alpha}{2}\tau_\alpha)Aexp(-i\frac{\xi_\alpha}{2}\tau_\alpha)+exp(i\frac{\xi_\alpha}{2}\tau_\alpha)d(exp(-i\frac{\xi_\alpha}{2}\tau_\alpha))$$
Presque par définition de la dérivée covariante, on a que :
\[
\left\{
\begin{array}{r c l}
\Psi_{i;\mu}^{'L}(x) &=& U(\vec{\xi})\Psi_{i;\mu}^L = exp(i\frac{\xi_\alpha}{2}\tau_\alpha)\Psi_{i;\mu}^L\\
\overline{\Psi_{i;\mu}^{'L}}(x) &=& \overline{\Psi_{i;\mu}^L}(x)U^\dagger(\vec{\xi})=\overline{\Psi_{i;\mu}^L}(x)exp(-i\frac{\xi_\alpha}{2}\tau_\alpha)\\
\psi_{i;\mu}^{'R} &=& \psi_{i;\mu}^R\\
\overline{\psi_{i;\mu}^{'R}} &=& \overline{\psi_{i;\mu}^R}\\
\end{array}
\right.
\]
ce qui entraîne que le lagrangien est invariant ($\delta\mathfrak{L}=0$). Or : 
$$\delta\mathfrak{L}=\frac{\partial \mathfrak{L}}{\partial \Psi^L_i}\delta\Psi^L_i+\frac{\partial \mathfrak{L}}{\partial \Psi^L_{i;\mu}}\delta\Psi^L_{i;\mu}$$
D'après l'équation d'Euler-Lagrange dont on requiert que les champs $\Psi_l^L$ sont des solutions, 
$$\forall i,\ \ \frac{\partial \mathfrak{L}}{\partial \Psi^L_i}-\partial_\mu^H(\frac{\partial \mathfrak{L}}{\partial \Psi^L_{i;\mu}})=0$$ l'équation précédente se réécrit en : 
$$\partial_\mu^H(\frac{\partial \mathfrak{L}}{\partial \Psi^L_{i;\mu}}\delta\Psi_i^L)=0$$ où la sommation sur $i$ est comme toujours, sous-entendue.
Pour la transformation considérée, en développant les exponentielles au voisinage de $\vec{\xi}=0$ : 
$$U(\vec{\xi})\Psi_{i}^L = \Psi_{i}^L+i\frac{\xi_\alpha}{2}\tau_\alpha\Psi_{i}^L+o(|\vec{\xi}|^2)$$ donc : 
$$\partial_\mu^H(\overline{\Psi_i^L}\gamma^\mu\frac{\xi_\alpha}{2}\tau_\alpha\Psi_{i}^L)=0$$
d'où trois courants conservés : 
$$\forall j\in[|1,3|],\ \ \ J_j^\mu=\frac{1}{2}\overline{\Psi_i^L}\gamma^\mu\tau_j\Psi_i^L$$
correspondant à des charges : 
$$\forall j\in[|1,3|],\ \ \ I_j^W=\int d^3\vec{x}J_j^0(x)=\frac{1}{2}\int d^3\vec{x}\Psi_i^{L\dagger}(x)\tau_j\Psi_i^L(x)$$
\begin{rmq}
	On retrouve les courants leptoniques $J^j$ et $J^{\dagger j}$ de la théorie IVB en faisant les transformations : 
	\[
	\left\{
	\begin{array}{r c l}
	J^j &=& 2(J_1^j-iJ_2^j) = \overline{\psi_l}\gamma^j(1-\gamma_5)\psi_{\nu_l}(x)\\
	J^{\dagger j} &=& 2(J_1^j+iJ_2^j) = \overline{\psi_{\nu_l}}\gamma^j(1-\gamma_5)\psi_{l}(x)
	\end{array}
	\right.
	\]
	et pour les quarks : 
	\[
	\left\{
	\begin{array}{r c l}
	J^j &=& 2(J_1^j-iJ_2^j) = \overline{\psi_2}\gamma^j(1-\gamma_5)\psi_{1}(x)\\
	J^{\dagger j} &=& 2(J_1^j+iJ_2^j) = \overline{\psi_{1}}\gamma^j(1-\gamma_5)\psi_{2}(x)
	\end{array}
	\right.
	\]
	Ce modèle prévoit même l'existence du courant 
	$$J_3^i=\frac{1}{2}\overline{\Psi_l^L}\gamma^i\tau_3\Psi_l^L=\frac{1}{2}[\overline{\psi_{1}^L}(x)\gamma^i\psi_{1}(x)-\overline{\psi^L_2}(x)\gamma^i\psi_2^L(x)]$$ qui est le \textbf{courant d'isospin faible}. C'est un courant neutre puisqu'il couple des particules de même charge (électrique) entre elles, au même titre que le courant électromagnétique : 
	$$s^i=-e\overline{\psi_l}(x)\gamma^i\psi_l(x)$$
	Le terme de droite est à un facteur de proportionnalité près (il faut changer l'échelle), le courant électromagnétique. On voit donc se profiler le fait que dans cette théorie, les interactions faibles et électromagnétiques seront reliées.
	\end{rmq}
On définit le courant d'hypercharge faible, pour les leptons :
$$J^i_Y(x)=\frac{s^i(x)}{e}-J^i_3(x)=-\frac{1}{2}\overline{\Psi_l^L}(x)\gamma^i\Psi_l^L(x)-\overline{\psi_l^R}(x)\gamma^i\psi_l^R(x)$$ ainsi que pour les quarks, 
$$J^i_Y(x)=\frac{s^i(x)}{e}-J^i_3(x)=\frac{1}{6}\overline{\psi_u^L}(x)\gamma^i\psi_u^L(x)+\frac{2}{3}\overline{\psi_u^R}(x)\gamma^i\psi_u^R(x)+\frac{1}{6}\overline{\psi_{d'}^L}(x)\gamma^i\psi_{d'}^L(x)-\frac{1}{3}\overline{\psi_{d'}^R}(x)\gamma^i\psi_{d'}^R(x)$$ et l'hypercharge faible : 
$$Y=\int d^3\vec{x}J^i_Y(x)$$ ce qui permet d'écrire la relation : 
$$Y=\frac{Q}{e}-I_3^W$$
On a les valeurs suivantes des charges électrofaibles pour les particules considérées, obtenues en utilisant la définition générale des charges et le fait que les états sont normalisés.
\begin{center}
	\begin{tabular}{|c|c|c|c|}
		\hline
		Particule & Charge électrique & $I_3^W$ & Y \\
		\hline
		\hline
		$lepton^{L,-}$ & -1 & $-\frac{1}{2}$ & $-\frac{1}{2}$\\
		$neutrino^L$ & 0 & $+\frac{1}{2}$ & $-\frac{1}{2}$\\
		$lepton^{R,-}$ & -1 & 0 & $-1$\\
		$neutrino^R$ & 0 & 0 & 0\\
		\hline
		$quark, up-type^L$ & $\frac{2}{3}$ & $\frac{1}{2}$ & $\frac{1}{6}$\\
		$quark, down-type^L$ & -$\frac{1}{3}$ & $-\frac{1}{2}$ & $\frac{1}{6}$\\
		$quark, up-type^R$ & $\frac{2}{3}$ & 0 & $\frac{2}{3}$\\
		$quark, down-type^R$ & -$\frac{1}{3}$ & 0 & $-\frac{1}{3}$\\
		\hline
	\end{tabular}
\end{center}

Lorsqu'on change de jauge pour le groupe $U(1)$, on a les transformations suivantes : 
\[
\left\{
\begin{array}{r c l}
\Psi_{l}^{L}(x) &\rightarrow& \Psi_{l}^{'L}(x)=e^{-i\kappa/2}\Psi_{l}^{L}(x)\\
\psi_{l}^{R} &\rightarrow& e^{-i\kappa}\psi_{l}^{R}\\
\psi_{\nu_l}^{R} &\rightarrow& \psi_{\nu_l}^{R}\\
\end{array}
\right.
\]
et de manière générale, si les quantons annihilés par $\psi$ sont d'hypercharge $Y$, on doit avoir : 
$$\psi(x)\rightarrow \psi'(x)=e^{i\kappa Y}\psi(x)$$

\begin{rmq}
	Les champs de particules appartenant à une famille de même isospin total et de même hypercharge sont des sections d'un même fibré associé puisque ils varient de la même manière lors d'un changement de jauge, c'est-à-dire qu'il sont sous la même représentation du groupe $G$.
\end{rmq}

\subsection{Le lagrangien des champs de jauge}
D'après le chapitre précédent, et le fait que la "bonne" densité de self-action des champs de jauge est donnée par : 
$$-\frac{1}{4}F_{\mu\nu}F^{\mu\nu}$$ où F est la courbure associée aux connexions que sont les potentiels de jauge : 
$$F_{\mu\nu}^\gamma=\partial_\mu A_\nu^\gamma -\partial_\nu A_\mu^\gamma +A_\mu^\alpha A_\nu^\beta f_{\alpha\beta}^\gamma$$
Compte tenu des différentes définitions des champs de jauge (à un coefficient près notamment), l'expression de la courbure pour les trois degrés de liberté de la connexion sur $SU(2)$ devient : 
$$F_{\mu\nu}^\gamma=\partial_\mu W_\nu^\gamma -\partial_\nu W_\mu^\gamma +g\epsilon _{\alpha\beta}^\gamma W_\mu^\alpha W_\nu^\beta$$ (pour $\gamma\in[|1,3|]$)
et celle de la courbure pour la connexion sur $U(1)$ : 
$$C_{\mu\nu}=\partial_\mu B_\nu -\partial_\nu B_\mu $$ comme en électromagnétisme.
Ainsi, le lagrangien total de la théorie s'écrit pour l'instant : 
$$\mathfrak{L}=-\frac{1}{4}C_{\mu\nu}C^{\mu\nu}-\frac{1}{4}F_{\mu\nu}^\gamma F^{\gamma\mu\nu}+\overline{\Psi_l^L}(x)\slashed{\partial}_L^H\Psi_l^L+\overline{\Psi_q^L}(x)\slashed{\partial}_q^H\Psi_q^L$$ $$+\overline{\psi_{\nu_l}^R}(x)\slashed{\partial}_{\nu_l}^H\psi_{\nu_l}^R(x)+\overline{\psi_l^R}(x)\slashed{\partial}_l^H\psi_l^R(x)+\overline{\psi_{q1}^R}(x)\slashed{\partial}_{q1}^H\psi_{q1}^R(x)+\overline{\psi_{q2}^R}(x)\slashed{\partial}_{q2}^H\psi_{q2}^R(x)$$ avec $$\Psi_l^L=\begin{pmatrix}
\psi_{\nu_l}^L(x)\\
\psi_l^L(x)
\end{pmatrix}\ \ \ \Psi_q^L=\begin{pmatrix}
\psi_{q1}^L(x)\\
\psi_{q2}^L(x)
\end{pmatrix}
$$

\subsection{Le problème de la masse}
Jusqu'à maintenant, on a supposé que tous les champs de particules étaient non massifs. Lors de la quantification, les champs de jauge que nous avons construits donnent aussi des quantas non massifs. Cependant, cette situation n'est pas satisfaisante, d'une part parce qu'une explication de la courte portée des interactions faibles est que les bosons vecteurs sont massifs, et d'autre part parce que les leptons et en particulier l'électron et le muon, avaient déjà été observés au moment de l'élaboration de la théorie et on avait pu \textbf{mesurer leur masse}, une masse non nulle.\\\\
Observons les conséquences qu'ont les transformations naïves et sans vraie explication (mathématiquement, vis-à-vis de la théorie établie précédemment) du lagrangien qui visent à donner artificiellement de la masse aux particules évoquées.

\paragraph{La masse des bosons}
Si on rajoute dans le lagrangien un terme en 
$$m_W^2W_\mu^\dagger(x)W^\mu(x)$$ on obtient la densité d'action de la théorie IVB, qui pose essentiellement deux problèmes : ce terme n'est pas invariant par symétrie de jauge $SU(2)\times U(1)$, \textbf{il est mal défini au sens de la jauge} ce qui conceptuellement est difficilement acceptable (aujourd'hui avec du recul et cette approche de la théorie). De plus, on retrouve les problèmes de non-renormalisabilité de la théorie IVB...

\paragraph{La masse des fermions}
Pour les fermions, les problème est le même. Si on rajoute dans le lagrangien un terme de la forme $$-m_l\overline{\psi_l}(x)\psi_l(x)$$ on perd l'invariance de jauge $SU(2)\times U(1)$, car : 
$$-m_l\overline{\psi_l}(x)\psi_l(x)=-m_l\overline{\psi_l}(x)[\frac{1}{2}(1+\gamma_5)+\frac{1}{2}(1-\gamma_5)]\psi_l(x)=-m_l[\overline{\psi_l}(x)^L\psi_l^R(x)+\overline{\psi_l}(x)^R\psi_l(x)^L]$$
et les champs gauches sont des isospineurs tandis que les champs droits sont des isoscalaires.\\\\

Le modèle de Higgs permet de résoudre ce problème : on suppose qu'un autre champ existe, de valeur moyenne dans le vide non nulle (c'est donc forcément un champ scalaire) : le champ de Higgs. Ce champ est une section d'un fibré associé de fibre $\C\oplus\C$ puisqu'il doit être un doublet d'isospin. Quoiqu'il en soit, l'introduction de cet objet rajoute des termes dans le lagrangien, qui sont les termes ci-dessus nécessaires pour pouvoir interpréter les particules comme massives, plus d'autres termes qui assurent la conservation globale des symétries de jauge, et la renormalisabilité de la théorie.

\section{Le mécanisme de Higgs et l'apparition de la masse}
\subsection{Brisure spontanée de symétrie}
Considérons un système de lagrangien $\mathcal{L}$ possédant une certaine symétrie. Si il n'y a qu'un seul état d'énergie minimale, alors on peut montrer qu'il est nécessairement invariant par la même symétrie. Cependant, si il y a plusieurs états d'énergie minimale, alors en choisissant \textbf{un état} d'énergie minimale, on brise la symétrie du système.\\
Pour fixer les idées, considérons une particule ponctuelle se déplaçant dans le plan complexe, et soumise au potentiel :
$$V(z)=\mu^2|z|^2+\lambda|z|^4$$
où $\lambda>0$ (parce que sinon le potentiel tend vers $-\infty$ en l'infini ou alors il est juste quadratique).
Dans ce potentiel, concernant les états d'énergie minimale, deux cas sont envisageables : 
\begin{enumerate}
	\item Si $\mu^2>0$, le potentiel admet seulement 0 comme point d'énergie minimale, et cet état respecte la symétrie du système par rotation autour de l'axe vertical passant par l'origine du plan.
	\item Si $\mu^2>0$, les états d'énergie extrémale sont donnés par : $$2\mu^2|z|+4\lambda|z|^3=0 \Leftrightarrow z=0\ ;\ |z|=\sqrt{-\frac{\mu^2}{2\lambda}}$$ si bien que les états d'énergie minimale forment un cercle centré en 0. C'est le potentiel célèbre en forme de "chapeau mexicain", de brisure spontanée de symétrie.
\end{enumerate}
Cette idée provient de l'étude des matériaux ferromagnétiques dont l'aimantation, à température suffisamment basse, brise la symétrie du système (aimantation spontanée, non nulle sans champ).\\
C'est Yoichiro Nambu (18/01/1921-05/07/2015) qui a introduit ces idées en physique des particules. Il a reçu le prix Nobel de physique en 2008 pour cette contribution majeure.

\subsection{Le mécanisme de Higgs}
Considérons une théorie de jauge, de groupe de jauge U(1), symétrie que nous allons briser. Pour pouvoir briser la symétrie, il faut qu'on soit dans les cas 2. du paragraphe précédent. Considérons donc un fibré principal $P=P(M,U(1))$ où $M$ est la variété d'espace-temps. Supposons que $U(1)$ agit sur $\C\equiv\R^2$ via la représentation : $$e^{i\alpha}\rightarrow e^{i\alpha}
\begin{pmatrix}
1 & 0 \\
0 & 1
\end{pmatrix}
$$
Soit $\phi$ un champ de particules complexe, de spin nul, qui est donc une section du fibré associé $P\times_\rho \R^2$, et vérifie l'équation de Klein-Gordon. Le lagrangien total du système s'écrit donc :
$$\mathcal{L}=-\frac{1}{4}F_{\mu\nu}F^{\mu\nu}+((\partial^H)_\mu\phi)^*((\partial^{H})^\mu\phi)-V(\phi)$$
où $\partial^H_\mu=\partial_\mu+iA_\mu$ est la dérivée covariante, $iF_{\mu\nu}=\partial_\mu A_\nu-\partial_\nu A_\mu$ est la courbure associée à la connexion $iA_\mu$, et où :
$$V(\phi)=m^2\phi^*\phi+\lambda(\phi^*\phi)^2$$
Si $m^2<0$ et $\lambda>0$, $V$ est minimale pour $$|\phi|=\sqrt{\frac{-m^2}{\lambda}}=\frac{v}{\sqrt{2}}$$
c'est le potentiel en forme de chapeau mexicain.

\begin{figure}[!h]
	\centering
	\includegraphics[scale=0.6]{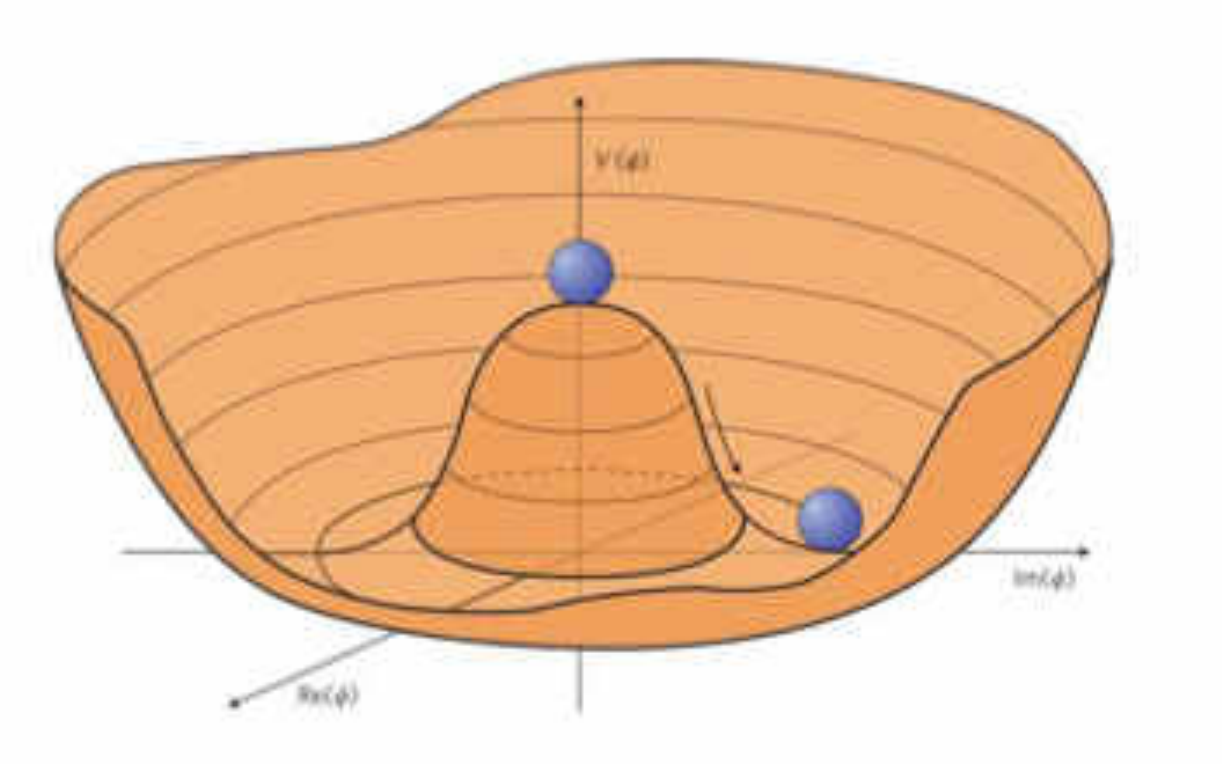}
	\caption{La forme du potentiel employé dans le mécanisme de Higgs}
\end{figure}

Quitte à travailler avec une jauge adaptée, on peut choisir le minimum comme le champ constant réel $$\phi=\frac{v}{\sqrt{2}}$$
Dans la théorie quantifiée, on doit avoir : $$<\phi>=\frac{v}{\sqrt{2}}$$ 
On a ici brisé la symétrie de jauge $U(1)$. On peut réécrire $\phi$ en : $$\phi(x)=\frac{1}{\sqrt{2}}(v+\sigma(x))e^{i\frac{\eta(x)}{v}}$$
mais quitte à changer de jauge on peut faire la transformation : 
$$\phi(x)\rightarrow \phi'(x)=e^{-i\frac{\eta(x)}{v}}\phi(x)=\frac{1}{\sqrt{2}}(v+\sigma(x))$$
qui implique également : 
$$A_\mu(x)\rightarrow A'_\mu(x)=A_\mu(x)+\frac{1}{v}\partial_\mu\eta(x)$$
\begin{rmq}
	De tels choix de jauge sont appelés jauges unitaires.
\end{rmq}
On peut désormais réécrire le lagrangien comme :
$$\mathcal{L}=\frac{1}{2}(\partial_\mu\sigma)^2-\frac{2\lambda v^2}{2}\sigma^2-\lambda v\sigma^3-\frac{\lambda}{4}\sigma^4$$
$$-\frac{1}{4}F_{\mu\nu}F^{\mu\nu}+\frac{1}{2}v^2A_\mu A^\mu$$
$$+v\sigma A_\mu A^\mu+\frac{1}{2}\sigma^2A_\mu A^\mu$$

ce qui permet d'interpréter les deux premiers termes de la première ligne comme la partie propre du lagrangien reliée au champ $\sigma$ qui apparait comme un champ scalaire (lagrangien de Klein-Gordon) massif, de masse $\sqrt{2\lambda v^2}$, la deuxième ligne comme le lagrangien d'un champ de bosons (Klein-Gordon) vecteurs \textbf{massifs}, de masse $v$, et les autres termes comme des termes d'interaction.
Le terme en $\sigma^3$ génère par exemple des vertex d'interaction du champ $\sigma$ avec lui même, le terme en $\sigma A_\mu A^\mu$ un vertex d'interaction du champ $\sigma$ avec deux lignes du champ de jauge ...

\subsection{La brisure de la symétrie $SU(2)$ dans la théorie de jauge électrofaible}
\paragraph{Le champ de Higgs}

Pour briser la symétrie de jauge SU(2) est introduit un doublet d'isospin $\frac{1}{2}$ et d'hypercharge $1$, noté :
$$\phi=\begin{pmatrix}
\phi^+ \\
\phi^0
\end{pmatrix}
= \frac{1}{\sqrt{2}}
\begin{pmatrix}
\phi_1+i\phi_2 \\
\phi_3 + i\phi_4 \\
\end{pmatrix}
$$
Les notations employées s'expliquent par la relation $$Q=I_3^W+Y$$ qui donne matriciellement : 
$$Q=\frac{1}{2}(\begin{pmatrix}
1 & 0 \\
0 & -1 
\end{pmatrix}+
\begin{pmatrix}
1 & 0 \\
0 & 1 \\
\end{pmatrix})=
\begin{pmatrix}
1 & 0 \\
0 & 0 
\end{pmatrix}
$$
donc la première composante du doublet de Higgs a une charge positive tandis que la deuxième composante est électriquement neutre.
La dérivée covariante du champ de Higgs s'écrit : 
$$ \partial^H_\mu\phi=(\begin{pmatrix}
\partial_\mu & 0 \\
0 & \partial_\mu
\end{pmatrix}+i\frac{g'}{2}\begin{pmatrix}
B_\mu & 0 \\
0 & B_\mu
\end{pmatrix}+i\frac{g}{2} W^\alpha_\mu\tau_\alpha)\phi$$
c'est-à-dire :
$$\partial_\mu^H=\begin{pmatrix}
\partial_\mu+\frac{ig'}{2}B_\mu+\frac{ig}{2}W_\mu^3 & \frac{ig}{2}(W_\mu^1-iW_\mu^2)\\
\frac{ig}{2}(W_\mu^1+iW_\mu^2) & \partial_\mu +\frac{ig'}{2}B_\mu -\frac{ig}{2}W_\mu^3
\end{pmatrix}$$
et en notant :

\[
\left\{
\begin{array}{r c l}
i\sqrt{g^2+g^{'2}}A_\mu &=& \frac{ig'}{2}B_\mu+\frac{ig}{2}W_\mu^3\\
i\sqrt{g^2+g^{'2}}Z_\mu &=& \frac{ig'}{2}B_\mu -\frac{ig}{2}W_\mu^3\\
\sqrt{2}W^+ &=& W_\mu^1-iW_\mu^2\\
\sqrt{2}W^- &=& W_\mu^1+iW_\mu^2\\
\end{array}
\right.
\]
 
on peut réécrire :
$$\partial_\mu^H=\begin{pmatrix}
\partial_\mu+i\sqrt{g^2+g^{'2}}A_\mu & \frac{ig}{\sqrt{2}}W^+\\
\frac{ig}{\sqrt{2}}W^- & \partial_\mu +i\sqrt{g^2+g^{'2}}Z_\mu
\end{pmatrix}$$

La partie du lagrangien relatif au champ de Higgs est :
$$\mathcal{L}^{Higgs}=((\partial^H)_\nu\phi)^\dagger((\partial^H)^\nu\phi)-\mu^2\phi^\dagger\phi-\lambda(\phi^\dagger\phi)^2$$
et la partie potentielle a un minimum pour $\phi^\dagger\phi=\frac{-\mu^2}{2\lambda}$. Sans perte de généralité (moyennant le choix d'une jauge adaptée), on peut prendre pour minimum :
$$\phi_{min}=\frac{1}{\sqrt{2}}\begin{pmatrix}
0\\
v
\end{pmatrix}$$ à condition d'avoir $$v^2=-\frac{\mu^2}{\lambda}$$
\begin{rmq}
	Le minimum est choisi de charge électrique nulle, pour que la symétrie $U(1)$ \textbf{ne soit pas brisée} par le mécanisme (et donc pour que le photon reste sans masse), comme nous allons le voir.
\end{rmq}

On peut choisir une jauge unitaire comme dans le cas de la théorie de jauge $U(1)$ étudiée ci-dessus, et alors : 
$$\phi(x)\rightarrow\frac{1}{\sqrt{2}}\begin{pmatrix}
0\\
v+\sigma(x)\\
\end{pmatrix}
$$
Le lagrangien $\mathcal{L}^{Higgs}$ devient alors :
$$\mathcal{L}^{Higgs}=\frac{1}{2}(\frac{ig}{2}W_\nu^-(v+\sigma)+(\partial_\nu-i\sqrt{g^2+g^{'2}}Z_\nu)(v+\sigma))\times(\frac{ig}{2}W^{+\nu}(v+\sigma)+(\partial^\nu+i\sqrt{g^2+g^{'2}}Z^\nu)(v+\sigma))$$
$$ - \frac{\mu^2}{2}(v+\sigma)^2-\frac{\lambda}{4}(v+\sigma)^4$$

\paragraph{La génération de la masse et l'apparition des bosons de jauge électrofaibles}

Réécrivons le lagrangien auquel nous sommes arrivés :
$$\mathcal{L}^{Higgs}=\mathcal{L}^{Higgs}_1 + \mathcal{L}^{Higgs}_2 + \mathcal{L}^{Higgs}_3$$
avec 
$$\mathcal{L}^{Higgs}_1=\frac{1}{2}(\partial_\nu\sigma)(\partial^\nu\sigma)-\frac{\mu^2}{2}\sigma^2-\lambda v\sigma^3 -\frac{\lambda}{4}\sigma^4$$ $$-\frac{\mu^2}{2}v^2 -\mu^2 v\sigma -\frac{\lambda}{4}v^4-\lambda v^3\sigma -\frac{3\lambda}{2}v^2\sigma^2$$
On peut tout d'abord s'affranchir des termes constants $-\frac{\mu^2}{2}v^2$ et $\frac{\lambda}{4}v^4$. De plus :
$$(-\frac{\mu^2}{2}-\frac{3\lambda}{2}v^2)=(\frac{\lambda v^2}{2}-\frac{3\lambda}{2}v^2)=-\lambda v^2$$
et 
$$\mu^2 v \sigma -\lambda v^3 \sigma = -v(\mu^2 - \lambda v^2)\sigma=0$$
donc finalement : 
$$\mathcal{L}^{Higgs}_1=\frac{1}{2}(\partial_\nu\sigma)(\partial^\nu\sigma)-\frac{1}{2}(2\lambda v^2)\sigma^2-\lambda v\sigma^3-\frac{\lambda}{4}\sigma^4$$
où les deux premiers termes sont compris comme le lagrangien du champ scalaire $\sigma$ (équation de Klein-Gordon) sans interaction, de masse $m_h^2=2\lambda v^2$, et les deux derniers comme des termes d'interaction du champs $\sigma$ avec lui-même.
Par ailleurs :
$$\mathcal{L}^{Higgs}_2=\frac{g^2v^2}{8}(W^-_\nu W^{+\nu})+\frac{v^2}{8}(g^2+g^{'2})Z_\nu Z^\nu$$

Posons $$m_W=\frac{gv}{2} \ \ m_Z=\frac{v\sqrt{g^2+g^{'2}}}{2}$$ ce terme du lagrangien devient : $$\mathcal{L}^{Higgs}_2=m_W^2(W^-_\nu W^{+\nu})+\frac{1}{2}m_Z^2 Z_\nu Z^\nu$$

Les champs $B$ et $W^3$ ont été couplés entre eux par :
$$\begin{pmatrix}
Z_\nu\\ A_\nu
\end{pmatrix}
=\begin{pmatrix}
cos\theta & -sin\theta\\
sin\theta & cos\theta
\end{pmatrix}
\begin{pmatrix}
W_\nu^3\\
B_\nu
\end{pmatrix}$$
avec $$cos\theta=\frac{g}{\sqrt{g+g'}} \ \ sin\theta=\frac{g'}{\sqrt{g+g'}}$$
où $\theta$ est le "weak mixing angle" ou angle de Weinberg, que l'expérience fixe à 
$$sin^2\theta=0.23122 \pm 0.00015$$
Cette manipulation permet de faire apparaître dans le lagrangien un champ de jauge non massif, dont on a besoin pour décrire les photons.
Il reste donc un champ de jauge complexe $W^+$ qu'on décrit également à l'aide de son adjoint $W^-=(W^+)^\dagger$, qui après quantification donne des quanta bosoniques de masse $m_W$, et un champ de jauge réel dont les quanta sont des bosons de masse $m_Z$. Remarquons que :
$$m_Z=\frac{m_W}{sin\theta}$$
mais que la masse exacte n'est pas prédite. Cependant, les masses mesurées des bosons W et Z, valant respectivement 80.40 GeV et 91.19 GeV, permettent de vérifier la cohérence de la théorie. De plus, si la mesure - par exemple - du temps de vie du muon permet d'accéder à la valeur de $v$ : $$v=\frac{2m_W}{g}=\frac{1}{\sqrt{\sqrt{2}G}}\approx246.22\ GeV$$ \textbf{il ne reste plus qu'un seul paramètre libre pour le champ de Higgs, par exemple, le paramètre $\lambda$, ou, de manière équivalent, la masse du boson de Higgs (le quanta du champ de Higgs)}.\\\\
Le modèle standard ne prévoit pas la valeur de cette constante, d'où le challenge pour observer le boson de Higgs : la seule information, avant la découverte, était que sa masse est comprise entre quelques dizaines de GeV et quelques TeV, soit une plage énorme à balayer. Cependant, une fois la masse fixée, il est possible de calculer les différents branching ratio de la désintégration du Higgs.

\begin{figure}[!h]
	\centering
	\includegraphics[scale=1.4]{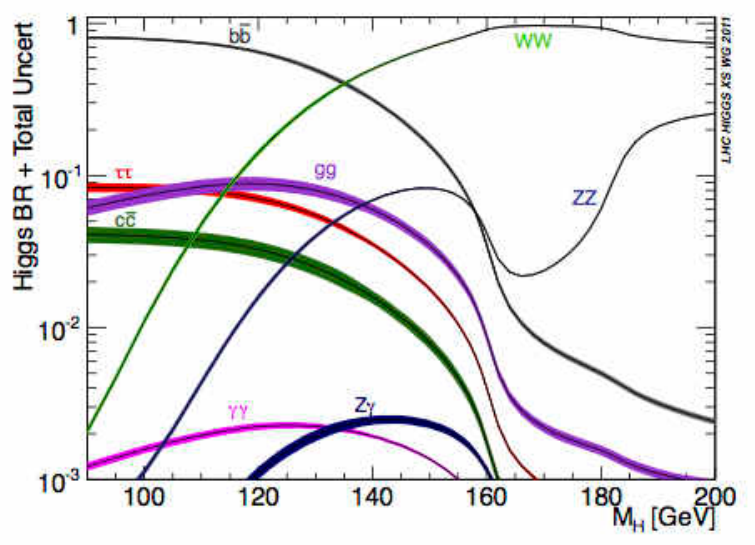}
	\caption{Diagramme montrant les différentes fractions de désintégration du boson de Higgs (en 2011, les recherches du LEP et les deux premières années de fonctionnement du LHC permettent de restreindre l'étude à la recherche d'un particle de masse comprise entre 115 et 200 GeV)}
\end{figure}

Le boson de Higgs a finalement été observé, et sa masse mesurée à 125 GeV. S'il s'agit bien du Higgs du modèle standard, on contraint le dernier paramètre libre de la théorie de Higgs, puisque : 
$$m_H^2=2\lambda v^2\Rightarrow\lambda=\frac{m_H^2}{2v^2}\approx 0.13$$

\paragraph{De la masse pour les fermions}
Nous avons vu pourquoi il était nécessaire de partir de fermions de masse nulle pour construire les interactions faibles. Cependant, il est des fermions (comme l'électron, le muon ...) dont on connait la masse, et qu'il serait absurde de décrire comme des particules non massives. Ce problème est résolu en couplant les fermions au champ de Higgs par des "couplages de Yukawa". Soit 
$$\Psi=\begin{pmatrix}
\psi_1\\
\psi_2
\end{pmatrix}
$$ un doublet d'isospin de fermions. On peut rendre le champ de fermions $\psi_2$ massif en rajoutant un terme au lagrangien de la forme :  
$$\mathcal{L}^{HL}=-G_l(\overline{\Psi}^L\begin{pmatrix}
\phi^+ \\
\phi^0
\end{pmatrix}\psi_2^R+\overline{\psi_2}\begin{pmatrix}
\phi^{+*} & \phi^{0*}
\end{pmatrix}\Psi^L)$$
où $$\overline{\Psi}^L\begin{pmatrix}
\phi^+ \\
\phi^0
\end{pmatrix}=\begin{pmatrix}
\overline{\psi_{1}^L} & \overline{\psi_{2}^L} 
\end{pmatrix}
\begin{pmatrix}
\phi^+ \\
\phi^0
\end{pmatrix}=\phi^+\cdot\overline{\psi_{1}^L}+\phi^0\cdot\overline{\psi_{2}^L} $$
et de même par la partie conjuguée hermitienne.
Il est clair que le terme ainsi rajouté est bien défini au sens de la jauge (en effet $\Psi^L$ est d'isospin total $\frac{1}{2}$, tout comme le doublet Higgs, alors que $\psi_2^R$ est d'isospin total nul ; pour l'hypercharge, $Y(\Psi^L)=1$, $Y(\phi)=1$ et $Y(\psi_2^R)=-2$), ce qui n'était pas le cas des termes "naïvement" rajoutés au lagrangien pour rendre les fermions massifs, à cause de la non invariance par symétrie de parité.\\
En jauge unitaire, 
$$\begin{pmatrix}
\phi^+ \\
\phi^0
\end{pmatrix}=\frac{1}{\sqrt{2}}\begin{pmatrix}
0 \\
v+\sigma(x)
\end{pmatrix}$$ ce qui fait apparaître les termes : 
$$\mathcal{L}^{HL}=-\frac{G_e}{\sqrt{2}}(\overline{\psi_2^L}(v+\sigma)\psi_2^L+\overline{\psi_2^R}(v+\sigma)\psi_2^R)$$
or $$\overline{\psi_2^L}\psi_2^R+\overline{\psi_2^L}\psi_2^R=\frac{1}{4}\overline{\psi_2}(1+\gamma_5)(1+\gamma_5)\psi_2+\frac{1}{4}\overline{\psi_2}(1-\gamma_5)(1-\gamma_5)\psi_2=\overline{\psi_2}\psi_2$$
d'où : 
$$\mathcal{L}^{HL}=-\frac{G_ev}{\sqrt{2}}\overline{\psi_2}\psi_2+-\frac{G_e}{\sqrt{2}}\overline{\psi_2}\psi_2\sigma$$
et la masse du fermion est donnée par $$m_2=\frac{G_ev}{2}$$ Pour pouvoir vérifier la véracité de cette égalité, il faut mesurer la constante adimensionnée $G_e$ grâce à l'autre terme qu'on rajoute et qui couple le champ de fermions au champ de Higgs.\\\\
Pour donner de la masse au champ $\psi_1$, on rajoute au lagrangien un terme de la forme : 
$$\mathcal{L}^{HL}=-G'_l(\overline{\Psi}^L\begin{pmatrix}
\phi^{0*} \\
-\phi^{+*}
\end{pmatrix}\psi_2^R+\overline{\psi_2}\begin{pmatrix}
\phi^0 & \phi^+
\end{pmatrix}\Psi^L)$$
où on peut réécrire : 
$$
\begin{pmatrix}
	\phi^{0*} \\
	-\phi^{+*}
\end{pmatrix}=
-i[\begin{pmatrix}
\phi^+ & \phi^0
\end{pmatrix}\begin{pmatrix}
0 &-i\\
i & 0
\end{pmatrix}]=-i[\phi^\dagger\tau_2]
$$

\paragraph{Cas des quarks}
Les quark gauches forment, comme les leptons, des doublets d'isospin : 
$$\begin{pmatrix}
\psi_u \\
\psi_{d'}
\end{pmatrix}_L
\ \ \begin{pmatrix}
\psi_c \\
\psi_{s'}
\end{pmatrix}_L
\ \ \begin{pmatrix}
\psi_t \\
\psi_{b'}
\end{pmatrix}_L
$$
d'isospin total $\frac{1}{2}$ et d'hypercharge $\frac{1}{3}$, et les quarks droits des singulets d'isospin : 
$$ \psi_u^R \ \ \psi_c^R \ \ \psi_t^R $$
d'isospin total 0 et d'hypercharge $\frac{4}{3}$, et
$$ \psi_{d'}^R \ \ \psi_{s'}^R \ \ \psi_{b'}^R $$
d'isospin total 0 et d'hypercharge $\frac{-2}{3}$.\\\\
Les couplages de Yukawa génèrent des termes de masse dans le lagrangien, de la forme :
$$\begin{pmatrix}
\overline{\psi_u} & \overline{\psi_c} & \overline{\psi_t} 
\end{pmatrix}
M_{up-type}
\begin{pmatrix}
\psi_u \\ \psi_c \\ \psi_t 
\end{pmatrix}$$ et 
$$\begin{pmatrix}
\overline{\psi_{d'}} & \overline{\psi_{s'}} & \overline{\psi_{b'}} 
\end{pmatrix}
M_{down-type}
\begin{pmatrix}
\psi_{d'} \\ \psi_{s'} \\ \psi_{b'} 
\end{pmatrix}$$
mais les deux matrices ne peuvent être diagonalisées simultanément, si bien que si on choisit de diagonaliser $M_{up-type}$, on obtient :
$$M_{down-type}=V\begin{pmatrix}
m_d & 0 & 0 \\
0 & m_s & 0 \\
0 & 0 & m_b \\
\end{pmatrix}V^\dagger
$$
et on a :
$$\begin{pmatrix}
\psi_{d'}\\
\psi_{s'} \\
\psi_{b'}
\end{pmatrix} = V \begin{pmatrix}
\psi_{d}\\
\psi_{s} \\
\psi_{b}
\end{pmatrix}
$$
et $V$ est la matrice unitaire de Cabbibo-Kobayashi-Maskawa, qui donne les probabilités de changement de saveurs des quarks par interaction faible, qui peuvent être mesurées. En 2014, le particle data group donne pour la matrice CKM : 
$$V=\begin{pmatrix}
0.97427\pm0.00014 & 0.22536 \pm 0.00061 & 0.00355 \pm 0.00015 \\
0.22522\pm0.00061 & 0.97343 \pm 0.00015 & 0.0414 \pm 0.0012 \\
0.00886\pm0.00033 & 0.0405 \pm 0.0012 & 0.99914 \pm 0.00005 
\end{pmatrix}
$$

\section{Au delà du modèle standard ...}

Le modèle standard est une théorie qui encore aujourd'hui, est difficile à dépasser. Des signes de "nouvelle physique" apparaissent depuis des années, mais c'est plutôt de l'inexplicable que des observations non concordantes avec le modèle Standard. Tout d'abord, le modèle standard n'explique pas la gravitation. En plus de cela, les problèmes liés à la matière noire, à la "baryon asymmetry", l'inflation de l'univers, la valeur des constantes cosmologiques, la hiérarchie de masse des fermions, pour ne citer qu'eux, ne trouvent pas de réponse dans le modèle standard, ce qui pousse à chercher une théorie englobant le modèle standard, et permettant de résoudre au moins certains des problèmes évoqués ci-dessus. On distingue plusieurs types de modèles étendant le modèle standard : 
\begin{enumerate}
	\item Le "$\kappa$ framework" ("kappa framework"), qui paramètre juste les déviations des constantes du modèle standard
	\item Les "Effective Field Theories" (EFT), qui sont à proprement parler des théories perturbatives du modèle standard
	\item Les modèles simplifiés, où on introduit seulement des corrections locales du modèle standard, qui peuvent être des prédictions de théories plus générales (supersymétriques par exemple)
	\item Les théories générales (UV pour ultraviolet) qui sont une extension du modèle standard également valables à haute énergie, au delà le l'énergie de brisure électrofaible, que nous n'évoquerons pas ici compte tenu de la complexité des idées introduites. Remarquons seulement que les extensions du domaine scalaire sont soit des théories supersymétriques (MSSM, Split SUSY, ...) soit des modèles introduisant une nouvelle dynamique forte qui rend le boson de Higgs composite.
\end{enumerate}

\subsection{Kappas frameworks}
Notons $g_{i,SM}$ les différents couplages du modèle standard, et $g_i$ les couplages vrais, mesurables. On définit : 
$$\kappa_i=\frac{g_i}{g_{i,SM}}$$
Si tous les $\kappa_i$ valent 1, alors il n'y a aucune déviation par rapport au modèle standard.
Ces modèles sont les plus simples extensions du SM, dans la mesure où la mesure des couplages se base sur les taux intégrés de désintégrations, ou les sections efficaces intégrées : il n'y a pas de nouvelles dynamiques. Par conséquent, les applications pour de la nouvelle physique sont fort limitées.
Pour l'instant, les couplages, s'ils diffèrent de ceux du modèle standard, en sont très proches et c'est pour cela qu'il fallait attendre le Run 2 ou 3 du LHC pour espérer mesurer des déviations.

\begin{figure}[!h]
	\centering
	\includegraphics[scale=0.55]{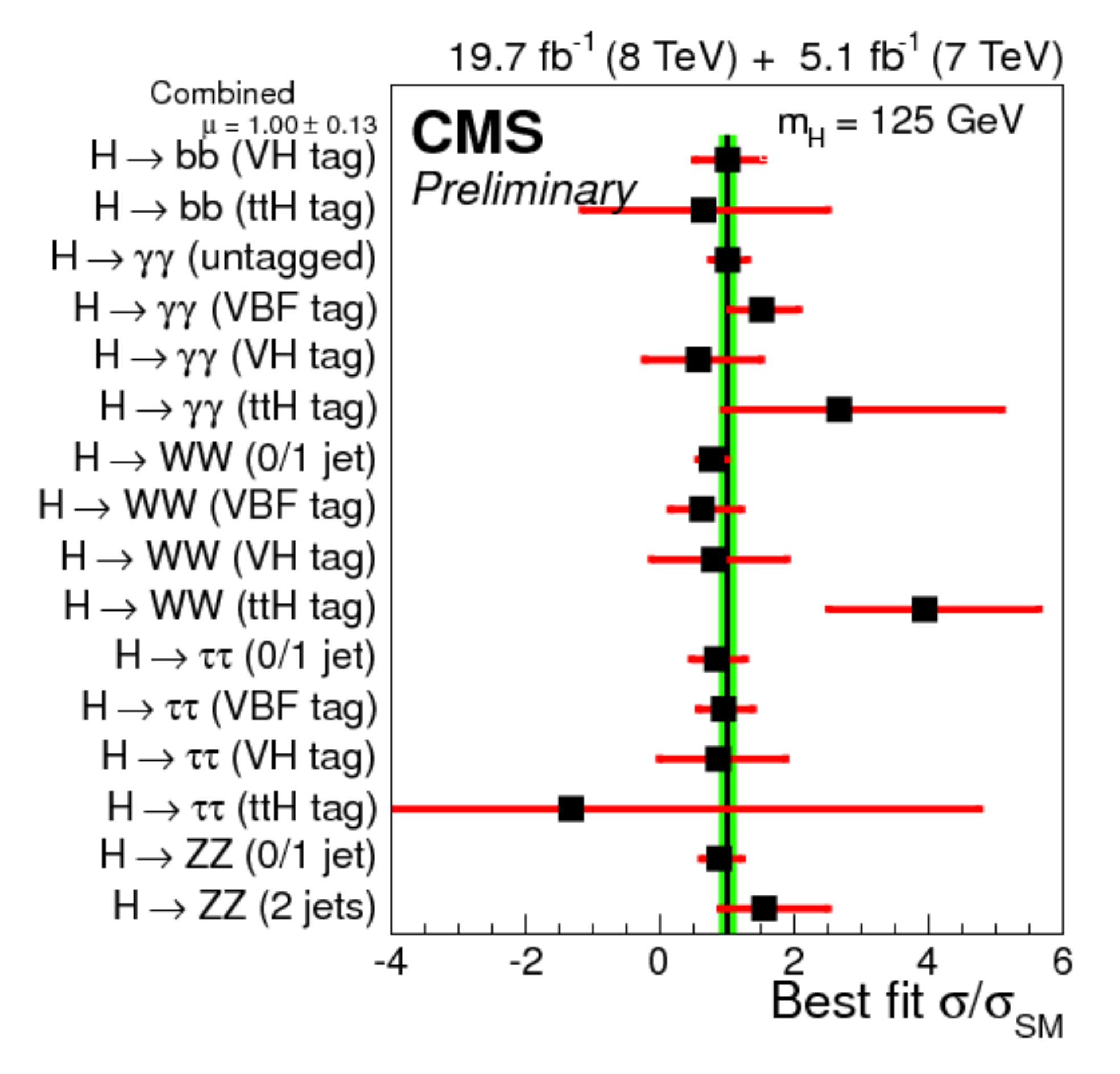}
	\caption{Best fit $\sigma/\sigma_{SM}$ pour la combinaison des principaux modes de désintégration du Higgs (à 8 TeV dans le centre de masse). Les barres horizontales donnent les $\pm1$ déviations standard incluant les erreurs statistiques et systématiques.}
\end{figure}

\subsection{Effective Field Theories}
Pour reprendre la définition de R. Ambrosio au First Annual Meeting of ITN HiggsTools le 17 Avril 2015, "An Effective field theory (EFT) is a field theory, designed to reproduce the
behavour of some underlying (in general, unknown) physical theory in some
limited regime. It focuses on the degrees of freedom relevant to that regime,
simplifying the problem though letting aside some important physics."

Par exemple, la mécanique newtonienne est une théorie effective, approximation de basse énergie de la relativité générale, qui est sans doute elle-même théorie effective d'une autre théorie plus générale, peut être la théorie des cordes. La théorie des interactions faible de Fermi est aussi effective, approximation de basse énergie de la théorie électrofaible.

\paragraph{Deux approches différentes des théories des champs effectives}

On distingue essentiellement deux approches, nommées respectivement \textbf{Top-Down approach} et \textbf{Bottom-Up approach}. Dans la première, on part d'une théorie complète, définie pour des énergies élevées, et on étudie son comportement à basse énergie. Les calculs sont en général simplifiés par le fait qu'il y a des découplages à basse énergie. Dans le deuxième cas, on part d'une théorie connue à basse énergie (par exemple le modèle standard), on étudie son comportement à haute énergie, on rajoute des opérateurs cohérents avec les symétries observées (par exemple de jauge), et on essaie de déterminer les inconnues grâce à l'expérience.\\\\
Supposons qu'il existe de nouvelles particules au delà d'une énergie $\Lambda$. En dessous de cette échelle d'énergie, la dynamique est décrite par le lagrangien : 
$$\mathcal{L}_{eff}=\mathcal{L}_{SM}+\sum_{d\geq5}\frac{c_n^{(d)}}{\Lambda^2}O_n^{(d)}(\phi_{SM})$$
où $d$ est la dimension (de masse, en unités naturelles) des nouveaux opérateurs introduits (le lagrangien standard est de dimension 4)
Weinberg a montré en 1959\cite{weinberg} qu'il n'y avait qu'un seul opérateur effectif de dimension 5, et qu'il violait la conservation du nombre leptonique. Dans le même article, il montre que le nombre baryonique n'est pas conservé par certains opérateurs de dimension 6.\\ 

On peut construire en fait \textbf{80 opérateurs de dimension 6} compatibles avec la symétrie de jauge du modèle standard et qui conservent le nombre leptonique et le nombre baryonique. Les équations du mouvement réduisent l'espace de ces opérateurs à un \textbf{espace de dimension 76}, et éventuellement 59 si on ne considère que des opérateurs CP-even. Cependant, c'est sans considérer les différentes saveurs de leptons et de quarks ! Sinon ce nombre s'élève à 2499.\\
\begin{center}
\color{red}Le lagrangien effectif de dimension 6 est donc le développement au premier ordre d'une théorie plus générale dont le modèle standard est le développement à l'ordre nul.\color{black}
\end{center}
Nous en verrons une partie lors de l'étude phénoménologique du couplage $\lambda$ du champ de Higgs.\\\\
Concernant le couplage lambda du champ de Higgs, de nouveaux termes peuvent apparaître, de la forme $$\frac{\rho}{\Lambda}(\phi_{SM}^\dagger\phi_{SM})^3$$ avec $\mu<<\Lambda$ où $\mu$ est le paramètre du potentiel de Higgs.
L'auto-couplage du champ de Higgs est modifié par cet ajout de la manière suivant : 
$$\delta\lambda_{hhh}=\frac{2\rho v^4}{m_h^2\Lambda^2}$$
Cette modification est intéressante car elle n'induit pas de décalage dans le couplage du champ de Higgs avec les fermions et les bosons faibles (qui ont déjà été mesurés et coïncident avec les prévisions du modèle standard). Le forme du potentiel de Higgs introduit pour pouvoir spontanément briser la symétrie de jauge $SU(2)$ n'a a priori aucune raison d'être un polynôme de degré 4, c'est cette idée qui se développe en fait dans ce modèle.

\subsection{Modèles simplifiés}

Le but est de complexifier de manière raisonnable le modèle standard pour arranger certaines de ses prédictions, comme présenté dans cette étude phénoménologique \cite{pheno}.

\paragraph{Doublet-singlet mixing}
L'extension la plus simple du secteur scalaire (du modèle de Higgs) est le doublet-singlet mixing. Introduisons un singulet d'isospin faible et de couleur, réel $\Phi_S$, qui mixe avec le doublet de Higgs $\Phi_{SM}$. Moyennant un bon choix de jauge : 
$$\Phi_{SM}=\begin{pmatrix}
0\\
\frac{1}{\sqrt{2}}(v+\phi_{SM}) 
\end{pmatrix} \ \ \ \Phi_S=\frac{1}{\sqrt{2}}(V+\phi_S)$$
On en déduit l'existence de deux bosons scalaires, un léger noté h et un lourd noté H, combinaisons linéaires :
$$\begin{pmatrix}
h=cos(\alpha)\phi_{SM}+sin(\alpha)\phi_S\\
H=-sin(\alpha)\phi_{SM}+cos(\alpha)\phi_S
\end{pmatrix}
$$
où $\alpha$ est le mixing angle. Pour que h ait des propriétés proches du Higgs du modèle standard, il faut que H soit très lourd (on est par conséquent dans la limite $v<<V$).\\
Dans cette limite, il est possible de montrer que : 
$$\delta\lambda_{hhh}=\frac{\lambda_{hhh}}{\lambda_{hhh}^{SM}}-1\approx -\frac{3}{2}s_\alpha^2$$ et 
$$\delta\lambda_{hff}=\frac{\lambda_{hff}}{\lambda_{hff}^{SM}}-1\approx -\frac{1}{2}s_\alpha^2$$
où $\lambda_{hhh}$ est le couplage trilinéaire du champ de Higgs et $\lambda_{hff}$ le couplage du champ de Higgs avec les champs de Dirac.

\paragraph{Doublet-doublet mixing}
Dans ce modèle, on introduit un deuxième doublet scalaire d'isospin, ce qui entraîne deux valeurs  moyennes $v_{1,2}$ des champs scalaires dans le vide. Posons $\beta:\frac{v_2}{v_1}$. On obtient deux états neutres h et H, pairs par symétrie CP (et on note $\alpha$ le mixing angle entre les deux), un état neutre A CP-odd, et un champ scalaire chargé $H^\pm$. Si $H$ est très lourd, $h$ se comporte comme le Higgs du modèle Standard comme dans le cas précédent. Dans cette limite, les scalaires restants $H$, $A$ et $H^\pm$ sont en dégénérescence de masse $m_A$. \\\\
Il y a principalement deux types de modèles 2 Higgs Doublet Model (2HDM), appelés type I et type II, qui diffèrent dans les couplages avec les quarks : dans le type I, tous les quarks couplent avec un des doublets, dans le type II les quarks de type up couplent avec l'un des doublets, et des quarks de type down couplent avec l'autre des doublets \cite{doublet}.\\
Dans les modèles de type II, les couplages sont modifiés selon : 
$$\delta\lambda_{hhh}\approx -\frac{2m_A^2}{m_h^2}cos^2(\beta-\alpha)$$ et 
$$\delta\lambda_{htt}\approx sin(\beta-\alpha) +\frac{cos(\beta-\alpha)}{tan\beta}-1$$

\subsection{Discussion}

De nombreux modèles sont donc proposés pour étendre le modèle standard. En effet, l'observation et l'expérience prouve par beaucoup d'aspects que le modèle standard n'est pas complet, il y a de nombreux phénomènes qui ne sont pas décrits. La difficulté vient du fait que pour l'instant, l'expérience ne nous dit pas où chercher la "nouvelle physique", seulement qu'il y en a une. Le LHC commence à arriver à des énergies suffisantes pour pouvoir infirmer certaines théories, et éventuellement faire une grande découverte pour en confirmer une, ou au moins, mettre la communauté scientifique sur une piste à suivre. Pour l'instant, maintenant que le boson de Higgs a été observé, il s'agit de mesurer ses propriétés avec une grande précision puisqu'elles dépendent beaucoup, comme on l'a vu, des extension de modèle standard. Cependant, jusqu'à aujourd'hui aucune grande déviation n'a été repérée. Soit les extensions proposées ne sont donc pas valables, soit les déviations sont trop faibles pour qu'on les observe avec les LHC. Heureusement, il reste à chercher : le couplage $\lambda$ du Higgs, \textbf{prévu par le modèle standard}, n'a pas encore été mesuré (c'est le sujet du prochain chapitre), et si on en croit certaines théories supersymétriques, la masse du superpartenaire du top, le stop, ne devrait pas être hors de portée ... C'est à voir pour les prochains mois ou années à venir.

\chapter{Étude du couplage lambda du champ BEH}

L'étude réalisée au laboratoire Leprince-Ringuet, dans le cadre de la collaboration CMS, du couplage trilinéaire du boson de Higgs se divise en quatre parties : tout d'abord, une étude au niveau générateur en considérant les leptons $\tau$ comme stables permet de sélectionner les variables cinétiques d'intérêt pour remonter jusqu'à la valeur de cette constante de couplage. Ensuite, une nouvelle étude est réalisée, cette fois en prenant en compte la désintégration des $\tau$, pour étudier l'impact de la perte d'information (due notamment aux neutrinos) sur ces variables cinétiques. La troisième étape consiste à utiliser non plus les informations générées, mais les informations reconstruites afin de s'approcher le plus possible de ce que le détecteur peut effectivement observer. Enfin, nous avons tenté de paramétrer les courbes intéressantes, dans le cas de l'étude sans la désintégration des leptons, pour obtenir une expression, au moins locale, de la surface ainsi obtenue.
Tout d'abord, présentons l'accélérateur LHC et l'expérience CMS de manière sommaire, pour bien réaliser comment sont effectuées les mesures, et quelles sont les informations auxquelles nous avons accès.
\section{Le LHC et le détecteur CMS}
\subsection{Le Large Hadron Collider}

Le LHC est un collisionneur de hadrons situé sous la frontière franco-suisse à proximité de Genève. C'est l'accélérateur le plus grand et le plus puissant construit à ce jour ; il est installé dans un tunnel circulaire de 3 mètres de diamètre et de 26.7 km de long, situé à des profondeurs comprises entre 50m et 170m, qui contenait auparavant le collisionneur LEP. Le LHC a été conçu pour faire des collisions de protons ou d'ions lourds (plomb) ; bien que les collisions soient plus compliquées que les collisions électron-positron puisque les protons sont des particules composites, il est possible d'atteindre des énergies bien plus importantes car les protons étant plus lourds, ils perdent moins d'énergie par rayonnement que les électrons.

L'accélérateur comprend des cavités accélératrices à haute fréquence, des cavités magnétiques pour la collimation, et des dipôles magnétiques pour courber le faisceau. Aux énergies atteintes il faut que les dipôles fournissent un champ magnétique de 8,3 T. Ce sont des supraconducteurs refroidis à l'hélium superfluide à 1,9 K (plus froid que la température de fond cosmologique). L'accélérateur comporte 1232 tels dipôles.

Deux faisceaux différents sont accélérés dans des directions opposées et se collisionnent en quatre points, autour desquels ont été bâties quatre expériences : les détecteurs polyvalents ATLAS (A Toroidal LHC Apparatus) et CMS (Compact Muon Solenoid) pour lesquels les objectifs principaux incluaient la recherche et l'étude des propriétés du champ de Higgs et la prospection de physique derrière le modèle standard ; l'expérience LHCb conçue pour étudier la physique du quark bottom et la violation de symétrie CP (non hermétique avec ses détecteurs sont placés de manière à observer des particules émises à très petit angle par rapport au faisceau, et l'appareil ALICE (A Large Ion Collider Experiment) pour étudier les interactions entre ions lourds et le plasma quarks-gluons. ATLAS et CMS sont situés aux antipodes de l'anneau, aux points de plus haute luminosité.

\begin{figure}[!h]
	\centering
	\includegraphics[scale=0.25]{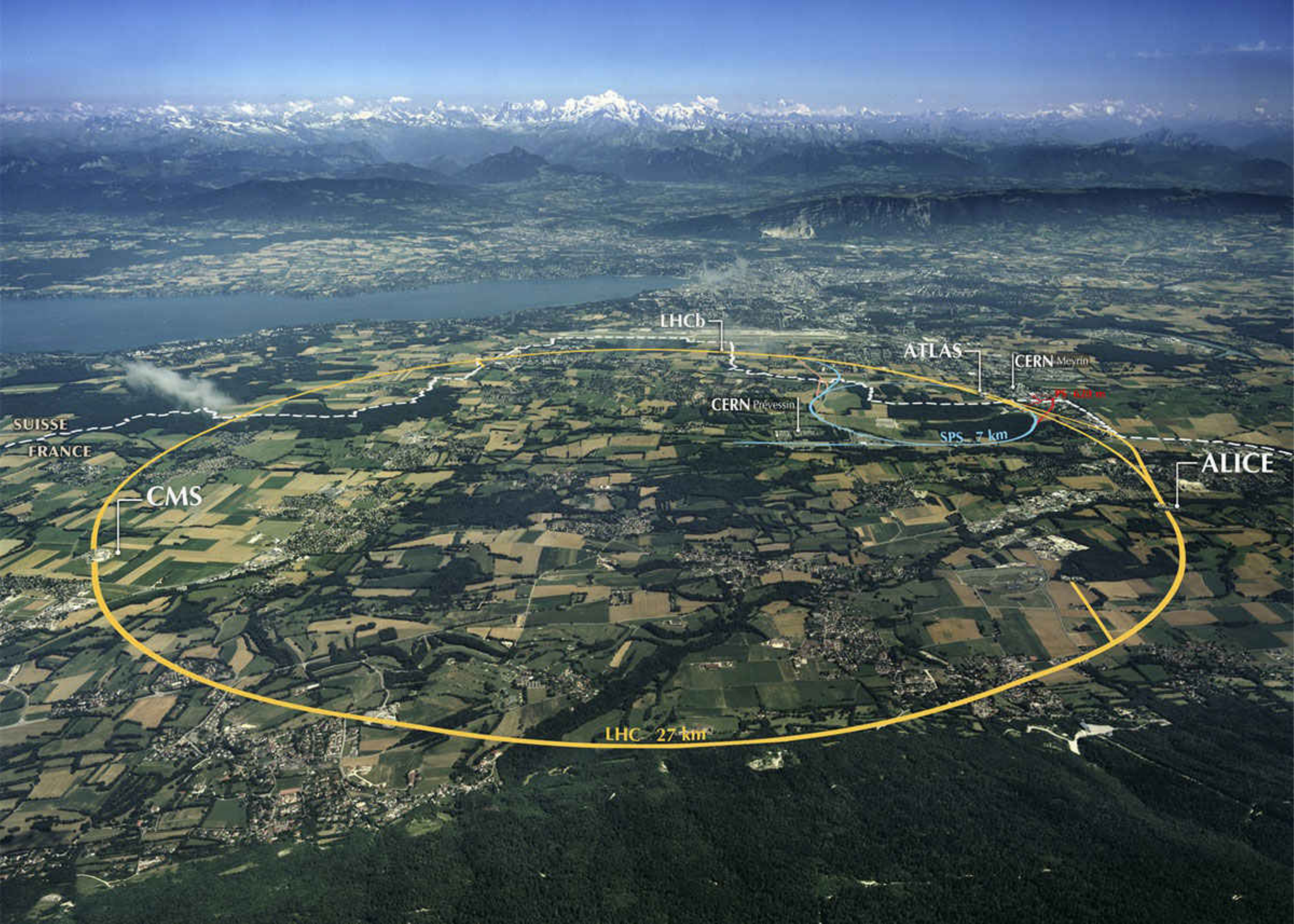}
	\caption{Vue aérienne du complexe du CERN (site de CERN)}
\end{figure}

Chacun des faisceaux est pulsé ; chaque paquet contient environ $10^{11}$ protons, et l'intervalle de temps entre deux paquets successifs est de 25 ns. La fréquence de collision est de 40 MHz, et la luminosité nominale instantanée du faisceau de $10^{34}$ $cm^{-2}s^{-1}$ (il y a potentiellement $10^{34}$ évènements par seconde pour une surface d'un $cm^2$ au sein des détecteurs du LHC).

Avant d'arriver jusqu'au LHC, les protons doivent être progressivement accélérés. Après ionisation des atomes d'hydrogènes, les protons parcourent le Linac 2, un accélérateur linéaire à la sortie duquel ils ont une énergie de 50 Mev (et ont gagné 5 \% en masse). Ils parcourent ensuite le Synchrotron à Protons (PS) qui leur confère une énergie de 26 GeV, puis le Super Synchrotron à Protons (SPS) à l'issue duquel ils sont injectés avec une énergie de 450 GeV dans le Large Hadron Collider (LHC), où ils atteignent une énergie de 6,5 TeV (et bientôt 7 TeV) avant de rentrer en collision avec les protons d'un autre paquet circulant en sens inverse.

\paragraph{Run 1} Le lancement opérationnel du LHC était prévu pour septembre 2008 mais suite à un problème électrique, 100 aimants de courbure ont subi un "quench", qui a causé l'endommagement de 53 aimants et libéré 6 tonnes d'hélium liquide dans le tunnel, brisant le vide poussé qui y régnait. Le LHC a repris les collisions le 20 novembre 2009, montant en moins de 10 jours à une énergie de 1,18 TeV par faisceau, battant ainsi le record du Tevatron de 0,98 TeV par faisceau détenu depuis plus de 8 ans. 
L'énergie des faisceaux a alors été progressivement augmentée pour atteindre 7 TeV en mars 2010. Le premier 'run' de collisions protons-protons s'est arrêté au début du mois de novembre 2010. Ont suivi des collisions d'ions lourds pendant deux mois, puis, après un arrêt technique, les collisions de proton ont été reprises en mars 2011, atteignant le 21 avril 2011 une luminosité record, de $4,67.10^{32}$ $cm^{-2}s^{-1}$, battant ainsi le record lui-encore détenu par le Tevatron.

\begin{figure}[!h]
	\centering
	\includegraphics[scale=0.13]{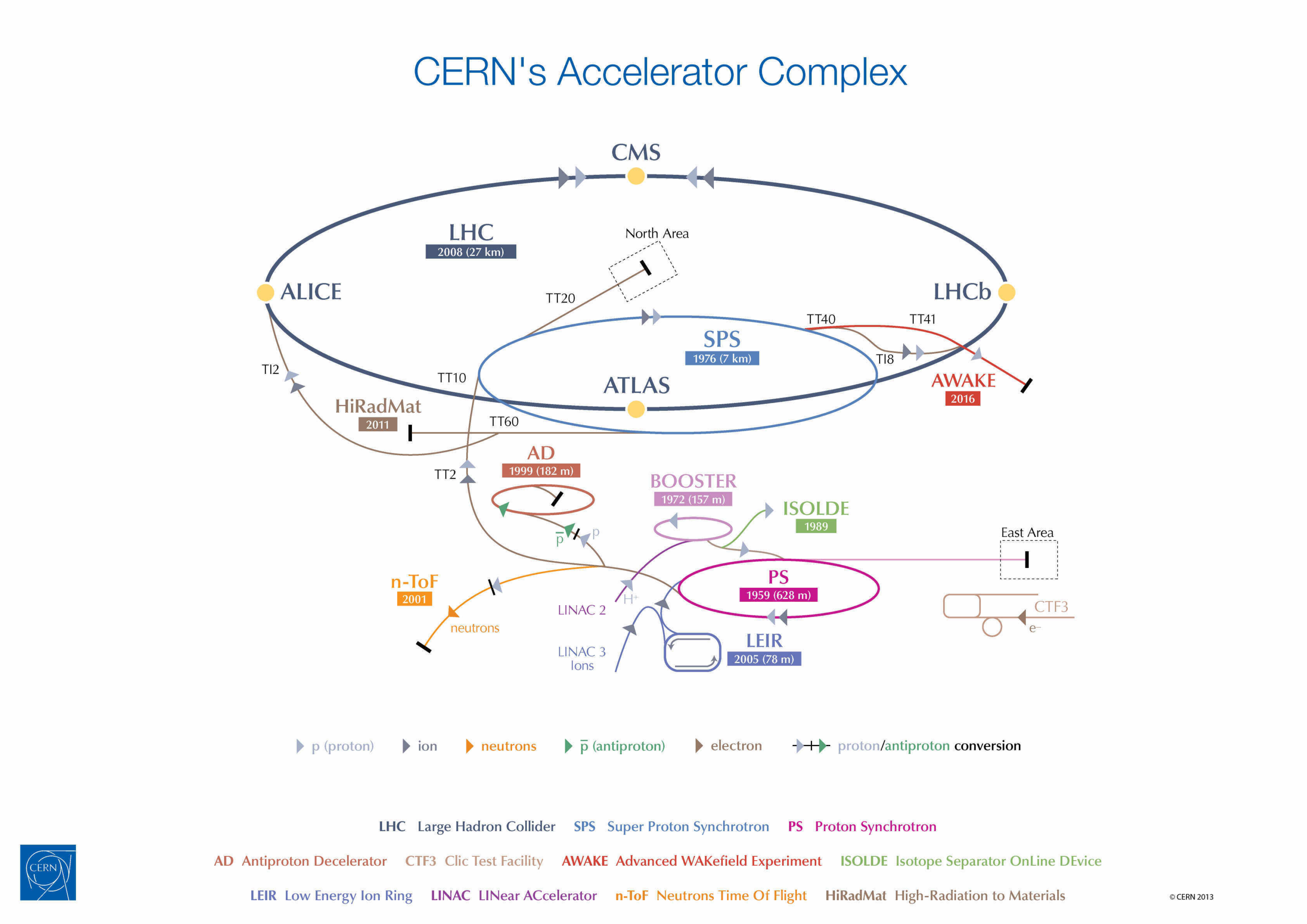}
	\caption{Schéma des accélérateurs du CERN et de leurs raccords}
\end{figure}

La première observation expérimentale de plasma quarks-gluons a été annoncée en mai 2011,  la découverte d'une nouvelle particule résonance de $b\overline{b}$ appelée $\chi_b(3P)$ en décembre 2011\cite{decouvchi}. Durant l'hiver, des modifications ont été apportées aux aimants du tunnel pour pouvoir élever l'énergie par proton, au moment de la collision, à 4 TeV (soit 8 TeV par couple de protons). Le 4 juillet 2012 la découverte d'une nouvelle particule élémentaire est annoncée par les deux détecteurs polyvalents du LHC : CMS rend compte d'un boson d'une masse de $125,3 ± 0,6$ $GeV$\cite{CMSHiggs} et ATLAS d'un boson d'une masse de $126,0±0,6$ $GeV$\cite{ATLASHiggs} (le 14 mars 2013, le CERN précise que ce boson est très vraisemblablement une particule scalaire, de parité paire ; qu'il s'agit donc d'un "boson de Higgs", en se basant sur l'étude des données recueillies l'année d'avant). Le 8 novembre 2012 la désintégration rare du boson $B^0_S$ en deux muons est observée pour la première fois. Elle constituait un test important des théories supersymétriques. Ces résultats sont compatibles à 3,5 sigmas avec le modèle standard, et imposent des contraintes fortes sur ses extensions\cite{Bdecay}.

Enfin, en février 2013, le LHC est arrêté pour le première longue pause technique, afin de modifier l'accélérateur pour des énergies et luminosités plus importantes.
 
\paragraph{Run 2}
Le LHC a repris les collisions en mai 2015 à une énergie de 13 TeV. La collaboration LHCb a annoncé le 14 juillet 2015 l'observation d'états pentaquarks en étudiant les données du run 1\cite{pentaq}. Au bout de ce run, fin 2018, le LHC devrait faire des collisions avec une énergie totale de 14 TeV dans le système du centre de masse (sa puissance nominale) et une luminosité d'environ $1,7*10^{34}$ $cm^{-2}s^{-1}$. L'une des principales lignes directrices des expériences ATLAS et CMS sera d'étudier les caractéristiques du boson de Higgs découvert depuis peu, et de tester les différentes extensions du modèle standard. Un deuxième arrêt long est prévu de 2019 à 2020.

\paragraph{Run 3}
Le Run 3 a d'ores et déjà été approuvé par la direction du CERN et devrait permettre d'atteindre une luminosité de $2,0*10^{34}$ $cm^{-2}s^{-1}$ à la fin de l'année 2023.

\paragraph{Et après ?}
Le fonctionnement du LHC est prévu jusqu'aux environs de 2030. Des modifications supplémentaires seront sans doute apportées à l'accélérateur pour avoir une meilleure résolution et fonctionner à une plus grande énergie.

\subsection{Le détecteur CMS}
\paragraph{Présentation générale}
Le détecteur CMS est un détecteur polyvalent installé dans la caverne souterraine de Cessy (en France). Il a été conçu pour étudier la physique à l'échelle de TeV, mieux comprendre le modèle standard, et  discriminer ses extensions, notamment les théories supersymétriques. CMS a d'ores et déjà été un acteur central de grandes découvertes, comme nous l'avons vu. Il a été conçu pour permettre une bonne reconstruction des états finaux di-jets et de l'énergie transverse manquante, et est hermétique jusqu'à une pseudo-rapidité de $\eta=5$. La pseudorapidité étant définie par :
$$\eta = \frac{1}{2}ln(\frac{|\vec{p}|+p_z}{|\vec{p}|-p_z})$$ 
cela correspond à un angle maximal de déviation par rapport au faisceau de $0,7°$ environ. Le diamètre du détecteur est 15 mètres, sa longueur 27,8 mètres.Ce volume est occupé par différents détecteurs pour une masse totale de 14000 tonnes.

\begin{figure}[!h]
	\centering
	\includegraphics[scale=0.6]{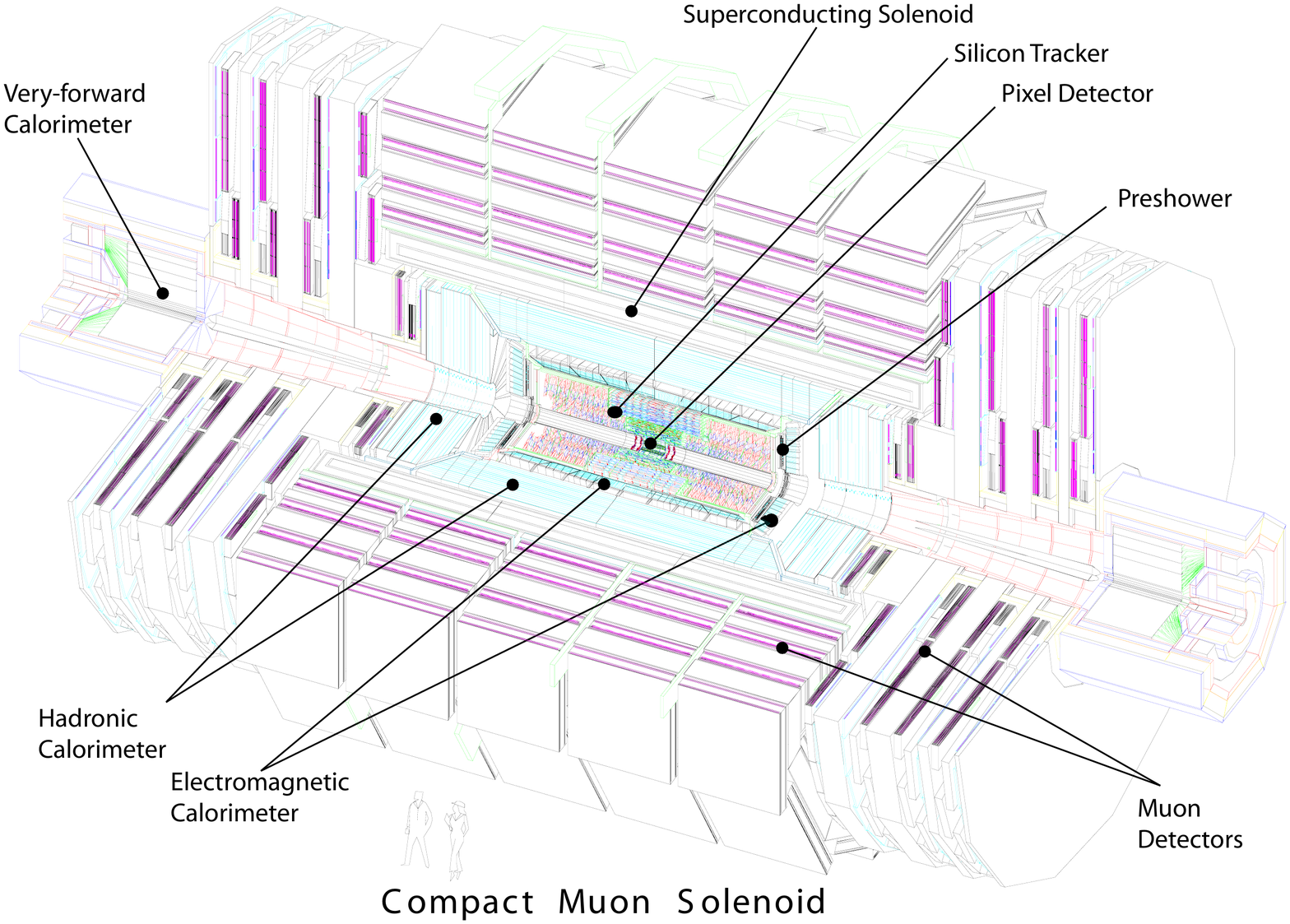}
	\caption{Coupe du détecteur CMS}
\end{figure}

Le volume central du détecteur est le "barrel", les deux éléments extrémaux sont dénommés "endcaps". La collaboration CMS utilise un repère orthonormé direct $(x,y,z)$ basé au point de collision, l'axe $x$ pointe vers le centre de l'anneau LHC, l'axe $y$ est dirigé vers le haut, et l'axe $z$ est porté par une tangente à la trajectoire des protons au point de collision, et pointe dans le sens anti-horaire. 
L'angle de longitude $\phi$ est défini dans le plan (x,y) comme l'angle par rapport à l'axe x ; la distance à l'axe du faisceau est notée r. L'angle de colatitude est l'angle $\theta$, dans le plan $(r,z)$ (de manière équivalente, comme on l'a vu, la pseudorapidité $\eta=-ln(tan(\theta/2))$ peut être employée). La norme de l'impulsion de la particule dans le plan transverse est notée $p_t$. 

Différents sous-détecteurs sont disposés en couches tout autour de la trajectoire du faisceau. 
Le cœur de CMS est le solénoïde niobium-titane supraconducteur produisant un champ de 3.8 T, à une température de 4.5 K. Le système de tracking et les calorimètres sont situés dans le cylindre qui supporte le solénoïde, tandis qu'un détecteur à muons situé à l'extérieur de celui-ci, utilise le retour de 2 T du champ du solénoïde dans la structure en acier qui entoure l'électroaimant.

Chacun des détecteurs permet de faire des mesures sur une classe de particules : le tracker mesure l'impulsion des particules chargées, le calorimètre électromagnétique permet de mesurer l'énergie des électrons et des photons, le calorimètre hadronique, l'énergie des hadrons chargés et neutres, et les détecteurs à muons servent à identifier ces derniers et à mesurer leur impulsion.

Les informations fournies par l'ensemble des détecteurs sont souvent redondantes mais cela permet d'augmenter la précision des mesures. 

\begin{figure}[!h]
	\centering
	\includegraphics[scale=0.45]{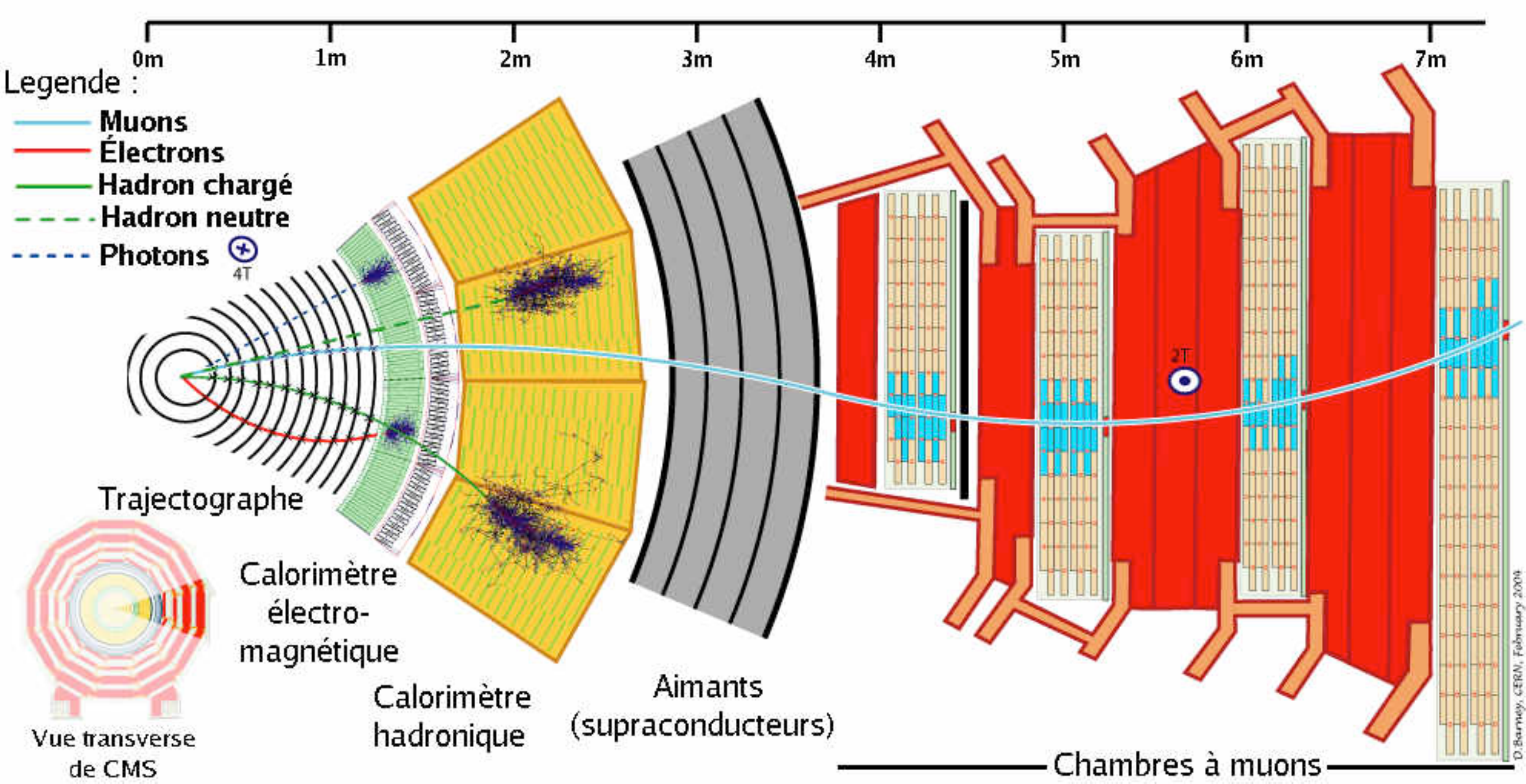}
	\caption{Disposition des différents sous-détecteurs dans CMS}
\end{figure}

\paragraph{Tracker}
Le système de traçage (ou tracker) des particules est situé juste autour du point d'interaction, le but étant de pouvoir observer des particules dont le temps de demi-vie est trop court pour qu'elles soient véritablement mesurées, mais suffisamment long pour qu'on observe un vol de quelques millimètres. C'est le cas notamment des quarks bottom $b$ et des leptons tau $\tau$. Par exemple, pour le lepton $\tau$ dont le temps de demi-vie est de $T=(290.3±0.5)\times 10^{-15} s$ (d'après le "Particle Physics Booklet" de Chinese Physics C), la particule parcourt une distance de vol $$l=\gamma cT\approx 1,5 mm$$ si son énergie est l'énergie de seuil de $30 Gev$ que nous prendrons plus tard pour notre étude. Bien sur, le tracker permet aussi de suivre les particules plus stables. \\\\
Le système est constitué de trois régions selon la distance au centre de collision. Tout d'abord, on trouve des pixels de silicone de $100\times150 \mu m$ (soit 66 millions par mètre carré), de résolution $10 \mu m$ dans le plan $(r,\phi)$ et $20 \mu m$ dans le plan $(r,z)$. Les deux autres "couches" sont constituées de bandes de silicone (c'est la taille des bandes qui les différentient). La résolution varie de $20$ à $50 \mu m$ dans le plan transverse, et de $200$ à $500 \mu m$ dans la direction $(r,z)$. Il y a en tout $9,6$ millions bandelettes.\\\\
Le matériau a été choisi pour que l'interaction perturbe le moins possible les particules, qui doivent ne laisser derrière elle qu'une toute petite fraction de leur énergie dans le tracker pour ne pas biaiser les mesures des calorimètres. Des algorithmes recueillent ensuite ces données puis les reconstituent en trajectoires de particules. Le tracker est localisé dans le solénoïde de CMS donc les trajectoires des particules chargées sont courbées par la force de Lorentz, et cela permet d'obtenir le rapport de leur charge avec leur masse effective.

\paragraph{Le calorimètre électromagnétique}
Autour du tracker est situé le calorimètre électromagnétique (ECAL). Il mesure l'énergie des électrons et des photons en provoquant des douches électromagnétiques dans les cellules du détecteur, ce qui entraîne l'absorption totale de l'énergie des particules incidentes et sa transformation en un signal, dans ce cas, un scintillement.\\\\
Dans CMS, le ECAL est composé de 75000 cristaux de tungstate de plomb ($PbWO_4$). Le choix du matériau est motivé pour l'optimisation de la reconstruction des particules. Chaque cristal fait environ $23$ cm de long. Le calorimètre a, comme CMS, une forme cylindrique. La résolution est très bonne ; pour des électrons de $45 GeV$, elle est d'environ $2 \%$ si les électrons sont détectés dans le barrel, et comprise entre $2 \%$ à $5 \%$ si les électrons sont détectés dans les endcaps. \\
Pour un système de deux photons ayant la topologie correspondant à deux photons issus de la désintégration d'un boson de Higgs, la résolution varie entre $1,1 \%$ et $2,6 \%$ dans le barrel, et entre $2,2 \%$ et $5 \%$ dans les endcaps.

\paragraph{Le calorimètre hadronique}
Les mesures de ECAL sont complémentées par les mesures du calorimètre hadronique (HCAL), qui reposent sur la transformation des hadrons en douches hadroniques. Ces mesures sont forcément moins précises que celles du ECAL, à cause de la structure des interactions fortes. Cependant, elle est indispensable puisque c'est la seule manière de faire des mesures sur les hadrons neutres, et aussi d'identifier et reconstruire les jets.
Le HCAL est constitué de cuivre, puisque le champ magnétique intense impose l'utilisation de matériaux non magnétiques, et d'un plastique actif pouvant scintiller (ce qui permet les mesures. Il y a un deuxième type de HCAL sur les bases du cylindre, composé d'absorbeurs en acier). Les mesures passent par la détection par des tubes photo-multiplicateurs de lumière Cherenkov produite par des fibres de quartz incrustées dans l'absorbeur.\\ 
Le HCAL est situé entre le ECAL et le solénoïde de CMS, cependant, cela ne laisse pas suffisamment de place pour pouvoir contenir entièrement la douche hadronique, si bien qu'un calorimètre complémentaire est rajouté autour de l'aimant.

\paragraph{Les détecteurs à muons}
Ce sont les détecteurs les plus en périphérie de CMS. $200$ fois plus lourds que les électrons, n'interagissant pas non plus par interaction forte, et avec un temps de vie relativement long, les muons pénètrent beaucoup plus loin dans le détecteur (parce qu'ils perdent moins d'énergie par rayonnement) que les électrons ou hadrons, et atteignent ces chambres à muons. L'impulsion des muons est mesurée par la courbure de leur trajectoire qui est suffisamment particulière pour être devenu un symbole de la collaboration CMS, présent sur son logo. En effet, le détecteur à muons est contenu dans une structure ferromagnétique jouant le rôle de réflecteur pour le champ du solénoïde. Ainsi, les trajectoire des muons est courbée d'abord dans un sens, puis dans l'autre ce qui permet une détection beaucoup plus efficace. Le détecteur comporte trois détecteurs à gaz différents.

\paragraph{Le trigger}
Le système de sélection des évènements, ou trigger and data acquisition system (TriDAQ), permet de sélectionner sur l'ensemble des collisions qui se produisent dans le détecteur, celles qui sont potentiellement les plus intéressantes. Le fréquence de collision est d'environ $40 MHz$, ce qui est beaucoup trop pour espérer stocker toute l'information, en vue d'analyses postérieures. Deux niveaux sont distingués dans le procédé. \\\\
Le premier niveau constitue le level-1 trigger, un système hardware rapide qui permet de ne garder que quelques milliers ou dizaines de milliers d'évènements par seconde. L'algorithme repose en majeure partie sur la conservation des évènements comportant des objets de haute impulsion transverse. Il ne bénéficie que d'informations dites de "basse granularité" et de basse résolution. La décision prend un temps caractéristique de $1 \mu s$, durant lequel les informations plus précises peuvent être collectées. Si l'évènement est sélectionné par le level-1, toutes les informations sont transmises au High Level Trigger (HLT). Le HLT a plus de temps pour "prendre une décision", et repose donc sur des algorithmes plus complexes ; certains permettent même de reconstituer des quarks $b$ ou des $\tau$. Si l'évènement en question est lui-aussi jugé digne d'intérêt par le HLT (il n'en reste que quelques dizaines par seconde en moyenne), toutes les informations nécessaires sont stockées pour une analyse off-line.

\section{Production di-Higgs non résonnante par fusion de gluons}

\subsection{Motivations de l'étude}

La production double Higgs est actuellement l'une des méthodes sur lesquelles se concentrent les espoirs pour la mesure du couplage trilinéaire $\lambda_{HHH}$ du champ de Higgs. Cette étude est considérée à faire durant la phase de haute luminosité de LHC (Run 3) mais il est possible qu'il y ait déjà des informations utiles à la fin du Run 2.\\

En particulier, il devrait être possible de tester les grandes déviations de $\lambda$ dans trois ans, grâce à une sensibilité accrue à des effets "Beyond the Standard Model" (BSM), modifiant de manière sensible le potentiel de Higgs. Entre autres :
\begin{enumerate}
	\item des couplages de Yukawa ne correspondant pas à ceux de Modèle Standard (SM)
	\item des interactions non linéaires $ttHH$ via le paramètre '$c_2$' (voir après)
	\item des opérateurs Higgs-gluons de dimension 6
	\item des champs scalaires légers interagissant par interaction forte
	\item des résonances extra-dimensionnelles
	\item des partenaires supersymétriques
\end{enumerate}

Les modes de production double-Higgs les plus probables sont au nombre de quatre. La fusion de gluons : 

\begin{figure}[!h]
	\centering
	\includegraphics{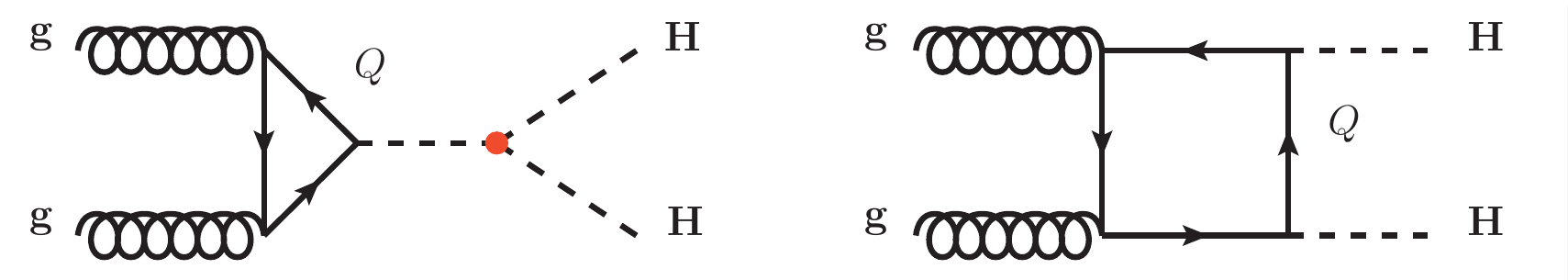}
	\caption{Production di-Higgs par fusion de gluons}
\end{figure}

est sensible, en plus de l'être au couplage lambda, à des particules lourdes 'colorées', et à des couplages top-Higgs anormaux.\\\\
La production par fusion de bosons vecteurs est quant'à elle sensible aux couplages anormaux entre boson de Higgs et bosons vecteurs :

\begin{figure}[!h]
	\centering
	\includegraphics{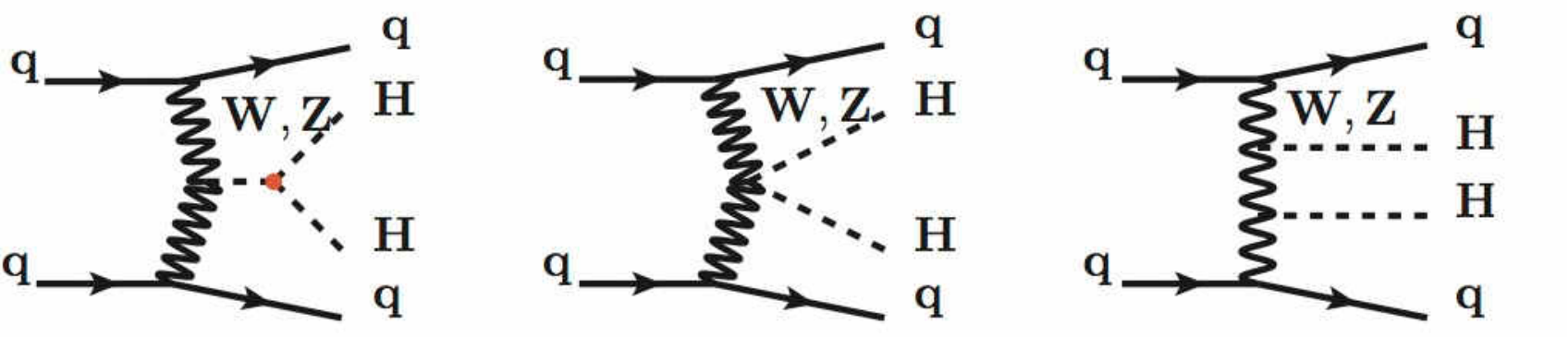}
	\caption{Production di-Higgs par fusion de bosons vecteurs}
\end{figure}

Les processus du type $gg\rightarrow ttHH$ : 

\begin{figure}[!h]
	\centering
	\includegraphics{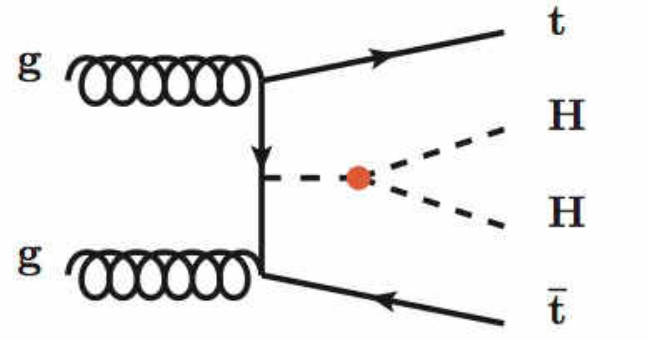}
	\caption{Production di-Higgs et $t\overline{t}$}
\end{figure}

dépendent également du couplage top-Higgs. Enfin la production 'directe' à partir de bosons vecteurs est sensible aux couplages non linéaires de type $VVHH$ (présent sur le graphe de Feynman le plus à droite) : 

\begin{figure}[!h]
	\centering
	\includegraphics[scale=0.75]{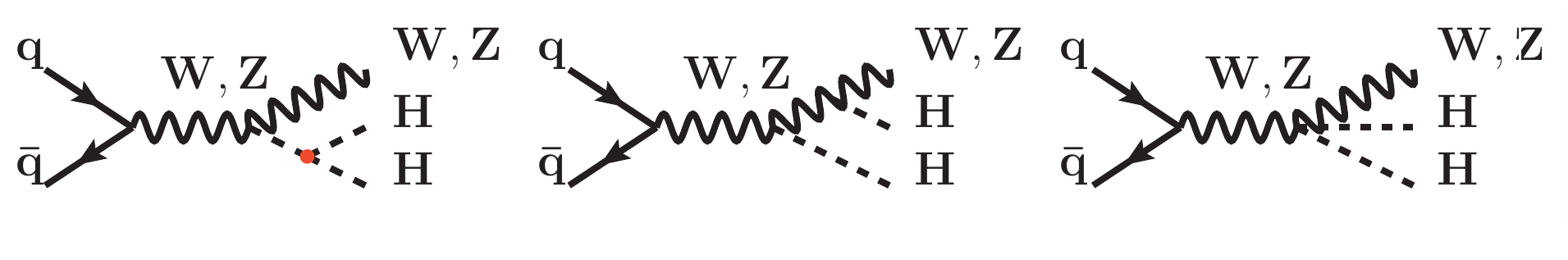}
	\caption{Production di-Higgs "directe" à partir de bosons vecteurs}
\end{figure}

Dans la suite, nous nous concentrerons sur les procédés par fusion de gluons.

\subsection{Lagrangien effectif}
Pour la production double-Higgs par fusion de gluons, les termes relevants dans le lagrangien des théories effectives en dimension 6 sont \cite{lagran}:
	$$\mathcal{L}_{hh}=-\frac{m_h^2}{2v}(1-\frac{3}{2}c_H+c_6)h^3+\frac{\alpha_sc_g}{4\pi}(\frac{h}{v}+\frac{h^2}{2v^2})G_{\mu\nu}^aG^{\mu\nu}_a$$ 
$$-[\frac{m_t}{v}(1-\frac{c_H}{2}+c_t)\overline{t}_Lt_Rh+...]-[\frac{m_t}{v^2}(\frac{3c_t}{2}-\frac{c_H}{2})\overline{t}_Lt_Rh^2 + ...]$$
avec dans ce somme de quatre termes, le premier qui correspond au couplage trilinéaire du Higgs, le deuxième aux interactions de contact Higgs-gluons, le troisième aux couplages Higgs-tops, et le quatrième à l'interaction non linéaire $ttHH$. On peut choisir une autre paramétrisation : 
$$\color{blue}\lambda=1-\frac{3}{2}c_H+c_6$$
$$\color{green}y_t=1-\frac{c_H}{2}+c_t$$
$$\color{red}c_2=\frac{3c_t}{2}-\frac{c_H}{2}$$
Ces paramètres interviennent de la façon suivante dans les différents procédés : 

\begin{figure}[!h]
	\centering
	\includegraphics[scale=0.6]{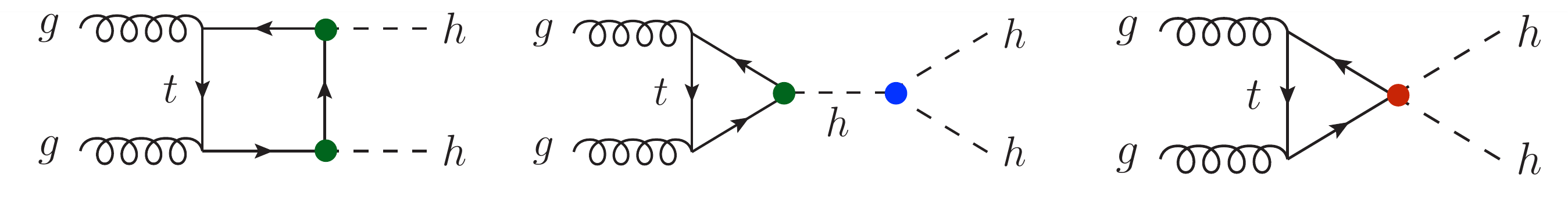}
	\caption{Modification des sommets colorés dans le cadre de la théorie effective}
\end{figure}

\paragraph{Cas du modèle standard}
Si on se place dans le cadre du modèle standard, les paramètres $\lambda$, $y_t$ et $c_2$ prennent leur valeur standard, ce qui implique en particulier que $c_2=0$. Le calcul de la section efficace des différents procédés est possible au troisième voire quatrième (séminaire Higgs Hunting 2015) ordre, ce qui permet d'obtenir le tableau suivant, issu des recommandations de "Higgs Higgs cross section working group" :
\begin{center}
\begin{tabular}{|c|c|}
	\hline
	Énergie dans le centre de masse (TeV) & $\sigma$ (femtobarn fb) \\
	\hline
	7 & 6.85 \\
	8 & 9.96 \\
	13 & 34.3 \\
	14 & 40.7 \\
	\hline
\end{tabular}
\end{center}
En comparaison, à 13 TeV, la section efficace de la contribution par fusion de bosons vecteurs à la production di-Higgs est de 6.8 fb, celle de ttHH de 0.7 fb. 

\paragraph{Cas du couplage $\lambda$ anormal}
Si on laisse le couplage $\lambda$ varier, tout en conservant $y_t^{SM}$ et $c_2=0$, la section efficace est modifiée de la manière suivante : 
\begin{center}
	\begin{tabular}{|c|c|}
		\hline
		$\lambda/\lambda^{SM}$ & $\sigma/\sigma^{SM}$ \\
		\hline
		-4 & 12 \\
		0 & 2.2 \\
		1 & 1 \\
		2.46 & 0.42 \\
		20 & 105 \\
		\hline
	\end{tabular}
\end{center}
et la contribution des différents modes de production peut être estimée par des méthodes Monte-Carlo (en utilisant MadGraph5) :

\begin{figure}[!h]
	\centering
	\includegraphics[scale=0.52]{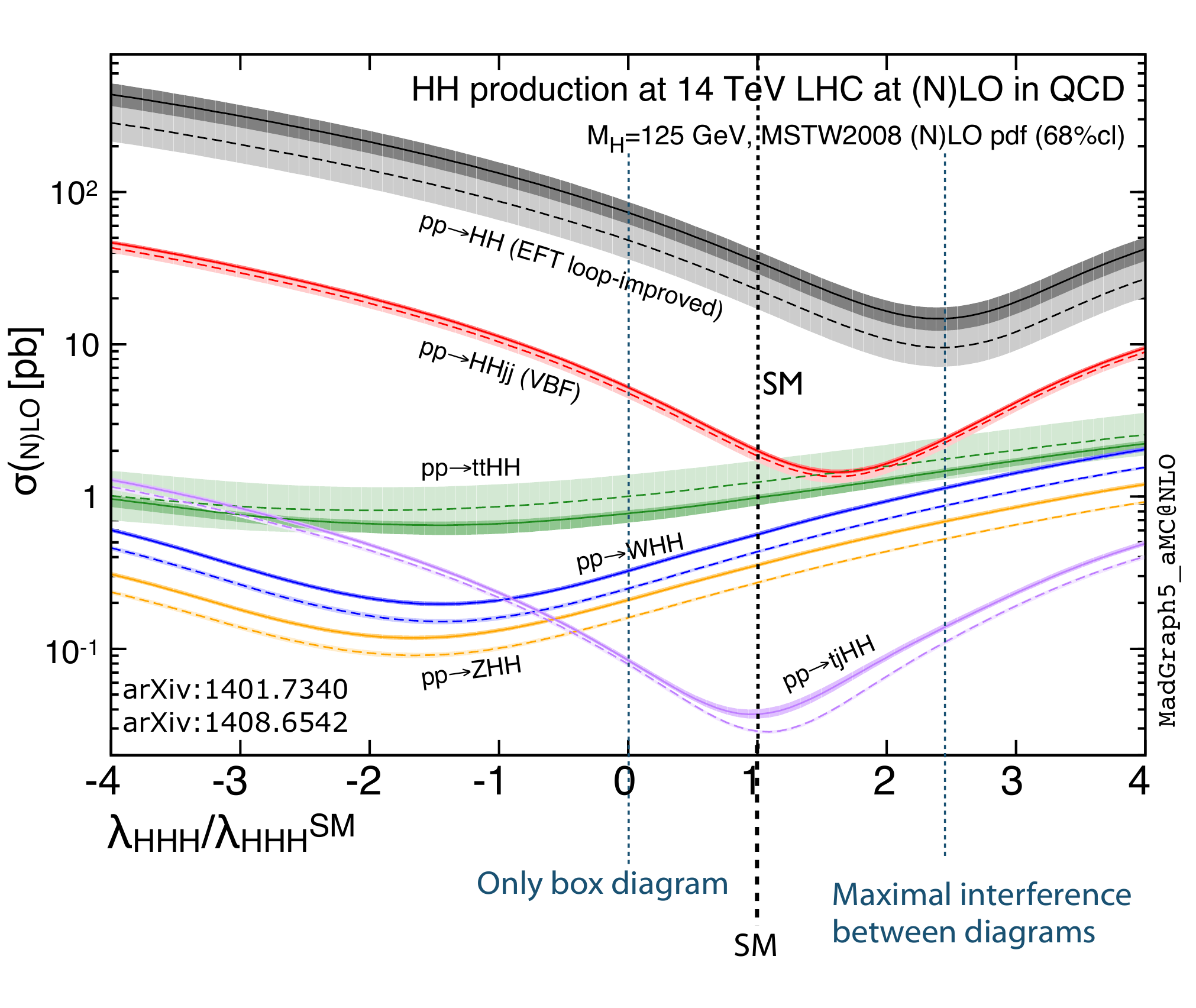}
	\caption{Évolution de la section efficace de différentes production di-Higgs en fonction de $\frac{\lambda}{\lambda^{SM}}$}
\end{figure}

\subsection{Sélection des évènements}
Il faut choisir à quels modes de désintégrations du système di-Higgs on s'intéresse, pour déterminer ce qui doit être mesuré dans les collisions de CMS pour pouvoir remonter au couplage $\lambda$. Ce choix dépend essentiellement de deux facteurs : tout d'abord, il faut que le branching ratio du canal ne soit pas trop faible, autrement dit il faut que la probabilité que le système se désintègre de manière voulue ne soit pas trop faible si on veut pourvoir l'observer. De plus, il faut que le fond d'évènements divers dont l'état final est proche de l'état final choisi soit relativement peut intense, en particulier pour le "irreductible background", c'est-à-dire les évènements pour lesquels l'état final est rigoureusement identique (d'un point de vue détecteur). Ces problèmes ont été posés et étudiés dans des études phénoménologiques \cite{Spira} \cite{Spannow}.

Dans la suite, les valeurs des sections efficaces et les branching ratios sont calculées en utilisant les prévisions du modèle standard. Cela se justifie par l'excellente concordance (actuelle) entre les mesures des expériences du LHC des propriétés du Higgs et les prévisions du modèle standard.

La section efficace de la production di-Higgs étant très faible, de l'ordre de quelques dizaines de femtobarns, il est important de choisir un canal de désintégration assez probable pour pouvoir recueillir suffisamment de données. Le "Higgs cross-section working group" donne le tableau suivant pour les ratios de désintégration les plus probables du boson de Higgs  :
\begin{center}
	\begin{tabular}{c|c}
		canal & BR (\%) \\
		\hline
		$b\overline{b}$ & 57.7 \\
		WW & 21.5 \\
		$\tau\tau$ & 6.32 \\
		ZZ & 2.64 \\
		$\gamma\gamma$ & 0.228 \\
	\end{tabular}
\end{center}

Le meilleur canal du point de vue du branching ratio est la désintégration en $b\overline{b}$. Cependant, la désintégration d'une paire de Higgs en quatre quarks bottoms est recouverte, au LHC, par du fond QCD extrêmement intense, ce qui exclut toute étude sous cette hypothèse. Cependant, pour maximiser le branching ratio tout en gardant une bonne maîtrise du fond QCD, on regarde en général des états finaux contenant deux quarks bottoms. On peut aussi également regarder des canaux contenant deux $\tau$ et deux bosons.
\begin{enumerate}
	\item Le canal $bb\tau\tau$ a un bon branching ratio, et le fond irréductible le plus important est le processus $t\overline{t}\rightarrow b\overline{b}\tau\nu_\tau\tau^-\overline{\nu}_\tau$ qui a une petite section efficace comparée à la production énorme de paires $t\overline{t}$ compte tenu du petit branching ratio de $W\rightarrow\tau\nu_\tau$ (qui est de $11,25±0,20\%$).
	\begin{figure}[!h]
		\centering
		\includegraphics{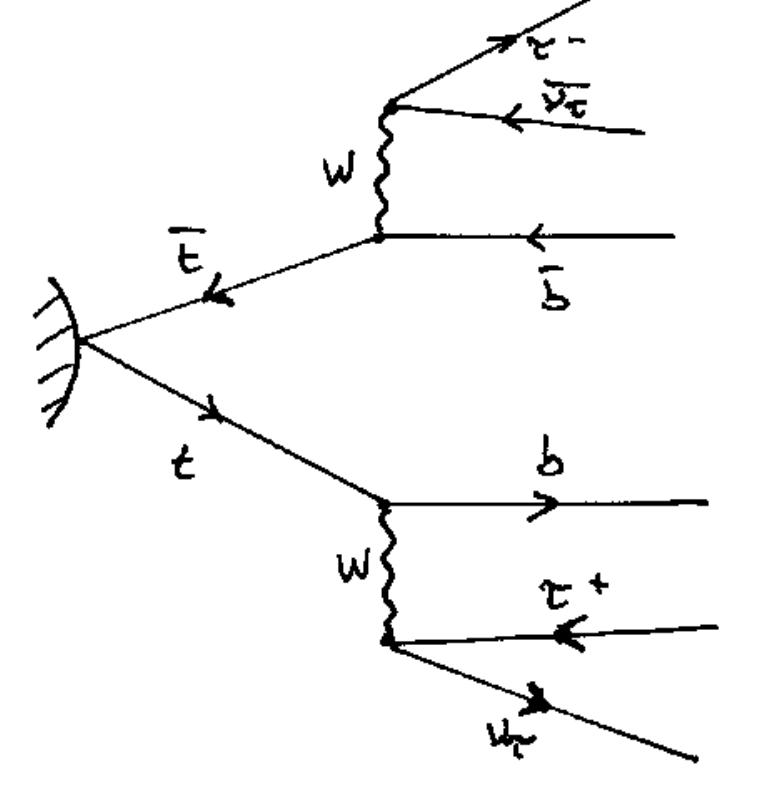}
		\caption{Diagramme de Feynman au  premier ordre pour $t\overline{t}\rightarrow b\overline{b}\tau\nu_\tau\tau^-\overline{\nu}_\tau$}
	\end{figure}
	\item Le canal $bb\gamma\gamma$ est très peu probable, mais le fond très peu intense et essentiellement du à des processus de QCD et à de la production d'un boson de Higgs avec une paire de top. Ainsi même si on s'attend à avoir très peu d'évènements, la sensibilité de cette topologie d'évènements reste intéressante.
	\item Le canal $bbWW$ est peu intéressant parce que très contaminé par le background $t\overline{t}$
	\item Il doit être possible de rejeter de manière très efficace le background pour le canal $bbZZ$, surtout pour un état final de quatre leptons. Cependant, le branching ration reste faible.
	\item $ZZ\tau\tau$ et $\gamma\gamma\tau\tau$ sont des canaux intéressants du point de vue du background mais ils n'ont pas pu être utilisés durant le Run 1 tellement leur branching ration est petit. Ils pourraient devenir très intéressant pour le HL-LHC.
	\item Le canal $WW\tau\tau$ a un bon branching ratio, mais la présence de bosons W entraîne une mauvaise reconstruction de la masse des Higgs.
\end{enumerate}

Le canal $bb\tau\tau$ semble donc le plus intéressant pour mesurer le couplage trilinéaire du champ de Higgs, comme cela a été souligné dans les études phénoménologiques référencées, pour les prochaines années d'expériences au LHC. L'étude qui suit se focalise donc sur ce canal.

\begin{figure}[!h]
	\centering
	\includegraphics[scale=0.8]{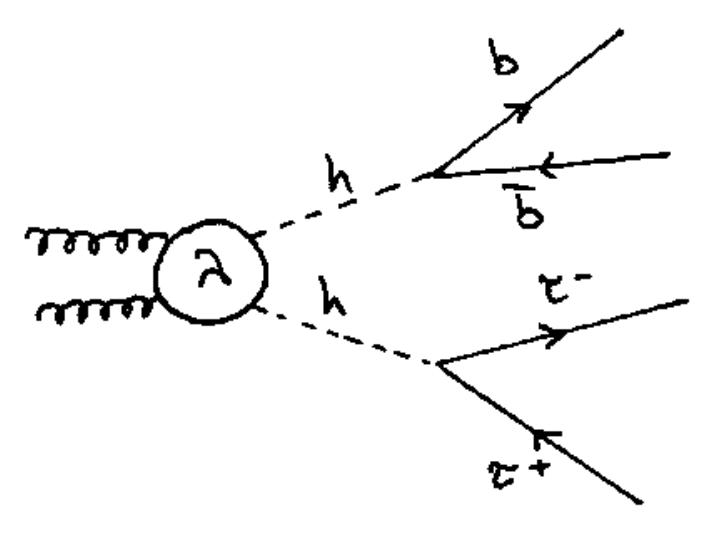}
	\caption{Le cadre de l'étude est le suivant : on fait varier le paramètre $\lambda$ qui modifie ce qui se passe dans la "boîte noire" qu'est le rond portant le $\lambda$. Le reste est fixé, c'est-à-dire que les générateurs Monte-Carlo ne simulent que des évènements qui débutent par fusion de gluons et s'achèvent avec la désintégration de système di-Higgs en $b\overline{b}\tau\tau$, suivie éventuellement de la désintégration des $\tau$.}
\end{figure}

\subsection{Simulation et génération des évènements}
Les fichiers utilisées sont des échantillons générés par des méthodes Monte-Carlo. Les processus durs de QCD sont générés par MadGraph \cite{madgraph} pour la production di-Higgs, et par powheg \cite{powheg} pour le fond irréductible $t\overline{t}$. Les parties de showering, d'hadronisation ou encore les corrections sont simulées par PYTHIA 8 \cite{pythia}. Pour la reconstruction des particules dans le détecteur CMS, GEANT4 \cite{geant} donne les résultats de l'interaction des particules avec le détecteur (ce qu'observe le tracker, combien d'énergie est déposée dans chaque calorimètre...), et enfin, l'environnement CMSSW \cite{cmssw} permet de simuler la réponse du détecteur et la reconstruction des particules.\\\\
Les informations sont stockées dans différents "Branch" d'un "Tree" (fichier ROOT \cite{root}), et traitées en utilisant l'environnement ROOT et l'interpréteur PyROOT \cite{pyroot}.

\section{Étude de la cinétique sans la désintégration des $\tau$}
\subsection{Reconstruction des évènements}
Le but de cette première étude à un niveau générateur, simpliste puisque nous ne prenons pas en compte la désintégration des leptons $\tau$, est de discriminer les différents invariants cinétiques du procédé et de déterminer ceux dont la distribution dépend visiblement de $\lambda$. Pour chaque particule, on connait les variables $p_t$ (impulsion transverse) qui est la projection de l'impulsion dans le plan transverse, et $\eta=\frac{1}{2}ln(\frac{|\vec{p}|+p_z}{|\vec{p}|-p_z})$ (pseudo-rapidité).\\\\
Dans les évènements candidats à de la production di-Higgs, une fraction non négligeable ne peut pas être enregistrée à cause des caractéristiques techniques du détecteur CMS : l'acceptance géométrique n'est pas parfaite dans la direction du faisceau, et il y a un seuil d'énergie au dessus duquel les particules doivent se trouver pour déclencher le trigger et être reconstituées de manière correcte.
Pour se faire une idée de cette fraction d'évènements perdus, on applique les sélections suivantes aux quarks bottom $b$ et $\overline{b}$ et aux leptons $\tau^-$ et $\tau^+$ :
$$
\left( \begin{array}{c}
p_t > 30 GeV \\
|\eta|<2.5 \\
\end{array} \right.
$$
Ces valeurs de seuil sont comparables à celles utilisées dans les études phénoménologiques (cf le papier de Luca pour les références). Le seuil sur $\eta$ provient directement de la géométrie de CMS, le seuil pour $p_t$ correspond au l'énergie minimale de trigger et de reconstruction des jets. Pour les désintégrations leptoniques des $\tau$, les seuils sont plus bas ; cependant, deux neutrinos sont émis durant cette désintégration ce qui entraine une perte significative d'énergie. On peut donc dire que la valeur de seuil choisie prend approximativement en compte ces particularités.\\

\subsection{Comparaison pour les différentes valeurs de $\lambda$}
Définissons l'efficacité $\epsilon$ de reconstruction, pour les données statistiques, comme le rapport entre le nombre d'évènements passant les coupures statistiques sur le nombre total d'évènements.

Les efficacités pour les différentes valeurs de $\lambda$, et pour le fond $t\overline{t}$ sont données ci-dessous :
\begin{center}
	\begin{tabular}{|c|c|c|c|c|c|}
		\hline
		& $\lambda/\lambda_{SM}=-4$ & $\lambda/\lambda_{SM}=1$ & $\lambda/\lambda_{SM}=2.46$ & $\lambda/\lambda_{SM}=20$ & $t\overline{t}$\\
		\hline
		Efficacité $\epsilon$ & 0.424 & 0.433 & 0.454 & 0.424 & 0.272\\
		\hline
	\end{tabular}
\end{center}

Les distributions en impulsion transverse et en rapidité du quarks et des leptons sont assez semblables pour les différentes valeurs de $\lambda$. Il faut s'intéresser aux invariants cinétiques des systèmes de particules, $h0=[b,\overline{b}]$, $h1=[\tau^-,\tau^+]$ et le système total $h=[b,\overline{b},\tau^-,\tau^+]$. 
On reconstitue les 4-vecteurs des higgs intermédiaires et du système total avec les relations :
$$
\left( \begin{array}{c}
p(h0)=p(b)+p(\overline{b}) \\
p(h1)=p(\tau^-)+p(\tau^+) \\
p(h) =p(h0)+p(h1) \\
\end{array} \right.
$$
Les diagrammes suivants sont la superposition des histogrammes pour chaque valeur étudiée de $\lambda$, mais aussi des histogrammes relatifs au fond $[t,\overline{t}]$. Le but ultime étant de remonter à la valeur de $\lambda$, il est important de sélectionner des grandeurs qui permettent bien de discriminer les évènements provenant de deux bosons de Higgs de ceux qui proviennent d'un quark top et d'un anti quark top. Pour pouvoir comparer uniquement la "forme" des distributions, les histogrammes sont normalisés.

Remarquons tout d'abord que les différentes distributions en rapidité ne sont pas vraiment adaptées à l'extraction du couplage $\lambda$, puisqu'elles n'en dépendent pas fortement. Par exemple, voici ces courbes pour le système $h1=[\tau^-,\tau^+]$ sans les coupures cinétiques, et pour l'ensemble des particules, avec les coupures cinétiques.

\begin{center}
	\includegraphics[scale=0.4]{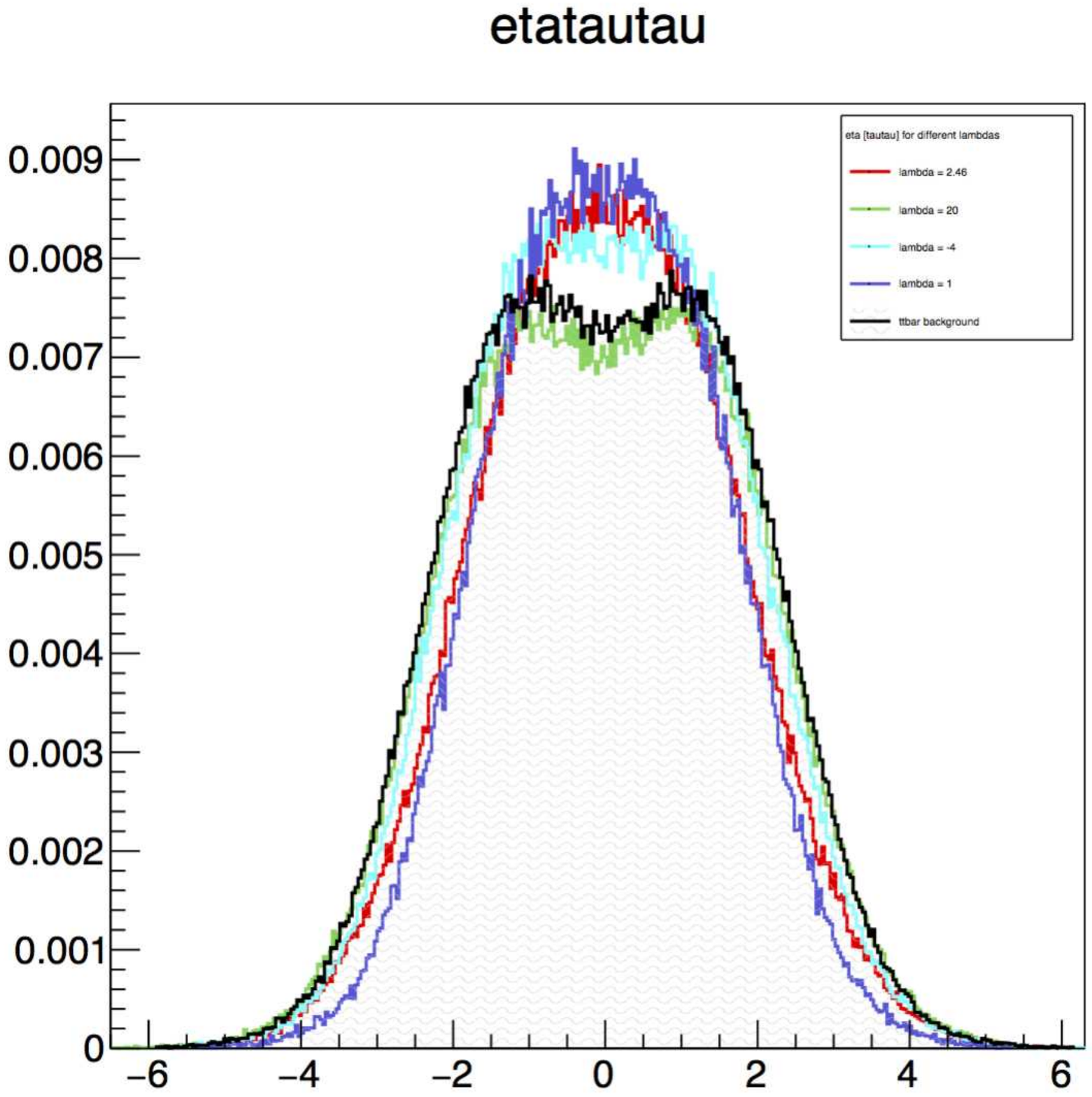}
	\includegraphics[scale=0.4]{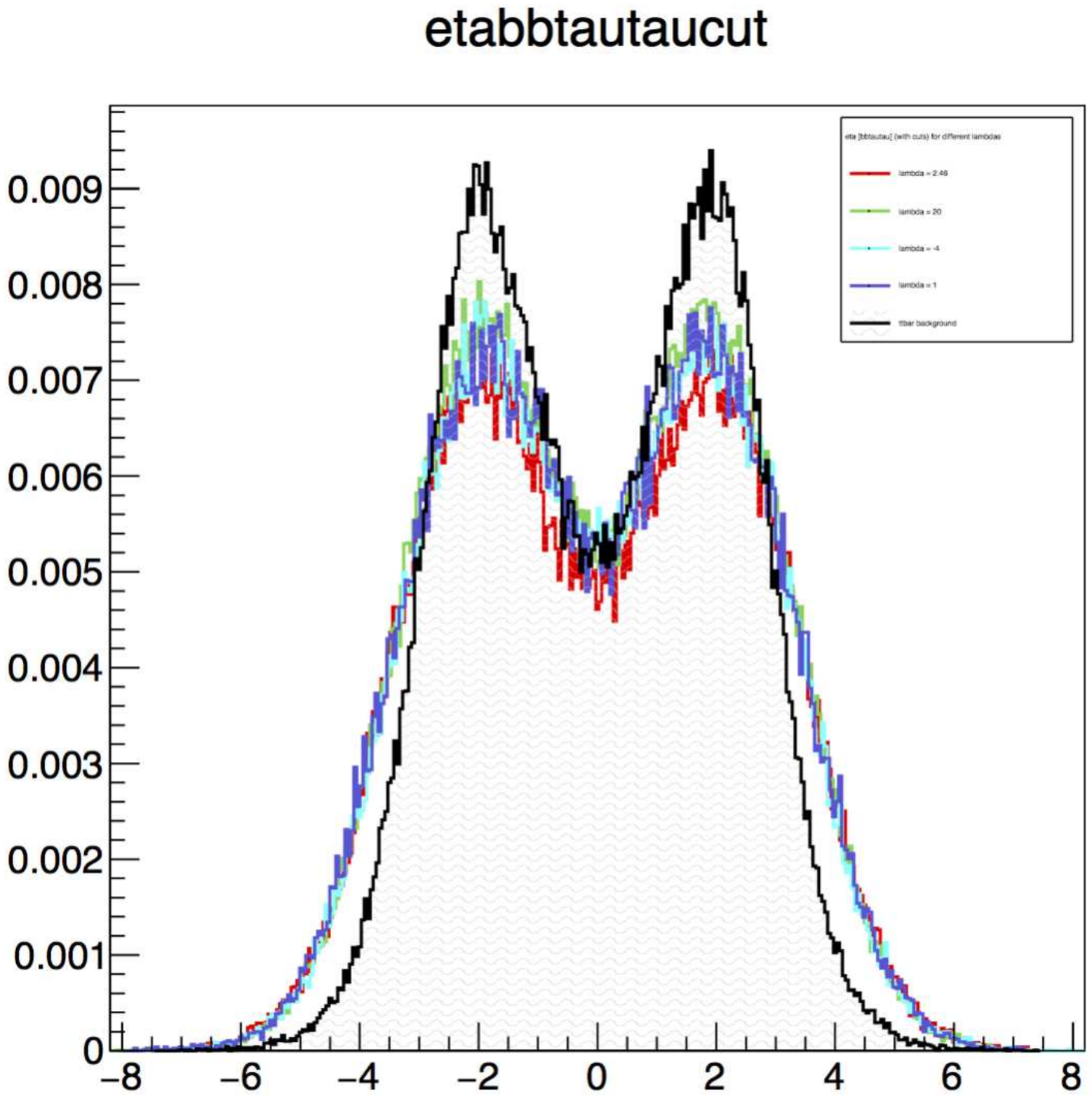}
\end{center}

Pour les autres invariants, nous avons adopté la disposition suivante : le diagramme de gauche est le résultat de la simulation Monte-Carlo, sans les coupures, comme ce que nous observerions en ayant un point de vue omniscient dans les collisions. A droite, les coupures cinétiques déjà évoquées ont été appliquées. Les quatre premiers diagrammes concernent le système $h1=[\tau^-,\tau^+]$, les quatre suivants le système total.\\\\

\begin{center}
	\includegraphics[scale=0.38]{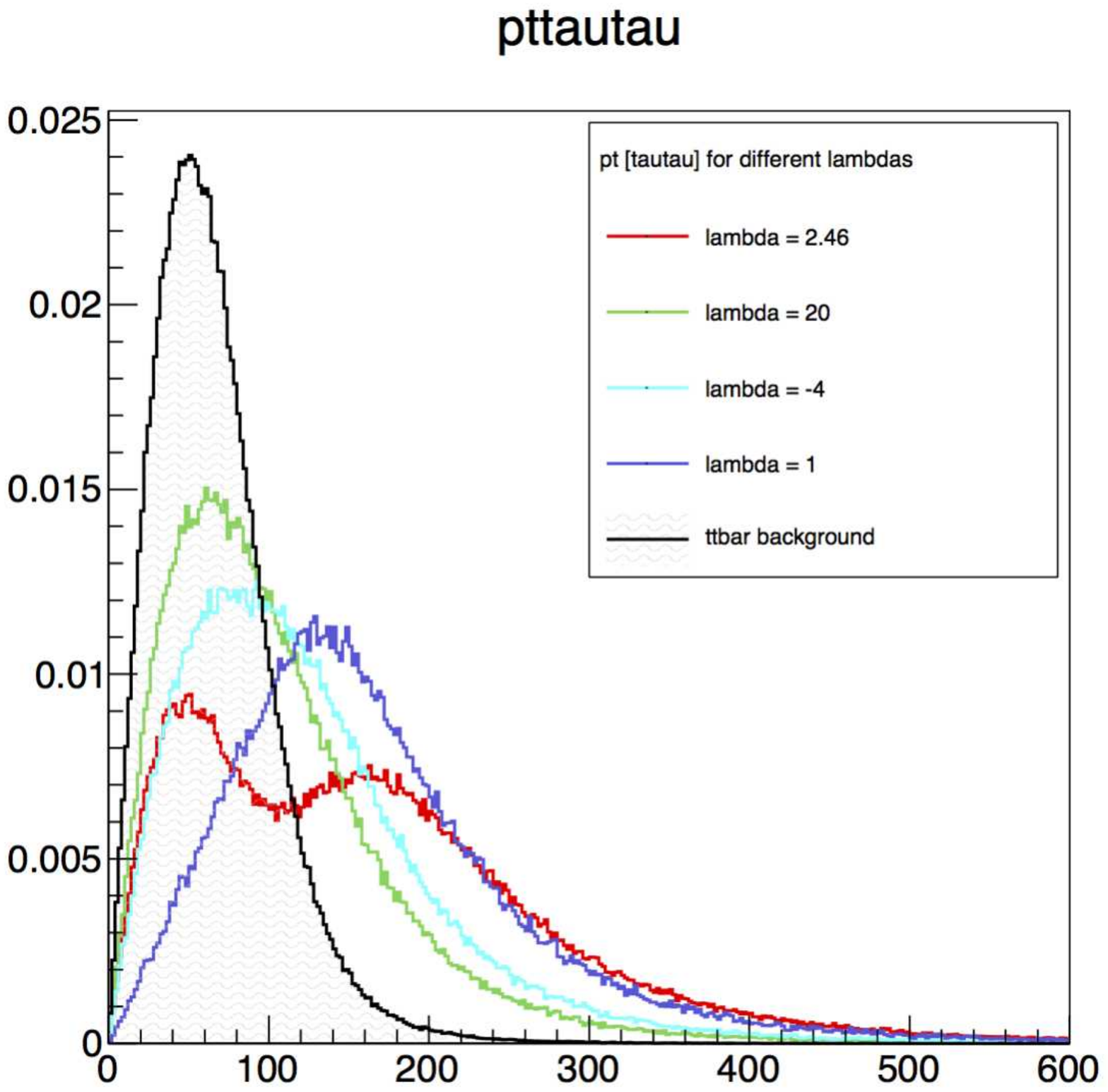}
	\includegraphics[scale=0.38]{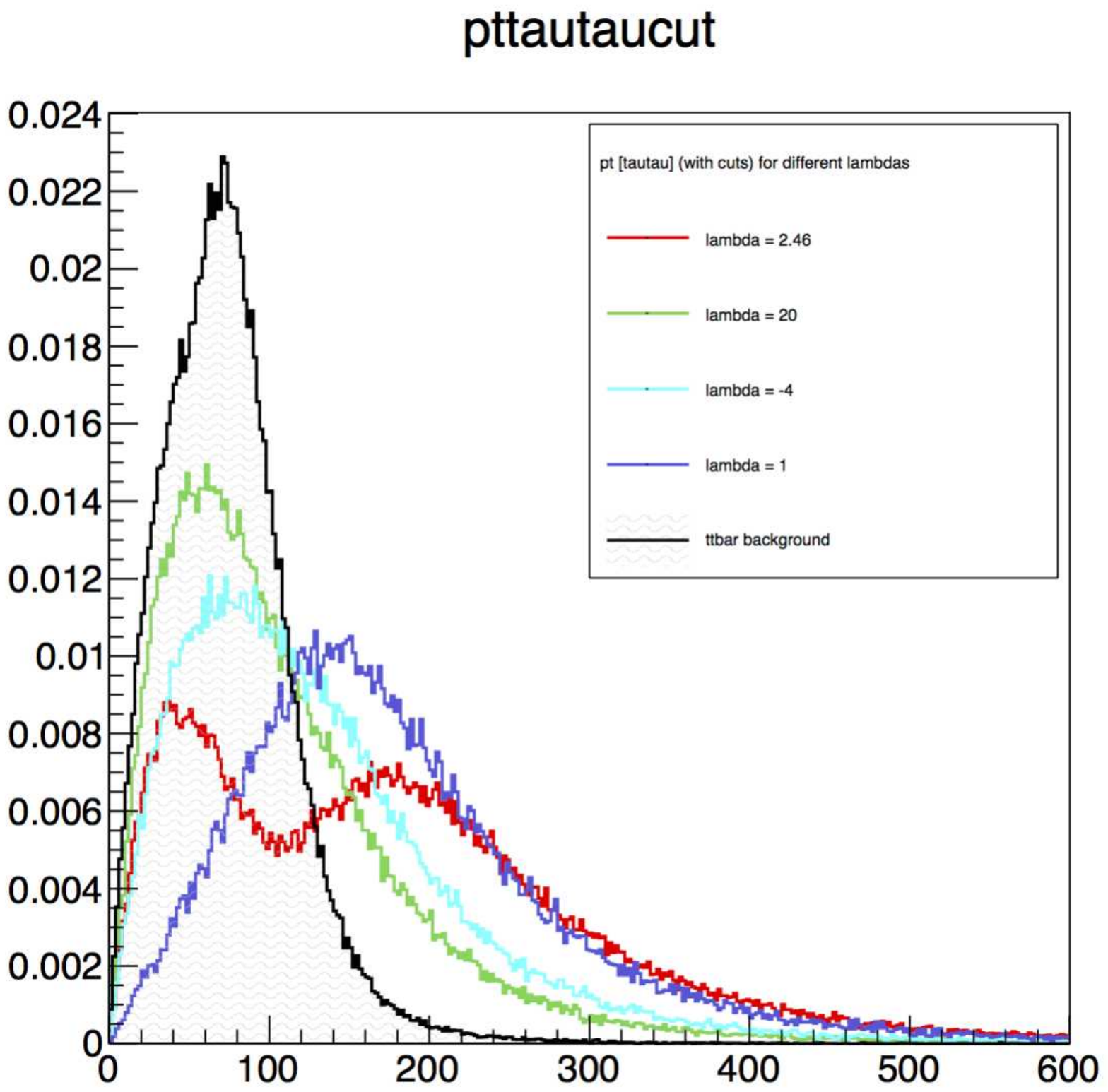}
\end{center}
\begin{center}
	\includegraphics[scale=0.38]{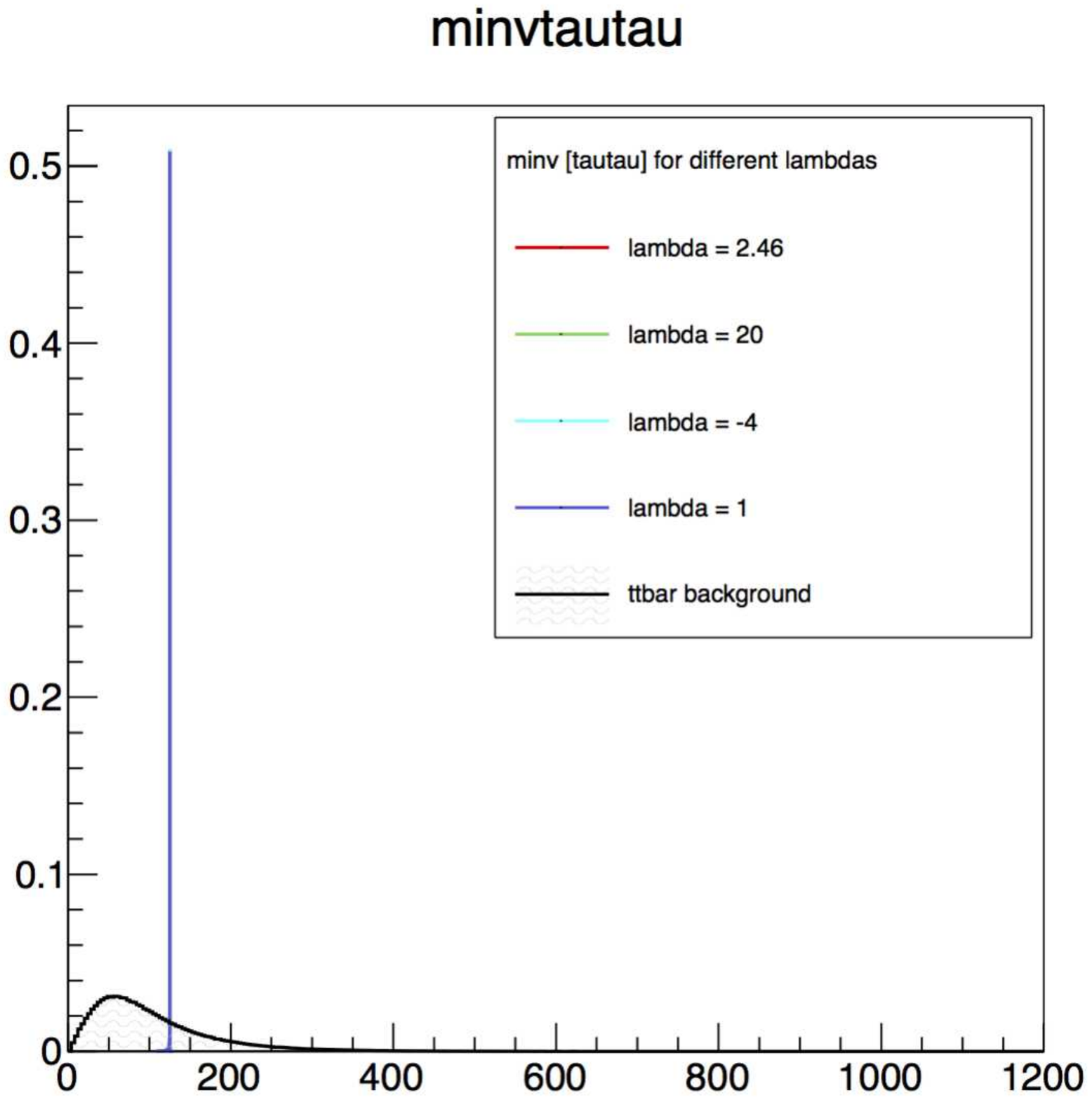}
	\includegraphics[scale=0.38]{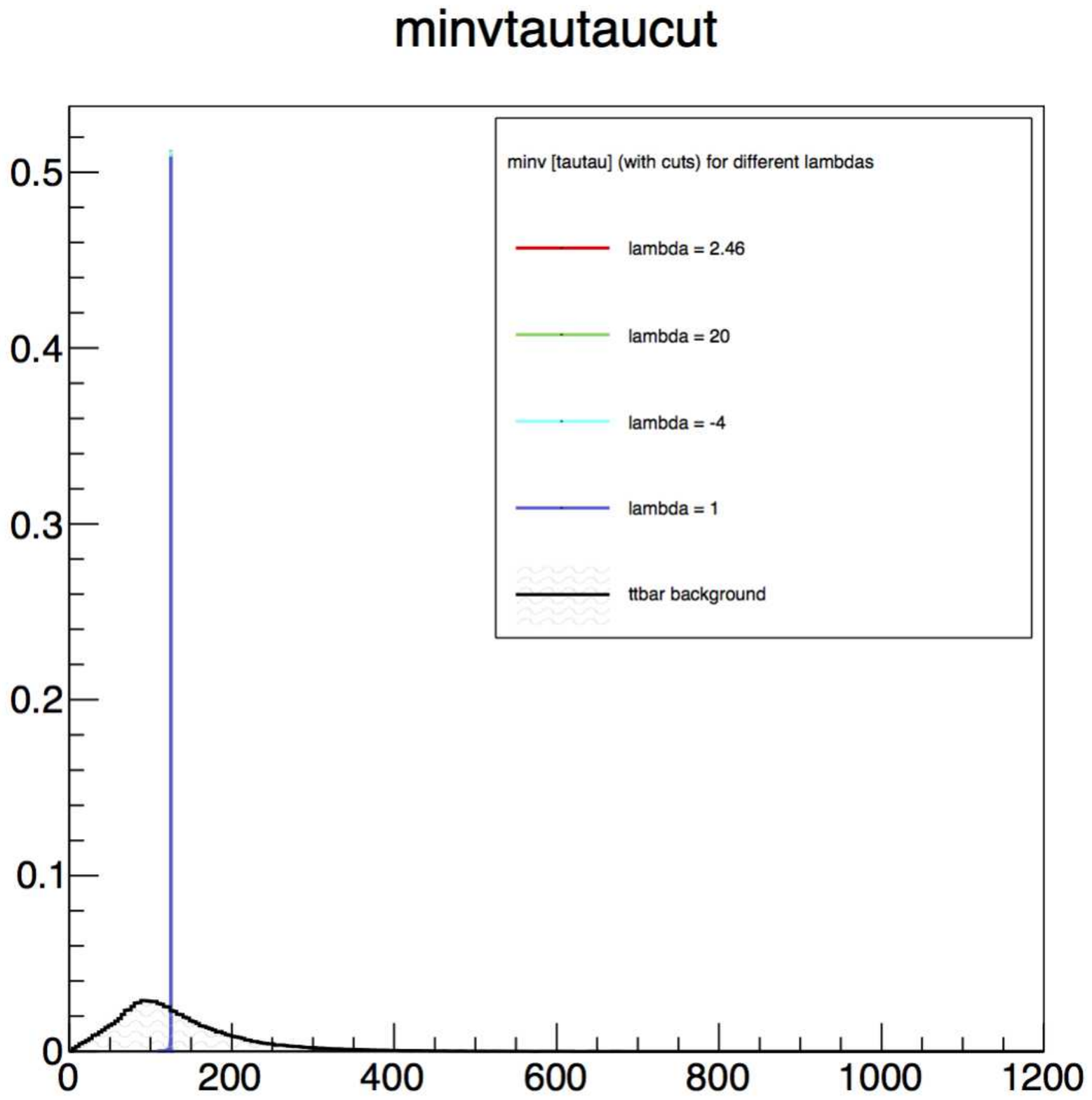}
\end{center}
\begin{center}
	\includegraphics[scale=0.38]{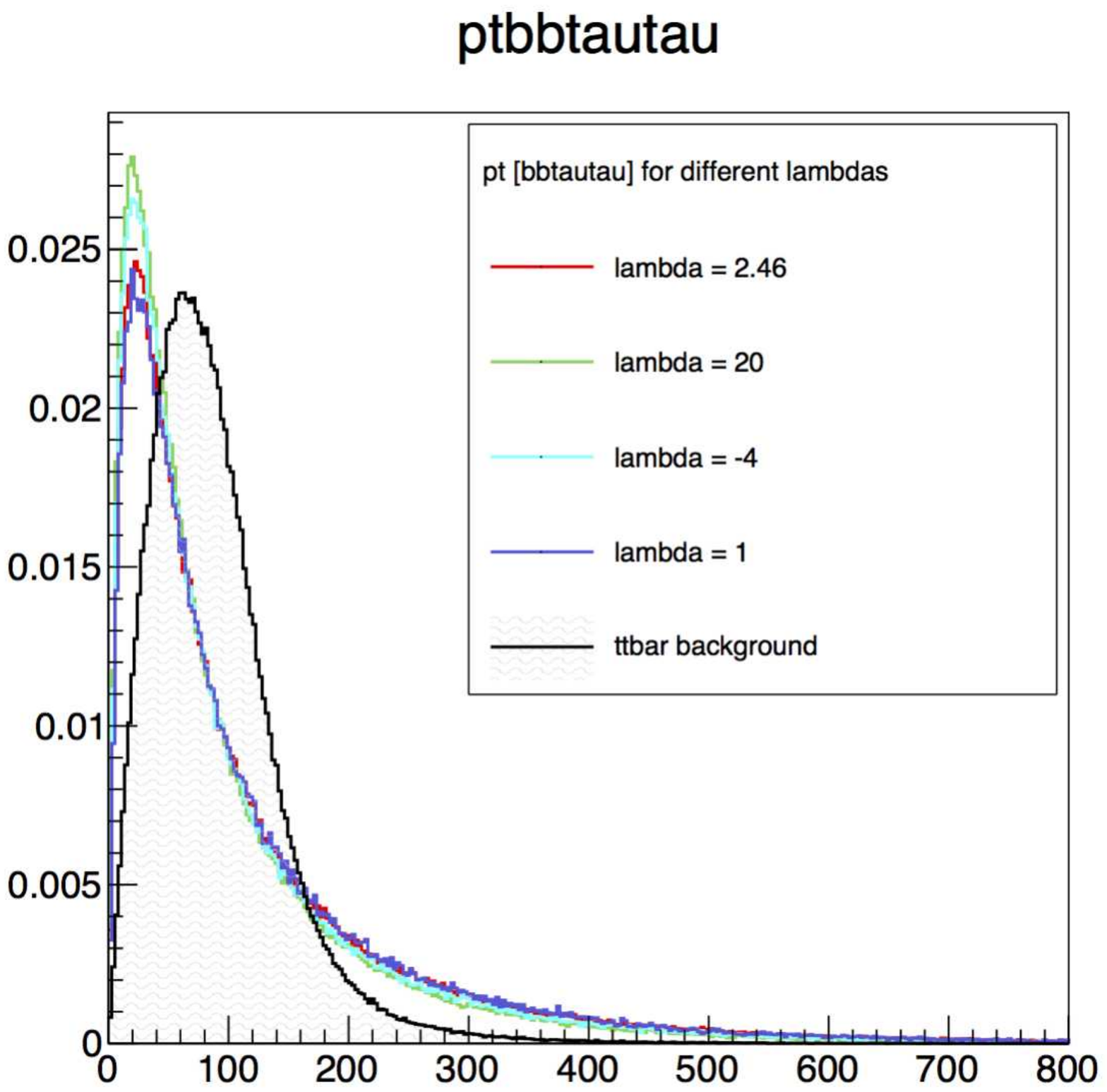}
	\includegraphics[scale=0.38]{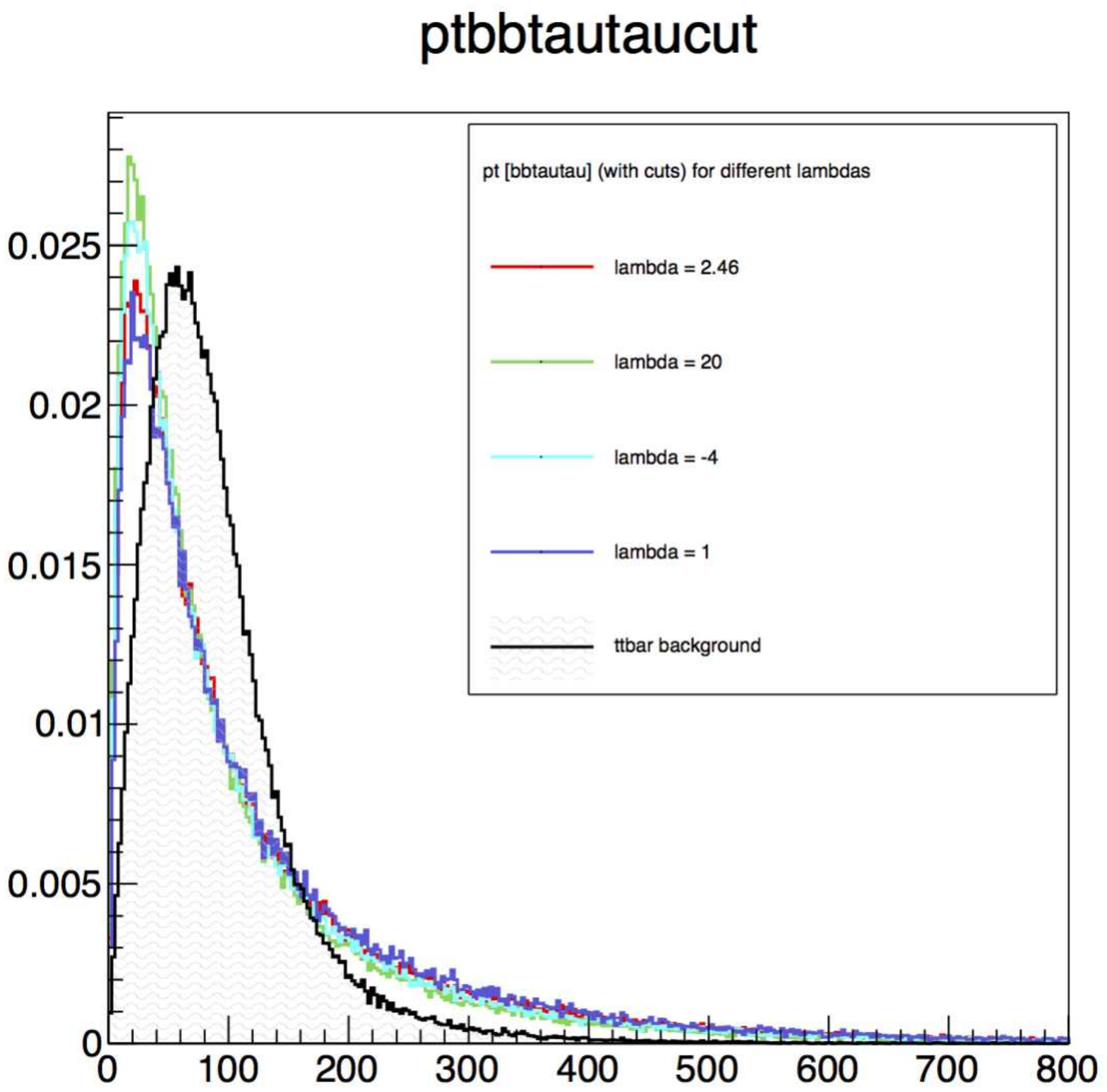}
\end{center}
\begin{center}
	\includegraphics[scale=0.38]{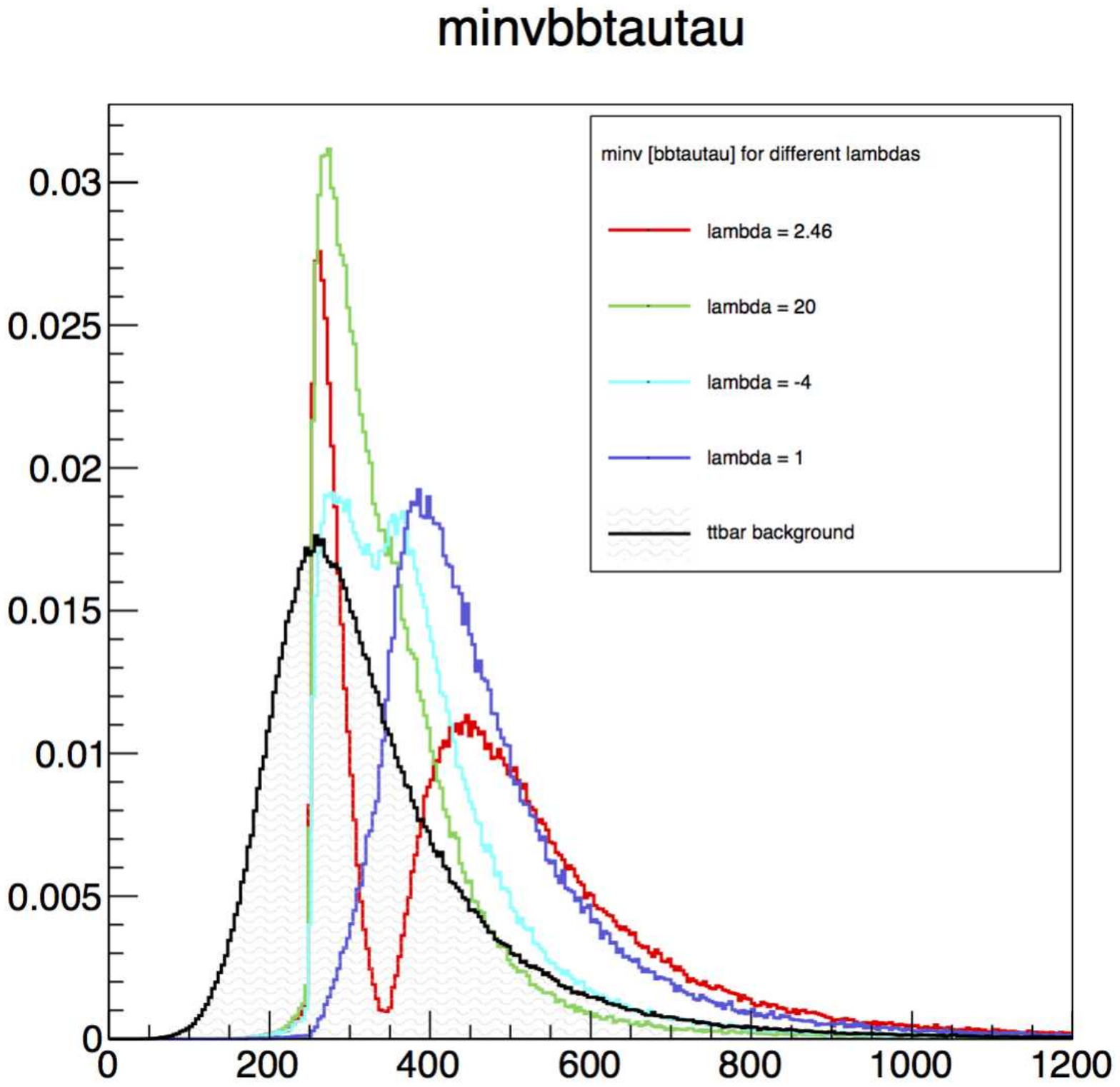}
	\includegraphics[scale=0.38]{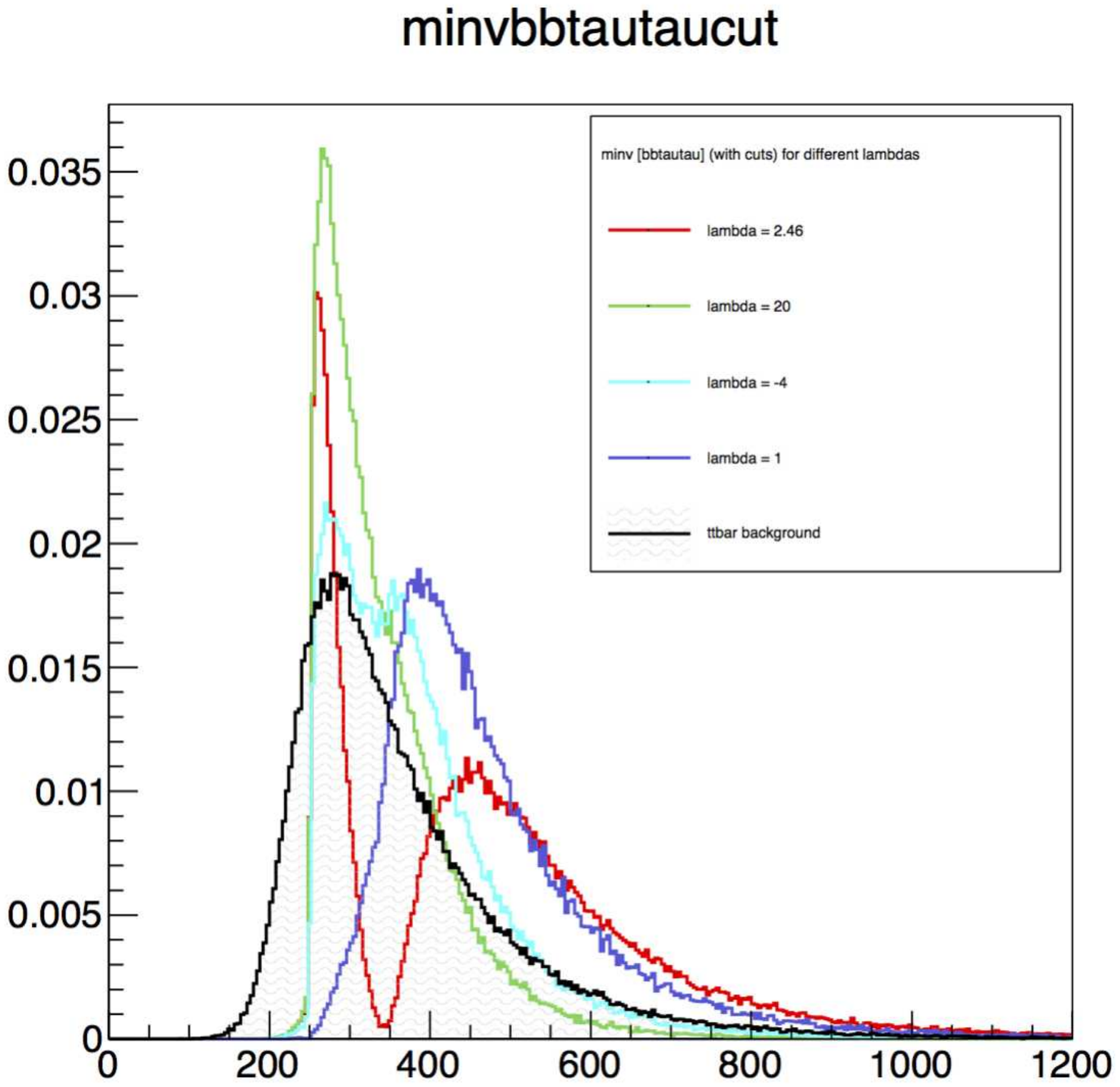}
\end{center}

Les diagrammes de l'impulsion transverse pour le système $[\tau^-,\tau^+]$ et de la masse invariante du système total montrent des différences significatives entre les histogrammes pour différentes valeurs du couplage $\lambda$. Nous nous y intéresserons presque essentiellement pour les études suivantes ou nous prendrons en compte la désintégration des deux $\tau$ finaux.\\

Les deux autres diagrammes pourraient quant'à eux être utiles pour discriminer le signal du bruit de fond, à condition que la différence de forme des distributions ne soit pas réduite par la prise en compte de la désintégration des leptons et l'utilisation de variables reconstruites.

\section{Étude de la cinétique avec la désintégration des $\tau$}
\subsection{Introduction et présentation des approximations}

Comme nous l'avons vu au dessus, les $\tau$ ont une distance de vol moyenne dans le détecteur de quelques millimètres si bien qu'ils se désintègrent avant d'arriver aux calorimètres. Le tableau suivant résume les possibilités de désintégration les plus probables pour les leptons $\tau$.
\begin{center}
	\begin{tabular}{|c|c|}
		\hline
		Mode de désintégration & Probabilité \\
		\hline
		$\tau^+\rightarrow e^++\overline{\nu_\tau}+\nu_e$ & 17,4 \% \\
		$\tau^+\rightarrow \mu^++\overline{\nu_\tau}+\nu_\mu$ & 17,8 \% \\
		$\tau^+\rightarrow$ hadrons & 64,8 \% \\
		\hline
	\end{tabular}
\end{center}
avec dans le mode $\tau\rightarrow$ hadrons, les désintégrations les plus probables :
\begin{center}
	\begin{tabular}{|c|c|}
		\hline
		Mode de désintégration & Probabilité \\
		\hline
		$\tau^+\rightarrow \pi^+\pi^0$ & 25,52 \% \\
		$\tau^+\rightarrow \pi^+$ & 10,83 \% \\
		$\tau^+\rightarrow \pi^+\pi^0\pi^0$ & 9,30 \% \\
		$\tau^+\rightarrow \pi^+\pi^-\pi^+$ & 8,99 \% \\
		$\tau^+\rightarrow \pi^+\pi^-\pi^+\pi^0$ & 2,70 \% \\
		$\tau^+\rightarrow \pi^+\pi^0\pi^0\pi^0$ & 1,05 \% \\
		\hline
	\end{tabular}
\end{center}

Ce qui donne le tableau suivant, pour un système de deux $[\tau^-,\tau^+]$, en notant $\tau_h$ les désintégrations hadroniques (jets), $\tau_e$ les désintégrations électroniques et $\tau_\mu$ les désintégrations muoniques. La section efficace du processus total est donné dans le meme tableau, pour des énergies différentes dans le système du centre de masse. Elle est calculée en utilisant : $$\sigma(gg\rightarrow b\overline{b}\tau_i\tau_j)=\sigma(gg\rightarrow hh)\times 2 \times BR(h\rightarrow b\overline{b})\times BR(h\rightarrow \tau\tau)\times BR(\tau\tau\rightarrow \tau_i\tau_j)$$
le facteur "2" étant purement combinatoire.
\begin{center}
	\begin{tabular}{|c|c|c|c|c|}
		\hline
		Canal & Fraction des évènements & $\sigma$ (8 Tev) [fb] & $\sigma$ (13 Tev) [fb] & $\sigma$ (14 Tev) [fb]\\
		\hline
		$\tau_\mu\tau_h$ & 23,1 \% & 0,13 & 0,578 & 0,686\\
		$\tau_e\tau_h$ & 22,6 \% & 0,13 & 0,565 & 0,671\\
		$\tau_h\tau_h$ & 42,0 \% & 0,24 & 1,05 & 1,25\\
		$\tau_e\tau_\mu$ & 6,2 \% & 0,036 & 0,155 & 0,184\\
		$\tau_e\tau_e$ & 3,0 \% & 0,018 & 0,075 & 0,089\\
		$\tau_\mu\tau_\mu$ & 3,2 \% & 0,018 & 0,080 & 0,095\\
		\hline
	\end{tabular}
\end{center}

Les quarks $b$, détectés comme jets, sont supposés connus parfaitement : on récupère le 4-vecteur généré par les méthodes Monte-Carlo.
A ce niveau d'étude, la philosophie est la suivante : on a un point de vue omniscient sur les particules produits de désintégration des leptons $\tau$, ce qui permet de remonter à des invariants cinétiques du système ou des sous-systèmes.
Comme des neutrinos sont émis durant la désintégration, la seule grandeur accessible expérimentalement est l'énergie manquante. Toujours au niveau générateur, on utilise en fait les 4-vecteurs des neutrinos pour calculer directement l'énergie manquante. Plusieurs invariants sont envisageables pour caractériser le système, qu'on peut obtenir à partir des 4-vecteurs des produits finaux et de l'énergie transverse manquante. On utilise donc, pour calculer les invariants du système, toutes les informations voulues sur les produits observables de la désintégration du système di-tau ($e$, $\mu$, pions) et l'énergie manquante calculée à l'aide des 4-vecteurs des neutrinos.

\begin{figure}[!h]
	\centering
	\includegraphics[scale=0.6]{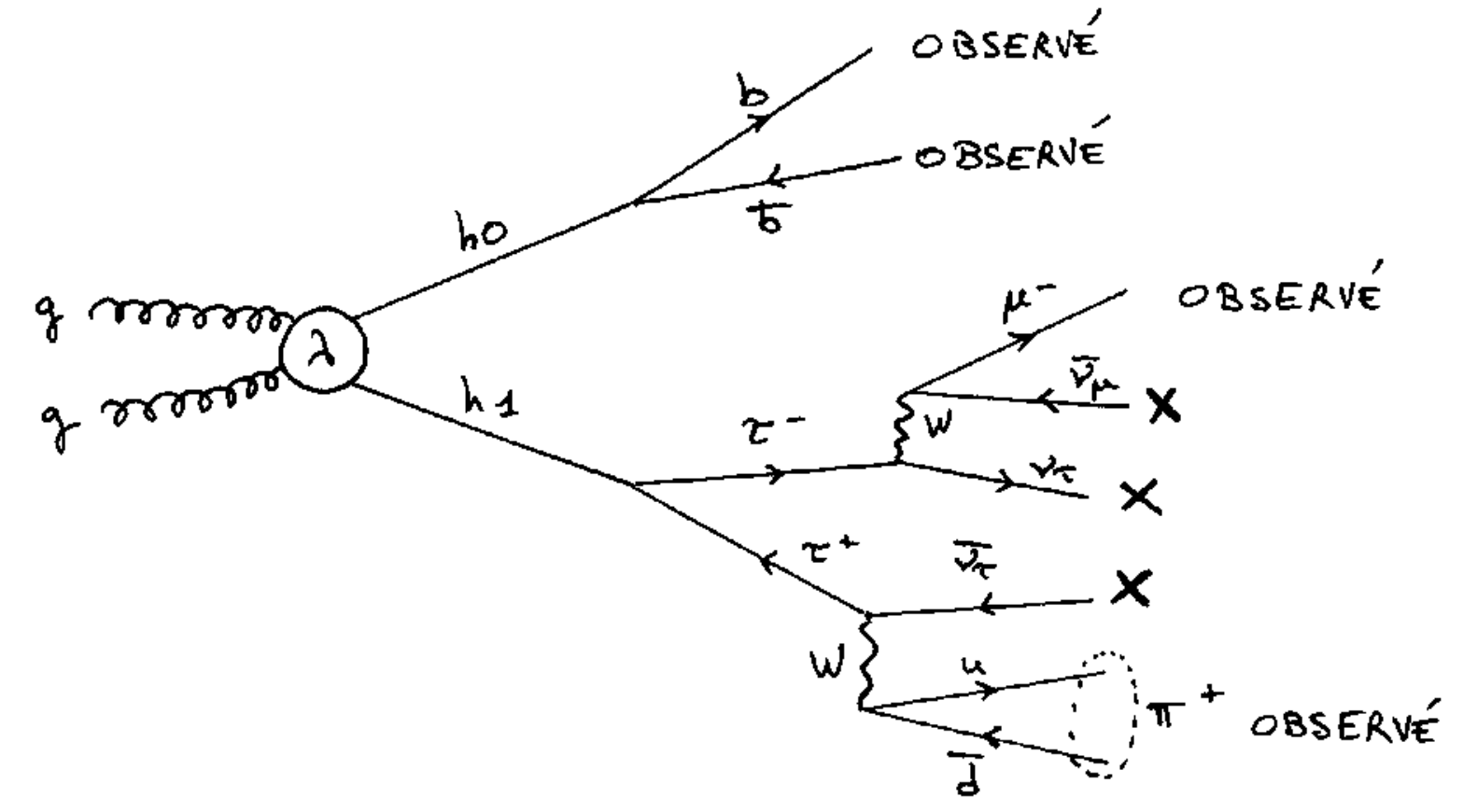}
	\caption{Exemple de processus considéré. La désintégration du système di-tau est ici semi-hadronique, l'un se désintègre en muon et l'autre en un pion $\pi^+$. Les croix correspondent à de l'information perdue, puisque les neutrinos ne sont pas détectés. Les particules ayant la mention "OBSERVÉ" sont supposées parfaitement détectées : on se sert des informations générateur  }
\end{figure}
\paragraph{Méthode des quantités visibles}
Cette première méthode est la plus simple, et de plus indépendante de toute mesure de l'énergie transverse manquante, ce qui la rend particulièrement intéressante. Il faut juste faire abstraction des neutrinos, comme s'ils n'étaient pas produits. On écrira : 
$$
\left( \begin{array}{c}
p(h0)=p(b)+p(\overline{b}) \\
p(h1)=p(\tau_i)+p(\tau_j) \\
p(h) =p(h0)+p(h1) \\
\end{array} \right.
$$
où $\tau_i$ et $\tau_j$ sont les produits de désintégration du $\tau^-$ et du $\tau^+$.\\
Il n'y a pas d'approximation à faire du type "collinear approximation", cependant, en raison de l'énergie emportée par les neutrinos, l'échelle d'énergie est décalée vers la gauche.

\paragraph{Méthode des quantités effectives}
Cela consiste juste à "re-scaler" l'échelle d'énergie en utilisant l'énergie transverse manquante : on "mesure" le 4-vecteur $$p^{miss}=(E^{miss}_x,E^{miss}_y,0,\sqrt{(E^{miss}_x)^2+(E^{miss}_y)^2})$$ et on pose :
$$
\left( \begin{array}{c}
p(h0)=p(b)+p(\overline{b}) \\
p(h1)=p(\tau_i)+p(\tau_j) + p^{miss}\\
p(h) =p(h0)+p(h1) \\
\end{array} \right.
$$
En raison de la nature de l'approximation qui est faite, les distributions ne sont pas améliorées par une telle manipulation, cependant, on peut au moins retomber sur la bonne échelle de masse.

\paragraph{Méthode de l'approximation colinéaire}
Pour des raisons d'hélicité, les produits de désintégration des $\tau$ ont de grandes chances d'être émis avec une impulsion proche de l'impulsion de $\tau$ dont ils proviennent. L'approximation colinéaire consiste à supposer que toutes les particules sont émises dans le même direction. Il faut donc résoudre le système linéaire qui permet d'associer aux deux système de neutrinos une fraction de l'énergie transverse manquante sous la forme d'un 4-vecteur dont la partie spatiale est alignée avec le produit de désintégration observé ($e$, $\mu$, jet), et telle que la somme de ces deux 4-vecteurs donne l'énergie transverse manquante (pour les composantes transverses).\\
Il faut inverser le système : 
$$
\begin{pmatrix}
	E^{miss}_x \\
	E^{miss}_y
\end{pmatrix}
=
\begin{pmatrix}
sin\theta cos\phi & sin\theta' cos\phi' \\
sin\theta sin\phi & sin\theta' sin\phi'
\end{pmatrix}
\begin{pmatrix}
p \\
p' 
\end{pmatrix}
$$
où p et p' sont la norme des 4-vecteurs des neutrinos, associés à la désintégration des deux $\tau$ respectifs. Pour cela, il faut que : 
$$sin\theta cos\phi sin\theta' sin\phi'-sin\theta sin\phi sin\theta' cos\phi' \neq 0$$
c'est-à-dire, si $sin\theta sin\theta' \neq 0$ (ce qui est forcément le cas si les particules sont détectées) : 
$$tan(\phi)\neq\tan(\phi')$$
Autrement dit, il ne faut pas que les particules observées, produits de désintégration des $\tau$, soient émises dans la même direction (même en sens opposé). En effet, si les neutrinos sont émis dans la même direction, en sens opposé, ils peuvent emporter une quantité quelconque d'énergie sans qu'on détecte une quelconque énergie transverse manquante.
Cela introduit une indétermination dans la sélection des évènements : ne faut-il que demander à ce que le déterminant de la matrice soit non nul ? Mais alors cela dépend de la précision avec lesquelles les variables "float" sont stockées ... Cependant, comment choisir la limite inférieure pour ce déterminant ? 
Dans notre étude, nous avons seulement demandé à ce que le déterminant soit non nul.\\\\
Enfin, un dernier problème que pose cette approximation colinéaire est que si on veut vraiment reconstruire le 4-vecteur des neutrinos manquants, il faut choisir le signe de la troisième coordonnée spatiale, et, à moins d'introduire de l'aléatoire qui n'a rien de "physique", on perd par exemple l'information de la rapidité par les $\tau$ et leurs particules parentes. \\
De toutes façons, nous avons vu que les distributions en rapidité ne permettaient pas de bien discriminer $\lambda$, nous nous intéressons seulement aux distributions jugées d'intérêt par l'étude précédente.

\subsection{Coupures cinétiques et efficacités}
Puisqu'on prend en compte les désintégrations des leptons $\tau$, nous pouvons facilement fixer des coupures cinétiques plus proches de celles qui sont en réalité imposées par le déclenchement du trigger et les algorithmes de reconstruction de CMS.
Voici donc les coupures utilisées pour les différents produits de désintégration du système $[\tau,\tau]$ :

\begin{center}
	\begin{tabular}{|c||c|c|}
		\hline
		$\tau_\mu\tau_h$ & $|\eta(\mu)|<2.1$, $|\eta(\tau_h)|<2.4$ & $p_T(\mu)>20$ GeV, $p_T(\tau_h)>30$ GeV \\
		$\tau_e\tau_h$ & $|\eta(e)|<2.1$, $|\eta(\tau_h)|<2.4$ & $p_T(e)>24$ GeV, $p_T(\tau_h)>30$ GeV \\
		$\tau_h\tau_h$ & $|\eta(\tau_h)|<2.1$ & $p_T(\tau_h)>45$ GeV \\
		$\tau_e\tau_\mu$ & $|\eta(\mu)|<2.1$, $|\eta(e)|<2.3$ & $p_T(l_1)>20$ GeV, $p_T(l_2)>10$ GeV \\
		$\tau_e\tau_e$ & $|\eta(e)|<2.3$ & $p_T(e_1)>20$ GeV, $p_T(e_2)>20$ GeV \\
		$\tau_\mu\tau_\mu$ & $|\eta(\mu_1)|<2.1$, $|\eta(\mu_2)|<2.4$ & $p_T(\mu_1)>20$ GeV, $p_T(\mu_2)>10$ GeV \\
		\hline
	\end{tabular}
\end{center}
où les indices 1 et 2 ordonnent les systèmes de deux leptons en appelant 1 le lepton d'impulsion transverse maximale.\\

Pour les quarks, détectés comme jets, on garde les valeurs de coupures cinétiques utilisées dans la première étude : 
\[
\left\{
\begin{array}{r c l}
p_T &>& 30 GeV\\
|\eta| &<& 2.5
\end{array}
\right.
\]
Les efficacité relevées pour les différentes valeurs de $\lambda$ sont les suivantes: 
\begin{center}
	\begin{tabular}{|c||c|c|c|c|c|}
		\hline
		& $\lambda=-4$ & $\lambda=-1$ & $\lambda=2.46$ & $\lambda=20$ & $t\overline{t}$ \\
		\hline
		globale & 0.143 & 0.187 & 0.190 & 0.128 & 0.073 \\
		$\tau_\mu\tau_h$ & 0.246 & 0.296 & 0.298 & 0.230 & 0.122 \\
		$\tau_e\tau_h$ & 0.207 & 0.253 & 0.258 & 0.188 & 0.121 \\
		$\tau_h\tau_h$ & 0.131 & 0.195 & 0.194 & 0.107 & 0.042 \\
		$\tau_e\tau_\mu$ & 0.364 & 0.411 & 0.401 & 0.343 & 0.262 \\
		$\tau_e\tau_e$ & 0.387 & 0.349 & 0.433 & 0.361 & 0.282 \\
		$\tau_\mu\tau_\mu$ & 0.369 & 0.425 & 0.418 & 0.365 & 0.271 \\
		\hline
	\end{tabular}
\end{center}
Dans ce tableau, l'efficacité globale est le ratio du nombre d'évènements dont tous les objets sont au dessus des coupures cinétiques avec le nombre total d'évènements, et chaque désintégration, par exemple $\tau_h\tau_h$, l'efficacité est la proportion d'évènements, parmi ceux aboutissant à une désintégration doublement hadronique, dont tous les produits de désintégration des $\tau$ (seulement !) sont au dessus des coupures. C'est pour cela que l'efficacité totale n'est pas le barycentre des efficacités individuelles.\\\\
Donnons maintenant les différents histogrammes obtenus à partir des informations provenant des quarks $b$ et des produits de désintégration des leptons $\tau$ seulement. Les diagrammes de gauche sont ceux qu'on aurait avec un détecteur parfait (sans coupures) tandis qu'à droite, les coupures cinétiques ont été appliquées. Les résultats des trois approximations sont présentées successivement, avec d'abord les quantités visibles, puis les quantités effectives, et enfin les résultats de l'approximation colinéaire. Il est intéressant de noter les différentes conséquences que ces approximations ont sur les distributions, et l'utilisation de plusieurs de ces approximations permet de lever certaines indéterminations (par exemple, dans le cadre de l'approximation colinéaire, les distributions de masse invariante du système total pour $\frac{\lambda}{\lambda^{SM}}$ et pour le fond $t\overline{t}$ sont très semblables, et il est plus intéressant de regarder les quantités visibles ou effectives, mais pour l'impulsion transverse totale, l'approximation colinéaire est plus intéressante).

\subsection{Comparaison des courbes pour les différentes valeurs de $\lambda$}

\paragraph{Méthode des quantités visibles}
\begin{center}
	\includegraphics[scale=0.37]{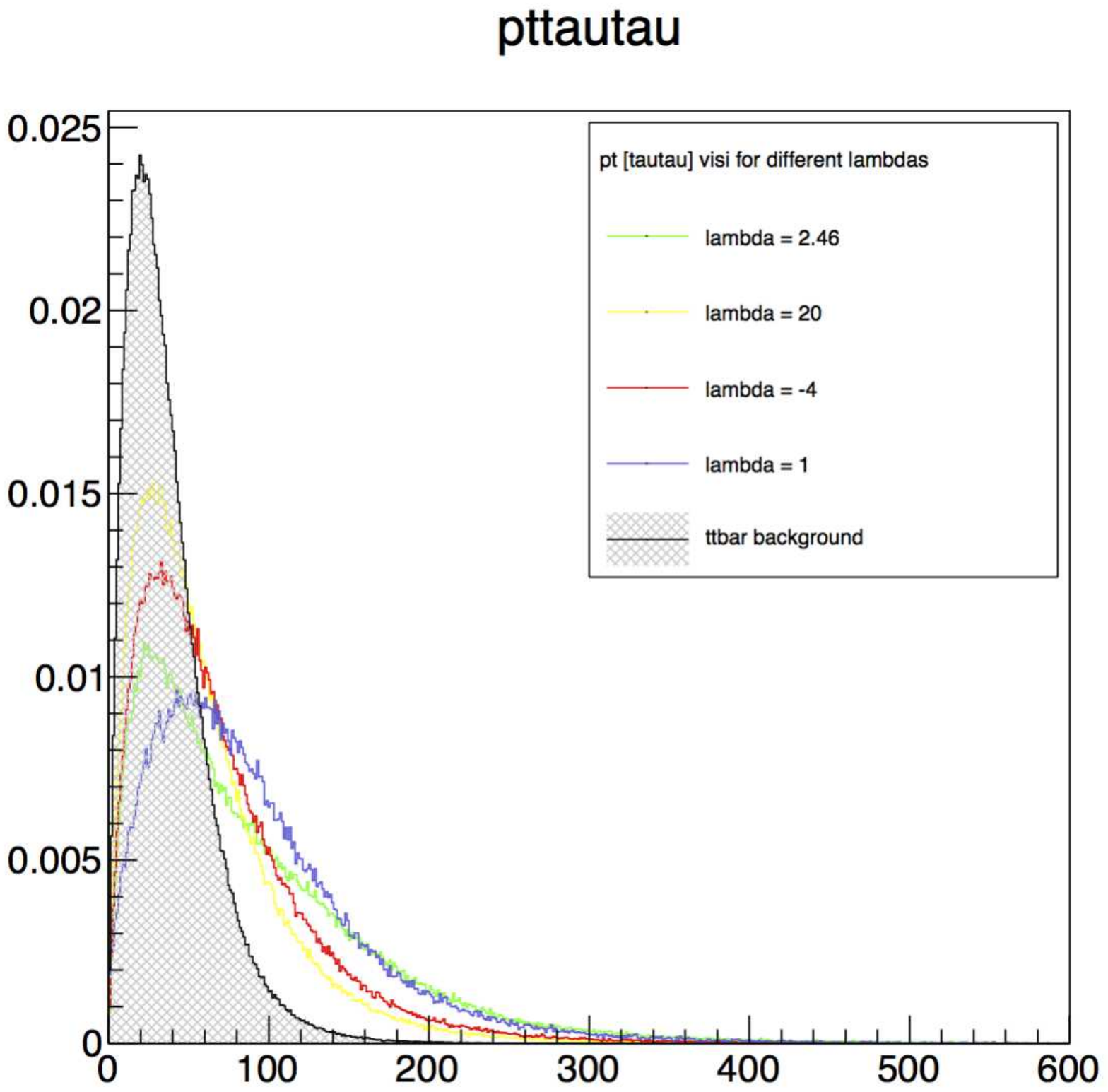}
	\includegraphics[scale=0.37]{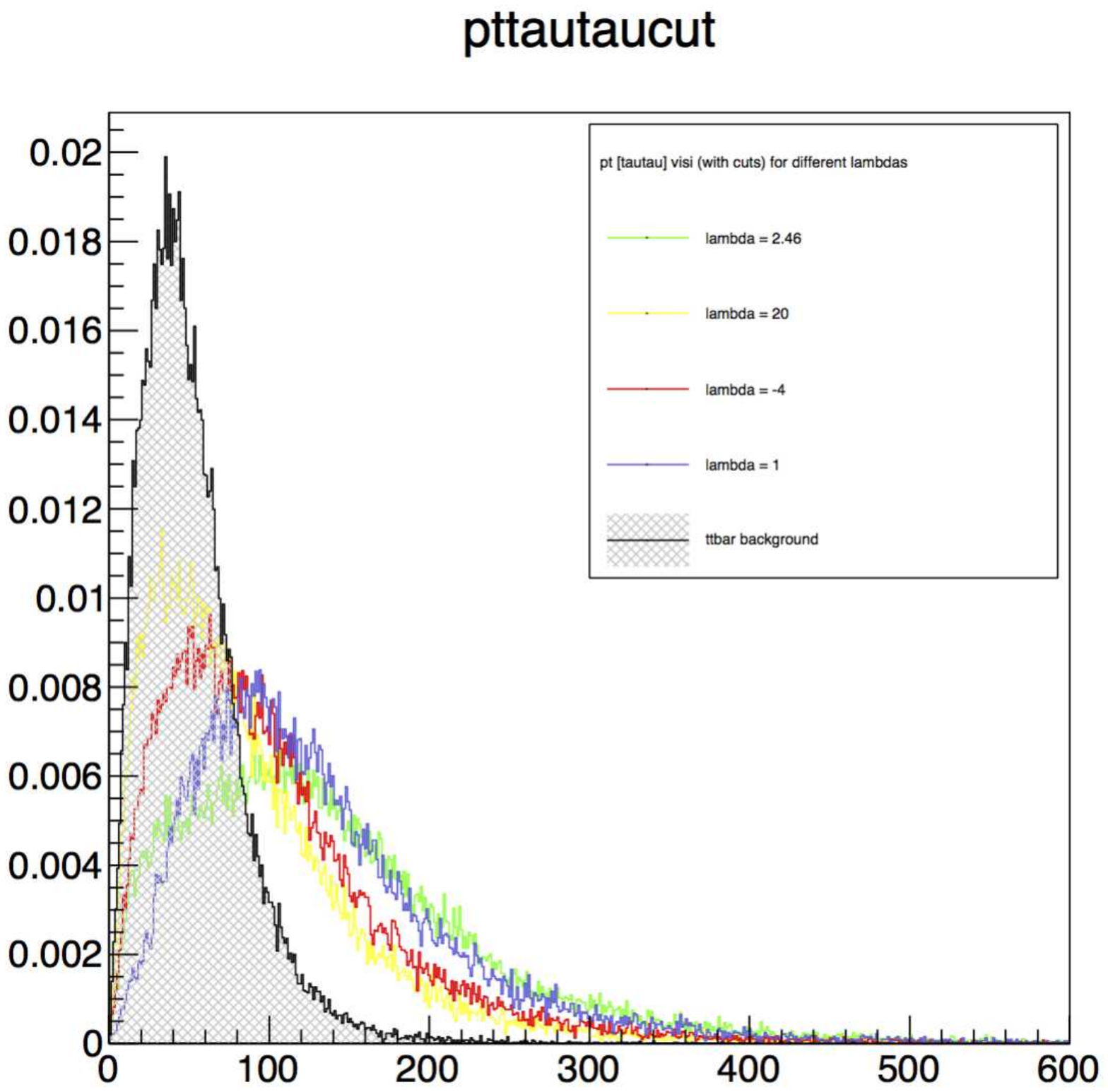}
\end{center}
\begin{center}
	\includegraphics[scale=0.37]{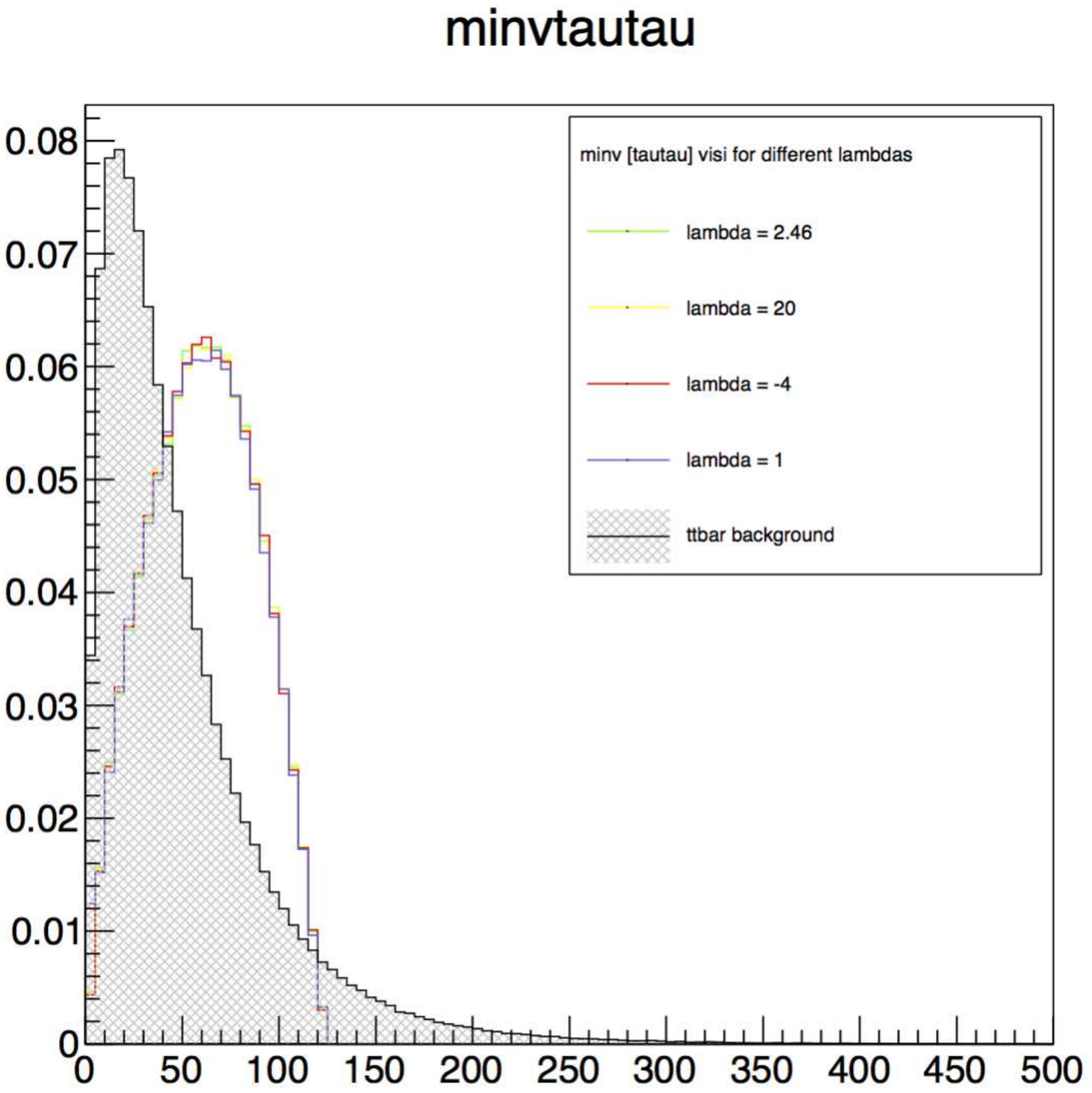}
	\includegraphics[scale=0.37]{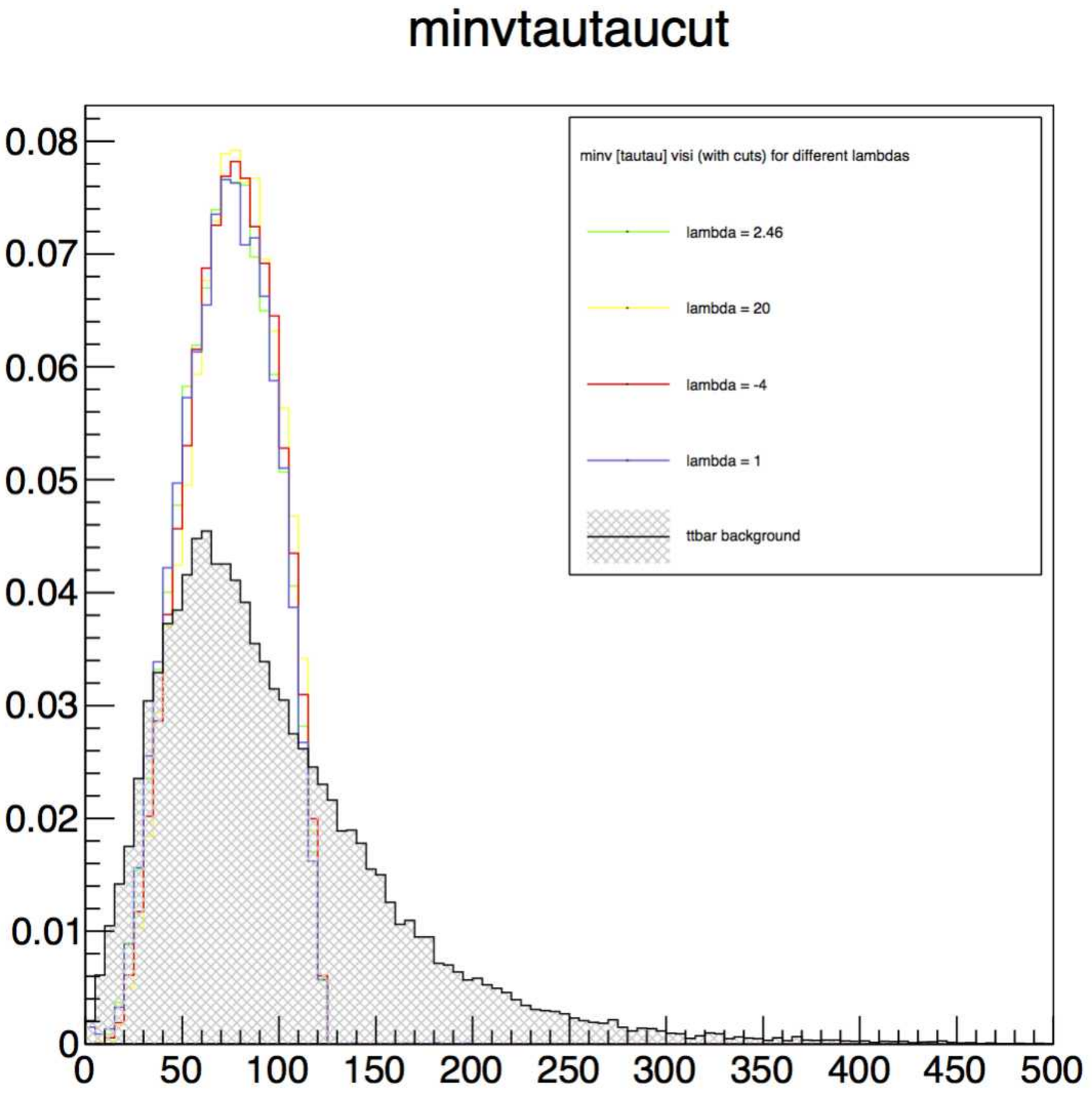}
\end{center}
\begin{center}
	\includegraphics[scale=0.37]{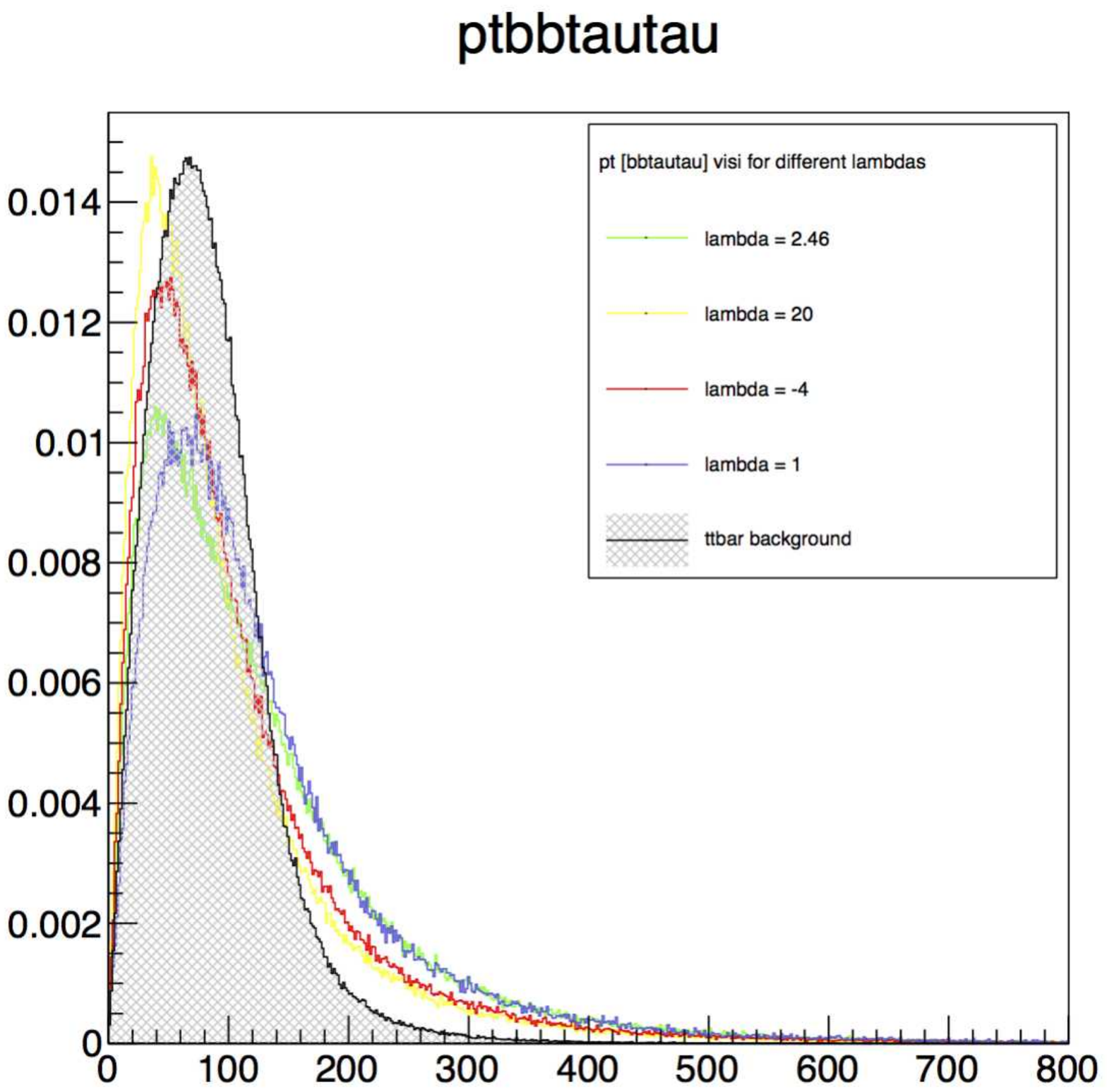}
	\includegraphics[scale=0.37]{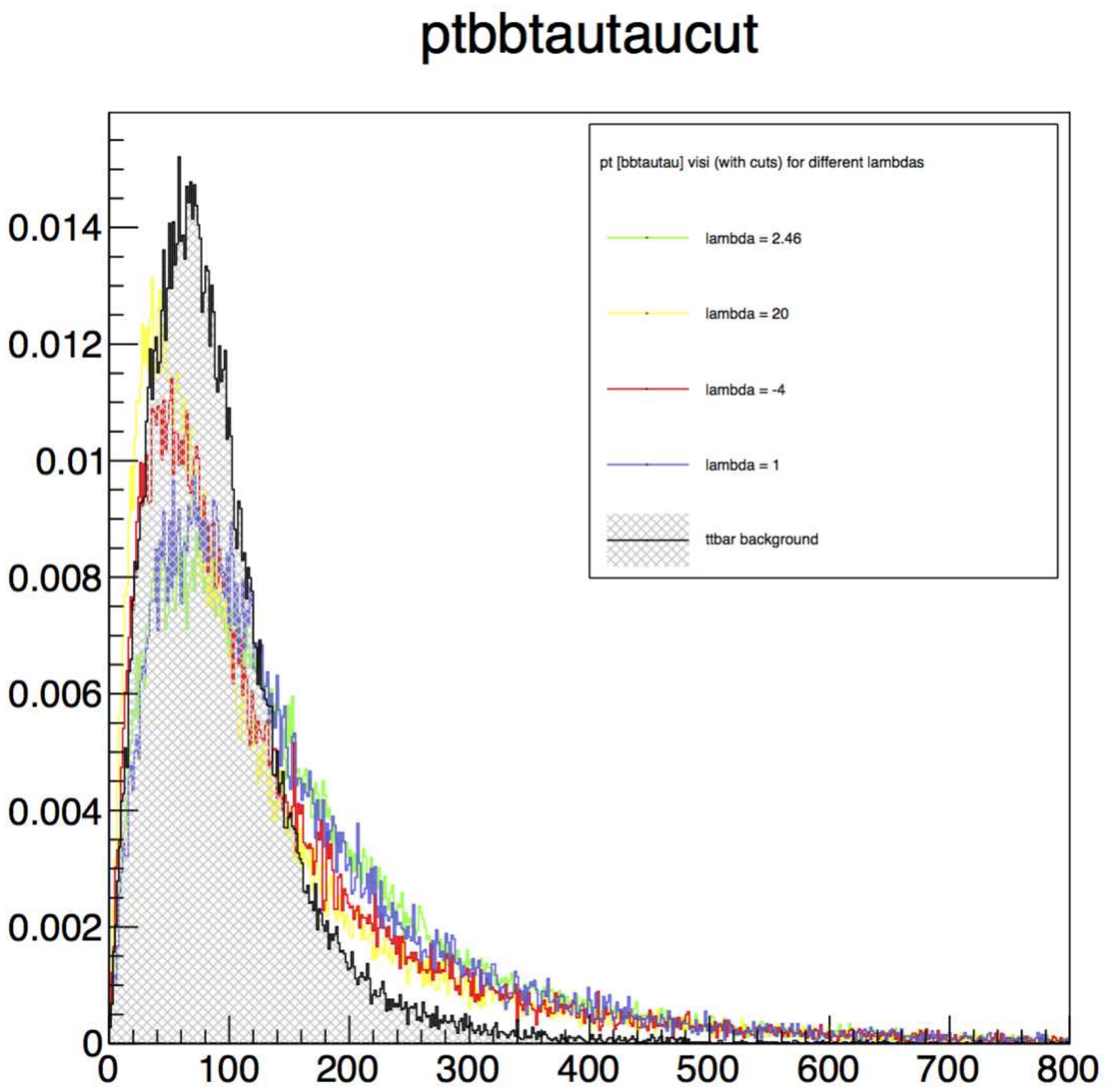}
\end{center}
\begin{center}
	\includegraphics[scale=0.37]{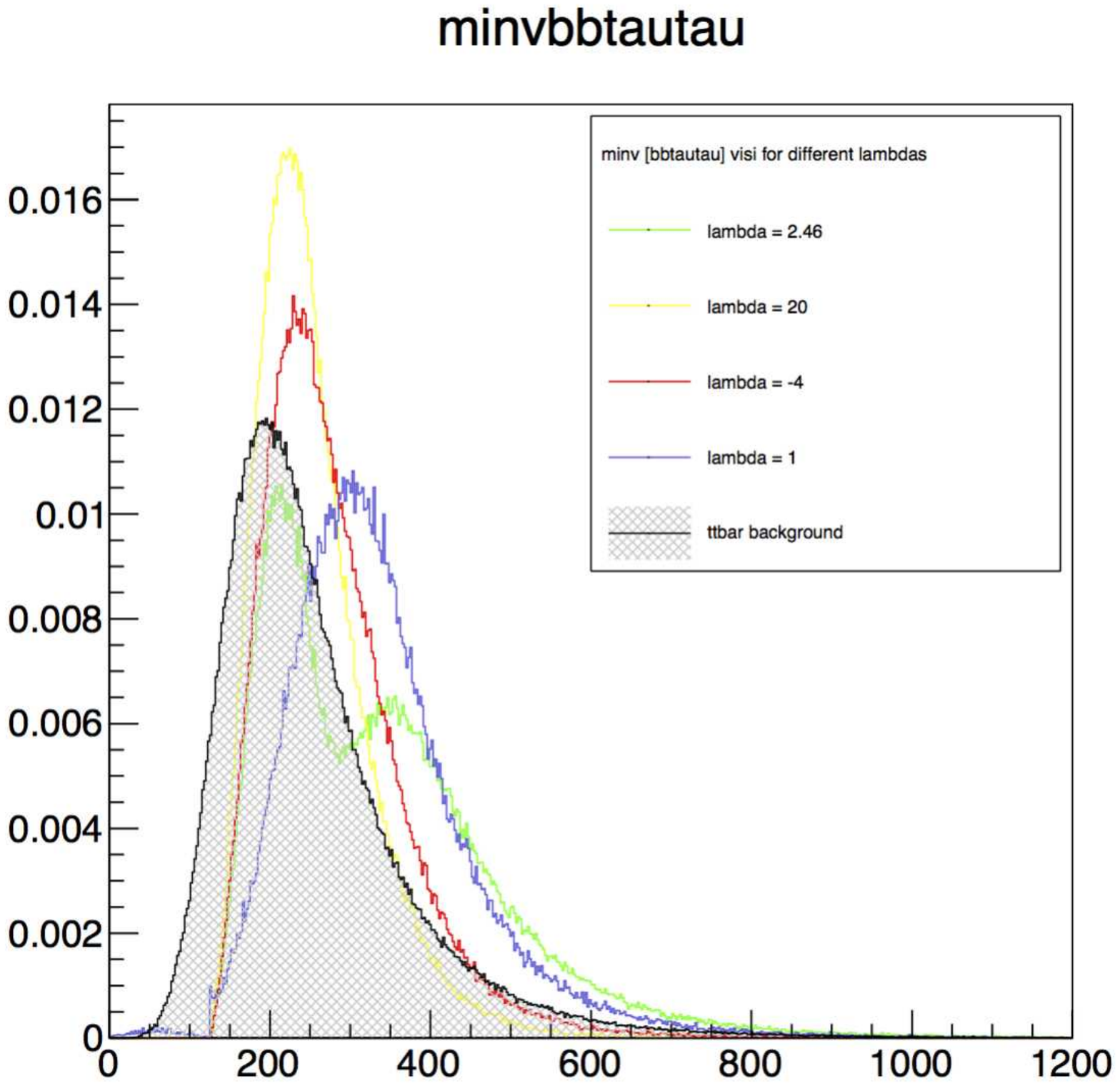}
	\includegraphics[scale=0.37]{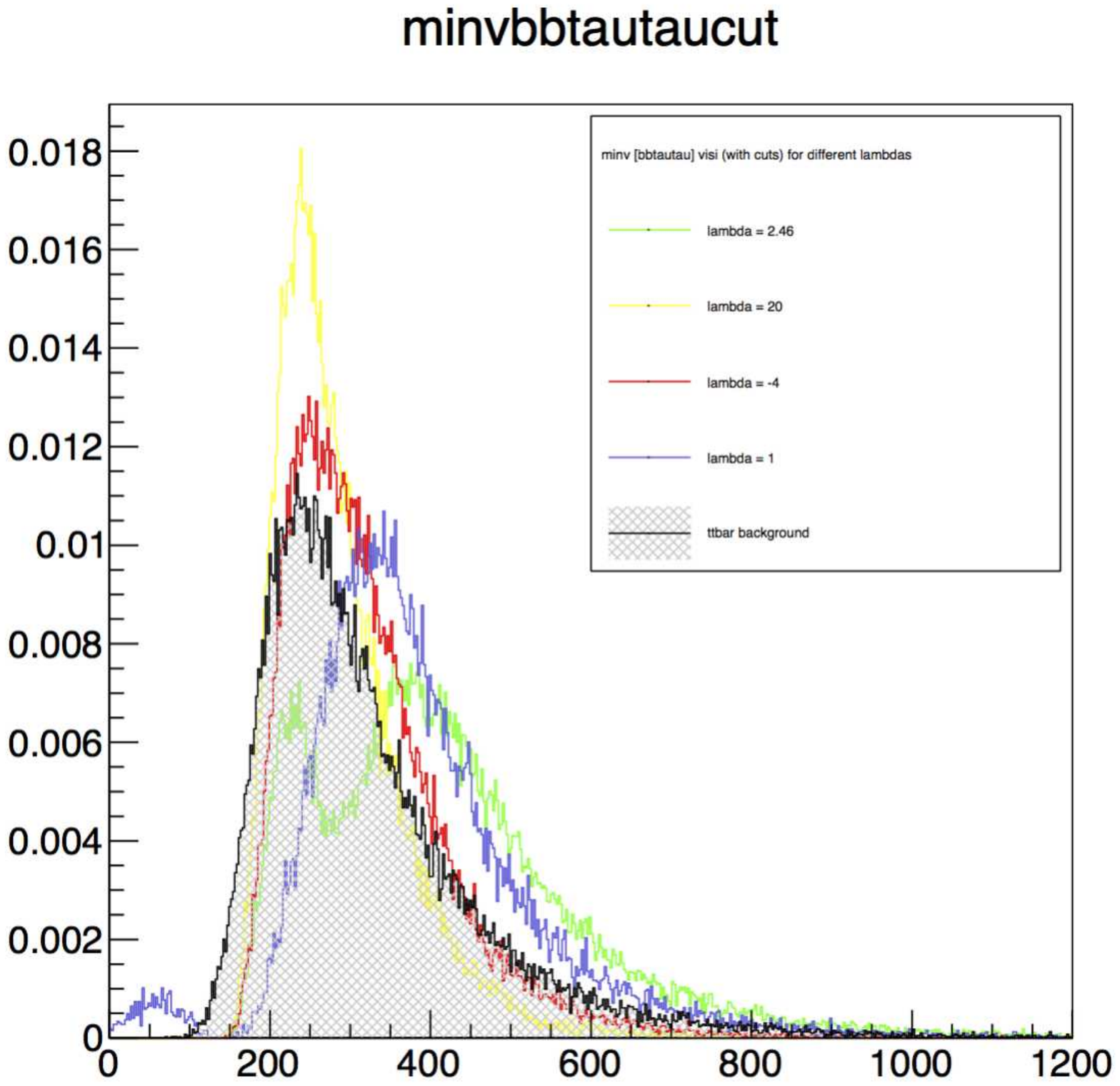}
\end{center}

\paragraph{Méthode des quantités effectives}
\begin{center}
	\includegraphics[scale=0.37]{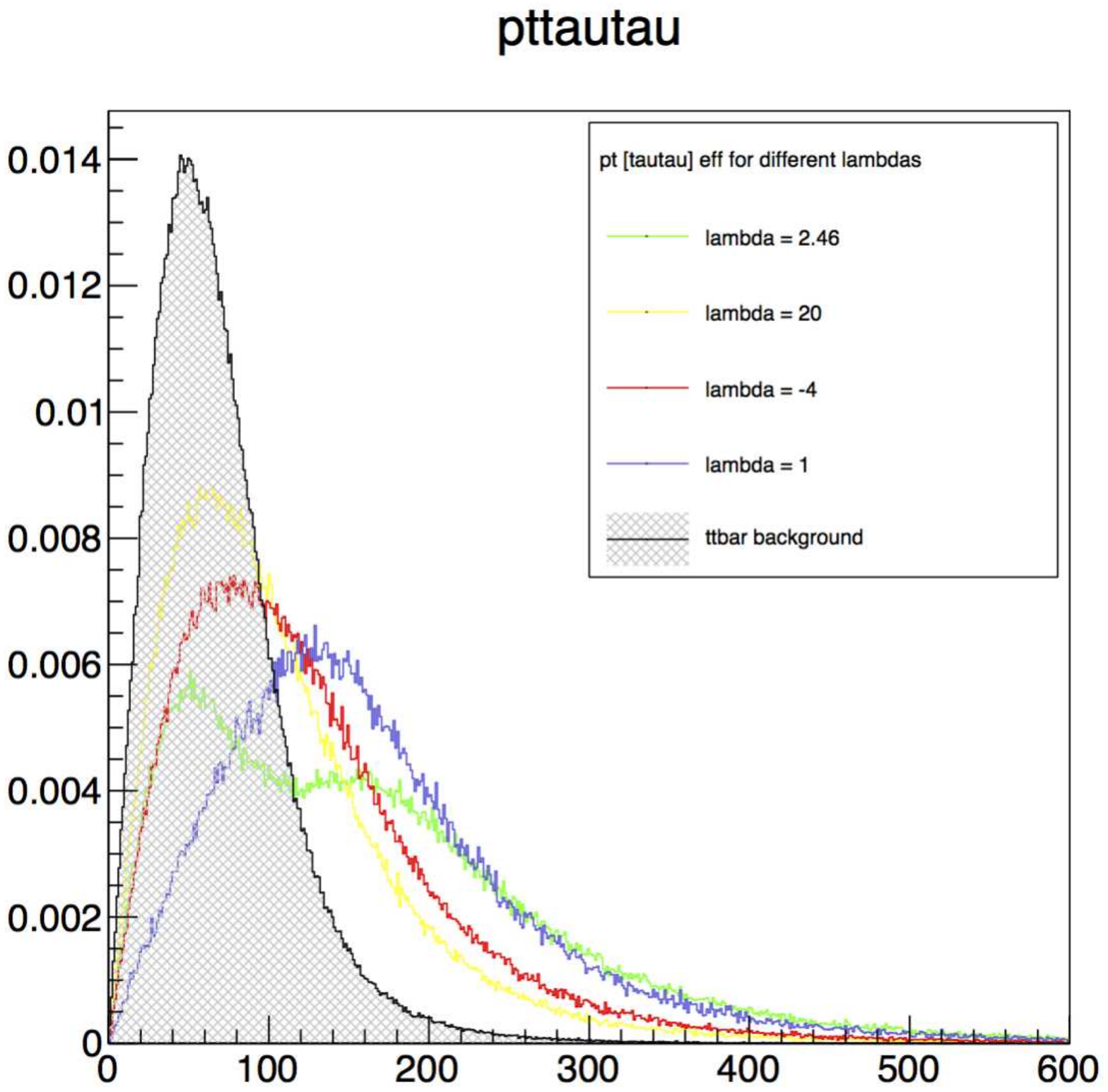}
	\includegraphics[scale=0.37]{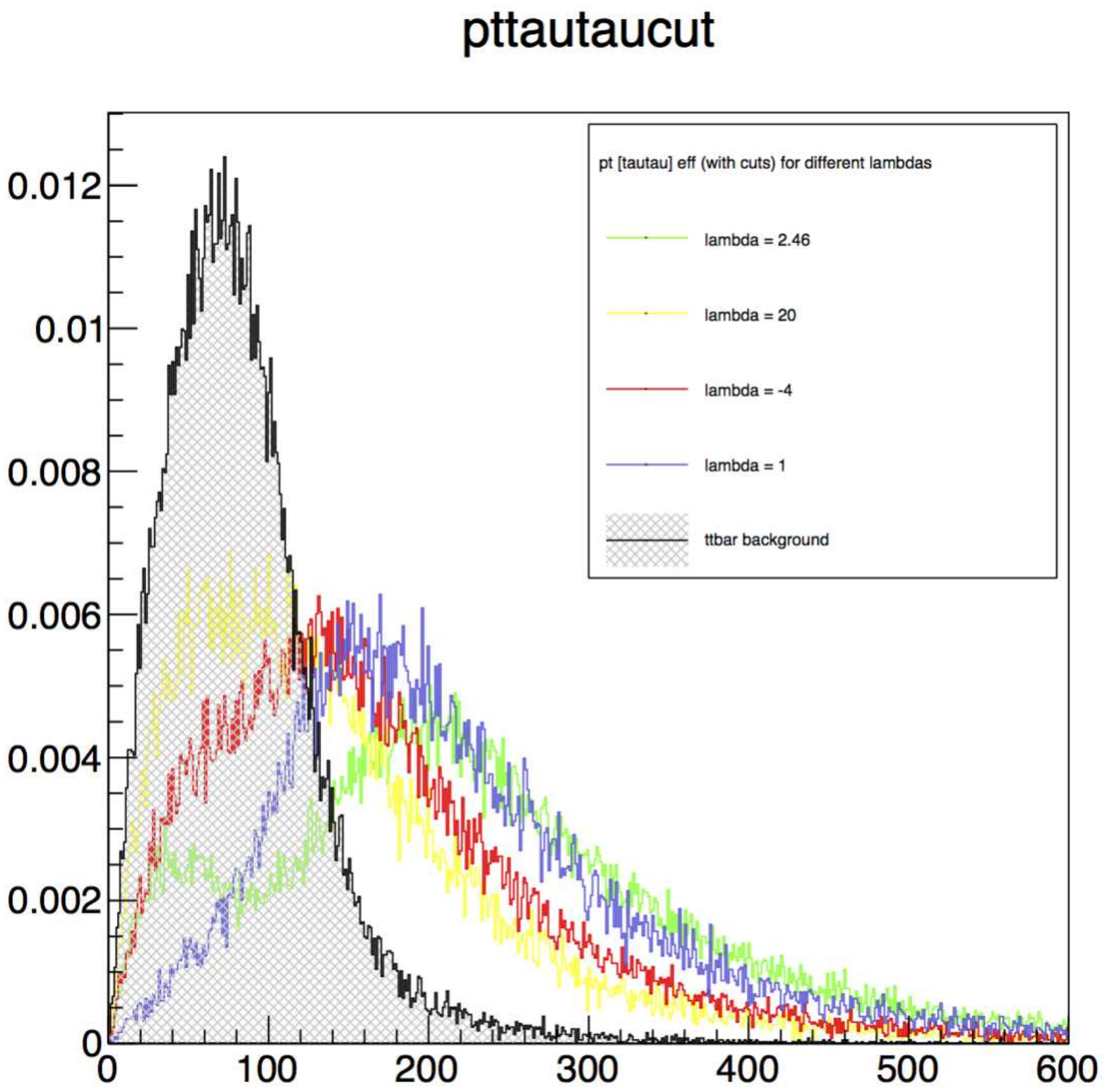}
\end{center}
\begin{center}
	\includegraphics[scale=0.37]{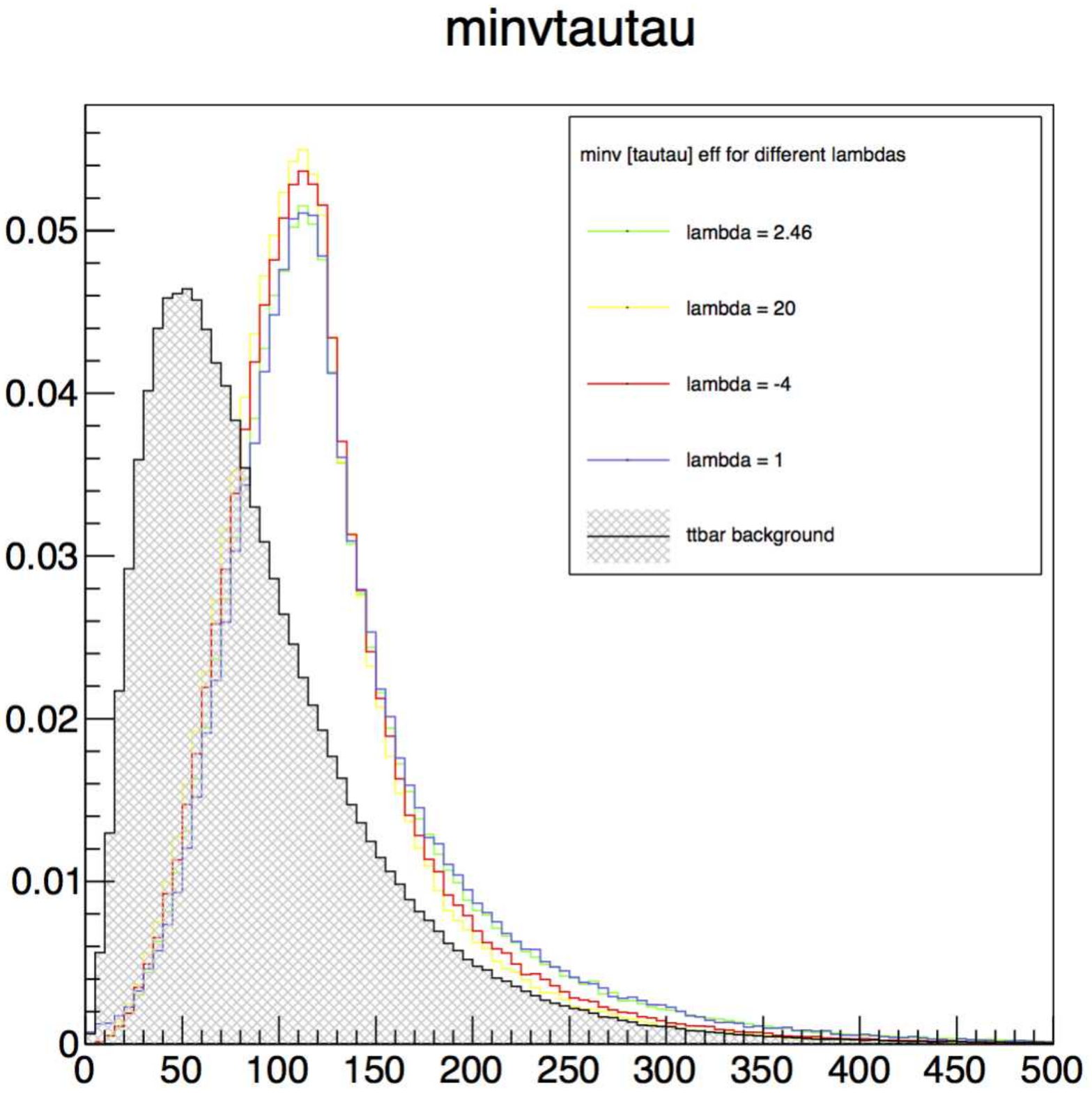}
	\includegraphics[scale=0.37]{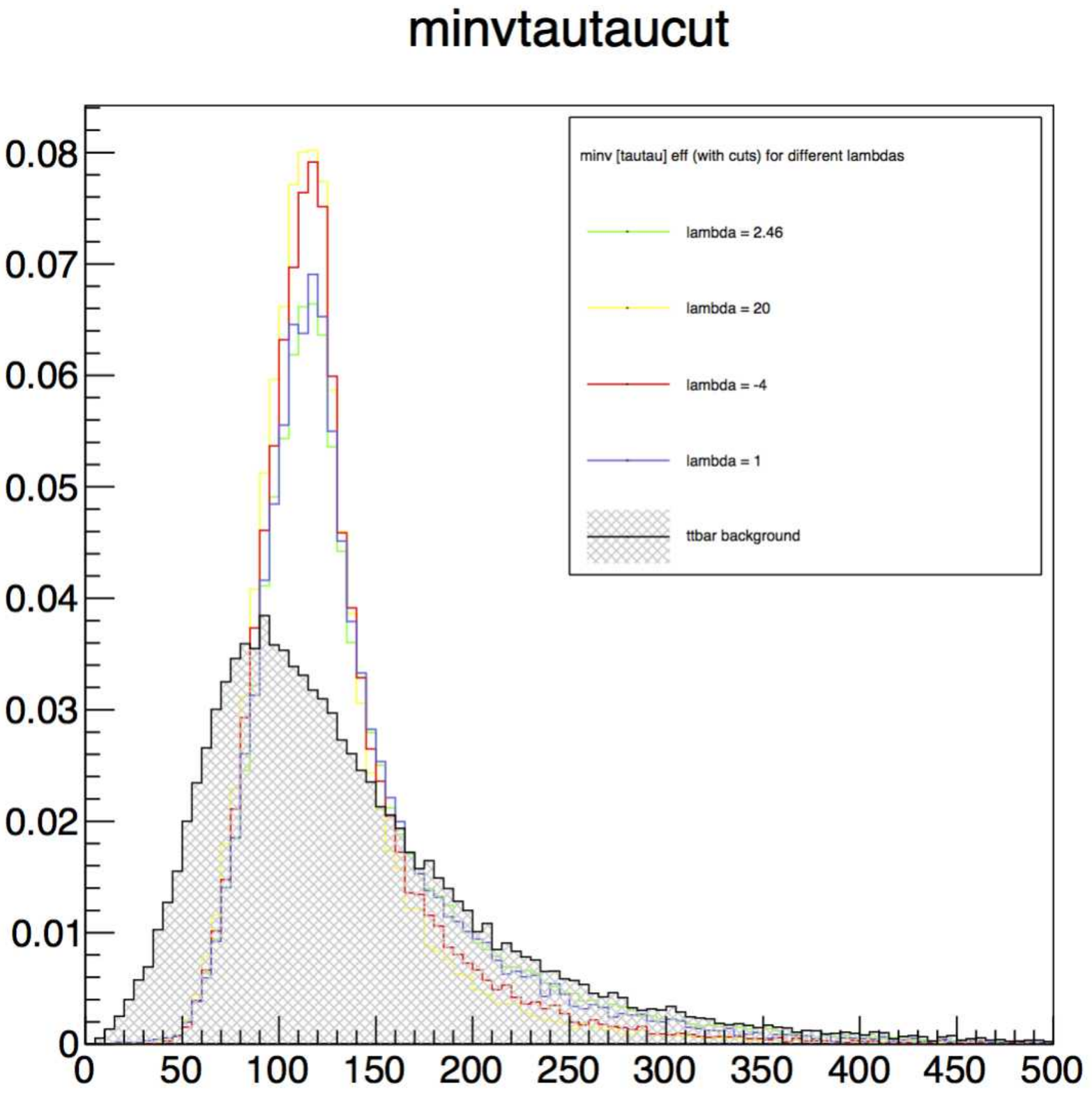}
\end{center}
\begin{center}
	\includegraphics[scale=0.37]{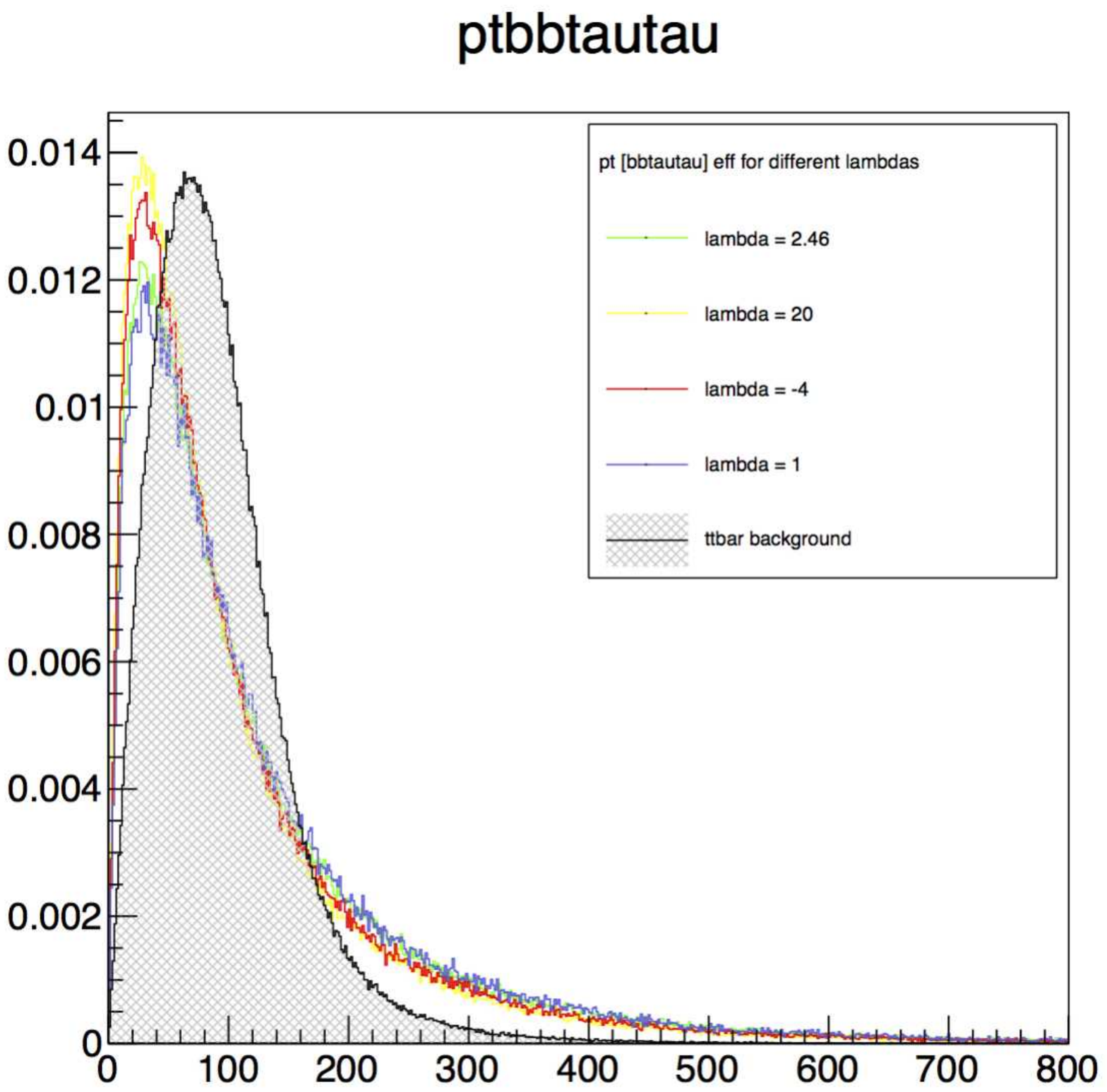}
	\includegraphics[scale=0.37]{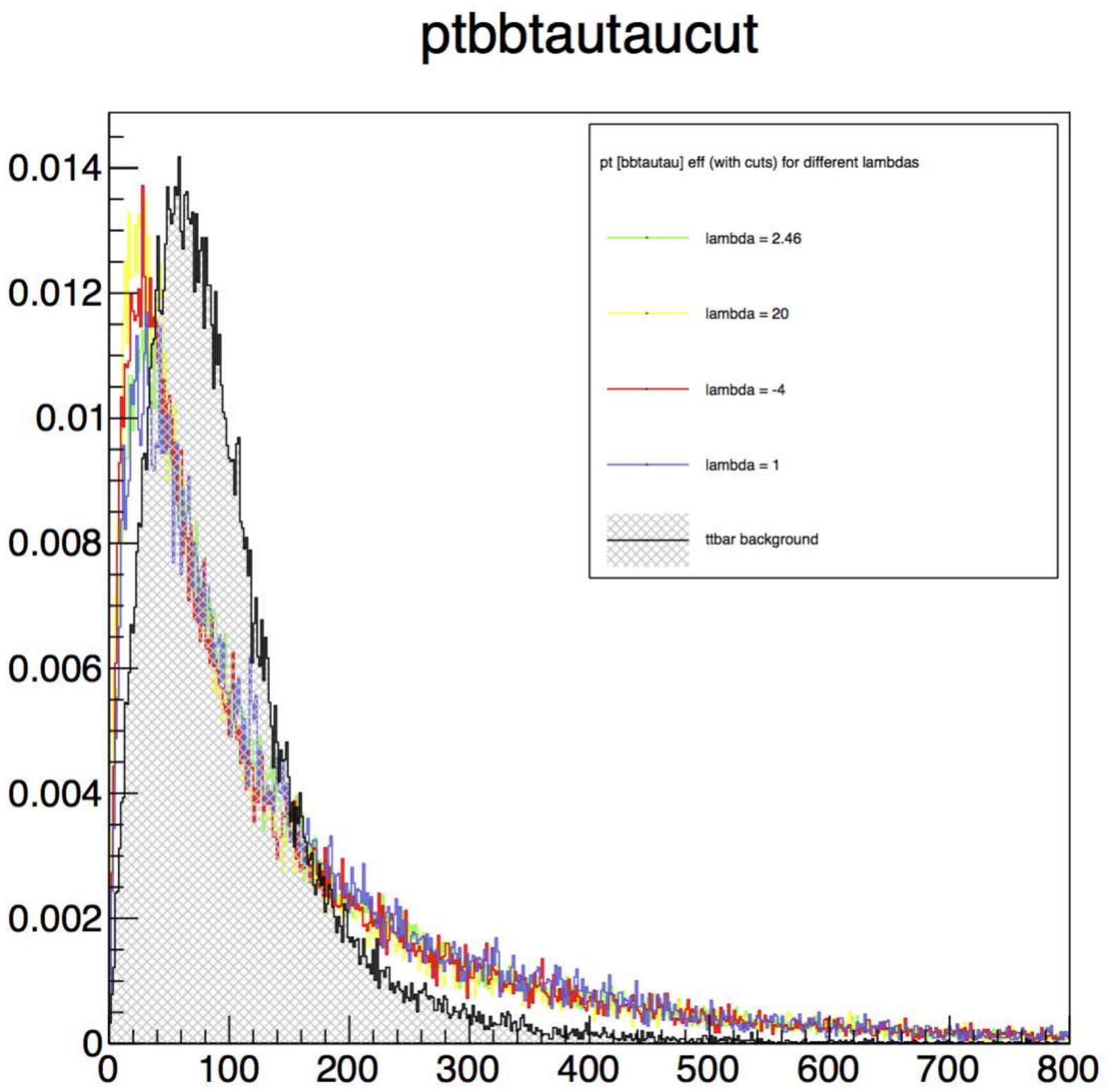}
\end{center}
\begin{center}
	\includegraphics[scale=0.37]{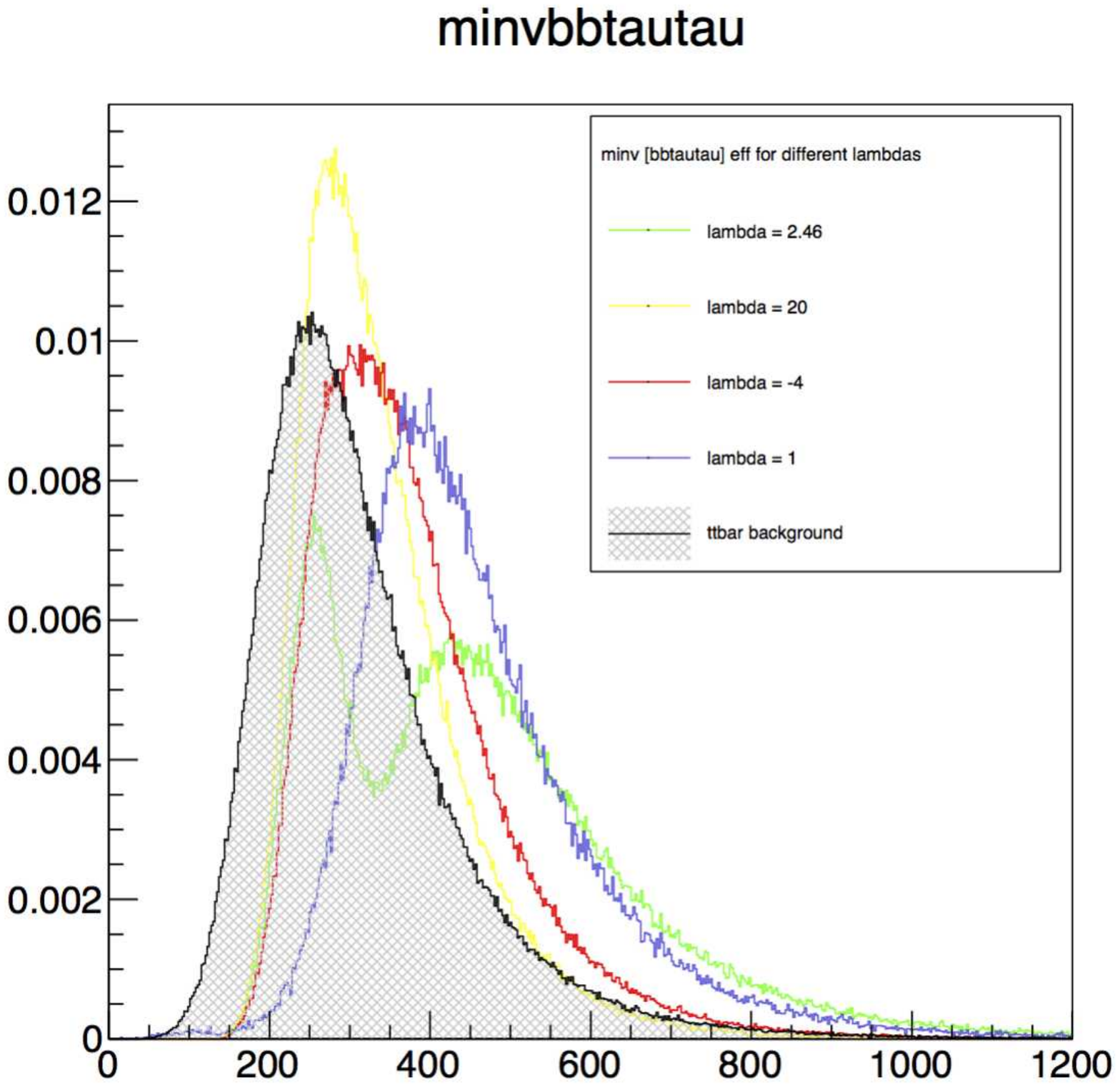}
	\includegraphics[scale=0.37]{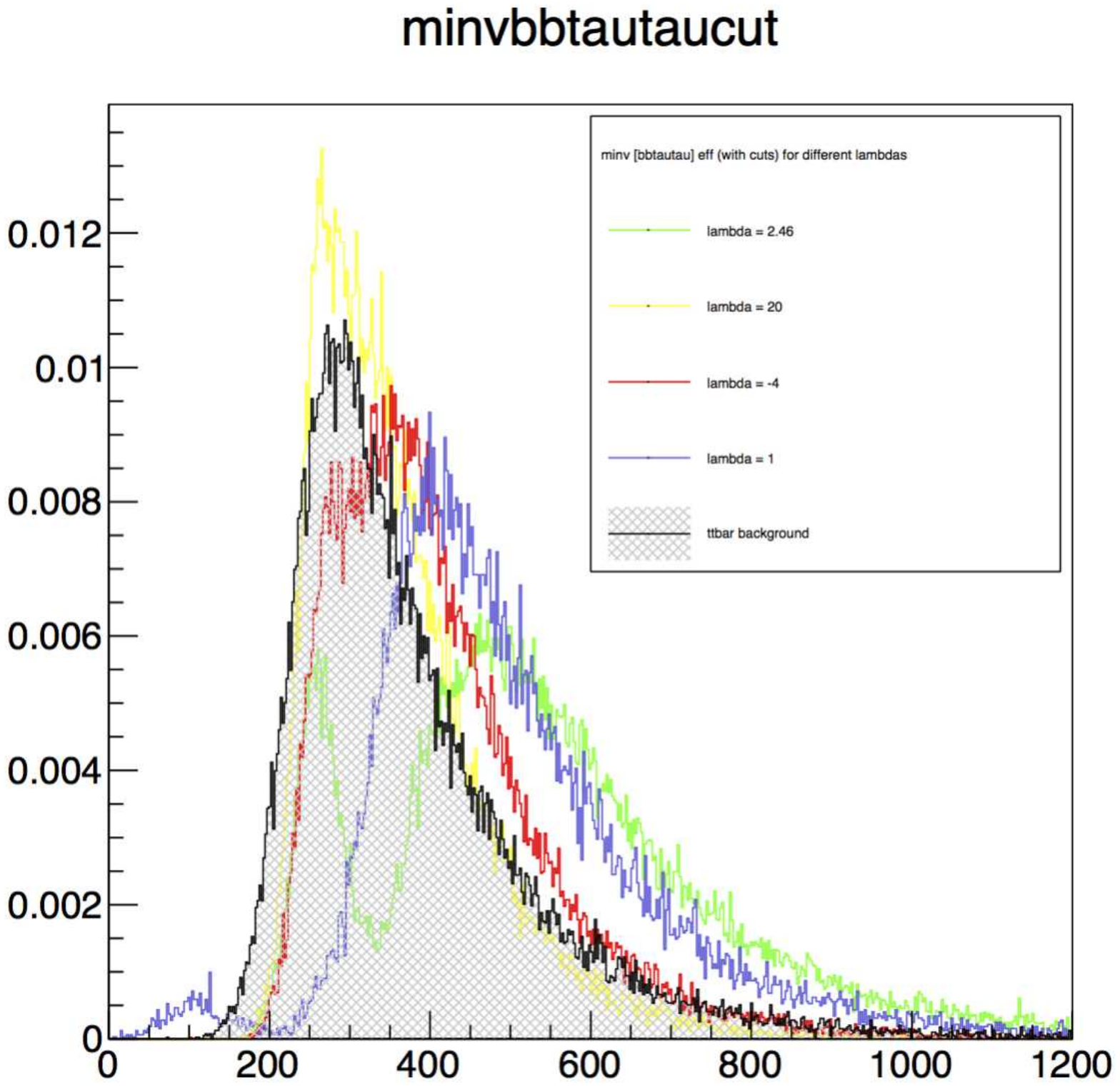}
\end{center}

\paragraph{Approximation colinéaire}
\begin{center}
	\includegraphics[scale=0.37]{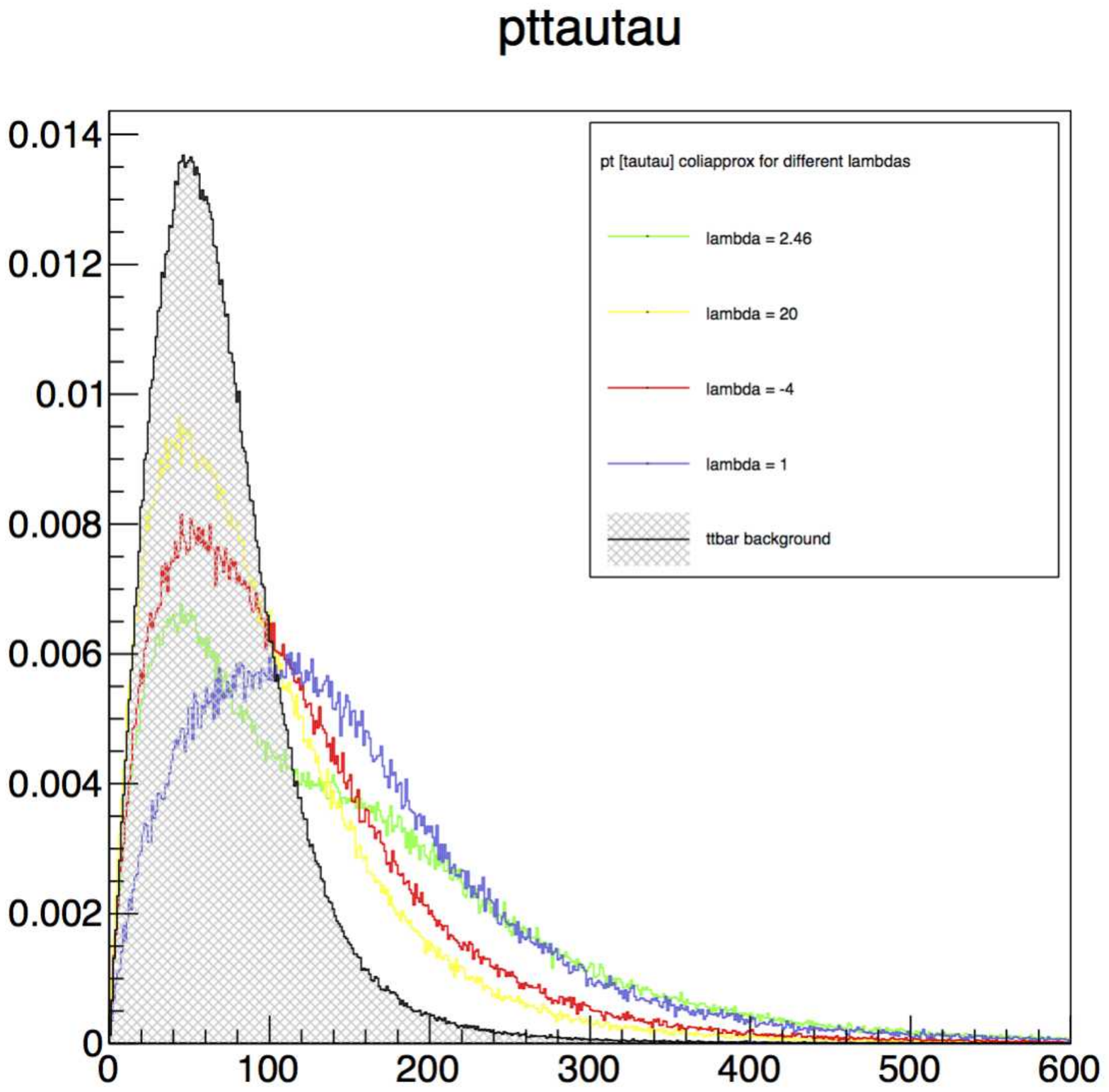}
	\includegraphics[scale=0.37]{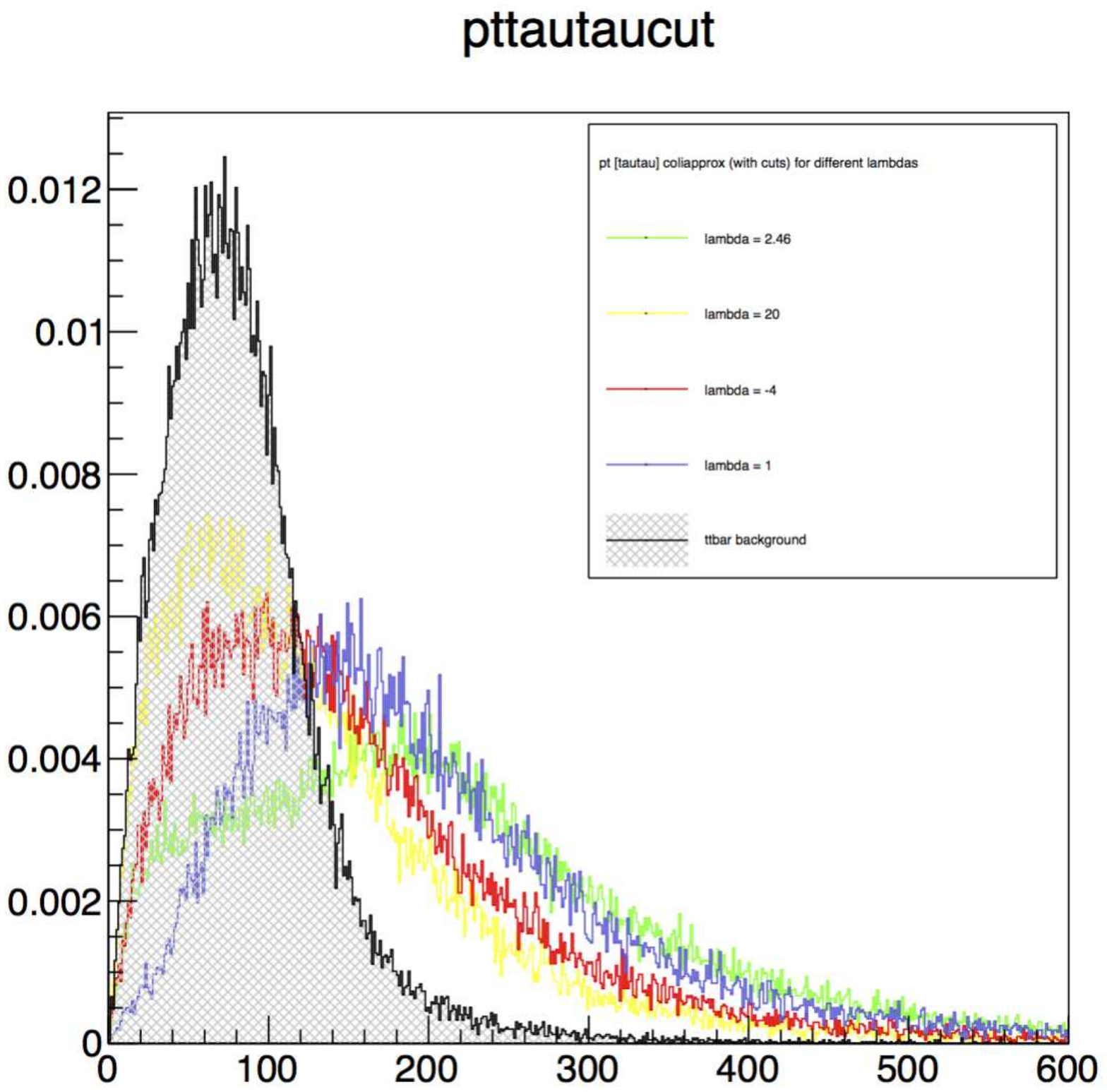}
\end{center}
\begin{center}
	\includegraphics[scale=0.37]{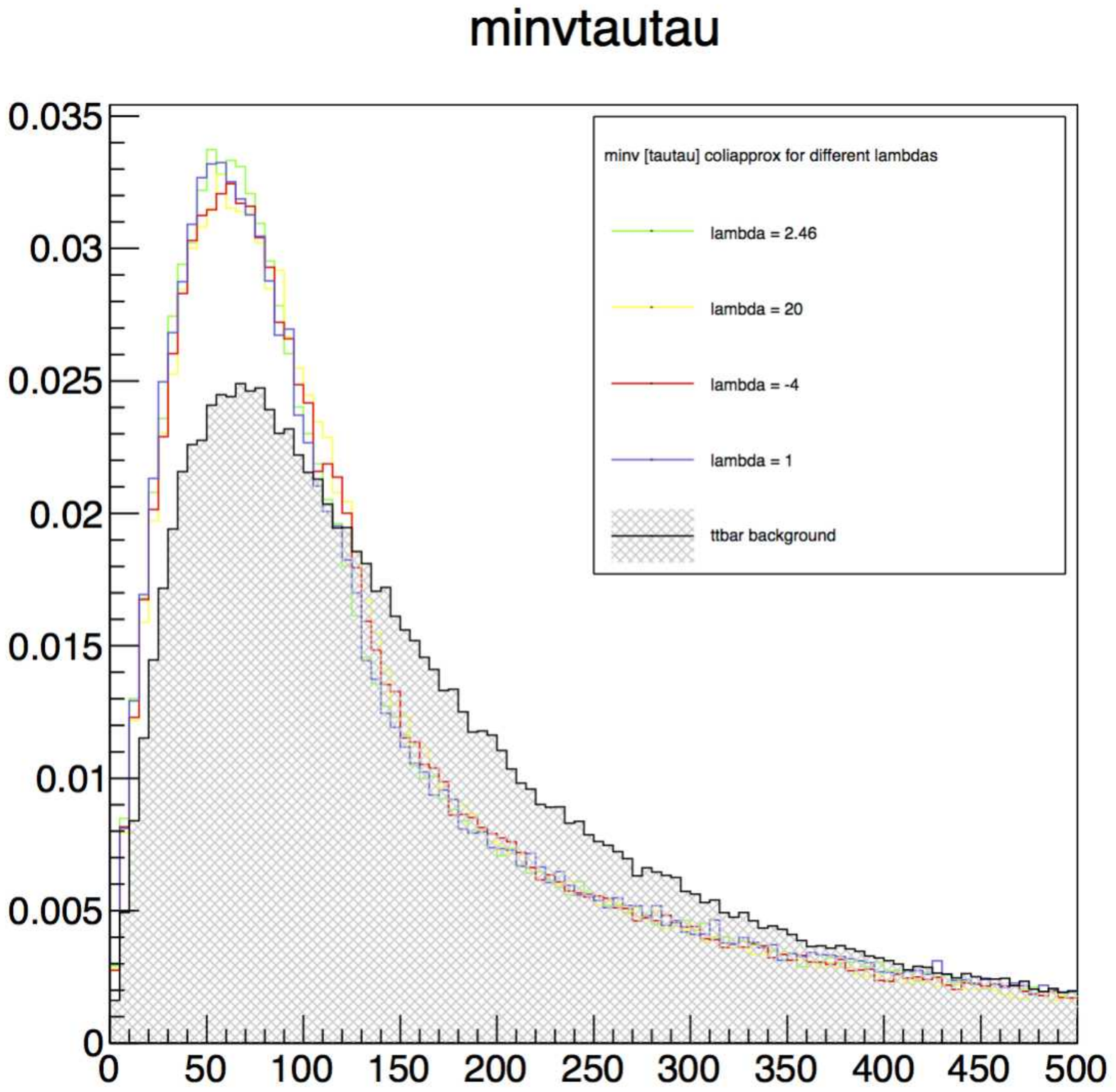}
	\includegraphics[scale=0.37]{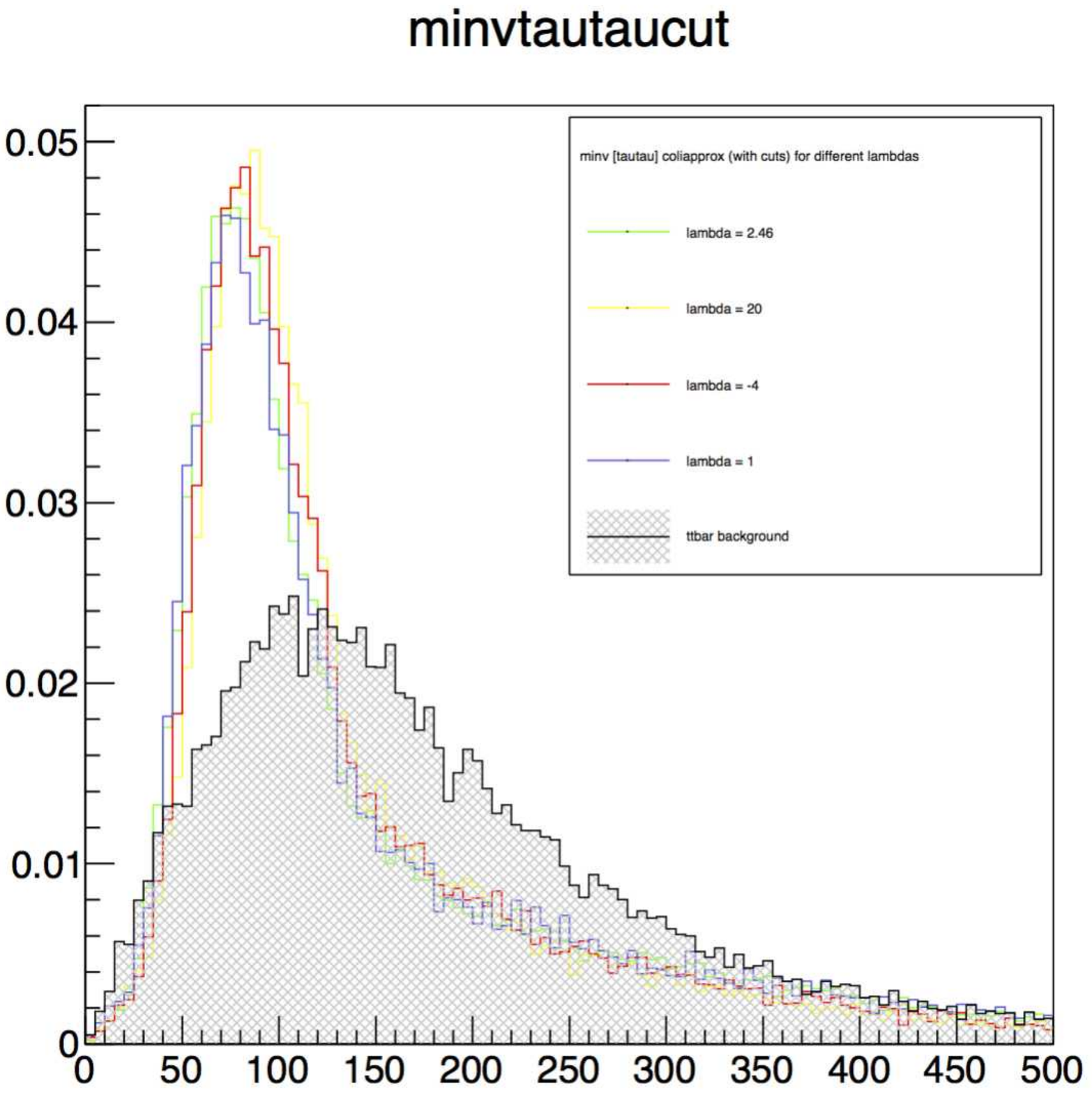}
\end{center}
\begin{center}
	\includegraphics[scale=0.37]{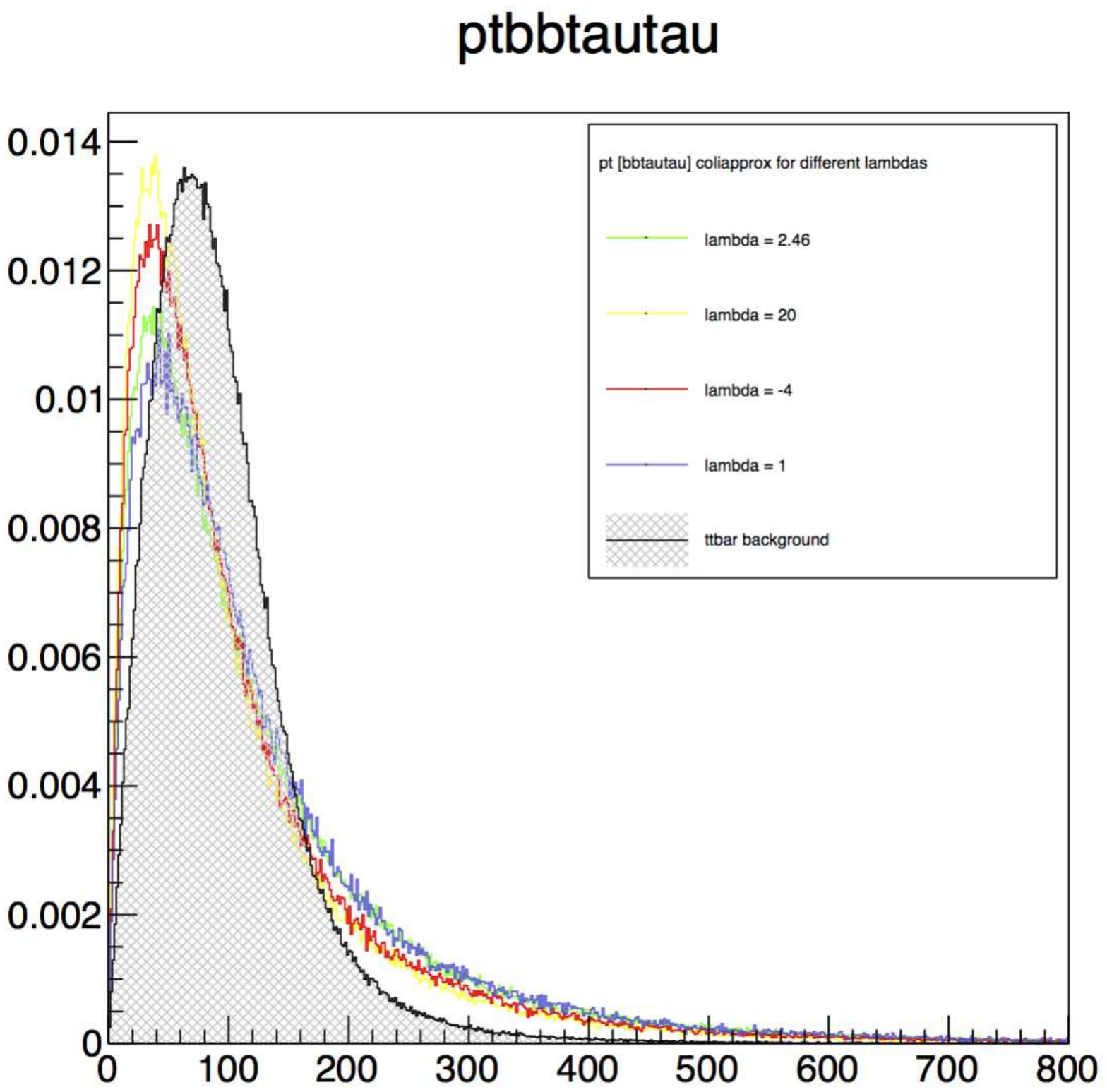}
	\includegraphics[scale=0.37]{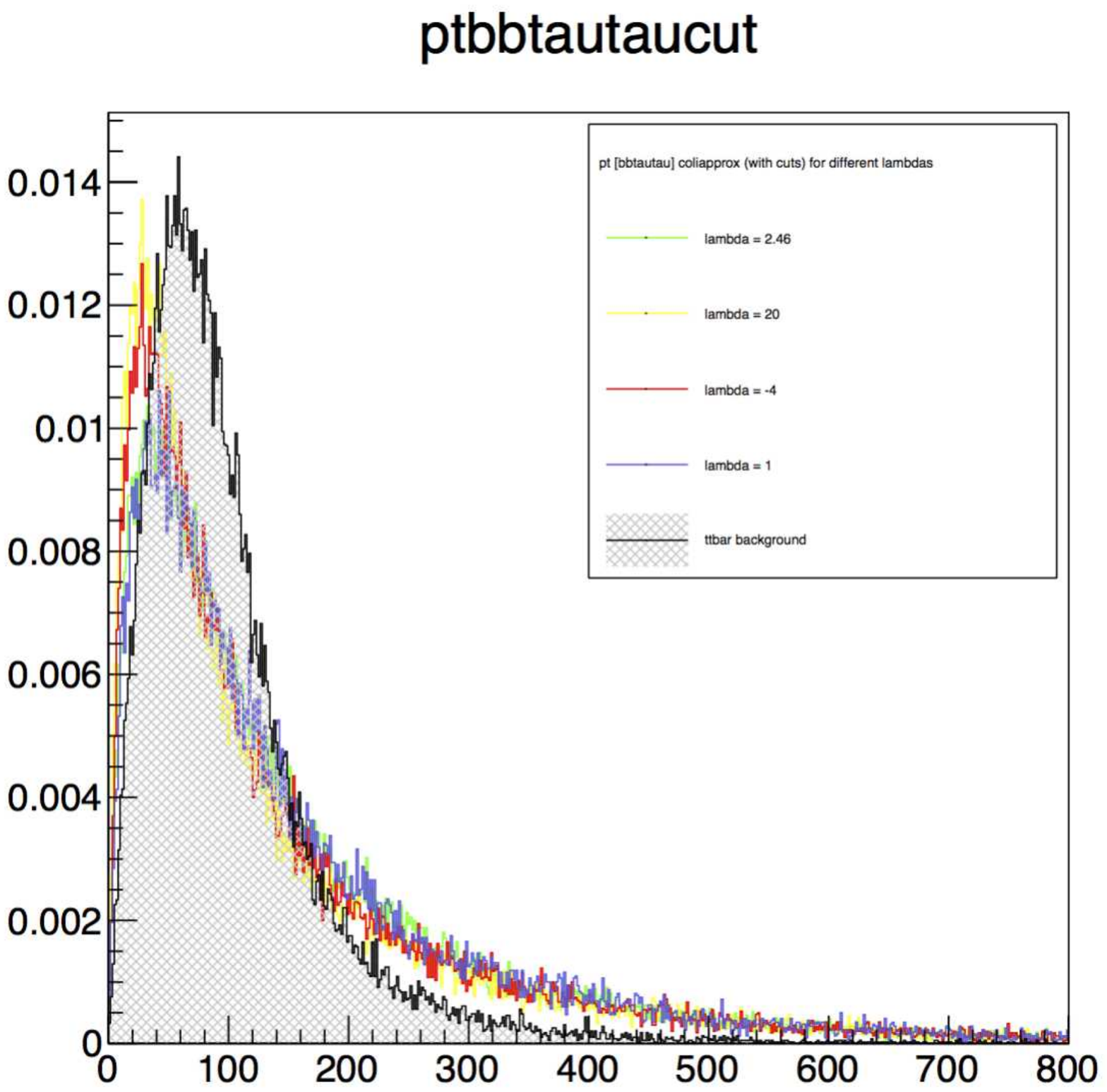}
\end{center}
\begin{center}
	\includegraphics[scale=0.37]{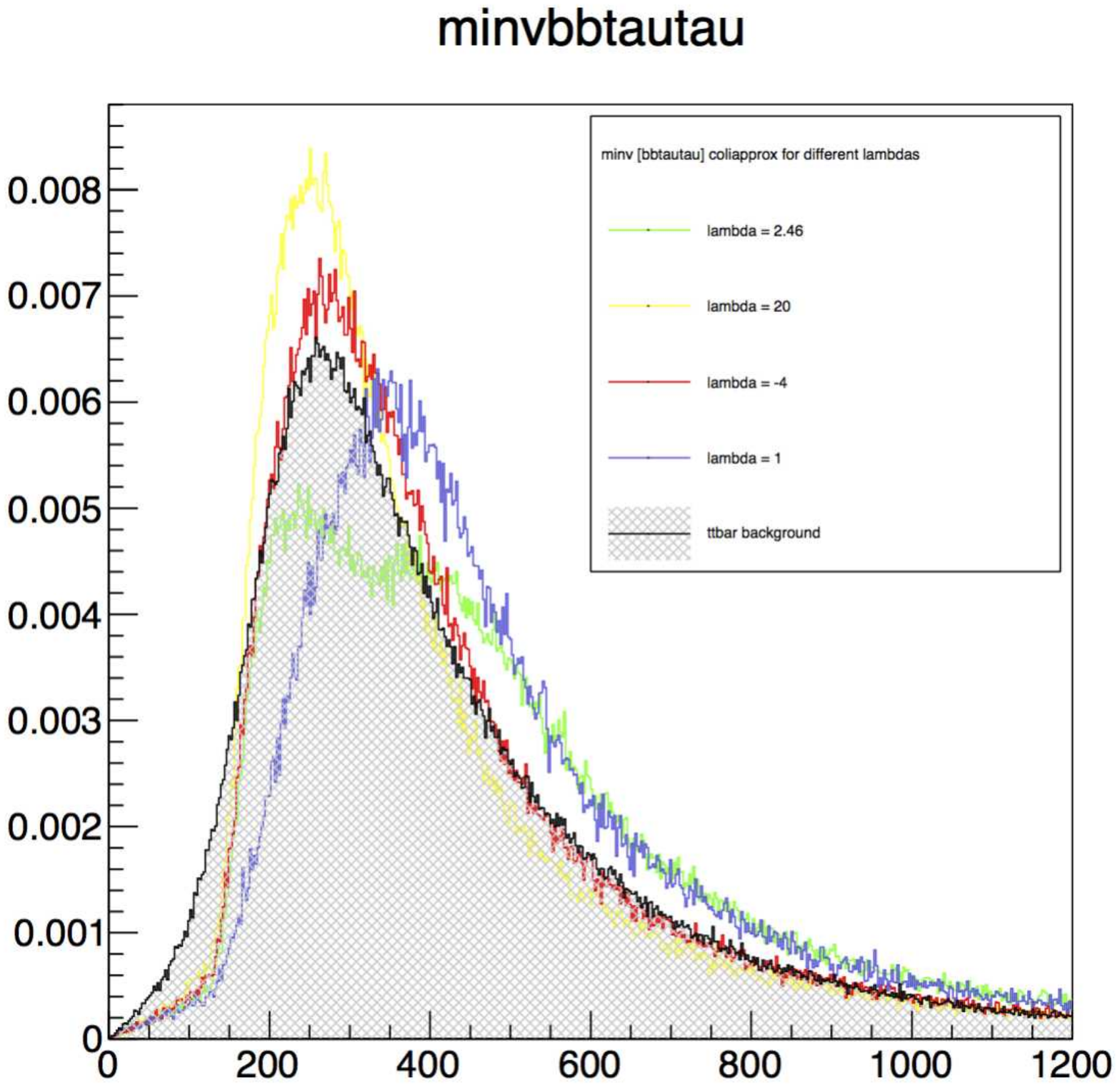}
	\includegraphics[scale=0.37]{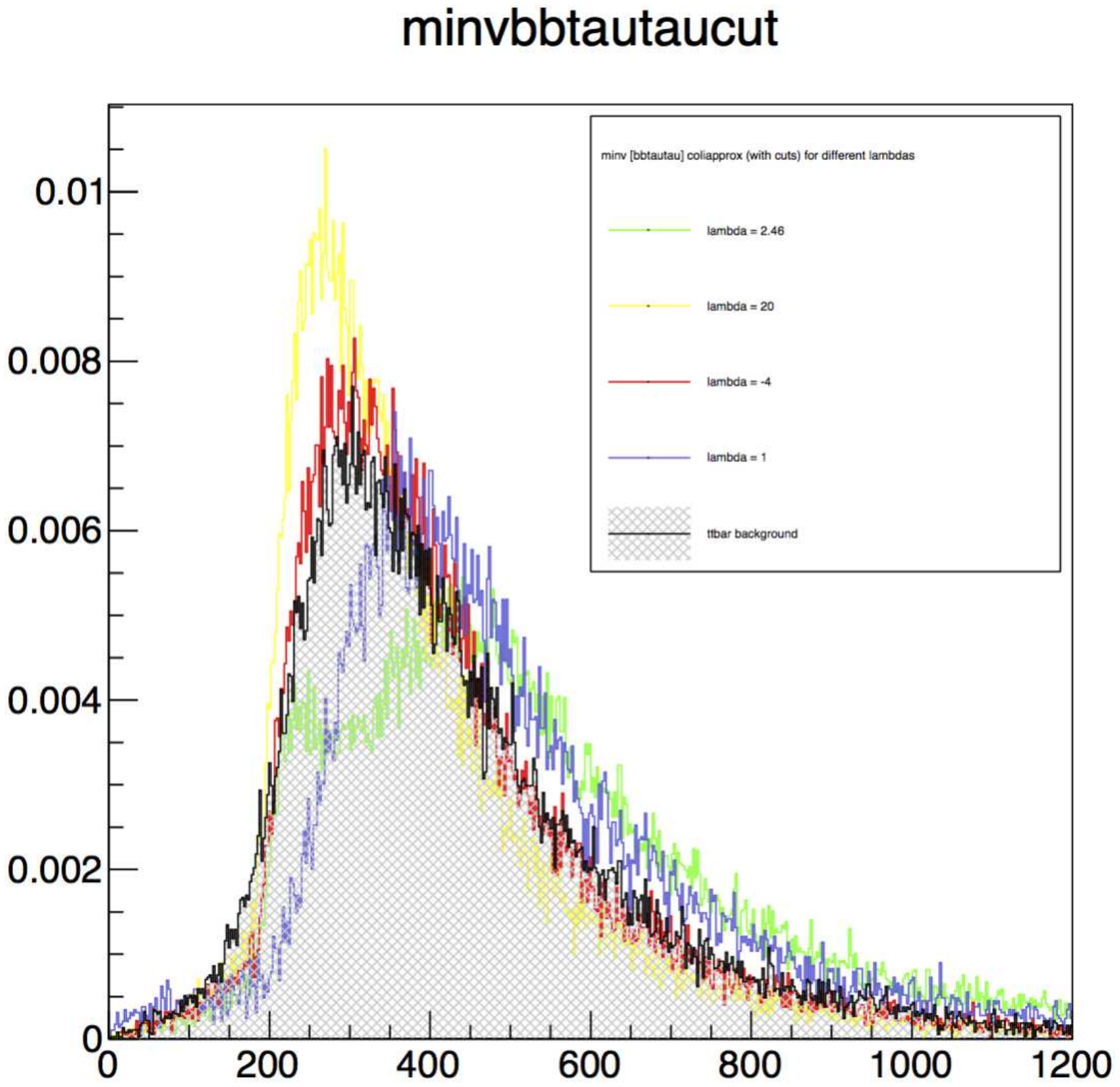}
\end{center}

\section{Utilisation des particules reconstruites}
\subsection{Reconstruction des particules}
L'étude étant faite à un niveau générateur, le fichier ROOT utilisées possède à la fois les informations vraies générées par les méthodes statistiques et les informations reconstruites. Le but est donc d'associer à chaque particule vraie celle reconstruite dans le détecteur, pour avoir une idée encore plus précise de l'impact de l'observation imparfaite et de la reconstruction sur les distributions étudiées.\\

La méthode employée est la suivante : on sait grâce aux informations générateur que dans telle collision, les système de deux $\tau$ se désintègre - par exemple - en un muon et un jet hadronique. On cherche donc une particule reconstruite en muon et une particule reconstruite comme un jet de désintégration de tau, et dont la direction de vol n'est pas trop différente du muon qu'on devrait effectivement observer si le détecteur était parfait, ainsi que la reconstruction. Le critère de sélection (l'impulsion étant paramétrée par les angles $\theta$ et $\phi$ est défini de la manière suivante. Posons : 
$$\Delta R=\sqrt{\Delta\theta^2+\Delta\phi^2}$$
Alors on cherche des particules reconstruites de la bonne nature dans un cône $\Delta R=0.5$ autour de l'impulsion de la particule "vraie".
S'il y en a plusieurs, il est probable que la reconstruction se fasse mal à cause de la proximité des traces laissées par les particules, on ne tient donc pas compte de la collision. S'il y en a aucune, on ne peux pas non plus tenir compte de la collision. \\\\
Une fois cette association réalisée, on prend pour calculer les invariants cinétiques de la collision qui ont été déterminés comme intéressant pas les études précédentes, non pas les quadrivecteurs des particules vraies mais les quadrivecteurs des particules telles qu'elles sont détectées. \\\\
Il faut redéfinir les efficacités. L'efficacité totale est le ratio d'évènements bien reconstruits sur le nombre total d'évènements. Pour chacune des désintégrations du système $[\tau,\tau]$, l'efficacité est le ratio du nombre de tels évènements bien reconstruits sur le nombre total d'évènements de ce type. On obtient : 
\begin{center}
	\begin{tabular}{|c||c|c|c|c|c|}
		\hline
		& $\lambda=-4$ & $\lambda=-1$ & $\lambda=2.46$ & $\lambda=20$ & $t\overline{t}$ \\
		\hline
		globale &0.160  & 0.183 & 0.177 & 0.152 & 0.135 \\
		$\tau_\mu\tau_h$ & 0.179 & 0.205 & 0.120 & 0.171 & 0.153 \\
		$\tau_e\tau_h$ & 0.156 & 0.179 & 0.175 & 0.149 & 0.131 \\
		$\tau_h\tau_h$ & 0.122 & 0.139 & 0.133 & 0.113 & 0.096 \\
		$\tau_e\tau_\mu$ & 0.273 & 0.303 & 0.292 & 0.251 & 0.239 \\
		$\tau_e\tau_e$ & 0.223 & 0.256 & 0.255 & 0.216 & 0.208 \\
		$\tau_\mu\tau_\mu$ & 0.301 & 0.354 & 0.330 & 0.294 & 0.289 \\
		\hline
	\end{tabular}
\end{center}

Si on compare ce tableau à celui obtenu dans l'étude générateur en prenant en compte la désintégration des $\tau$, on voit que les efficacités individuelles de chacun des modes de désintégration sont plus basses dans le cas de la reconstruction, pour les raisons de non-idéalité de cette étape. Cependant, l'efficacité globale et pour certaines valeurs du fond irréductible est plus grande avec la reconstruction. Cela s'explique par le fait qu'en réalité, les coupures cinétiques correspondant le mieux à la reconstruction des quarks bottom comprennent une coupure à 20 GeV sur l'impulsion transverse, et non à 30 GeV comme ce nous avons utilisé dans l'étude précédente. Il y aurait également une modification à apporter concernant la reconstruction des muons, qui est meilleure que ce qui était simulé par les coupures, et celle des électrons, moins bonne que celle attendue (on s'en rend compte à l'aide de la dernière colonne de ce tableau).\\\\
Dans un premier temps, nous allons mettre en exergue la modification de l'allure de la distribution des quantités cinétiques sélectionnées à cause des processus physiques (désintégrations avec émission de neutrinos et reconstruction) et matériels (reconstruction imparfaite), en juxtaposant les distributions "idéales" obtenues lors de la première étude (sans coupures) (à gauche) et les distributions obtenues avec la reconstruction (à droite).\\
Dans un deuxième temps, nous montrerons les effets de la reconstruction seulement : comme auparavant ou nous juxtaposions histogrammes obtenus avec et sans coupures, nous mettrons côte-à-côte les diagrammes obtenus, sans coupure, au niveau générateur et en considérant les désintégrations des $\tau$ (à gauche), et les histogrammes obtenus en considérant la reconstruction (à droite).

\subsection{Étude de l'efficacité totale de l'observation}

\paragraph{Méthode des quantités visibles}
\begin{center}
	\includegraphics[scale=0.35]{figure4o.pdf}
	\includegraphics[scale=0.35]{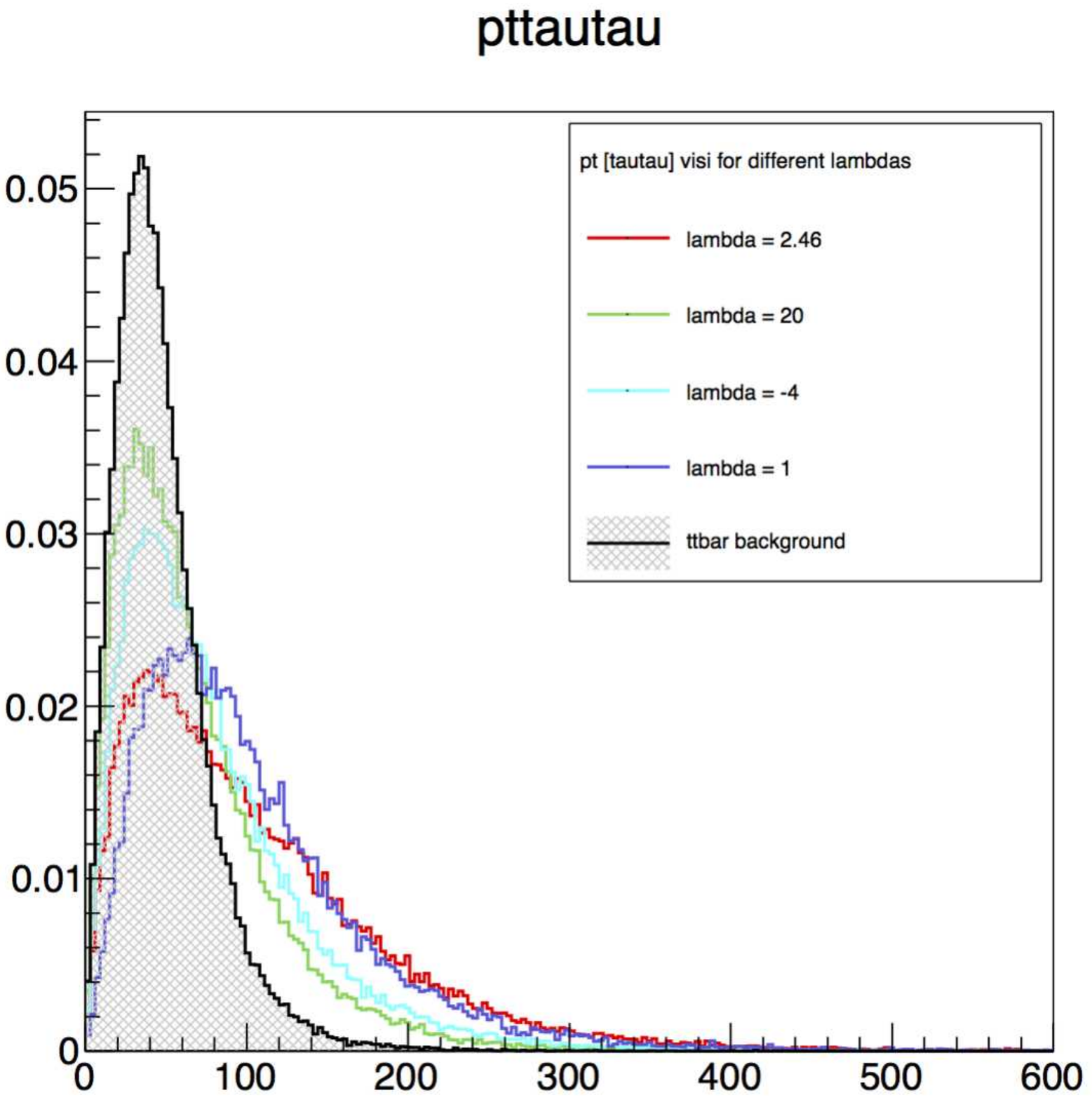}
\end{center}
\begin{center}
	\includegraphics[scale=0.35]{figure4u.pdf}
	\includegraphics[scale=0.35]{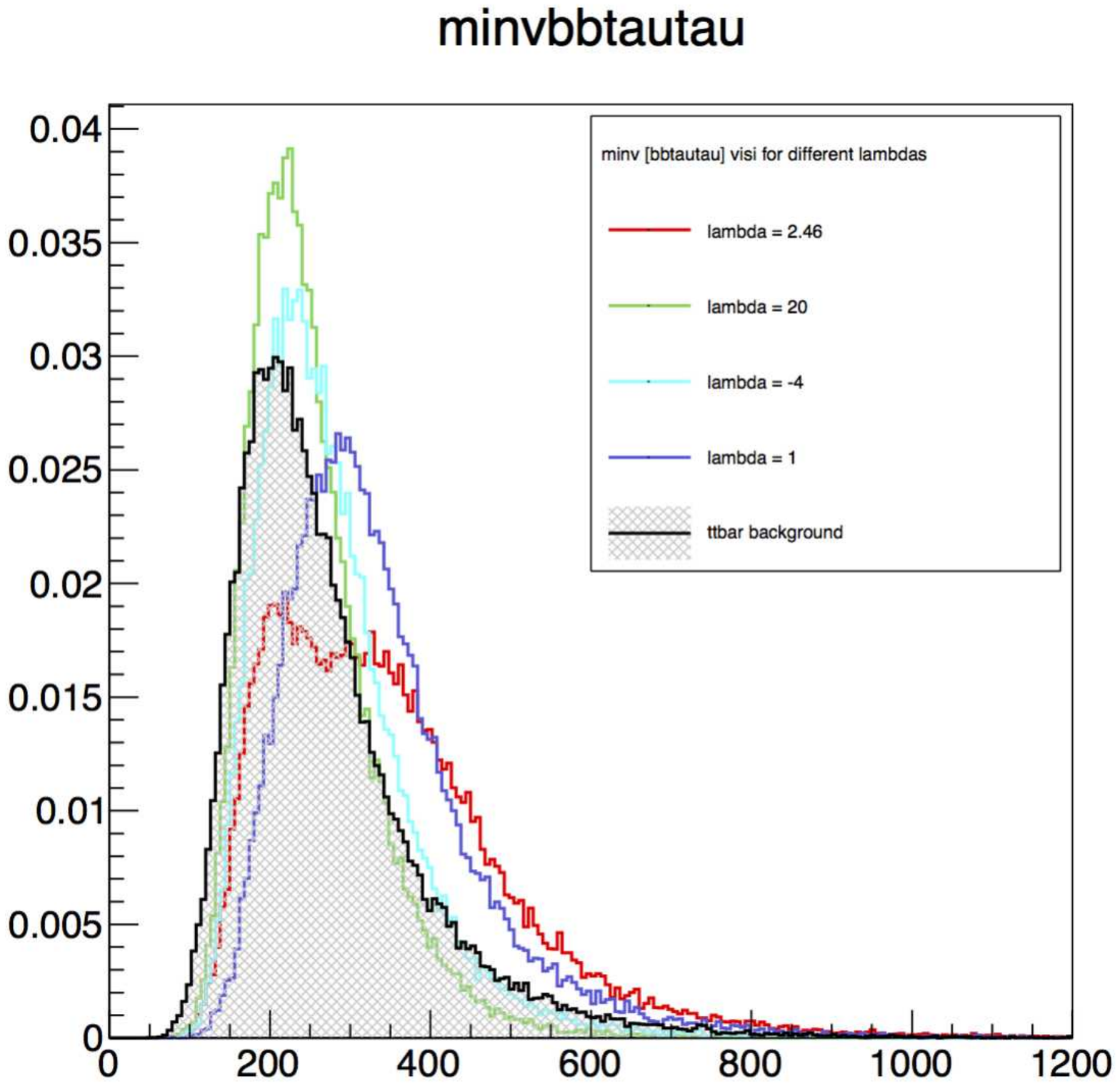}
\end{center}
\paragraph{Méthode des quantités effectives}
\begin{center}
	\includegraphics[scale=0.35]{figure4o.pdf}
	\includegraphics[scale=0.35]{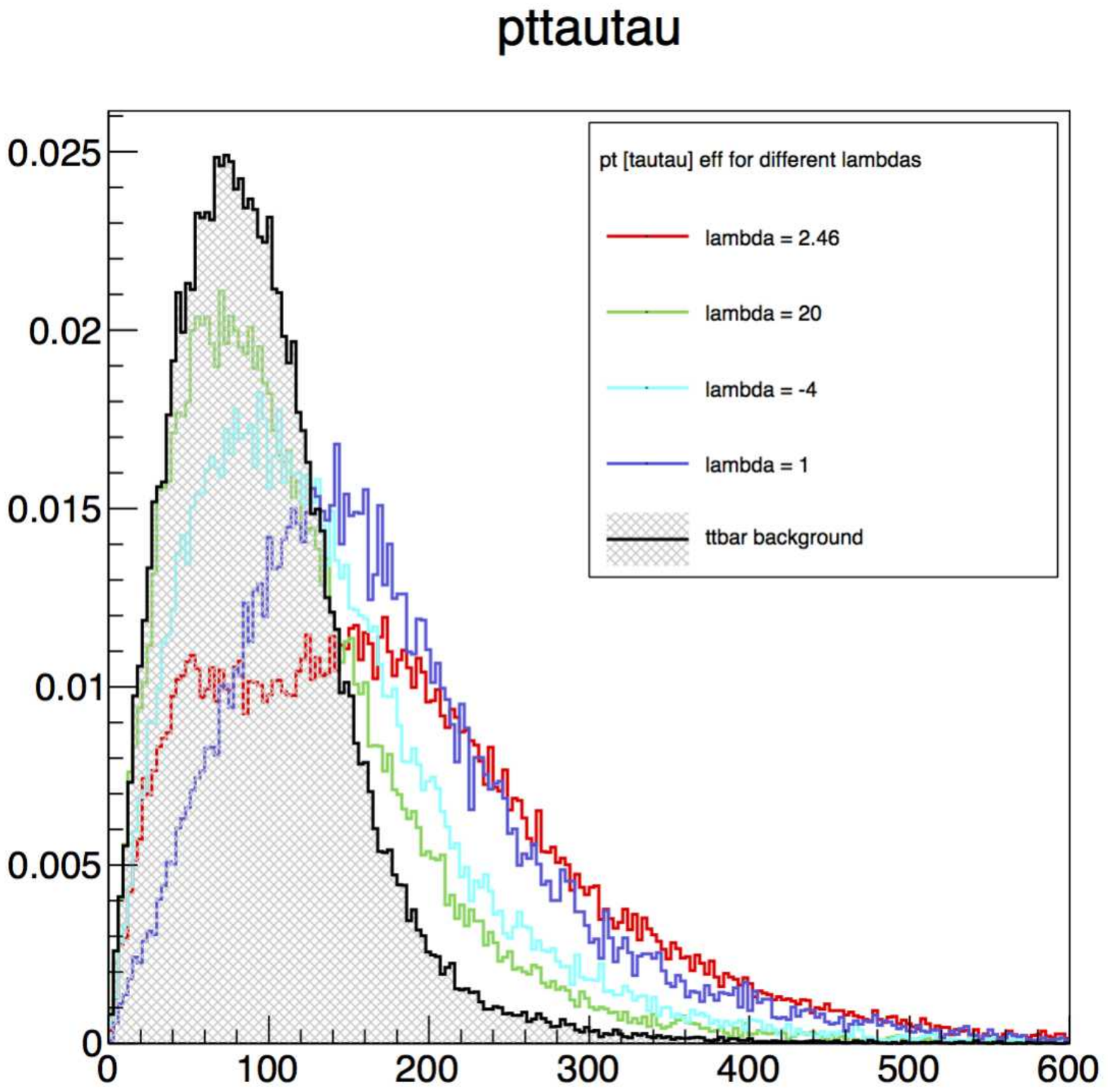}
\end{center}
\begin{center}
	\includegraphics[scale=0.35]{figure4u.pdf}
	\includegraphics[scale=0.35]{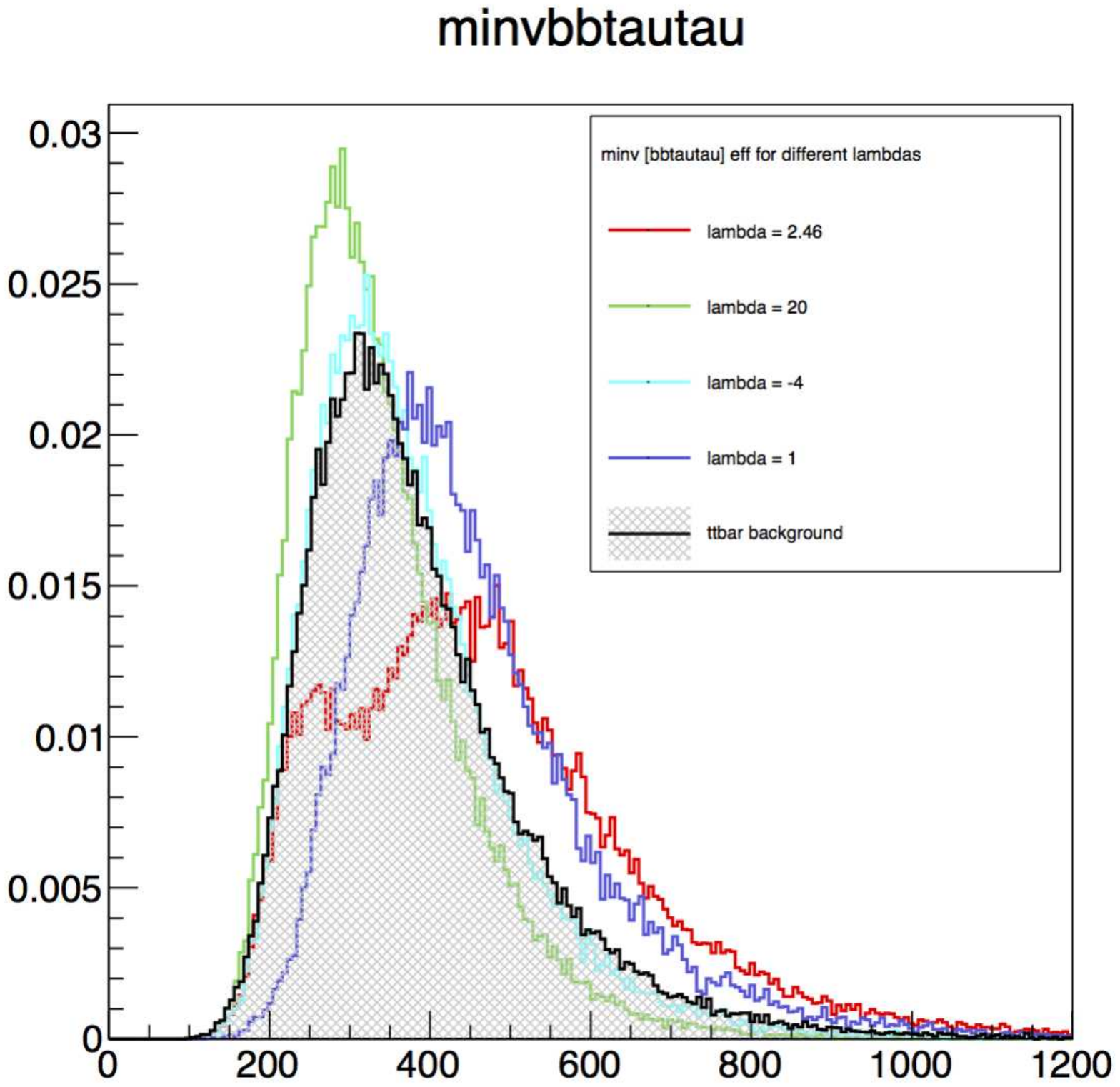}
\end{center}
\paragraph{Approximation colinéaire}
\begin{center}
	\includegraphics[scale=0.35]{figure4o.pdf}
	\includegraphics[scale=0.35]{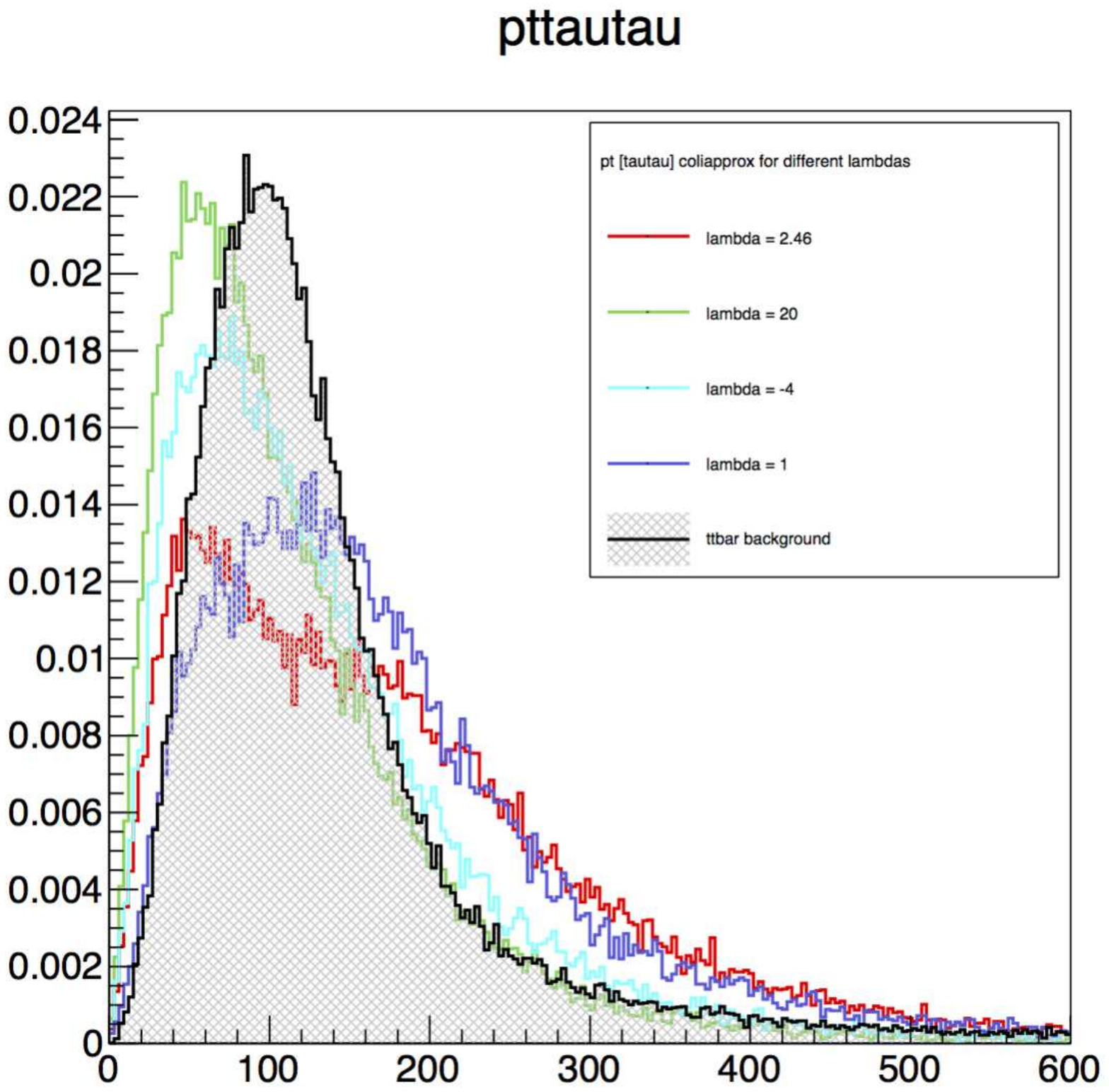}
\end{center}
\begin{center}
	\includegraphics[scale=0.35]{figure4u.pdf}
	\includegraphics[scale=0.35]{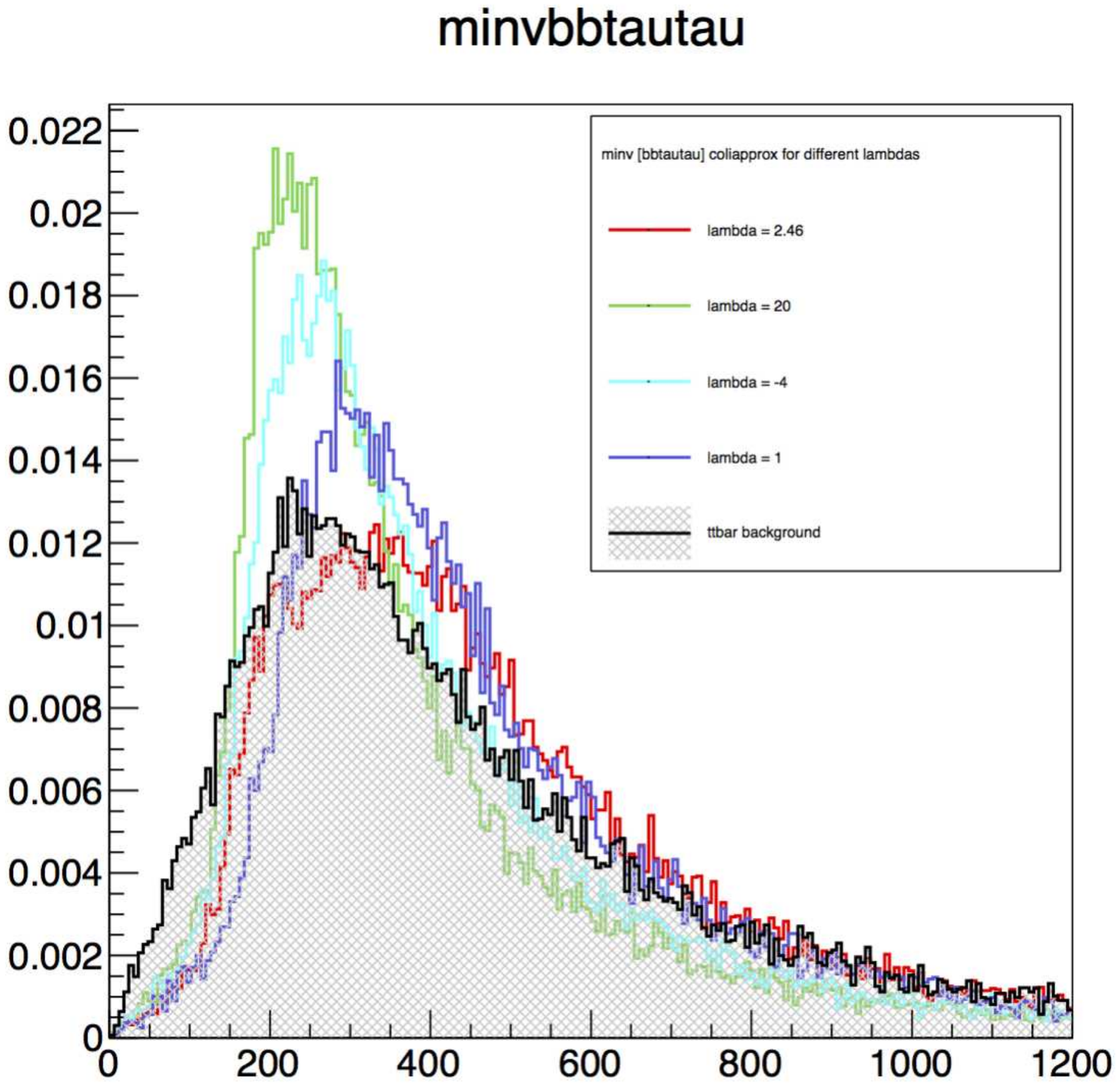}
\end{center}

\subsection{Étude spécifique de l'efficacité de reconstruction}
\paragraph{Méthode des quantités visibles}
\begin{center}
	\includegraphics[scale=0.35]{figure5a.pdf}
	\includegraphics[scale=0.35]{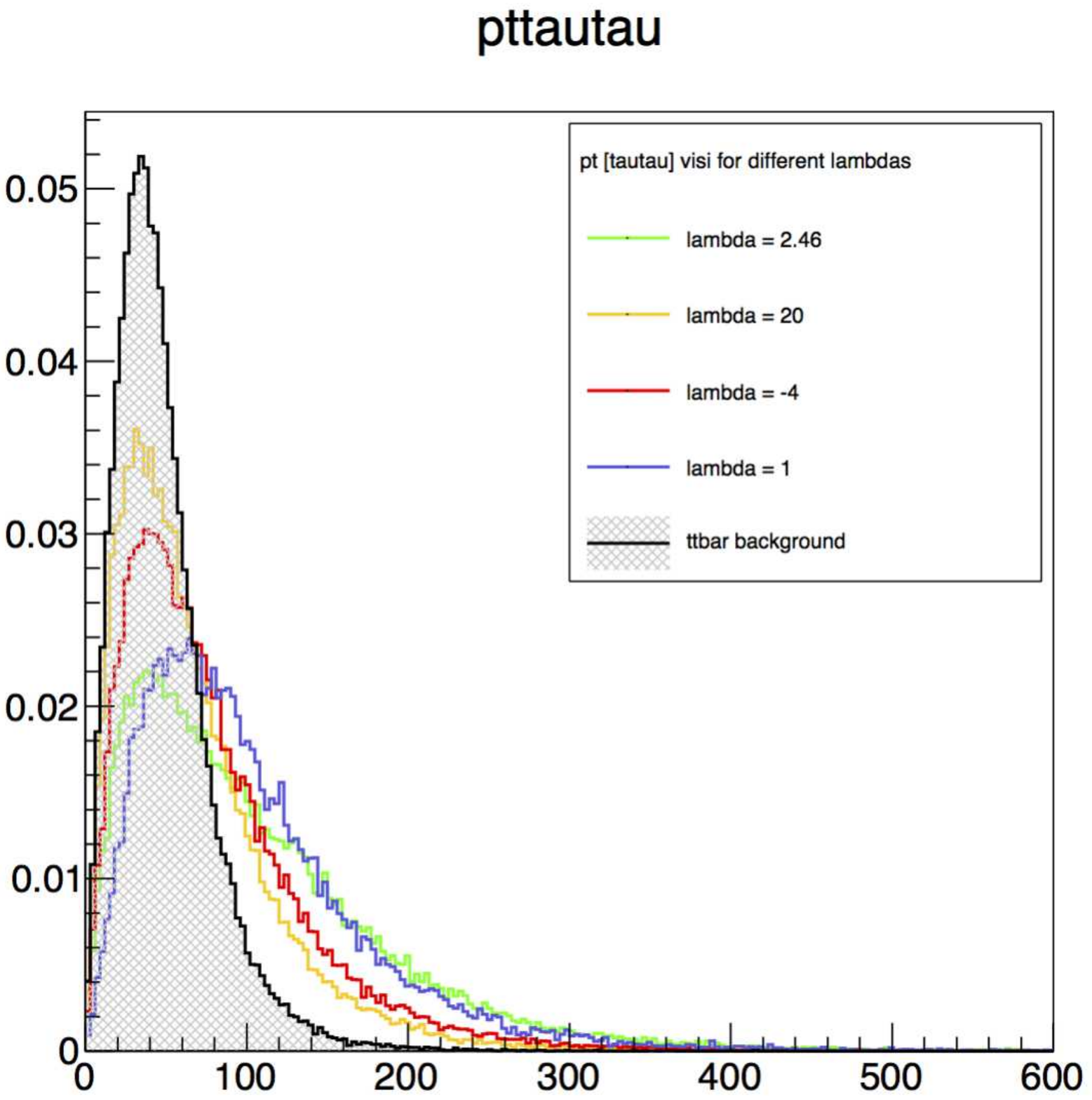}
\end{center}
\begin{center}
	\includegraphics[scale=0.35]{figure5g.pdf}
	\includegraphics[scale=0.35]{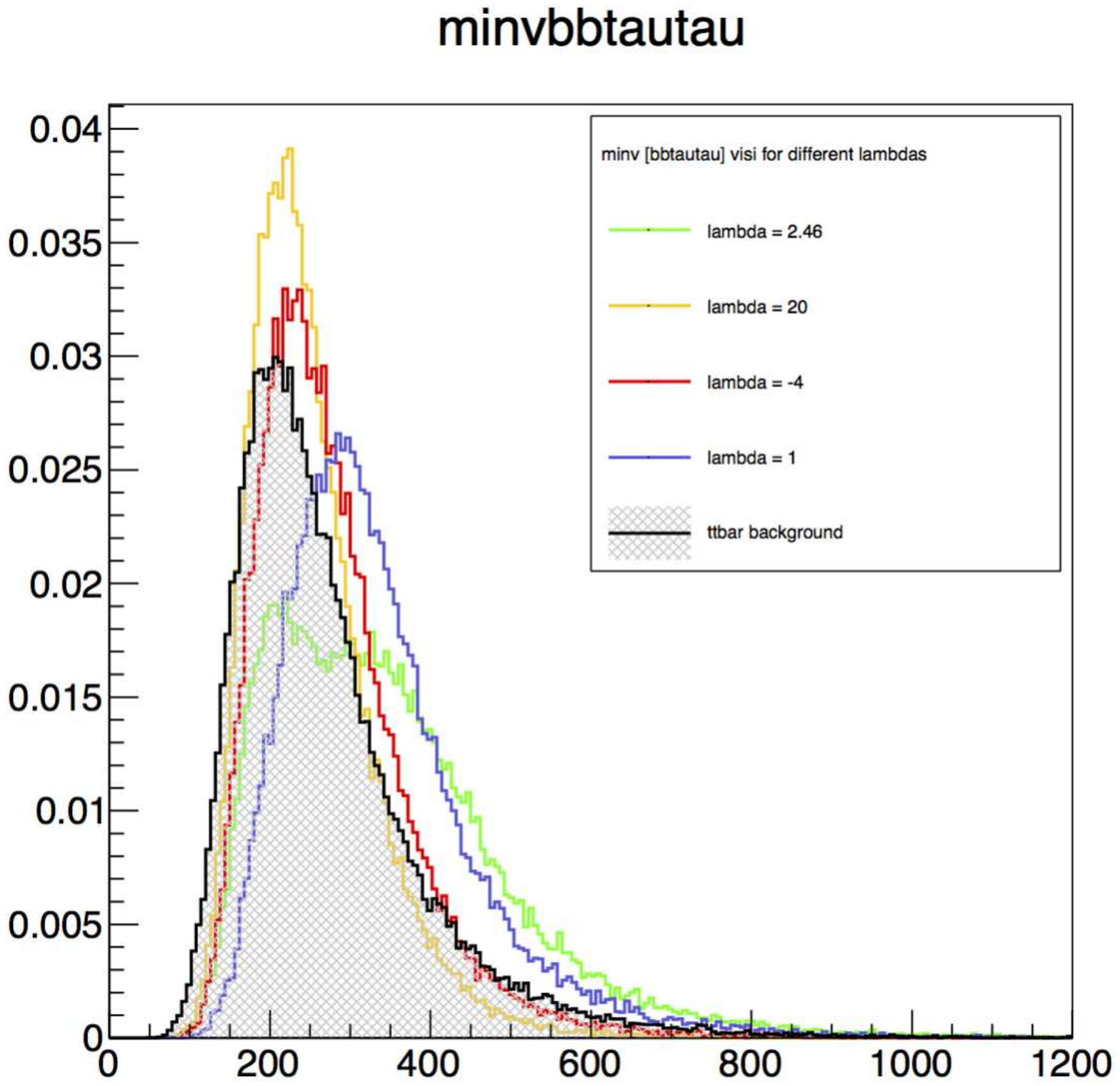}
\end{center}
\paragraph{Méthode des quantités effectives}
\begin{center}
	\includegraphics[scale=0.35]{figure5i.pdf}
	\includegraphics[scale=0.35]{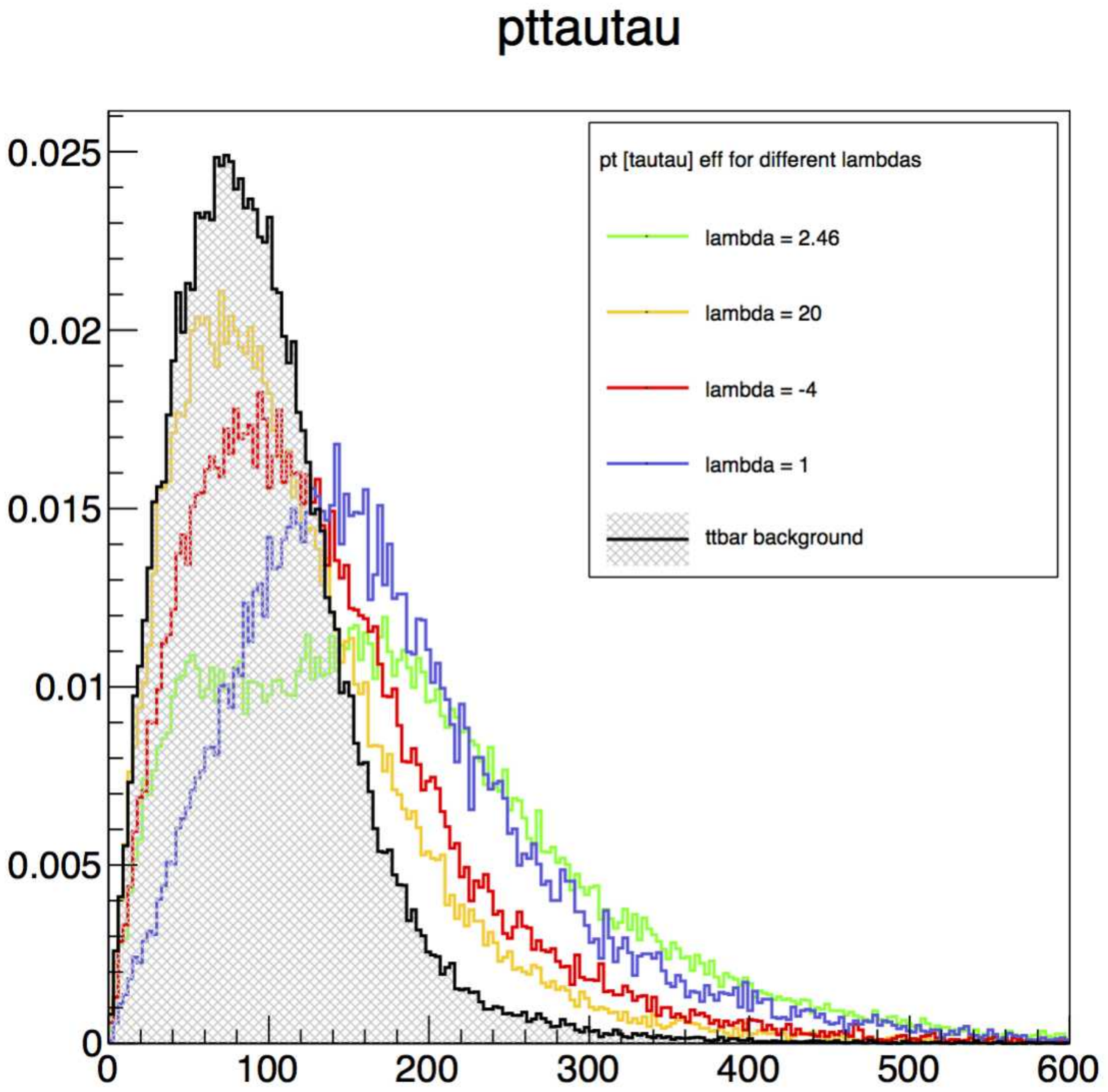}
\end{center}
\begin{center}
	\includegraphics[scale=0.35]{figure5o.pdf}
	\includegraphics[scale=0.35]{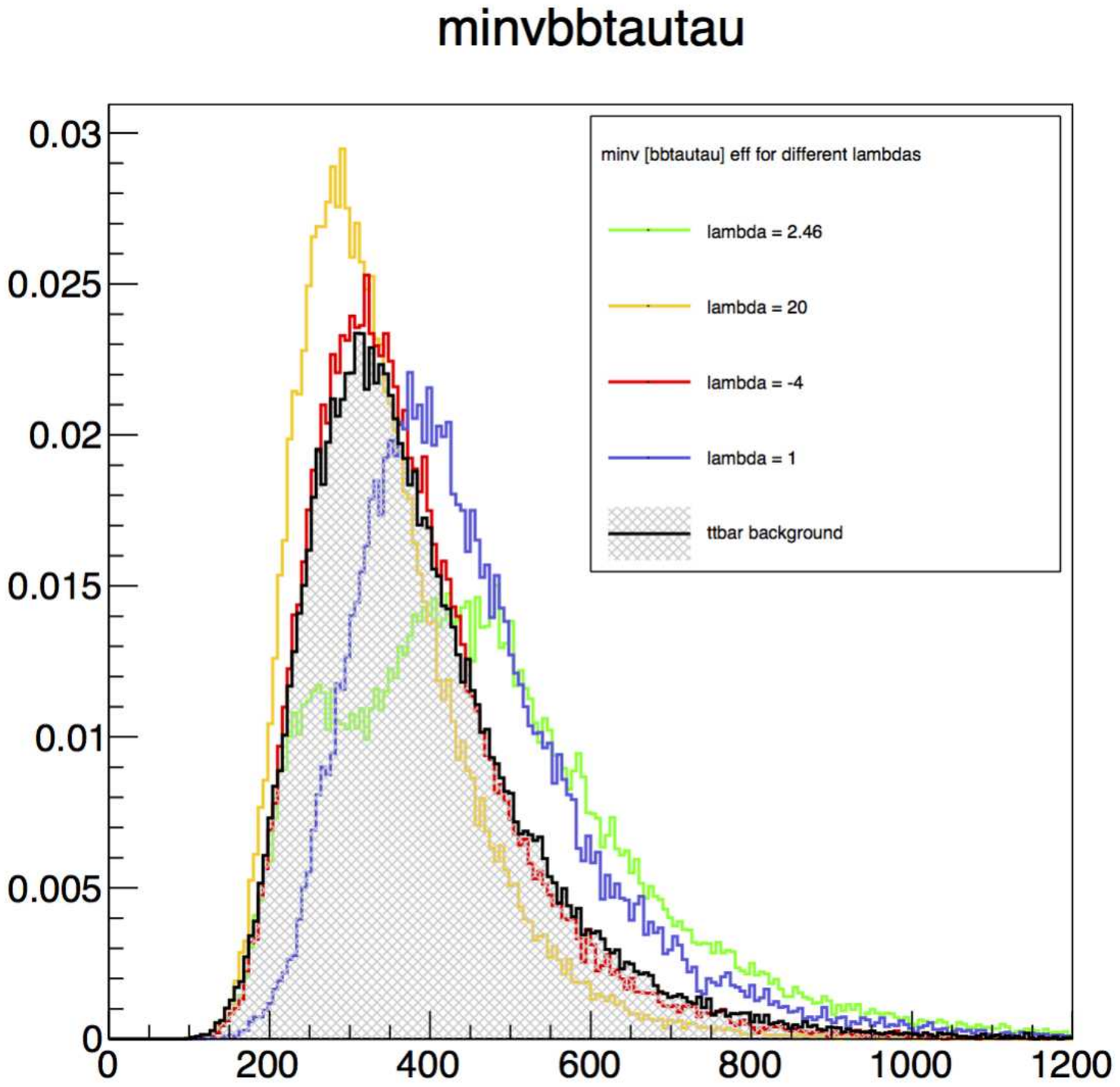}
\end{center}
\paragraph{Approximation colinéaire}
\begin{center}
	\includegraphics[scale=0.35]{figure5q.pdf}
	\includegraphics[scale=0.35]{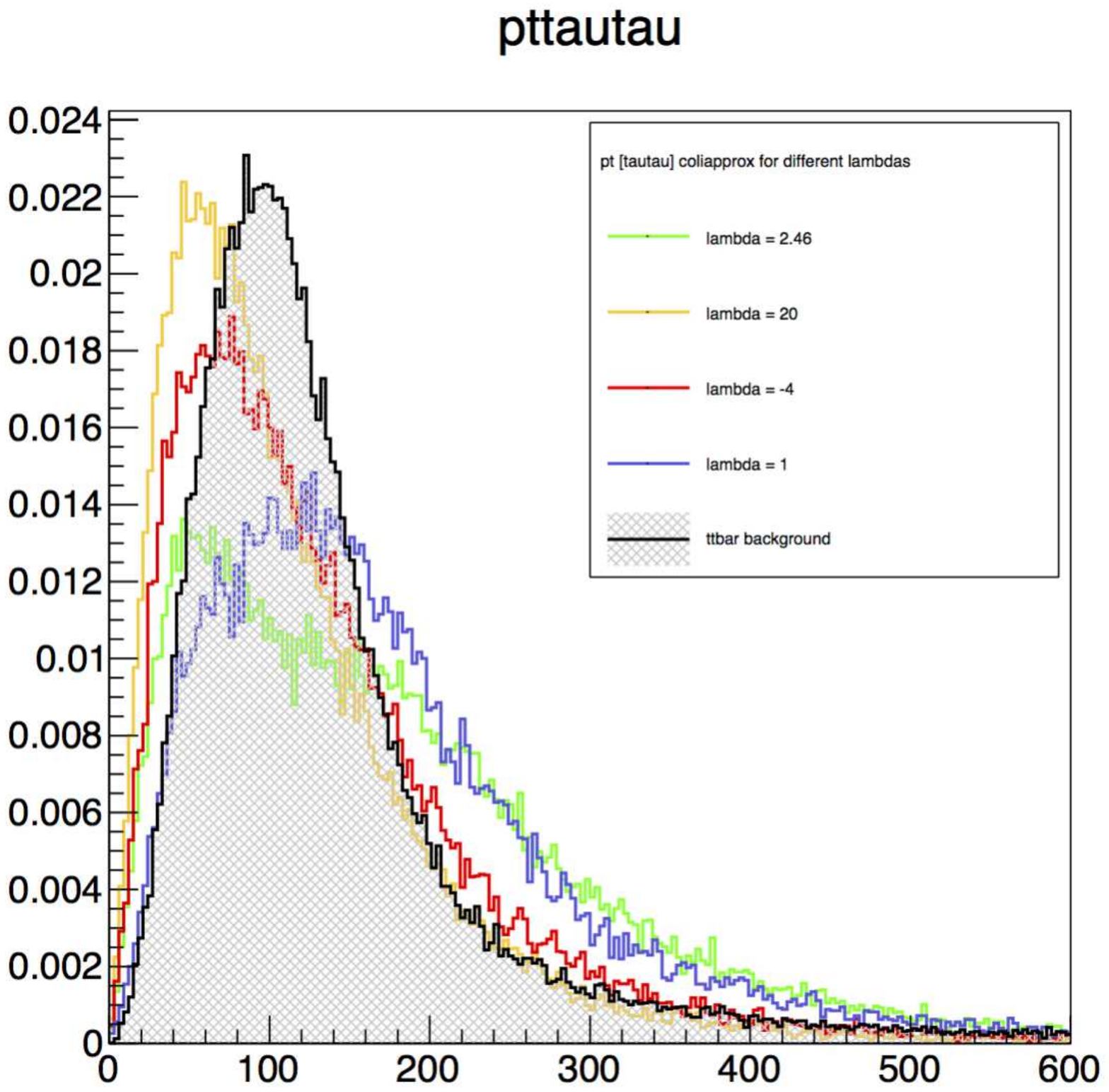}
\end{center}
\begin{center}
	\includegraphics[scale=0.35]{figure5w.pdf}
	\includegraphics[scale=0.35]{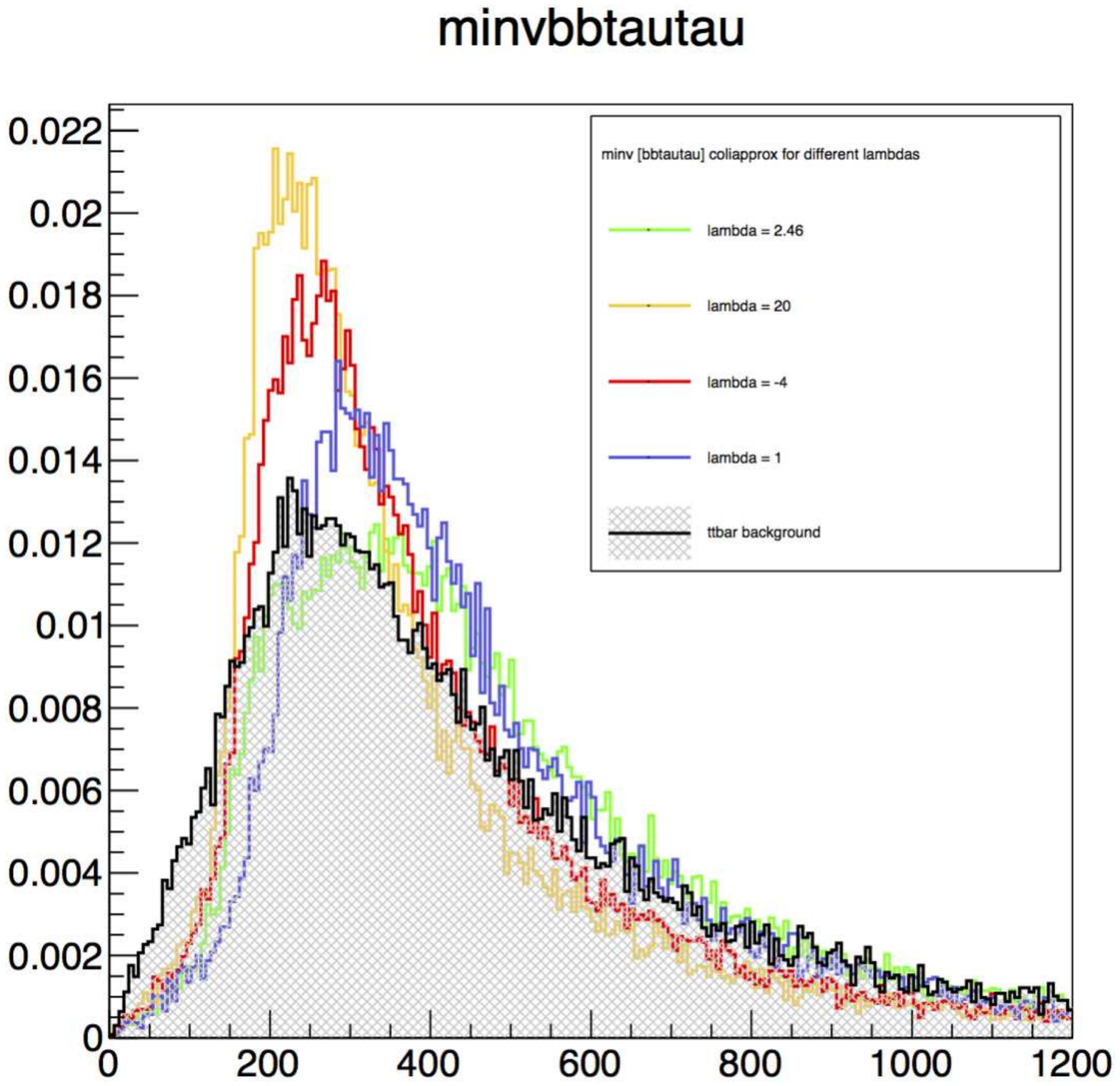}
\end{center}

\section{Paramétrisation des grandeurs cinétiques}
Nous avons ensuite essayé de paramétrer les distributions obtenues, afin d'obtenir une expression analytique au moins locale de la surface (locale en $\lambda$).
Pour pouvoir générer des fichiers pour davantage de valeurs de $\lambda$, il a fallu choisir une méthode rapide. La solution a été de produire des fichiers d'environ $100000$ collisions avec MadGraph, sans tenir compte des effets radiatifs, de la désintégration des $\tau$, etc ... Pour le reste, l'étude est menée exactement comme dans la première partie.
Dans un premier temps, les fichiers générés correspondaient aux valeurs de $\lambda$ suivantes : $\lambda=-15, -4, 0, 1, 2.46, 4, 10$.
\begin{figure}[!h]
	\centering
	\includegraphics[scale=0.37]{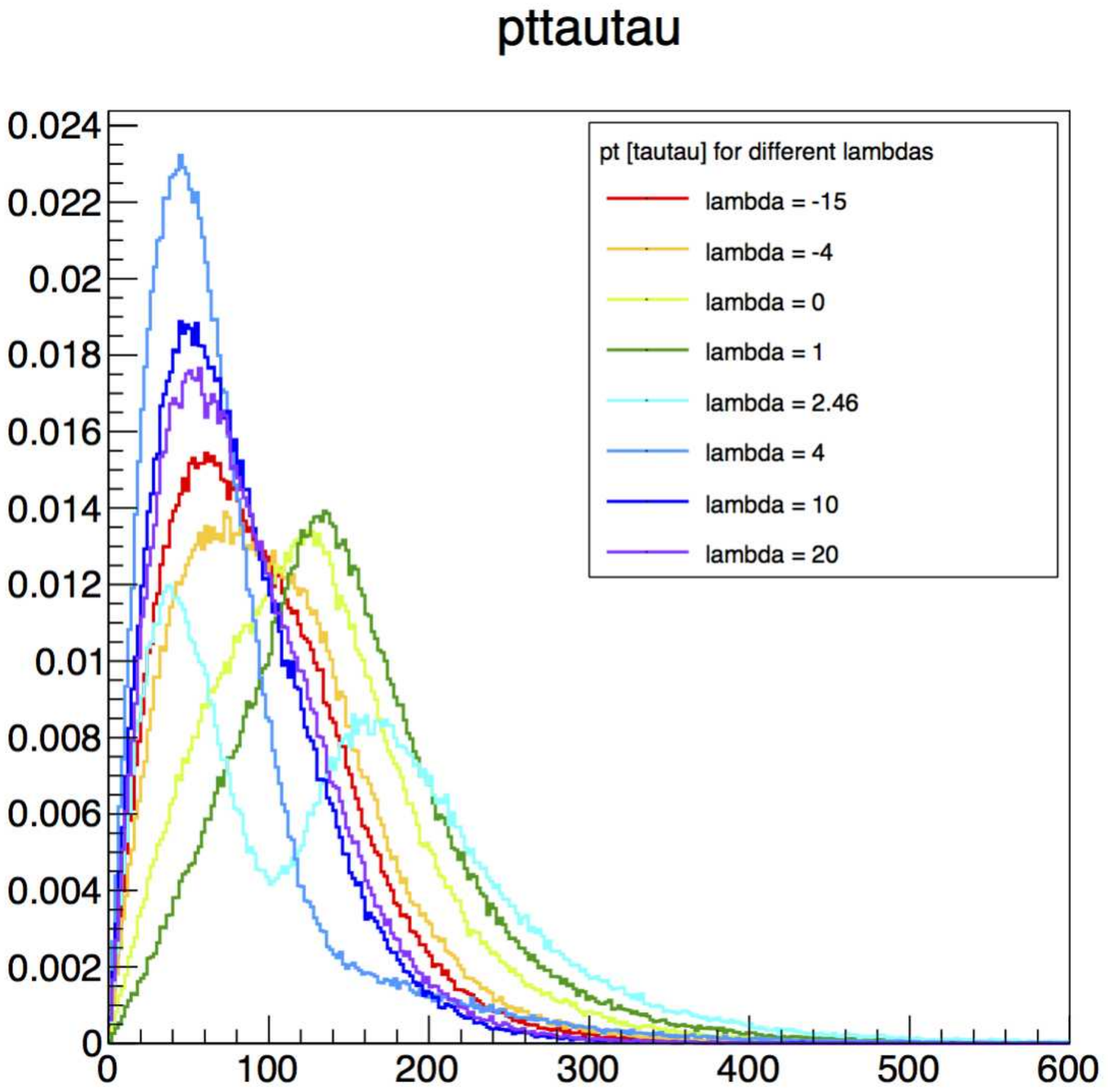}
	\includegraphics[scale=0.37]{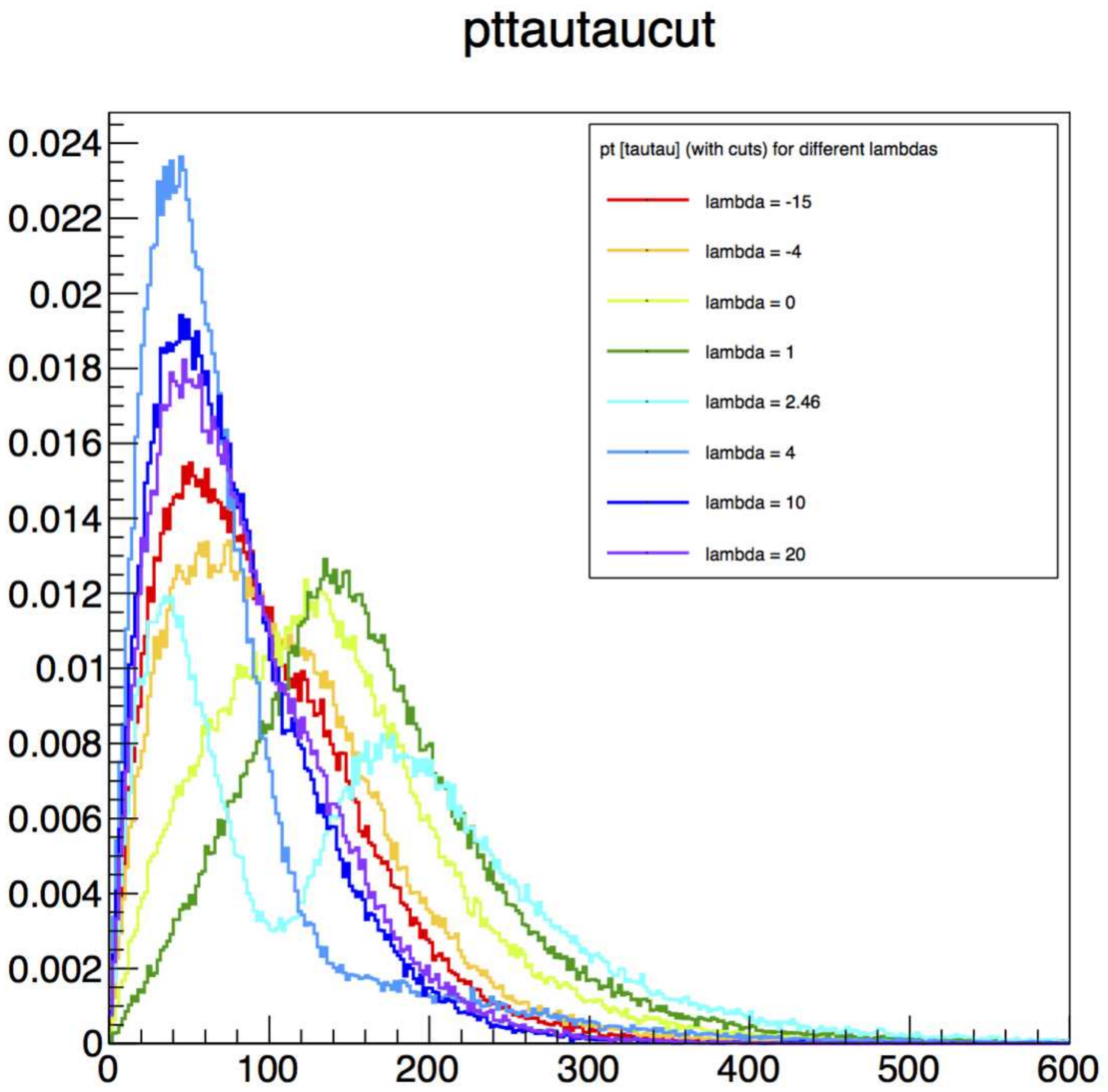}
	\caption{Superposition des histogrammes en impulsion transverse du système di-tau, pour différentes valeurs de $\lambda$. A gauche, tous les évènements ont été pris en compte, à droite, les coupures cinétiques de la première étude ont été appliquées.}
\end{figure}
\begin{figure}[!h]
	\centering
	\includegraphics[scale=0.37]{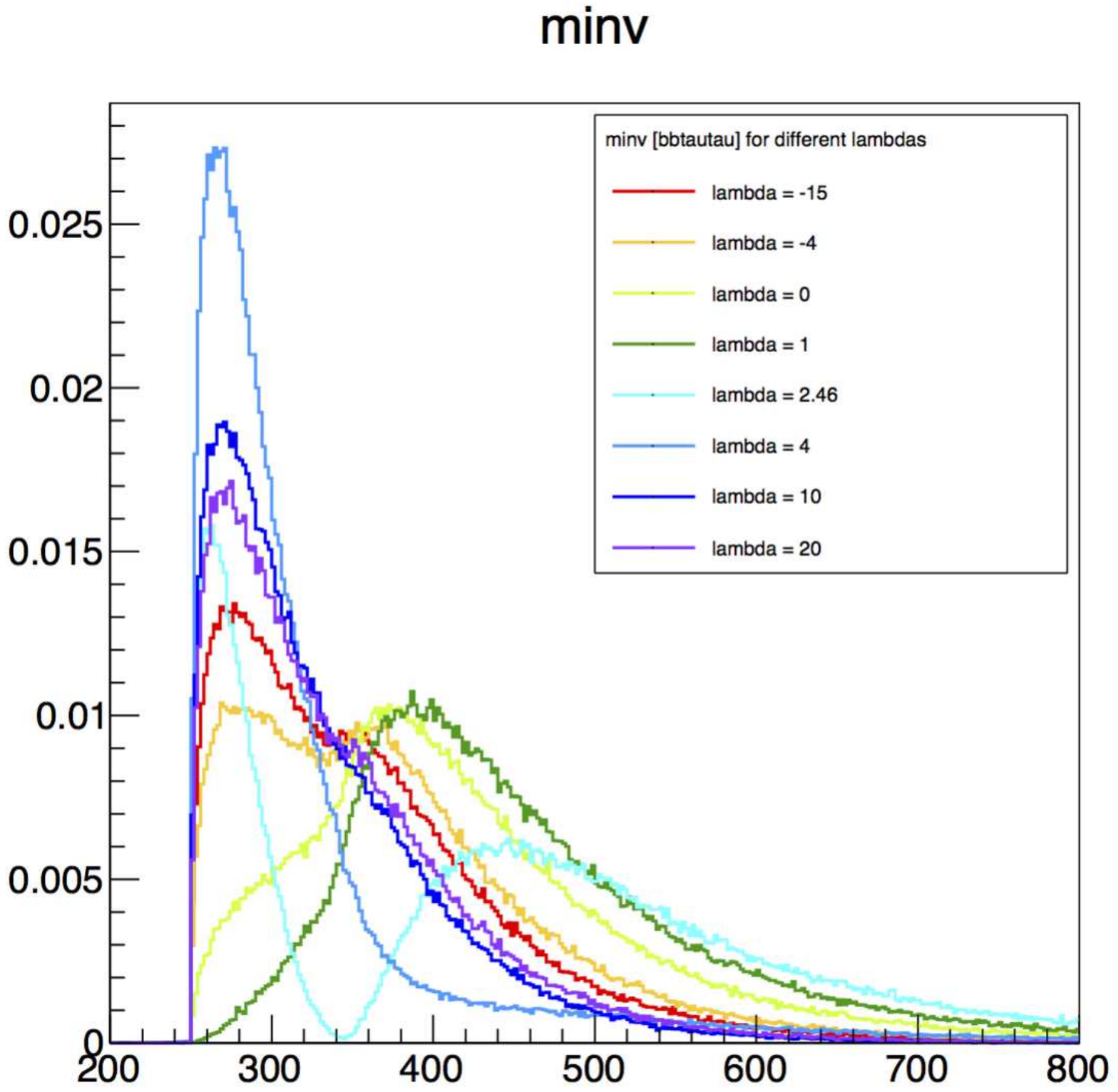}
	\includegraphics[scale=0.37]{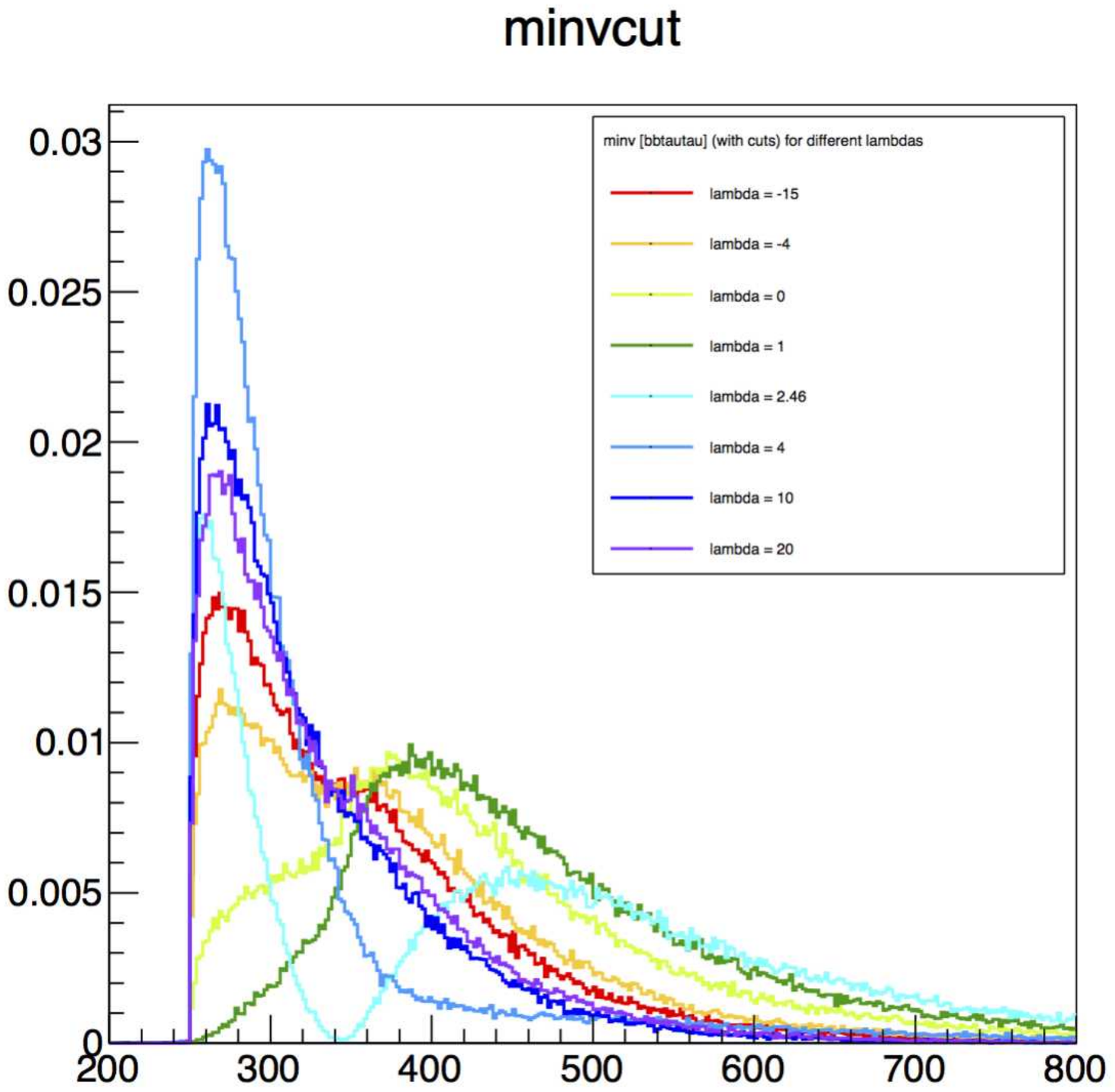}
	\caption{Superposition des histogrammes de masse invariante du système entier, pour différentes valeurs de $\lambda$. A gauche, tous les évènements ont été pris en compte, à droite, les coupures cinétiques de la première étude ont été appliquées.}
\end{figure}

Ces diagrammes semblent indiquer que les distributions que nous avons jugées dignes d'intérêt varient peu lorsque $\lambda$ s'éloigne de la région $[0;4]$. Physiquement, cela signifie que la contribution du diagramme de Feynman comportant la boucle top triangulaire devient très importante devant la contribution du diagramme avec la boucle carrée. Ce même argument fait que l'on s'attend également à trouver des distributions a peu près identiques pour des très grandes valeurs de $\lambda$ (plus grand que 10) ou des très petites valeurs (en dessous de $\frac{\lambda}{\lambda^{SM}}=-4$). Cela semble être le cas au vu des courbes tracées. \\\\
Par ailleurs, la forme des distributions varie énormément dans la région $[0;4]$, ce qui rend la paramétrisation difficile dans cette zone de valeurs pour $\lambda$. Nous reviendrons ensuite à l'étude de la surface "$Energie\times\lambda$" dans cette zone, concentrons nous pour l'instant sur les régions périphériques.

\subsection{Distribution de l'impulsion du système $[\tau,\tau]$}
Pour la distribution en impulsion, nous avons choisi une fonction à 7 paramètres donnée par : 
$$f:x\mapsto\frac{x*y2}{1+exp(\frac{x-y0}{y1})}+\frac{x*y5+y6}{1+exp(\frac{x-y3}{y4})}$$ pour laquelle nous avons obtenu les valeurs suivantes avec ROOT fit : 
\begin{center}
	\begin{tabular}{|c|c|c|c|c|c|c|c|c|}
	\hline
	$\lambda$ & -15 & -4 & 0 & 1 & 2.46 & 4 & 10 & 20 \\
	\hline
	\hline
	y0 & $-21\pm4$ & $-51\pm6$ & $83 \pm 3$ & $93\pm3$ & $-693\pm4$ & $40.7\pm0.7$ & $-4\pm3$ & $-7\pm2$ \\
	y1 & $43.5\pm0.2$ & $64\pm1$ & $49.1\pm0.2$ & $52.8\pm0.2$ & $62.7\pm0.1$ & $20.5\pm0.2$ & $36.8\pm0.2$ & $38.2\pm0.2$ \\
	y2 & $240\pm20$ & $220\pm20$ & $46\pm3$ & $45\pm2$ & $(92\pm7)\cdot10^{5}$ & $166\pm3$ & $280\pm10$ & $250\pm10$ \\
	y3 & $170\pm3$ & $223\pm5$ & $97\pm2$ & $85\pm1$ & $115.3\pm0.5$ & $200\pm6$ & $330\pm30$ & $170\pm3$ \\
	y4 & $18\pm2$ & $62\pm1$ & $143\pm1$ & $21.8\pm0.6$ & $22.3\pm0.3$ & $65\pm1$ & $50\pm6$ & $14\pm 2$ \\
	y5 & $2.3\pm0.2$ & $-3.1\pm0.4$ & $-11\pm1$ & $-26\pm1$ & $-38.0\pm0.7$ & $1.4\pm0.1$ & $0.08\pm0.02$ & $1.0\pm0.1$ \\
	y6 & $-27\pm7$ & $3\pm7$ & $1\pm4$ & $-10\pm1$ & $-9\pm8$ & $40\pm10$ & $-30\pm7$ & $-33\pm8$ \\
	\hline	
	\end{tabular}
\end{center}
on peut alors tracer l'évolution des différents paramètres en fonction de $\lambda$.
\begin{figure}[!h]
	\centering
	\includegraphics[scale=0.6]{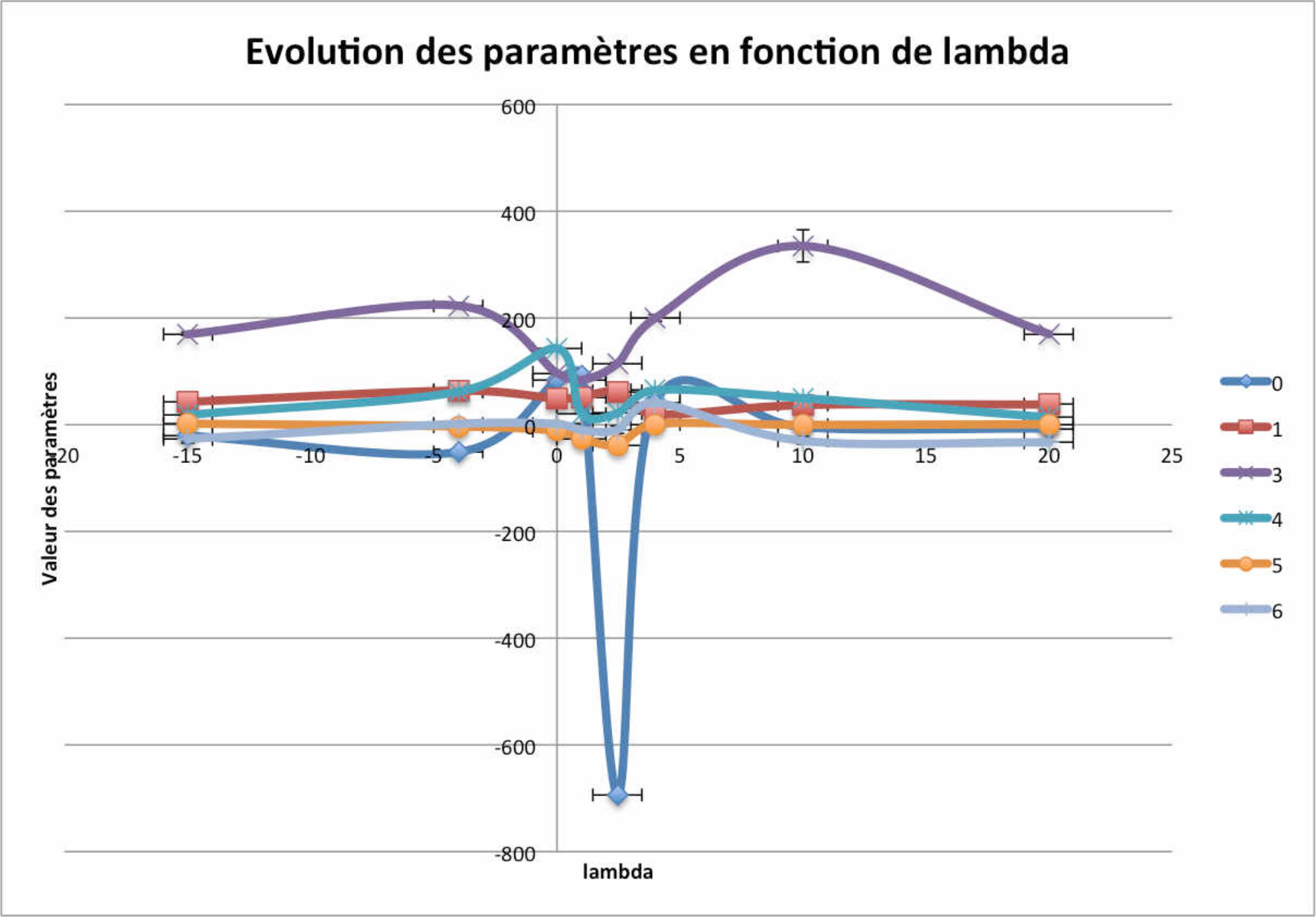}
	\caption{Sur ce graphique nous avons représenté seulement six des sept paramètres, avec les marges d'erreur, qui sont la plupart du temps "cachées" par l'épaisseur de la courbe. On voit que le comportement varie peu aux extrémités du domaine d'étude, et beaucoup au voisinage de l'unité, ce qui signifie que la famille de fonctions utilisées pour approcher les courbes n'est pas très adaptée à la description lorsque $\lambda\approx1$.}
\end{figure}

\paragraph{Étude pour $\lambda\in[-15,-4]$ :}
	Les différents paramètres se comportent régulièrement dans cette plage de valeurs, on peut donc avoir une approximation de la distribution en impulsion transverse du système di-tau en réalisant une homotopie linéaire de l'une des courbes à l'autre (on suit, pour chaque paramètre, la droite qui relie la valeur en $\lambda=-15$ à la valeur en $\lambda=-4$). On a ainsi l'équation locale (en $\lambda$) d'un surface qui approxime la "vraie" surface.
\begin{figure}[!h]
	\centering
	\includegraphics[scale=0.6]{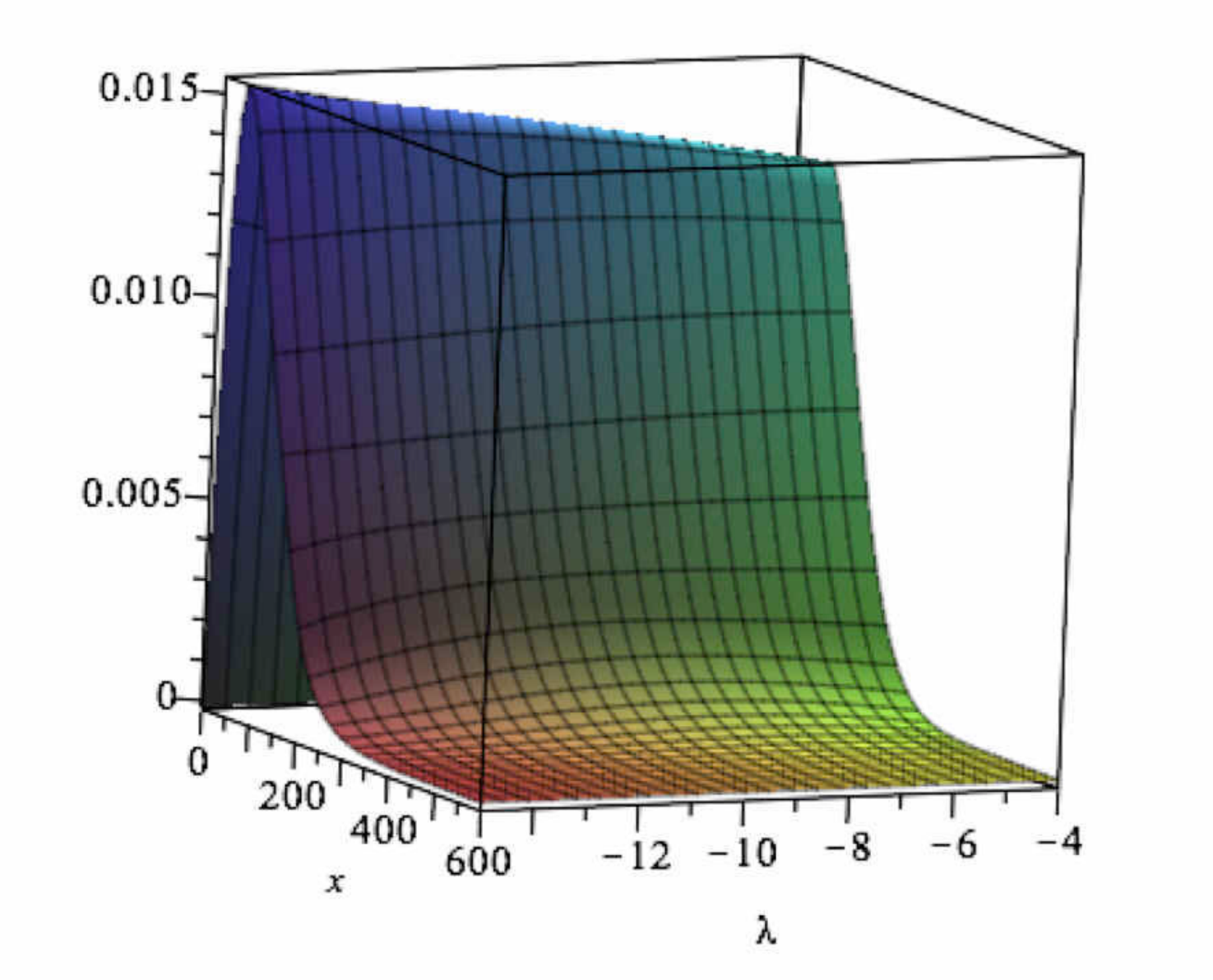}
	\caption{Surface obtenue grâce à l'approximation linéaire pour la distribution en impulsion transverse du système di-tau et pour $\lambda\in[-15,-4]$}
\end{figure}

\paragraph{Étude pour $\lambda\in[10,20]$ :}
De la même manière, la forme de fonctions choisie permet d'avoir une surface qui approche la surface réelle "distribution en impulsion / lambda" pour $\lambda\in[10,20]$.
\begin{figure}[!h]
	\centering
	\includegraphics[scale=0.6]{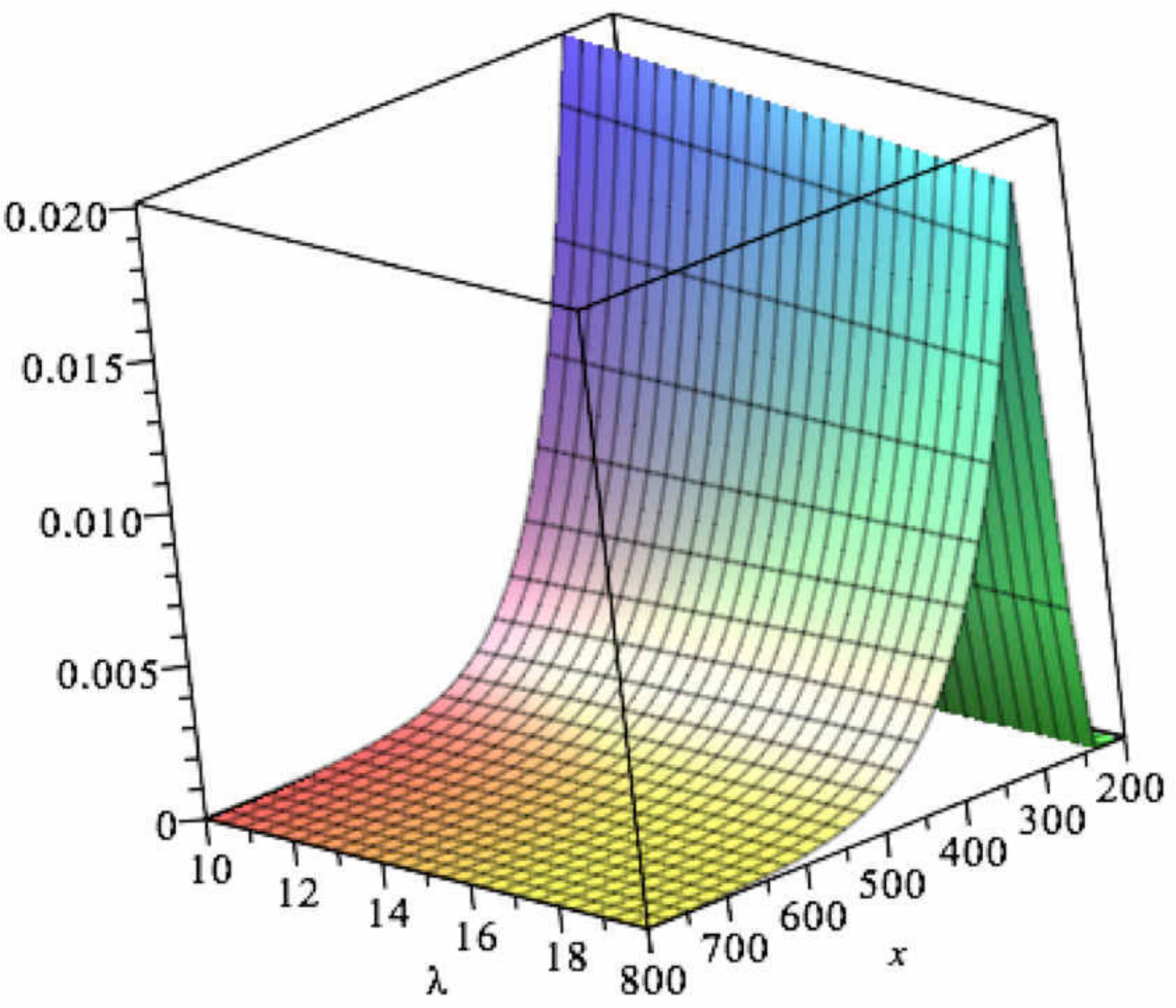}
	\caption{Surface obtenue grâce à l'approximation linéaire pour la distribution en impulsion transverse du système di-tau et pour $\lambda\in[10,20]$}
\end{figure}

\subsection{Distribution de masse invariante du système total}
Pour la distribution de masse invariante, nous avons choisi une fonction à 8 paramètres de la forme :
$$f:x\mapsto(1+erf(\frac{x-y0}{y1}))\frac{y2}{1+exp(\frac{x-y0}{y3})}+(1+erf(\frac{x-y4}{y5}))\frac{y6}{1+exp(\frac{x-y4}{y7})}$$
pour laquelle nous avons obtenu les valeurs suivantes par approximation avec ROOT Fit :
\begin{center}
	\begin{tabular}{|c|c|c|c|c|}
		\hline
		lambda & -15 & -4 & 0 & 1\\
		\hline
		\hline
		y0 & $252.7\pm0.1$ & $253.8\pm0.1$ & $254.9\pm0.3$ & $291.2\pm0.8$ \\
		y1 & $3.8\pm0.1$ & $3.9\pm0.1$ & $5.8\pm0.4$ & $30.0\pm0.1$ \\
		y2 & $4200\pm20$ & $3160\pm20$ & $1200\pm20$ & $730\pm30$ \\
		y3 & $87.1\pm0.3$ & $103.3\pm0.4$ & $154\pm1$ & $180\pm2$ \\
		y4 & $403\pm4$ & $373\pm3$ & $351\pm1$ & $358.1\pm0.6$ \\
		y5 & $78\pm3$ & $62\pm2$ & $49.6\pm0.1$ & $35\pm1$\\
		y6 & $620\pm30$ & $1300\pm30$ & $2540\pm30$ & $2670\pm40$ \\
		y7 & $21\pm2$ & $38\pm2$ & $75.9\pm0.1$ & $95.8\pm0.8$ \\
		\hline 
	\end{tabular}
\end{center}
et 
\begin{center}
	\begin{tabular}{|c|c|c|c|c|}
		\hline
		lambda & 2.46 & 4 & 10 & 20 \\
		\hline
		\hline
		y0 & $250\pm160$ & $252.98\pm0.06$ & $252.45\pm0.07$ & $252.66\pm0.08$\\
		y1 & $4\pm10$ & $4.09\pm0.07$ & $3.63\pm0.08$ & $3.72\pm0.09$\\
		y2 & $6000\pm3000$ & $10630\pm50$ & $6300\pm40$ & $5660\pm50$\\
		y3 & $20\pm40$ & $34.4\pm0.2$ & $60.6\pm0.7$ & $62\pm2$\\
		y4 & $390\pm80$ & $419\pm3$ & $350\pm2$ & $342\pm3$\\
		y5 & $40\pm10$ & $93\pm3$ & $75\pm4$ & $55\pm2$ \\
		y6 & $2000\pm3000$ & $314\pm 4$ & $120\pm30$ & $491\pm90$\\
		y7 & $150\pm100$ & $159\pm2$  & $122\pm7$ & $91\pm3$\\
		\hline 
	\end{tabular}
\end{center}
\begin{figure}[!h]
	\centering
	\includegraphics[scale=0.65]{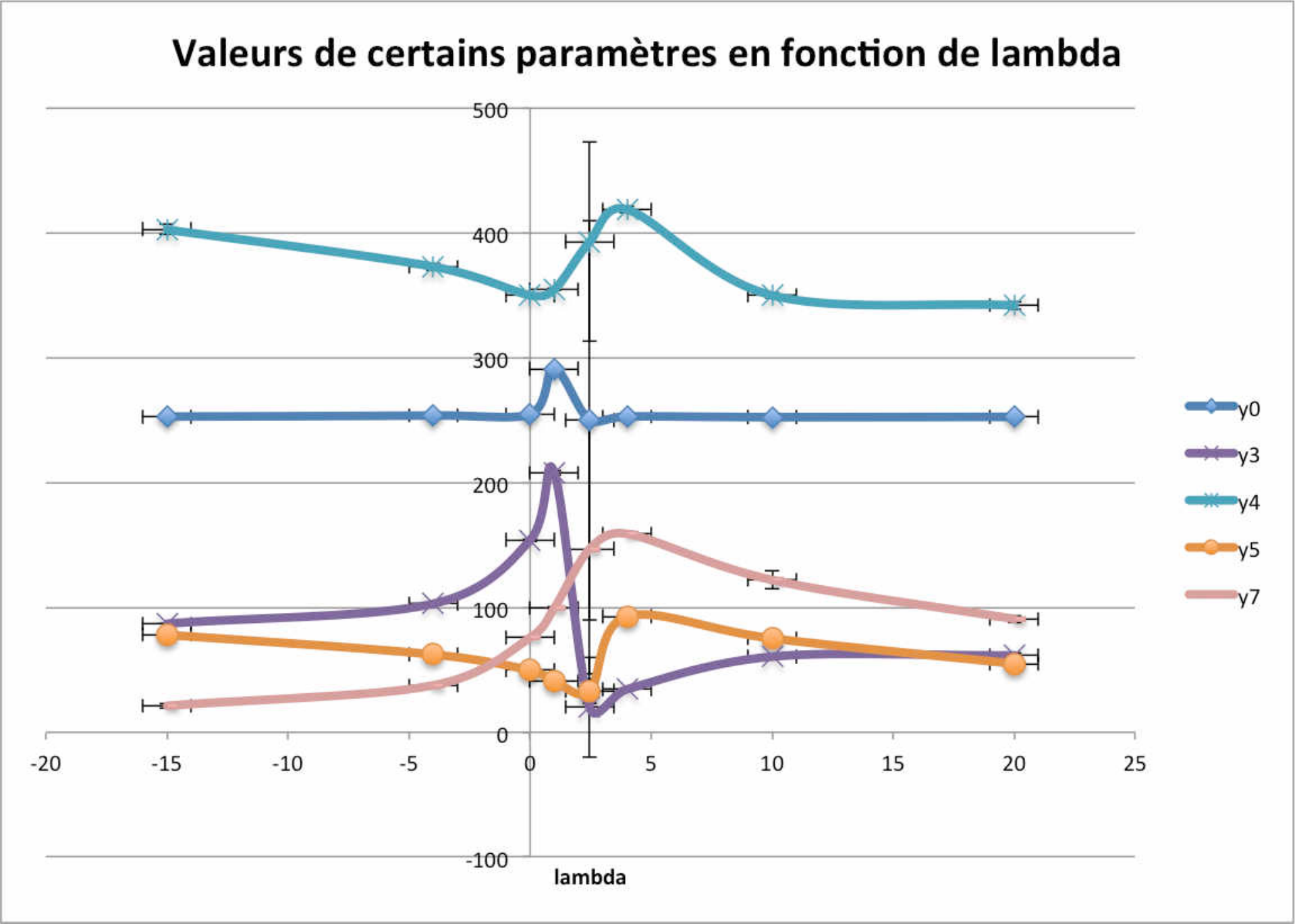}
	\caption{Valeurs des paramètres y0, y3, y4, y5 et y7 pour les différentes valeurs de $\lambda$ étudiées. On a placé les barres d'erreurs.}
\end{figure}
Les remarques formulées pour l'étude des distributions de l'impulsion transverse tiennent encore ici. Donnons seulement le "résultat", c'est-à-dire la forme de la surface obtenue dans chacun des cas.
\begin{figure}[!h]
	\centering
	\includegraphics[scale=0.6]{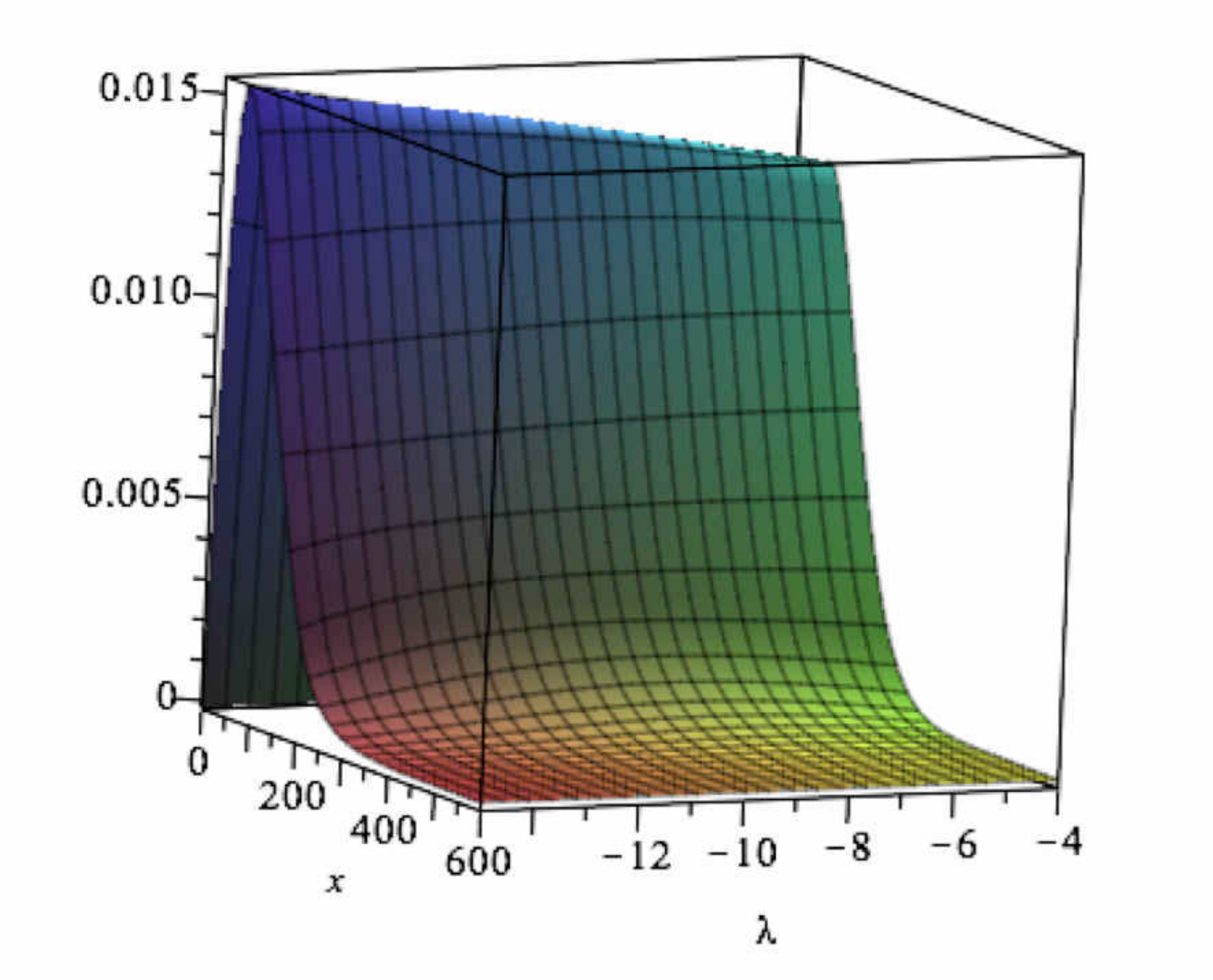}
	\includegraphics[scale=0.6]{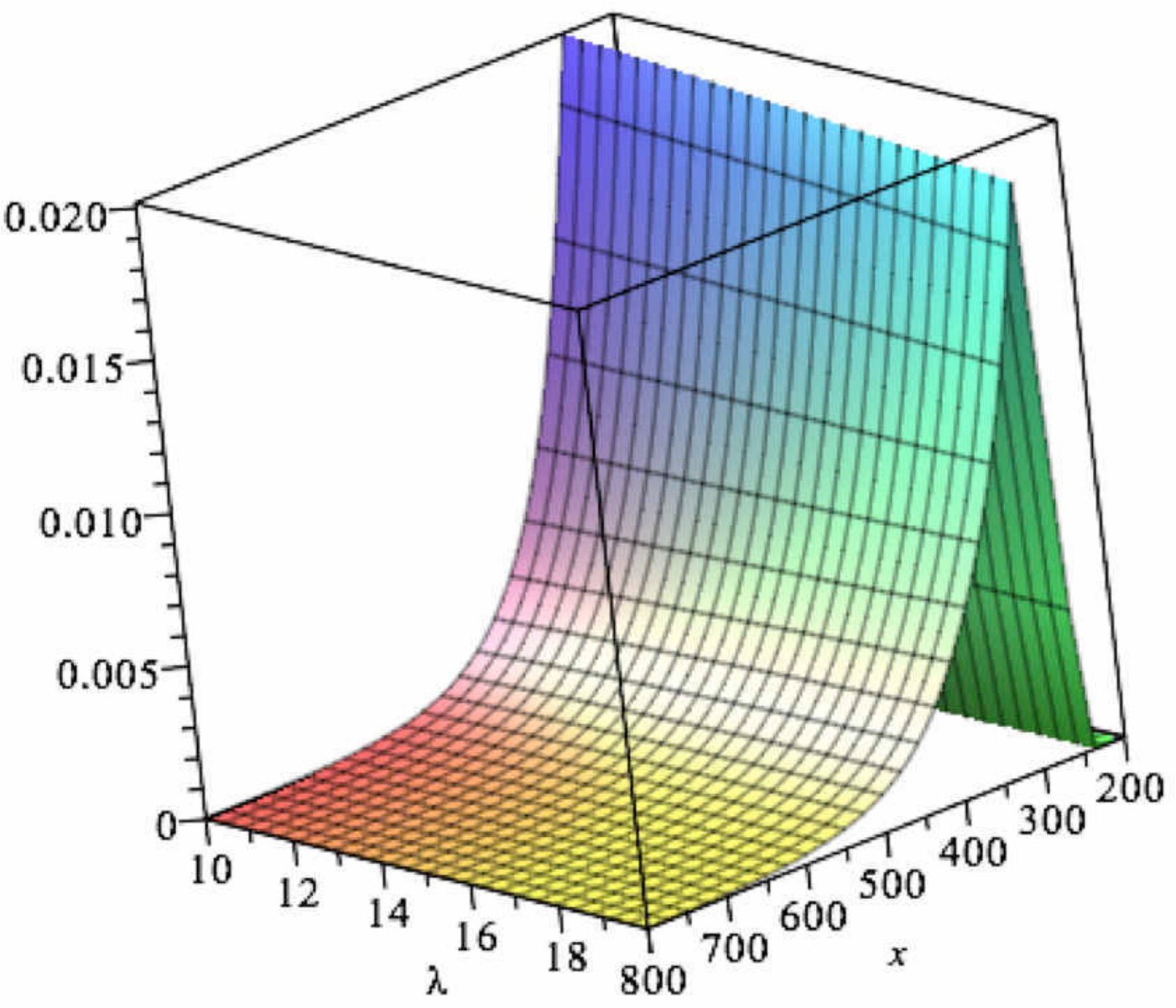}
	\caption{Surfaces obtenues grâce à l'approximation linéaire pour la distribution en masse invariante du système entier et pour $\lambda\in[-15,-4]$ ainsi que $\lambda\in[10,20]$}
\end{figure}

\subsection{Distributions au voisinage de $\frac{\lambda}{\lambda^{SM}}=1$}
Dans le cadre de la théorie effective en dimension 6 que nous étudions, il est naturel de vouloir s'intéresser au comportement des constantes fondamentales pour des valeurs proches de celles du modèle standard. C'est pourquoi on s'intéresse au distributions au voisinage de $\frac{\lambda}{\lambda^{SM}}=1$. Cependant, comme on l'a vu dans les paragraphes précédents, c'est dans cette région que le comportement est difficile a prévoir, en raison de l'interférence des deux diagrammes de Feynman intervenant au premier ordre. De nouveaux fichiers de collisions ont été produits, pour les valeurs $\lambda= 1.5,\ 2,\ 3,\ 3.5$. La superposition des tous les histogrammes donne les diagrammes suivants.
\begin{figure}[!h]
	\centering
	\includegraphics[scale=0.37]{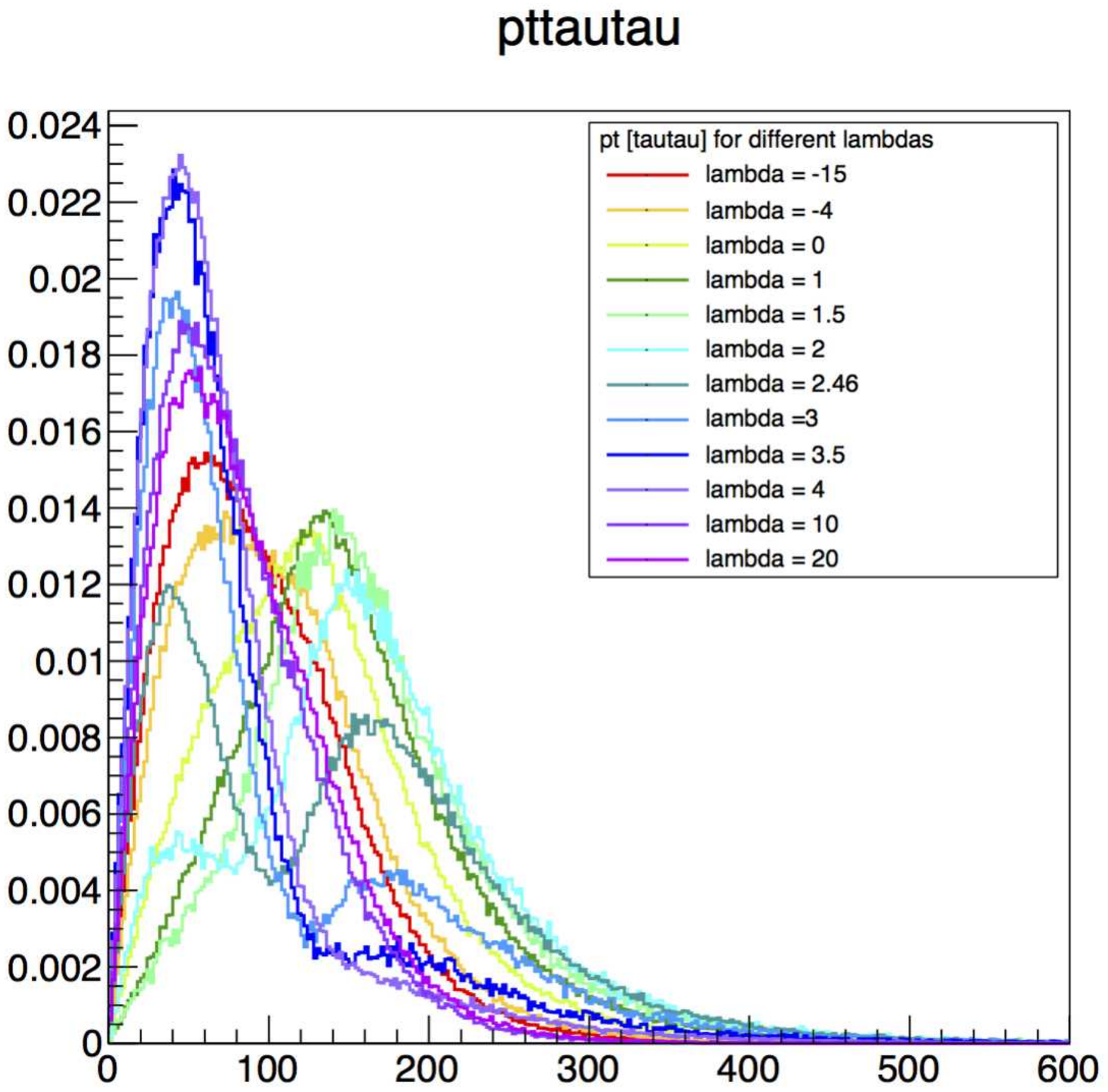}
	\includegraphics[scale=0.37]{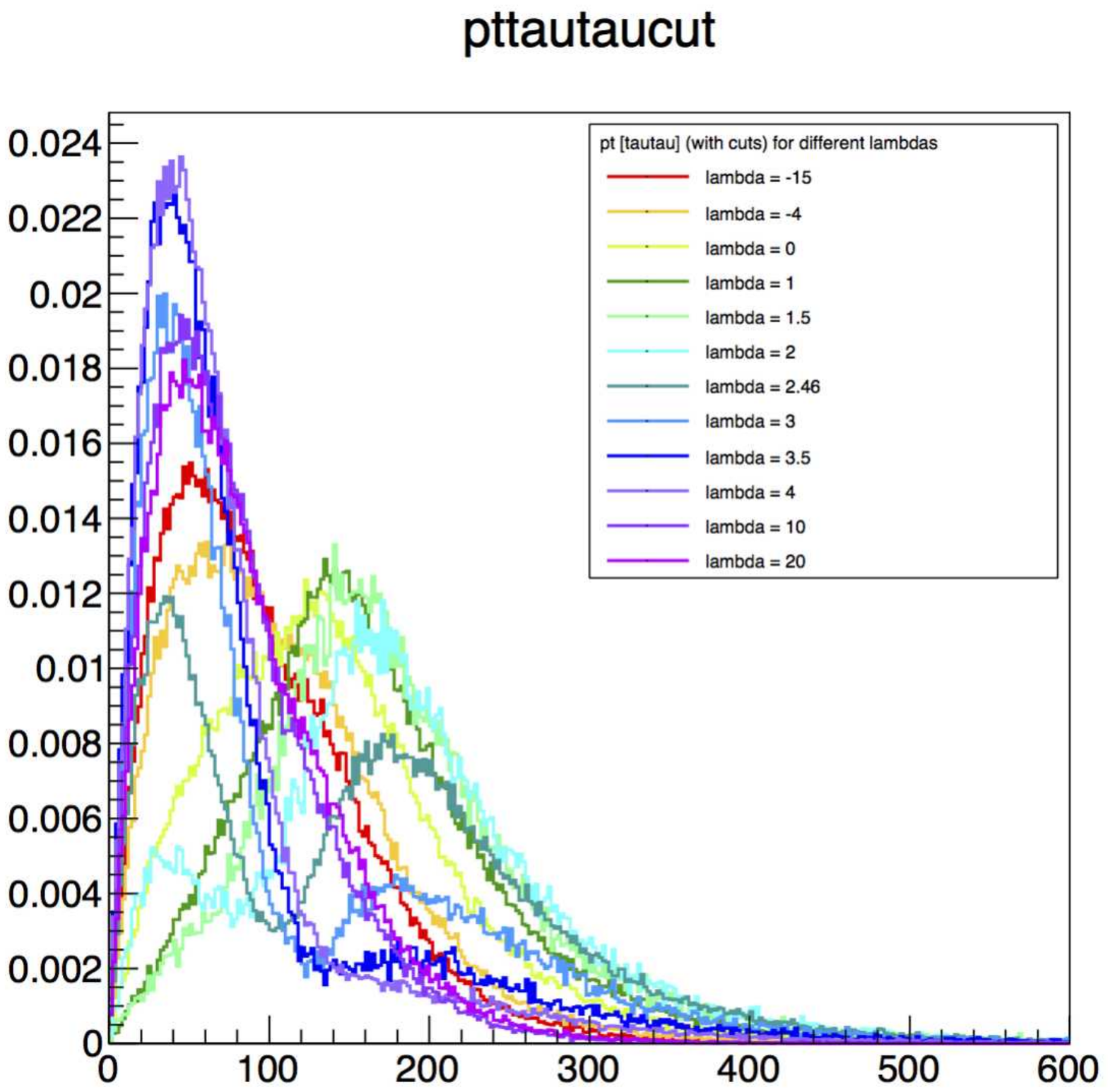}
	\includegraphics[scale=0.37]{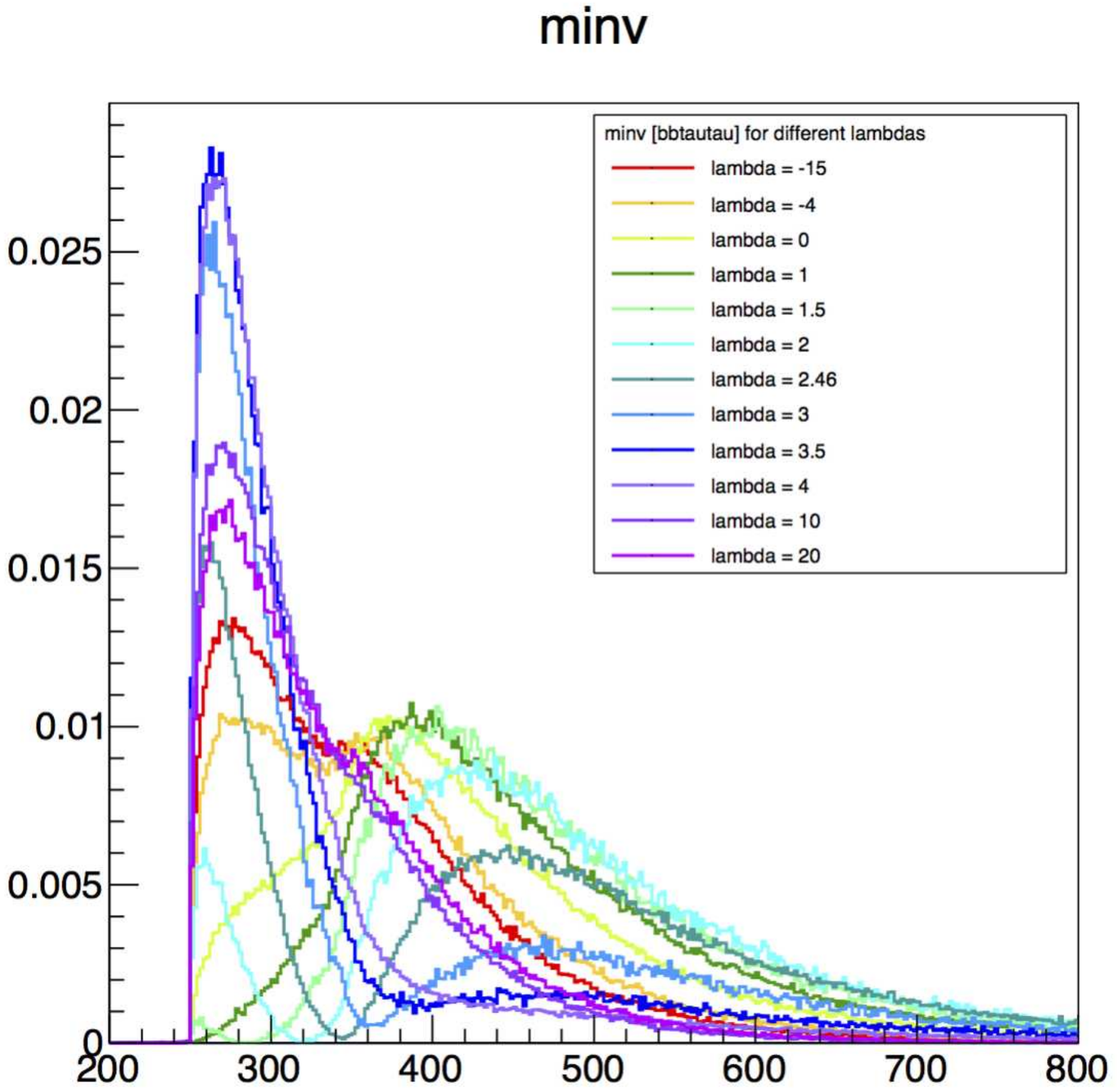}
	\includegraphics[scale=0.37]{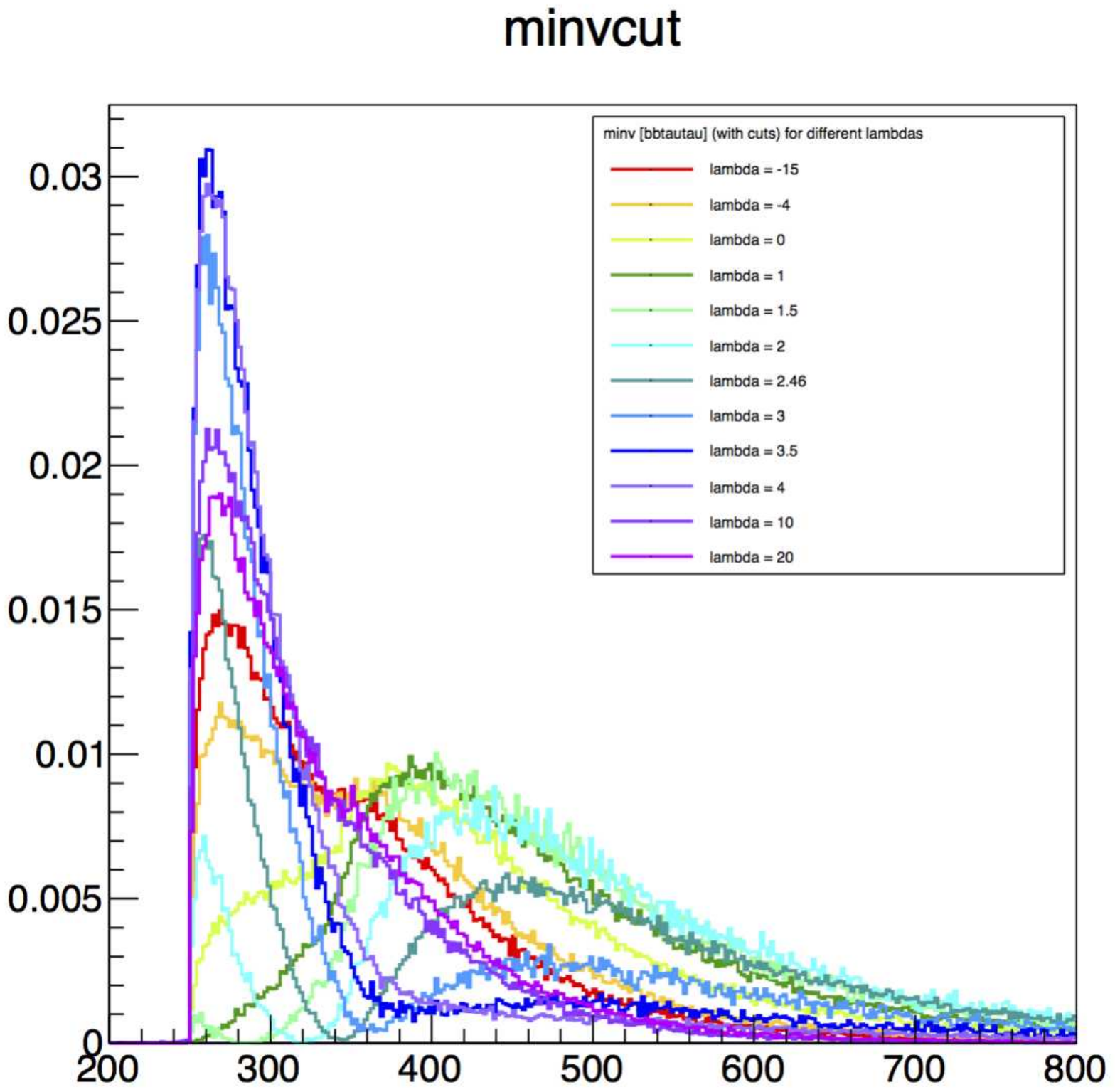}
	\caption{Superposition des distributions en impulsion transverse du système di-tau et en masse invariante du système entier, pour l'ensemble des valeurs de $\lambda$. Ces courbes donnent une meilleure intuition au sujet du comportement de ces invariantes cinétiques pour des valeurs de $\frac{\lambda}{\lambda^{SM}}$ proches de l'unité.}
\end{figure}
	Afin de mieux percevoir ce qui se passe au voisinage de $\frac{\lambda}{\lambda^{SM}}=1$, les fichiers correspondant aux valeurs $\frac{\lambda}{\lambda^{SM}}=0.5,\ 1.25$ ont été générés et la superposition des histogrammes obtenus pour les valeurs $\lambda=0,\ 0.5,\ 1,\ 1.25,\ 1.5\ et\ 2$ est donnée ci-dessous. Le fichier correspondant à $\frac{\lambda}{\lambda^{SM}}=0.75$ a en réalité été produit pour $\frac{\lambda}{\lambda^{SM}}=0.5$ d'où la superposition exacte entre les deux courbes (aux fluctuations statistiques près). C'est également pour cette raison que nous n'y prêtons pas attention. Les courbes correspondant au rapport $\frac{\lambda}{\lambda^{SM}}=0.5$ est donc bien présente sur les diagrammes, derrière celle intitulée $\lambda=0.75$ même si en réalité il s'agit bien de $\lambda=0.5$. Par ailleurs, malgré ce qui est indiqué dans la légende, les histogrammes pour $\lambda=0$ ne sont pas sur les diagrammes suivants.
\begin{figure}[!h]
	\centering
	\includegraphics[scale=0.37]{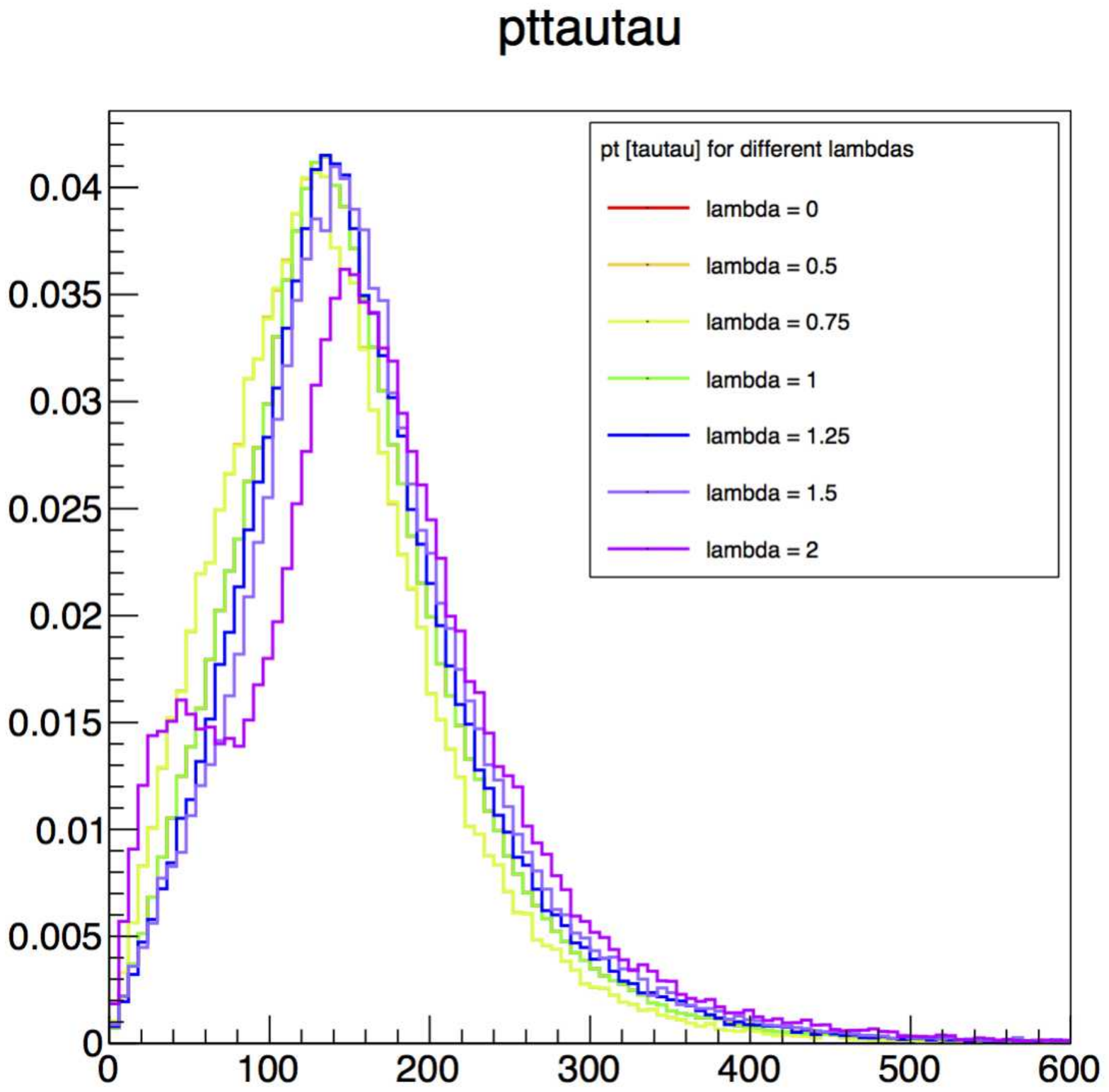}
	\includegraphics[scale=0.37]{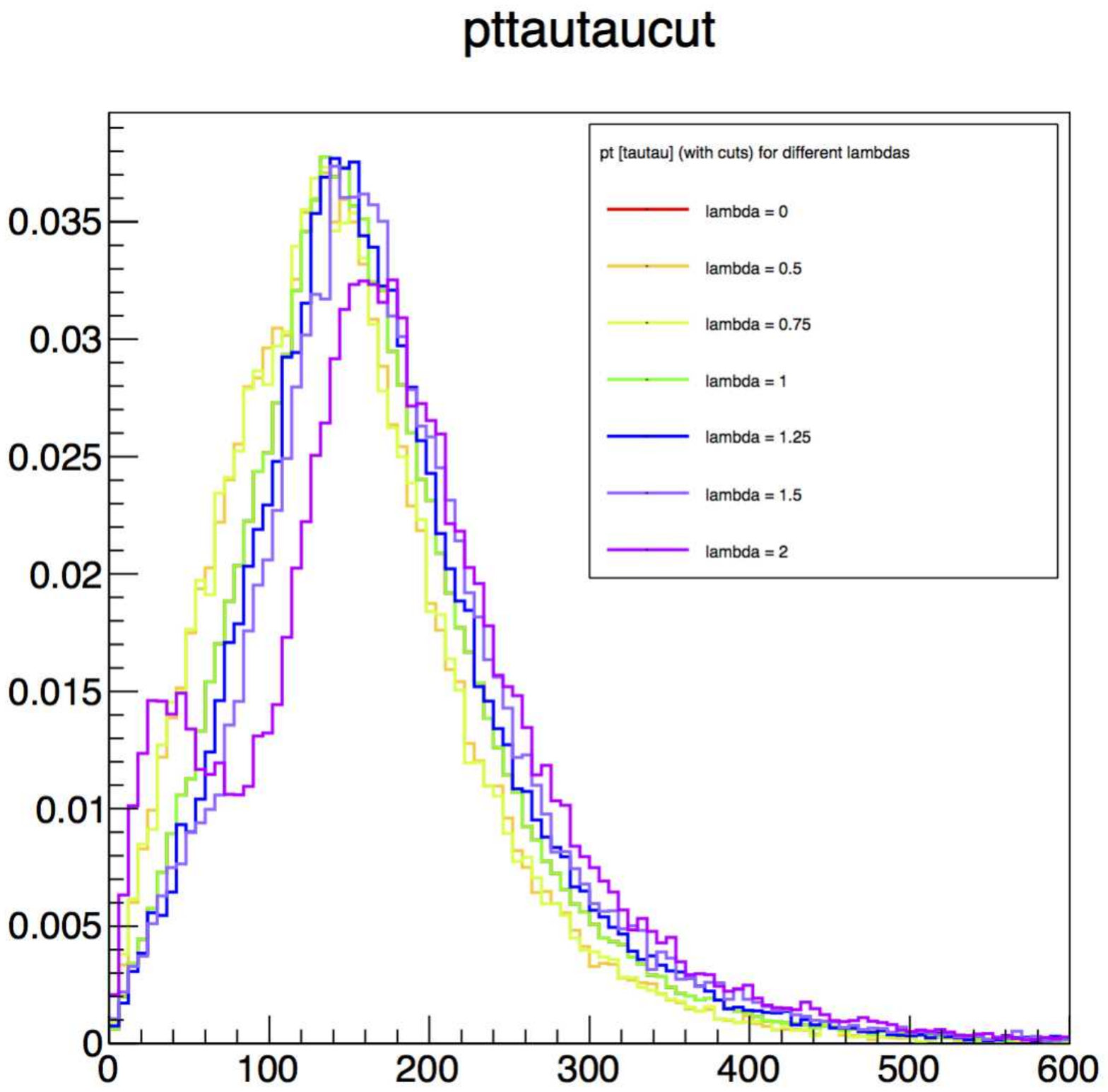}
	\includegraphics[scale=0.37]{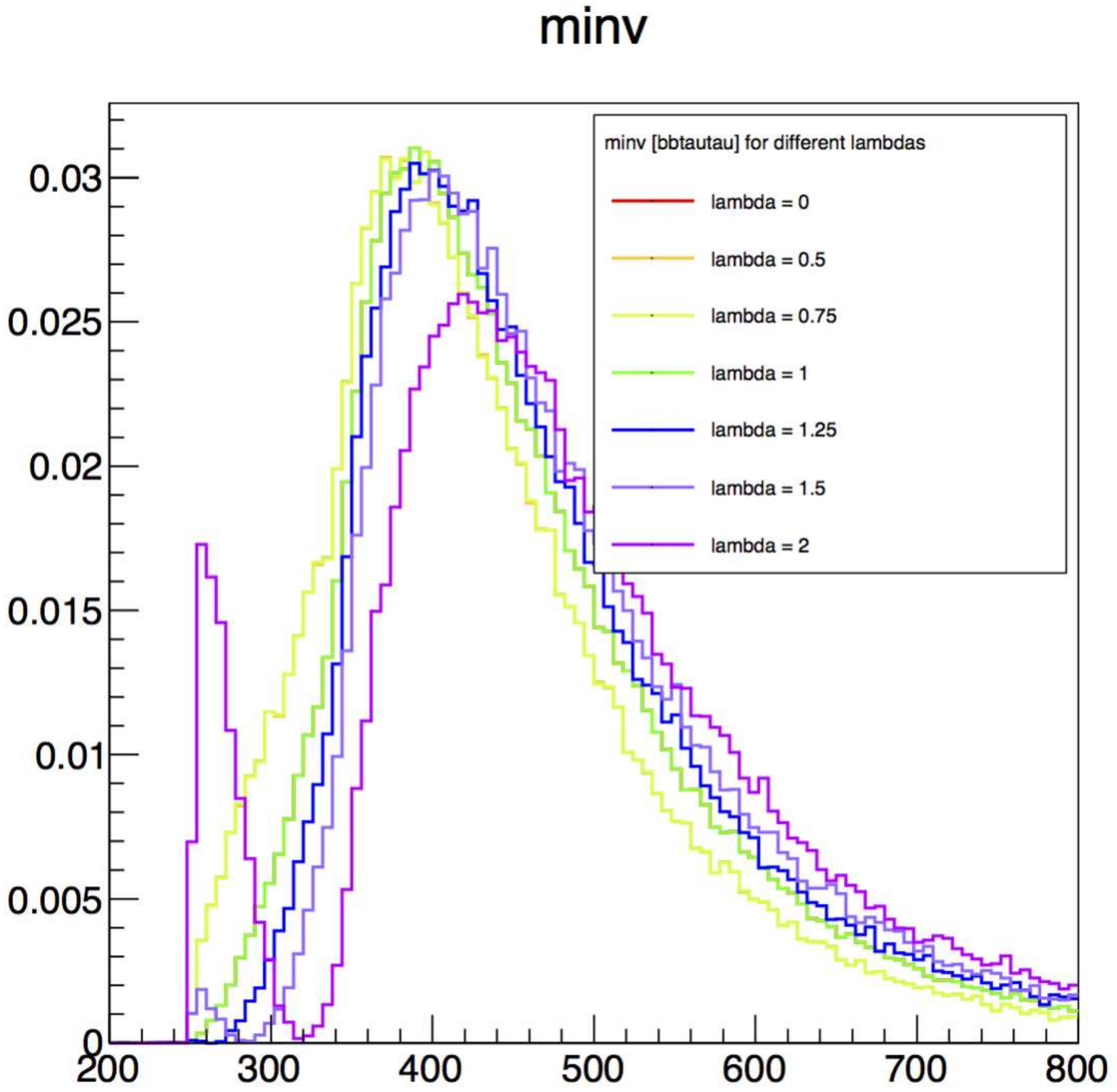}
	\includegraphics[scale=0.37]{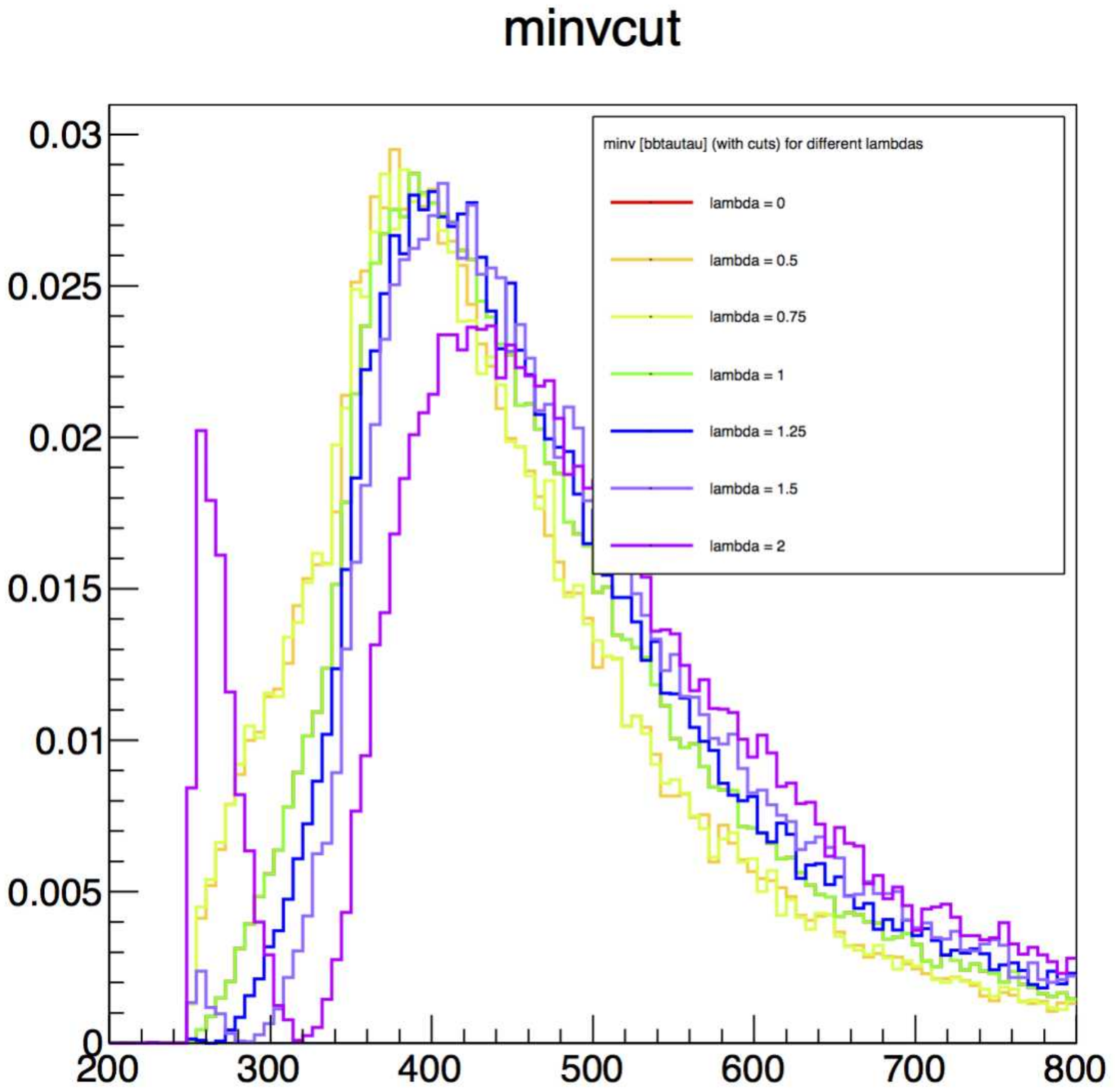}
	\caption{Superposition des distributions en impulsion transverse du système di-tau et en masse invariante du système entier, pour les valeurs de $\lambda$ les plus proches de celle du SM. Grâce à ces tracés, on voit plus distinctement ce que fait varier une petite modification de $\lambda$ au voisinage de $\lambda^{SM}$.}
\end{figure}
Il n'a pas été possible de réaliser, dans le temps imparti, les "fit" de ces courbes. La section précédente a cependant montré que les fonctions utilisées précédemment ne convenaient pas pour décrire les histogrammes dans cette région de $\lambda$ (certains paramètres divergent). Il faut donc essayer de trouver de nouvelles distributions qui décrivent mieux les variations observées.

\paragraph{Paramétrisation de pttautau} Pour les distributions de l'impulsion transverse du système di-tau, le courbe peut être décomposée en deux parties, une perturbation d'un bout de droite, pour $p_T\in[0,150]$ environ, et une décroissance typiquement exponentielle. La perturbation du segment reliant $(0,0)$ à $(150,0.041)$ peut être paramétrée par un sinus de période environ 300, si bien que l'équation de la première partie de courbe est : 
$$y(p_T)=\frac{0.041}{150}p_T+a(\lambda)*sin(\frac{\pi p_T}{150})$$
avec le paramètre $a$ variant typiquement dans $[-0.002,0.002]$, avec $a(0.5)\approx0.002$ et $a(1.5)\approx-0.002$. Dans le cadre de notre approximation linéaire, on considèrera pour l'homotopie l'application $$a:\lambda\rightarrow 0.004(1-\lambda)$$
Pour la deuxième partie de courbe, on prend une portion de gaussienne centrée en $150$ et de variance proportionnelle à $p$, donc : $$p_T\rightarrow 0.041e^{-(p_T-150)^2/(10000p)}$$ ce qui permet de choisir p de l'ordre de l'unité.\\\\
Il faut ensuite raccorder les deux portions, mais ces formes sont susceptibles de bien décrire les variations des distributions en impulsion transverse du système di-tau au voisinage de $\lambda^{SM}$.

\paragraph{Paramétrisation de la masse invariante du système}
Les distributions de masse invariante du système total présentent les variations suivantes au voisinage de $\lambda^{SM}$ : de minv=250 GeV à minv=350 GeV environ, on observe une diminution rapide lorsque $\frac{\lambda}{\lambda^{SM}}$ augmente. De plus, le maximum de la courbe se déplace légèrement vers les grandes masses invariantes, et la décroissance exponentielle est d'autant plus rapide que le rapport des $\lambda$ est petit.
On peut imaginer, comme pour les distributions en impulsion transverse, paramétrer les variations au voisinage de $\lambda^{SM}$ en rajoutant un sinus de période 200 environ, et en jouant sur le facteur multiplicatif pour la décroissance exponentielle.

\subsection{Discussion}
Cette dernière étude a permis d'obtenir des formes analytiques approchées pour les distributions intéressantes pour la mesure du couplage $\lambda$ du champ BEH. Bien qu'il ne s'agisse sans doute pas de la véritable expression de ces distributions, cela permet au moins de se faire une idée de l'influence du couplage $\lambda$ sur les différents invariants cinétiques de la production di-Higgs au sein de collisionneurs protons-protons comme le LHC. Même si les caractéristiques actuelles du LHC ne permettent pas de telles mesures, de grands espoirs sont portés dans la phase de haute luminosité HL-LHC (Run 3). Le paramètre $\lambda$ est la dernière information manquante au sujet de Higgs du modèle standard, et sa mesure un test crucial pour le modèle standard ainsi que ses extensions. Compte tenu des informations expérimentales actuelles (le Higgs est très proche de celui prévu par le modèle standard), l'étude prend tout son sens dans les extensions du modèle standard comme le kappa-framework et surtout, les théories des champs effectives. Les écarts par rapport au modèle standard sont peut être trop faible pour qu'ils puissent être mesurés au HL-LHC, ce sera alors le rôle de futurs collisionneurs circulaires, fonctionnant à une centaine de TeV dans le centre de masse, d'étudier ces déviations.

\end{document}